\documentclass[10pt,reqno]{amsart}
\usepackage[a4paper, margin=2.5cm]{geometry}
\usepackage{amsthm,amssymb,amstext,graphicx,comment}
\usepackage[foot]{amsaddr}
\usepackage{bbm}
\usepackage{enumerate}
\usepackage{amscd}
\usepackage{mathtools}
\usepackage{xcolor}
\usepackage{hyperref}
\hypersetup{
  colorlinks=true,
  pdfpagelayout=TwoColumnRight,
  pdftitle={},
  pdfauthor={Sven Karbach},
  pdfsubject={},
  pdfkeywords={}
}
\usepackage{enumitem}
\usepackage{empheq}
\usepackage[utf8]{inputenc}
\usepackage[T1]{fontenc}
\usepackage{lmodern}
\usepackage{soul}
\usepackage[capitalise]{cleveref}

\DeclareMathOperator\dom{dom}
\DeclareMathOperator\Tr{Tr}

\DeclareMathOperator\lin{lin}

\newcommand*\D{\mathop{}\!\textnormal{d}}

\newcommand*\E{\mathop{}\!\textnormal{e}}
\newcommand*\I{\mathop{}\!\textnormal{i}}

\numberwithin{equation}{section}

\newtheorem{theorem}{Theorem}[section]
\newtheorem*{theorem*}{Theorem}
\newtheorem{lemma}[theorem]{Lemma}
\newtheorem{proposition}[theorem]{Proposition}
\newtheorem*{proposition*}{Proposition}
\newtheorem{corollary}[theorem]{Corollary}
\theoremstyle{plain}

\theoremstyle{definition}
\newtheorem{definition}[theorem]{Definition}
\newtheorem{remark}[theorem]{Remark}

\newtheoremstyle{example}
  {.3\baselineskip}
  {.3\baselineskip}
  {\normalsize}  
  {0pt}       
  {\bfseries} 
  {.}         
  {5pt plus 1pt minus 1pt} 
  {}          
\theoremstyle{example}
\newtheorem{example}[theorem]{Example}

\newtheorem*{assumption*}{\assumptionnumber}
\providecommand{\assumptionnumber}{}
\makeatletter
\newenvironment{assumption}[1]
 {%
  \renewcommand{\assumptionnumber}{Assumption $\mathfrak{#1}$}%
  \begin{assumption*}%
  \protected@edef\@currentlabel{$\mathfrak{#1}$}%
 }
 {%
  \end{assumption*}
 }
\makeatother

\newlist{theoremenum}{enumerate}{1}
\setlist[theoremenum]{label=\roman*), ref=\textup{\thetheorem~\roman*)}}
\crefalias{theoremenumi}{theorem}

\newlist{propenum}{enumerate}{1}
\setlist[propenum]{label=\roman*), ref=\textup{\thetheorem~\roman*)}}
\crefalias{propenumi}{proposition}

\newlist{defenum}{enumerate}{1}
\setlist[defenum]{label=\roman*), ref=\textup{\thedefinition~\roman*)}}
\crefalias{defenumi}{definition}

\newlist{lemenum}{enumerate}{1}
\setlist[lemenum]{label=\roman*), ref=\textup{\thelemma~\roman*)}}
\crefalias{lemenumi}{lemma}

\newlist{remenum}{enumerate}{1}
\setlist[remenum]{label=\roman*), ref=\textup{\theremark~\roman*)}}
\crefalias{remenumi}{remark}

\def\e{\operatorname{e}} 
\def\eps{\varepsilon}
\renewcommand{\epsilon}{\eps}

\renewcommand{\MR}{\mathbb{R}}

\newcommand{\MN}{\mathbb{N}}
\newcommand{\MP}{\mathbb{P}}

\newcommand{\MS}{\mathbb{S}}
\newcommand{\MF}{\mathbb{F}}

\newcommand{\R}{\MR}

\newcommand{\cF}{\mathcal{F}}
\newcommand{\cA}{\mathcal{A}}
\newcommand{\cB}{\mathcal{B}}

\newcommand{\cD}{\mathcal{D}}
\newcommand{\cE}{\mathcal{E}}
\newcommand{\cH}{\mathcal{H}}

\newcommand{\cL}{\mathcal{L}}

\newcommand{\cR}{\mathcal{R}}
\newcommand{\cV}{\mathcal{V}}

\newcommand{\cG}{\mathcal{G}}

\newcommand{\cT}{\mathcal{T}}

\newcommand{\sP}{\mathsf{P}}

\newcommand{\df}{\coloneqq}

\newcommand{\one}{\mathbf{1}}
\newcommand{\interior}[1]{({\kern0pt#1})^{\textnormal{o}}}
\newcommand{\set}[1]{\left\{ #1\right\}}
\newcommand{\norm}[1]{\|#1\|}

\newcommand{\EX}[1]{\mathbb{E}\left[#1\right]}
\newcommand{\EXspec}[2]{\mathbb{E}_{#1}\left[#2\right]}

\newcommand{\cHplus}{\cH^{+}}

\newcommand{\cHpluso}{\cHplus\setminus \{0\}}
\newcommand{\MRplus}{\MR^{+}}
\newcommand{\dm}{m(\D\xi)}
\newcommand{\dmu}{\mu(\D\xi)}

\newcommand{\dmuk}{\mu^{(k)}(\D\xi)}
\newcommand{\be}{\mathbf{e}}
\newcommand{\bP}{\mathbf{P}}

\begin{document}

\title[Heat-modulated affine stochastic volatility models]{Heat-modulated
  affine stochastic volatility models for forward curve dynamics} 
\address{Korteweg-de Vries Institute for Mathematics and Institute for
  Informatics at the University of Amsterdam.}

\maketitle{}

\vspace{-5mm}

\begin{center}
  \begin{tabular}{c}
    \textsc{Sven Karbach} \\
   \small[\texttt{\MakeLowercase{sven@karbach.org}}] \\
  \end{tabular}
\end{center}

\begin{abstract}
  We present a function-valued stochastic volatility model designed to capture
  the continuous-time evolution of forward curves in fixed-income or commodity
  markets. The dynamics of the (logarithmic) forward curves are defined by a
  Heath-Jarrow-Morton-Musiela stochastic partial differential equation
  modulated by an instantaneous volatility process that describes the
  second-order moment structure of forwards with different time-to-maturity. We
  propose to model the operator-valued instantaneous covariance by an affine
  process on the cone of positive Hilbert--Schmidt operators with drift given
  by the Lyapunov operator of the Laplacian. The resulting Hilbert-valued
  stochastic covariance model admits an affine transform formula and we derive the associated Feynman--Kac martingale identity.
  Furthermore, we analyze a numerically feasible spectral Galerkin approximation
  of the associated operator-valued generalized Riccati equations and derive
  finite-rank error bounds of the moment-generating function of the model.
  A pricing robustness statement for exponential flow-forward functionals is
  formulated and numerical
  illustrations are carried out in the case of a heat-modulated version of the BNS volatility model.\\

\end{abstract}



\vspace{-1mm}

\section{Introduction}\label{sec:introduction}

Since the work of Black and Scholes~\cite{BS73}, stochastic volatility has become
a central tool for explaining option prices and managing derivative risk. The
constant-volatility benchmark is still useful as a quotation model, but observed
volatility is time-varying and stochastic; this has motivated models such as
Heston~\cite{Heston1993}, SABR~\cite{HKLW02}, Barndorff-Nielsen--Shephard
models~\cite{BNS02}, local volatility models~\cite{Dupire1994}, and more
recently rough volatility models~\cite{GJR18}. For portfolios and multi-asset derivatives, the corresponding object is a
stochastic covariance process. Classical finite-dimensional approaches include
positive-semidefinite Ornstein--Uhlenbeck models~\cite{BNS13, MKPS12, BK23},
Wishart-type diffusion models~\cite{DFGT07, GS10}, and affine jump-diffusions on
positive-semidefinite matrices~\cite{LT08, CFMT11}. This paper develops an
infinite-rank affine stochastic covariance framework for forward curve dynamics,
viewed as the functional analogue of these finite-dimensional affine
covariance models.\par{}

Function-valued models play a crucial role in finance, especially in modeling
the dynamics of forward rates and forward price curves in fixed-income and
commodity markets. The key idea is to directly model the family of forward
prices for all maturity dates simultaneously, rather than using a parametric
model for the underlying spot asset and then deriving the forward prices through
a spot-forward relation. This direct modeling approach was initially applied in
interest rate markets by Heath, Jarrow, and Morton~\cite{HJM92}, hence it is
referred to as the HJM approach. Since then, this methodology has been extended
to various other markets, including commodities~\cite{BBK08} and stock
options~\cite{KP15}. In commodity markets, the HJM modeling approach is particularly advantageous in
energy and weather related markets where establishing a valid relation between
spot and forward prices is challenging. Indeed, in these markets the
conventional buy-and-hold hedging strategy for commodity forwards may be
impractical due to limitations such as non- or limited storability (as in the
case of electricity and gas, see e.g., \cite{BBK08}) or the underlying assets
being simply non-tradable (as for example wind speed or temperature indexes, see
e.g., \cite{BSB12}). However, futures, forwards, and other derivatives written
on day-ahead power prices or index-based underlyings are sometimes liquidly
traded. For instance, power futures from day-ahead to, in some areas, 3-year
delivery periods are liquid, and renewable energy forward production curves can
be tracked with reasonable high frequency given a flow of weather
predictions. In these markets, the focus therefore shifts to directly modeling
the forward dynamics following the HJM modeling paradigm. The limitations in trading power or weather index-based underlyings not only
pose challenges in spot-forward modeling, but also contribute to high and
complex maturity-specific volatilities in the forward dynamics. These
complexities arise from factors such as intermittency, forward-looking
information, spatial and network risk, political decisions, and weather
forecasts, see also the discussion in~\cite[Section 1.3]{BenKru23}. The general
objective of this paper is to contribute to the mathematical formulation of a
framework to model these complex volatility structures within continuous-time
functional models, recently put forward in a number of papers~\cite{BRS18,
  BS18,BS24, CKK22b, FK24}.\par{}

One key motivation of this line of research is to enhance the understanding of
the mathematics of options written on forward curves. Similar to the progression
from the ordinary Black-Scholes model to stochastic volatility models in the univariate
setting, we advocate for stochastic volatility models for functional forward
curve dynamics that: (1) can jointly capture the stylized facts of forward curve
dynamics and their implied or realized volatilities, and (2) are tractable
enough to efficiently price and hedge European-style options written on the
forwards.

\subsection{The Heath-Jarrow-Morton-Musiela Modeling Approach}
To present our main contributions and review the relevant literature, we begin
with a brief introduction to the Heath-Jarrow-Morton-Musiela (HJMM) type approach
we use for modeling forward curve dynamics in fixed-income or commodity
markets. For any time $t\geq 0$ and maturity date $T\geq t$, we denote by $F(t,T)$ the
forward price at time $t$ maturing at time $T$. In other words, $F(t,T)$ is the
at time $t$ agreed-upon price for the purchase or sale of the underlying asset
at the future date $T$. By introducing the \emph{Musiela parametrization}
$x \equiv T-t$, we can express forward prices as functions of
the \emph{time-to-maturity}. Specifically, we define $f_t(x) \equiv F(t,t+x)$ in
case of an arithmetic model, or $f_{t}(x)\equiv \ln(F(t,t+x))$ in case of a
geometric model, where $x,t\geq 0$. The mapping $x \mapsto f_{t}(x)$ is then
referred to as the (logarithmic) \emph{forward price curve} at time $t\geq 0$.

Following the Heath-Jarrow-Morton-Musiela modeling framework~\cite{CT06} the
dynamics of the risk-neutral (logarithmic) forward prices can be described by a
stochastic partial differential equation of hyperbolic type given as:
\begin{align}\label{eq:HJMM}
  \begin{cases}
    \D f_{t}(x)&=\big(\frac{\partial}{\partial x}f_{t}(x)+g_{t}(x)\big)\D
  t+\sum_{i=1}^{d}\sigma_{t}^{(i)}(x)\D W_{t}^{(i)},\quad \forall t,x>0,\\
    ~~~f_{0}(x)&=F(0,x), \quad\forall x>0,    
  \end{cases}
\end{align}
where $\{(W^{(i)}_{t})_{t\geq 0}\colon i=1,\ldots,d\}$ is a family of
$d \in\mathbb{N} \cup\set{+\infty}$ independent real-valued Brownian
motions. The drift term $(g_{t}(x))_{x,t\geq 0}$ in the HJMM
equation~\eqref{eq:HJMM} either vanishes in the arithmetic case, or adheres to a
certain HJM no-arbitrage condition in the geometric case, see e.g.~\cite[Section
2.4]{CT06}, ensuring that the forward price processes $(F(t,T))_{0\leq t\leq T}$
are martingales for all maturities $T$. 

We refer to the diffusion coefficients
$\set{(\sigma_{t}^{(i)}(x))_{x,t\geq 0}\colon i=1,\ldots,d}$ in
equation~\eqref{eq:HJMM} as the \textit{instantaneous volatilities} of the
forward curve dynamics. The family of instantaneous
volatilities characterizes the covariance structure between forwards with
different time-to-maturity. The number $d\in\MN\cup\set{+\infty}$
represents the \emph{rank of the noise}, corresponding to the number of
independent risk factors driving the forward curve dynamics. If $d$ is finite we
refer to~\eqref{eq:HJMM} as a \emph{finite-rank HJM} model; otherwise it is called
an \emph{infinite-rank HJM} model.\par{}

\subsection{Stochastic Covariance Models for Forward Price Curve
  Dynamics}\label{sec:intro-stoch-covariance-models} 

Within the general framework of the HJMM model~\eqref{eq:HJMM}, we retain the
flexibility to model the instantaneous volatilities
$\sigma_{t}(x)$ and decide on the rank $d$. Similar to the univariate case, we can
distinguish between three subclasses of models:  

\begin{enumerate}
    \item In the \emph{Gauss-Markov HJM} class, $\sigma_{t}(x)$ is chosen as
      a deterministic function of $t$ and $x$, independent of the
      forward prices $(f_{t}(x))_{t, x \geq 0}$. This is the forward curve
      equivalent of a (time-dependent) Black-Scholes model. 
    
    \item In the \emph{Markovian HJM} class, we assume that $\sigma(\omega,t
      ,f_{t}(x)) = \sigma(t, f_{t}(x))$, meaning the randomness arises solely
      through the dependence on the forward price itself. This approach is
      analogous to the local volatility models in the univariate case. 
     
    \item In a \emph{stochastic volatility model}, we assume that
      $(\sigma_{t}(x))_{t, x \geq 0}$ is both stochastic and time-dependent,
      governed by a separate stochastic differential equation. 
\end{enumerate}

The modeling approach we adopt in this paper is the third type: a
function-valued stochastic volatility model. Due to its infinite-dimensional
nature, it is more appropriately termed a \emph{stochastic covariance
  model}. More precisely, we refer to the joint process
$(f_{t}, X_{t})_{t \geq 0}$ as a stochastic covariance model for
forward curve dynamics, where $(f_{t})_{t \geq 0}$ is given by~\eqref{eq:HJMM}
and is modulated by a volatility operator $(\sigma_t)_{t\ge0}$ satisfying
$X_t=\sigma_t\sigma_t^*$ in the $H$-valued formulation. These types of
stochastic covariance models have been recently studied
in~\cite{BS18, BRS18, CKK22b, FK24, Kar22, BE24}.

\subsubsection{Non-Parametric Modeling of Forward Prices and Volatility}

In such function-valued stochastic covariance models, we interpret the HJMM
equation~\eqref{eq:HJMM} as a stochastic differential equation formulated within
a suitable space of forward curves (see Section~\ref{sec:forward-curve-dynamics}
below for details). This approach is non-parametric, as it leverages the entire
initial forward curve, represented by $x \mapsto f_{0}(x)$, along with a
\emph{instantaneous covariance process} $(X_{t})_{t \geq 0}$, as
primary model inputs. In this function-valued framework, $X_t$, for
$t \geq 0$, represents the covariance of $f_{t}$ across all time-to-maturities
and is defined as a positive, self-adjoint Hilbert--Schmidt operator acting on a
specific space of forward curves. In kernel settings it can be identified with a
covariance kernel on the time-to-maturity domain, see~\eqref{eq:kernel-representation}
below. This non-parametric stochastic modeling
has the advantage of capturing the continuous-time dynamics of the forward
curves while at the same time fitting well the realized and implied covariance
structure with only few hyperparameters. The initial forward curve
$x \mapsto f_{0}(x)$ can be extracted from market data using a smoothing method,
as detailed in~\cite[Section 1.6.3]{CT06}. Meanwhile, the volatility component
can be calibrated to option price data~\cite{BDL21}, estimated from realized
term structure data~\cite{BSV22a, BSV22b,
  schroers2024robustfunctionaldataanalysis} or fitted to both simultaneously.

\subsubsection{The Rank and Structure of Forward Curve Volatility: Empirics and Theory}
Empirical studies have shown that energy forwards exhibit idiosyncratic risk
across different maturities, characterized by high-dimensional noise and
distinctive stochastic correlation structures~\cite{Fre08,AKW10}. For instance,
a principal component analysis (PCA) of the \emph{NordPool} power futures market
conducted in~\cite{BBK08} revealed that more than ten risk factors are necessary
to explain 95\% of the variance. Similar results have been reported in other
studies, such as~\cite{Fre08} and~\cite{KO05}. These findings in energy markets contrast with the interest rate markets, where PCA indicates that only a small
number of risk factors is required to explain the variance~\cite{CT06}.
However, \cite{CrumpGospodinov2022} demonstrate that the
empirically observed low-dimensional factor structures in bond returns can largely be explained by
the strong local correlation of prices across nearby maturities, rather than by a genuinely low-dimensional data-generating process. As a consequence, conclusions in favor of intrinsically
low-rank models for interest rate markets must be treated with caution.
Moreover, empirical evidence in~\cite{schroers2024dynamicallyconsistentanalysisrealized} supports the
presence of stochastic volatility effects and points to the relevance of high-dimensional HJM-type
models for capturing the dynamics of forward curves and their volatility. Taken together,
these findings suggest that relying solely on PCA-based dimension reduction may obscure important
sources of risk.
In particular, ad-hoc truncation based on explained variance can lead to misspecified dynamics and
suboptimal modeling choices for pricing, hedging, and risk management in forward markets, due to the following theoretical and practical
considerations:
\begin{itemize}
\item \emph{Maturity-Specific Risk}: As observed in~\cite{CT06}, classical finite-rank HJM
  models are complete, allowing for perfect hedging of interest rate
  derivatives. This redundancy even allows to choose tenors of the hedging
  instruments that are independent of the derivative’s expiry
  date, leading to unrealistic hedging strategies. This issue can only be
  adequately addressed in the case of \emph{maturity-specific risk},
  where the volatility is modeled with infinite rank. 
\item \emph{Higher-Order Moments}: As argued by~\cite{Con05}, discarding
  risk factors based on a PCA in an a priori fashion can introduce biases when
  computing quantities that are nonlinear functions of the forward curve. This
  is particularly relevant for pricing and hedging forward curve derivatives
  with nonlinear payoffs, or for computing quantiles and
  value-at-risk. Therefore, for applications in pricing, hedging, and risk
  management, it is advisable to first model with maturity-specific
  risk, before analyzing finite-rank approximations and their
  associated errors relative to the sound full-rank model.
\item \emph{Covariance representation}:
Even for a moderate number of risk factors $d \in \mathbb{N}$, there are $d \times (d-1)/2$
stochastic correlations to be computed, which poses a high-dimensional problem, especially in energy
markets with up to ten identified risk factors. In the limit as $d\to\infty$, the instantaneous covariance structure admits a
natural functional representation in terms of a $L^{2}$-kernel, see
equation~\eqref{eq:kernel-representation} below. From both a modeling and
calibration perspective, it is therefore more practical to work with
(semi-)parametric families of covariance kernels and to analyze their
finite-rank approximations. Such an approach provides greater interpretability
and flexibility than directly specifying high-dimensional stochastic covariance
matrices, while retaining sufficient expressive power to capture
maturity-specific dependence structures observed in the data. 
\end{itemize}
This modeling philosophy aligns with the perspective articulated
in~\cite{Con05}, who advocate modeling term structure dynamics
``without imposing ad hoc restrictions on the number of factors driving
different rates or the form of the volatility functions. More importantly, the
shapes and variances of principal components are obtained as a result of the
model rather than an input.''.

\subsection{Main Contributions and Related Literature}

This paper's primary contribution is the development of a novel stochastic
covariance model, $(f_{t}, X_{t})_{t \geq 0}$, for forward curve
dynamics in fixed-income or commodity markets. The forward curve process,
$(f_{t})_{t \geq 0}$, is modeled using an infinite-rank HJMM
framework~\eqref{eq:HJMM}, while the covariance process,
$(X_{t})_{t \geq 0}$, is represented as an affine process taking values
	in the cone of positive Hilbert--Schmidt operators. Under some extra
	integrability assumptions, the covariance process is additionally shown to be
	trace class, which is the regularity level needed for the cylindrical-noise
	interpretation of the HJMM dynamics. The resulting framework is analytically
	tractable and the accompanying Feynman--Kac martingale
	identity is established directly (Proposition~\ref{prop:affine-fk}), and it
	remains compatible with a separate geometric HJM drift-correction layer when
	such a layer is imposed.

\subsubsection{Affine Transform Formula and Finite-Rank Approximations}

The proposed stochastic covariance model $(f_{t}, X_t)_{t \geq 0}$ is
analytically tractable through an \emph{affine transform formula}. In the
driftless $H$-valued formulation considered in the main theorem, and using the
affine Feynman--Kac martingale identity established in
Proposition~\ref{prop:affine-fk}, the joint Fourier--Laplace transform of
$(f_{t}, X_t)_{t \geq 0}$ takes the following form:
\begin{align}\label{eq:affine-transform-intro}
  \mathbb{E}\left[\E^{\langle f_{t}, u_{1}\rangle_{H} - \langle X_{t}, u_{2}\rangle}\right]
  = \E^{-\Phi(t,u) + \langle f_{0}, \psi_{1}(t,u)\rangle_{H} - \langle X_{0}, \psi_{2}(t,u)\rangle},
\end{align}
where $\Phi(\cdot, u_{1}, u_{2})$, $\psi_1(\cdot, u_{1}, u_{2})$, and
$\psi_2(\cdot, u_{1}, u_{2})$ are the unique solutions to an associated system
of \emph{generalized Riccati equations}. Here, $u_{1}\in\I H$ is an imaginary
Fourier loading for the forward curve and $u_{2}\in\cHplus$ is a positive
Hilbert--Schmidt Laplace loading. Real or more general complex loadings are used
only under the separate analytic-continuation assumptions in the pricing layer.
The affine
transform formula~\eqref{eq:affine-transform-intro} conveniently expresses the joint
Fourier--Laplace transform of the model $(f_{t}, X_t)_{t \geq 0}$ as an
exponential function of the initial values $(f_{0}, X_0)$ and
quasi-explicitly up to the solution of a system of associated deterministic
differential equations.\par{}

Below, in Theorem~\ref{thm:main-convergence}, we establish
the existence of a certain class of \emph{irregular affine processes} on the
cone of positive self-adjoint Hilbert--Schmidt operators, where the drift is governed
by the Lyapunov operator of some unbounded coercive operator, as e.g. the Laplacian,
under the compact-containment condition
Assumption~\ref{assump:cV-compact-containment}. Under additional parameter
assumptions we also obtain trace-class regularity. We show in
Theorem~\ref{thm:heat-affine-model} that this operator-valued process is
well-suited for modeling the instantaneous covariance process of forward curve
dynamics and that, under the $H$-valued stochastic-integrability condition for
the forward equation (the affine Feynman--Kac martingale identity for
the limiting covariance process is no longer assumed: it is established in
Proposition~\ref{prop:affine-fk}), the joint model, henceforth called the
\emph{heat-modulated affine model}, satisfies the affine transform
formula~\eqref{eq:affine-transform-intro}. The latter martingale condition is
automatic for the finite-dimensional Galerkin systems and is verified explicitly
in the Ornstein--Uhlenbeck/BNS specialization considered below.\par

The heat-modulated affine model $(f_{t}, X_t)_{t \geq 0}$ is also numerically
tractable, as it allows for finite-rank approximations
$(f_{t}^{d}, X_t^d)_{t \geq 0}$, where
$(X_t^d)_{t \geq 0}$ denotes an affine process on positive operators
of rank $d$. Moreover, these finite-rank approximations satisfy a similar
affine transform formula to~\eqref{eq:affine-transform-intro}, with
$\Phi(\cdot, u_{1}, u_{2})$, $\psi_1(\cdot, u_{1}, u_{2})$, and
$\psi_2(\cdot, u_{1}, u_{2})$ replaced by their respective spectral Galerkin approximations
of rank $d$. This result is detailed in
Theorem~\ref{thm:finite-dim-approx-ScoV}, where we also provide explicit
convergence rates in the canonical weighted finite-horizon setting. Furthermore,
we examine a conditional robustness result for Fourier prices of exponential
curve functionals with respect to these finite-rank approximations in
Proposition~\ref{prop:robustness}; this is not an unconditional pricing theorem,
but a conditional layer based on external complex-transform and
dominated-stability hypotheses. General robustness results for approximations of the
instantaneous covariance in Hilbert space-valued stochastic covariance models were
discussed in~\cite{BE24}. The proofs of our main results on the existence of the positive
operator-valued affine processes are adapted from the work
in~\cite{karbach2023finiterank}. In contrast to~\cite{karbach2023finiterank},
where the focus is on affine processes with bounded drift operators, we here
study a class of \emph{irregular affine processes}, characterized by the
presence of an unbounded drift operator given by the Lyapunov operator of the
Laplacian. While such processes are irregular in the terminology of affine
theory, they exhibit a regularizing effect on the associated forward curve
dynamics, as shown and discussed in
Proposition~\ref{prop:cV-refinement} (under
Assumption~\ref{assump:cV-compact-containment}) and Remark~\ref{rem:regularity}
below. Affine stochastic volatility models in Hilbert spaces were initially
introduced in~\cite{CKK22b}, with the subclass of Hilbert-valued
Barndorff--Nielsen--Shephard models presented in~\cite{BRS18, BS24}. The affine class has been shown to be highly flexible,
particularly in the volatility component, allowing for L\'evy noise as well as
state-dependent and drifted jump dynamics of infinite variation. Finite-dimensional approximations
of affine processes were explored in~\cite{karbach2023finiterank, STY20}, and the long-term behavior of affine stochastic
covariance models in Hilbert spaces was investigated in~\cite{FK24}. A Hilbert
space-valued extension of the Heston stochastic volatility model was introduced
in~\cite{BS18}, and \emph{a heat-modulated version} was examined
in~\cite{BLDP22}. This article can be considered as an analogous heat-modulation
extension of the affine class in~\cite{CKK22b}. 
\par{}

\subsubsection{Regularization by Heat Modulation}\label{sec:heat-modulation}

The main idea behind heat modulation in stochastic covariance models for
forward curve dynamics is to use the instantaneous covariance process
$(X_t)_{t\ge0}$ to \emph{regularize the noise} in the HJMM
equation~\eqref{eq:HJMM}. This is achieved by augmenting the drift of the
instantaneous covariance process with the Lyapunov operator associated with the
Laplacian. More precisely, the Lyapunov operator of the Laplacian $\Delta$, acting on a suitable
forward curve space, is defined as
$$
  L_{\Delta}(\cdot)\coloneqq \Delta(\cdot)+(\cdot)\Delta^{*}.
$$
The resulting class of models is referred to as
\emph{heat-modulated stochastic covariance models}, reflecting the fact that the Laplacian generates
the heat semigroup $(T(t))_{t\geq 0}$.
The use of the Laplacian as a regularizing operator in forward curve modeling has already been
proposed by Cont~\cite{Con05}. There, a second-order
(parabolic) differential operator in the maturity direction is introduced
directly into the forward rate dynamics. Based on both empirical evidence and
statistical fitting considerations, Cont argues that such a parabolic structure
leads to several desirable modeling properties, including smoothness in
maturity, mean reversion of deformations, a parsimonious representation of
principal components, and an improved empirical fit to yield curve data. However, the resulting model is not arbitrage-free and therefore cannot be used directly for the pricing and hedging of forward curve derivatives.\newline{}
In contrast, the present approach incorporates the Laplacian indirectly through
the dynamics of the instantaneous covariance process. This preserves the
regularizing effects emphasized in~\cite{Con05} while remaining compatible with
the HJM modeling paradigm: in the driftless arithmetic specification used in our
main affine-transform theorem the forward dynamics carry no HJM drift
correction, while in a geometric specification the no-arbitrage drift correction
has to be imposed separately as a compatible extension layer. Combined with
finite-rank approximability, this gives tractable driftless benchmark dynamics
and a route to geometric pricing once the additional HJM drift and complex
transform assumptions are imposed.

This shift in perspective has several important consequences. First, following~\cite{Con05,Dou14}, the noise driving the forward curve dynamics is most naturally modeled
by infinite-rank cylindrical Brownian noise. In contrast to a
$\mathcal{Q}$-Brownian motion with a fixed covariance operator $\mathcal{Q}$,
the cylindrical formulation ensures that all information about the covariance
structure is encoded solely in the operator $X_t$. This avoids
confounding the covariance of the noise with that of the driving Brownian
motion, which is empirically indistinguishable in the data. In particular, the
instantaneous volatilities $\sigma_t^{(i)}$ in~\eqref{eq:HJMM} can be interpreted
as principal components of $\sigma_t\,\mathrm{d}W_t$ for each $t\ge0$, and can
be estimated effectively from data; see~\cite[Section~1.7]{CT06}. Moreover, the
cylindrical setting eliminates the need for restrictive commutativity
assumptions between instantaneous covariance operators and a background
covariance $\mathcal{Q}$, as required, for example, in
\cite[Assumption~2.11]{CKK22b} or~\cite[Proposition~3.2]{BRS18}. For these
reasons, the cylindrical framework is both more natural and more convenient for
modeling covariance structures in forward curve spaces.

From a mathematical perspective, however, the use of cylindrical noise raises
substantial analytical challenges. For the HJMM equation to be well posed as an
SDE on a Hilbert space, the instantaneous volatility must be Hilbert--Schmidt
valued, or equivalently, the instantaneous covariance operators must be of
trace class; see~\cite{GM11}. Since the space of trace-class operators is only a
Banach space, establishing affine dynamics directly on the trace-class cone is
technically demanding. A key contribution of this paper is therefore to work on
a Gelfand triple of Hilbert--Schmidt operators and to exploit two
\emph{distinct} regularizing effects of the Laplacian-driven drift.
It is worth separating them clearly, since they have different sources.
First, the Lyapunov semigroup associated with the Laplacian has a genuine
\emph{smoothing} effect: for every fixed $t>0$ it maps $\cH$ into the smaller
space $\cV=\cL_2(V^*,H)\cap\cL_2(H,V)$ at the deterministic level
(Lemma~\ref{lem:Lyapunov-semigroup-smoothing}). Passing this regularity to
the stochastic process $X_t$ requires additional $\cV$ jump-moment
hypotheses on $(m,\mu)$ that absorb the mode-by-mode bracket (Proposition
\ref{prop:cV-refinement}); the conclusion is that under those hypotheses,
$X_t\in\cV\cap\cHplus$ $\MP_x$-a.s.\ at every fixed $t>0$. The mechanism is illustrated
numerically in Figure~\ref{fig:heat-smoothing} below: high modes of $X_t$ are
damped at rate $2\omega_k$, so the covariance kernel relaxes towards the
smoother stationary profile of the L\'evy subordinator. Second, under additional
trace first-moment hypotheses on the jump characteristics $(m,\mu)$, the
covariance process is trace-class
(Theorem~\ref{thm:main-convergence}(a)); here the Lyapunov term contributes only
through its dissipativity (it does not increase the trace), and the finiteness is
supplied by the trace-moment hypotheses rather than by the smoothing itself.
We emphasize that the two are independent: $\cV$-regularity of $X_t$ does not
imply trace class, since $\sum_n\|X_te_n\|_V^2<\infty$ does not control
$\sum_n\langle X_te_n,e_n\rangle_H$. Trace-class regularity is the level
compatible with cylindrical noise; together with an integrated trace moment it
also gives a direct sufficient condition for the $H$-valued stochastic
convolution used in the joint transform theorem (see
Remark~\ref{rem:H-valued-trace-sufficient}).

Finally, from a modeling standpoint, it is unrealistic to assume that the noise
itself, cylindrical or otherwise, takes values in the same space of smooth forward
curves. Forward curves observed in practice are constructed by smoothing prices
or rates across maturities; see~\cite[Section~1.2]{Fil01}. The noise component,
by contrast, is latent and should not be assumed to possess the same degree of
regularity. Consequently, the instantaneous volatility operators must
\emph{regularize} the rough maturity-wise noise. Beyond the
operator-level regularization discussed above, the heat semigroup induces
a regularizing effect at the function space level, smoothing the noise.
We caution, however, that pathwise regularity of the forward curve
$f_t$ in the regular space $V$ is governed by the volatility, not the covariance:
it requires $X_t^{1/2}\in\cL_2(H,V)$, which is a \emph{separate} and strictly
stronger hypothesis than the $\cV$-regularity of $X_t$
(see Remark~\ref{rem:square-root-regularity}(iii)--(iv)). The
$H$-valued formulation used in Theorem~\ref{thm:heat-affine-model} avoids this
gap by working in the pivot space. Related
regularization effects have recently been studied in~\cite{BLDP22} for
infinite-dimensional Heston-type stochastic volatility models.

In summary, while~\cite{Con05} introduces parabolic regularization directly into
the forward curve dynamics, we obtain heat-based regularization by introducing
the Laplacian at the level of the stochastic covariance component. This retains
the structural benefits of smoothing in maturity and decay of higher-order modes
while preserving an affine structure in the covariance dynamics. From a financial
modeling perspective, this distinction is useful because the regularization acts
on the cylindrical-noise covariance rather than by changing the hyperbolic
Musiela transport term itself; geometric no-arbitrage pricing can then be added
as a separate HJM drift layer when the corresponding drift and complex-transform
assumptions are available.

\section{Stochastic Volatility Modulated HJMM and Affine
  Processes}\label{sec:stoch-covar-model}

In this section, we introduce the heat-modulated affine stochastic volatility
model. A stochastic volatility model comprises two key components: the price
process and the instantaneous (co)variance process. We begin by addressing these
components separately. In Section~\ref{sec:forward-curve-dynamics}, we examine a
generic stochastic volatility-modulated HJMM model to describe the
spatio-temporal dynamics of logarithmic forward curves in fixed-income or
commodity markets. Subsequently, in
Section~\ref{sec:admissible-irregular-affine-processes}, we present a particular
class of operator-valued affine processes designed to model the instantaneous
covariance process associated with the forward price curve dynamics.

\subsection{The Dynamics of Forward Curves}\label{sec:forward-curve-dynamics}

To interpret the HJMM stochastic partial differential equation~\eqref{eq:HJMM},
we view it as a stochastic differential equation taking values in some Hilbert
spaces of viable forward curves. In the following two paragraphs, we recall
the definition and a few essential properties of a class of forward curve spaces
introduced in~\cite[Chapter 4 \& 5]{Fil01} and~\cite{Con05, Dou14}, and subsequently
give sufficient conditions for the existence of mild solutions to the HJMM equation.

\subsubsection{The State Spaces}\label{sec:the-state-space}

Let $[0,\Theta_{\mathrm{max}}]$ denote the interval of all time-to-maturities, where
$\Theta_{\mathrm{max}}\leq \infty$ is the maximal time-to-maturity of all available
forward contracts in a market. In an idealized setting, one may assume
$\Theta_{\mathrm{max}}=\infty$. Empirical and economic considerations suggest that
any forward curve $[0,\Theta_{\mathrm{max}}]\ni x\mapsto f_t(x)$ exhibits a certain
degree of smoothness with respect to time-to-maturity, admits a limit at the long
end of the curve (often referred to as the \emph{long rate} or
\emph{long-term price}), and displays a characteristic flattening for large
maturities; see, for example,~\cite{Con05,Dou14} and
\cite[Section~1.2]{Fil01}.

These features are naturally captured by the following class of weighted Sobolev
spaces. Let $\beta>0$ and define $H_\beta(0,\Theta_{\mathrm{max}})$ as the space of
all absolutely continuous functions
$f:[0,\Theta_{\mathrm{max}})\to\mathbb R$ such that
$$
\|f\|_\beta
:= \left(|f(0)|^2 + \int_0^{\Theta_{\mathrm{max}}}
\mathrm e^{\beta x}\,|f'(x)|^2\,\mathrm dx\right)^{1/2}
<\infty,
$$
where $f'$ denotes the weak derivative. We equip $H_\beta(0,\Theta_{\mathrm{max}})$
with the inner product
$$
\langle f,g\rangle_\beta
:= f(0)g(0) + \int_0^{\Theta_{\mathrm{max}}}
\mathrm e^{\beta x} f'(x)g'(x)\,\mathrm dx .
$$
When $\Theta_{\mathrm{max}}<\infty$, functions in
$H_\beta(0,\Theta_{\mathrm{max}})$ are understood through their continuous
representatives on $[0,\Theta_{\mathrm{max}}]$. For the case
$\Theta_{\mathrm{max}}=\infty$, this space was proposed in~\cite{Fil01}
as a state space for the HJMM equation and further discussed in the context of
commodity markets in~\cite{BK14}. The exponential weight penalizes irregularities of the forward curve at large
maturities, while the derivative term enforces asymptotic flatness. The explicit
inclusion of the value $f(0)$ fixes the short-end level of the forward curve.
If $\Theta_{\mathrm{max}}<\infty$, the space
$H_\beta(0,\Theta_{\mathrm{max}})$ admits the natural decomposition
$$
H_\beta(0,\Theta_{\mathrm{max}})
\cong V_{0,\beta}\oplus\mathbb R,
\qquad
V_{0,\beta}:=\{u\in H^1(0,\Theta_{\mathrm{max}},\mathrm e^{\beta x}\mathrm dx):u(0)=0\},
$$
	via the isomorphism $f\mapsto (f-f(0),f(0))$. On an infinite horizon, the same short-end anchored space
	$V_{0,\beta}$ is still well defined, but it should not be confused with the
	long-end level decomposition. Indeed, the exponential weight implies the
	existence of
$f(\infty)\df\lim_{x\to\infty}f(x)$, and a long-end decomposition uses instead
\[
  V_{\infty,\beta}:=\{u\in H_\beta(0,\infty):\lim_{x\to\infty}u(x)=0\},
  \qquad
	f\longmapsto (f-f(\infty),f(\infty)).
	\]
	This is a different anchoring convention from $V_{0,\beta}\oplus\MR$. In the
	finite-interval Laplacian example and in the numerical implementation below we
	use the short-end anchored space $V_{0,\beta}\oplus\MR$ on
	$(0,\Theta_{\max})$. Here,
$H^1(0,\Theta_{\mathrm{max}},\mathrm e^{\beta x}\mathrm dx)$ denotes the weighted
Sobolev space of (equivalence classes of) functions $u$ with
$u,u'\in L^2(0,\Theta_{\mathrm{max}},\mathrm e^{\beta x}\mathrm dx)$.
The anchored Sobolev component $V_{0,\beta}$ controls the shape and regularity of
the curve, while the one–dimensional component determines its overall
short-end level in the finite-interval convention. The anchored Sobolev space
$V_{0,\beta}$ enforces a homogeneous boundary condition at the short end of the
maturity axis while leaving the long end free. Since the trace operator at $0$ is continuous on
$H^1(0,\Theta_{\mathrm{max}},\mathrm e^{\beta x}\D x)$ for both finite and infinite horizons, this
definition is well posed for all $\Theta_{\mathrm{max}}\in(0,\infty]$.
If $\Theta_{\mathrm{max}}=\infty$, the exponential weighting guarantees that the
limit $\lim_{x\to\infty} f(x)$ exists and is finite, providing a natural
interpretation as the long rate. In this case, we write $H_\beta$ for
$H_\beta(0,\infty)$. Moreover, $(H_\beta,\langle\cdot,\cdot\rangle_\beta)$ is a
separable Hilbert space and the long-end decomposition gives the compact
embedding
\[
H_\beta\ni f\longmapsto (f-f(\infty),f(\infty))\in
L^2(\mathbb R_+,\mathrm e^{\gamma x}\mathrm dx)\oplus\mathbb R,
\qquad 0<\gamma<\beta,
\]
see~\cite[Theorem~6]{Tap13}.
If $\Theta_{\mathrm{max}}<\infty$, the weight $\mathrm e^{\beta x}$ is bounded
above and below on $(0,\Theta_{\mathrm{max}})$, hence the weighted and unweighted
Sobolev norms are equivalent. In particular, compactness follows from Rellich’s
theorem~\cite{Rel30}: the embedding $V_{0,\beta}\hookrightarrow L^2(0,\Theta_{\mathrm{max}})$ is
compact, and therefore the embedding
\[
H_\beta(0,\Theta_{\mathrm{max}})\ni f \longmapsto (f-f(0),f(0))\in
L^2(0,\Theta_{\mathrm{max}},\mathrm e^{\beta x}\mathrm dx)\oplus\mathbb R
\]
is compact as well.
The Hilbert space $V=H_{\beta}(0,\Theta_{\mathrm{max}})$ has further
convenient properties: For every finite $x \in [0,\Theta_{\mathrm{max}})$, the point evaluation map
$\delta_{x} \colon V \to \mathbb{R}$ is a continuous linear functional
on $V$, see e.g., \cite[Theorem 2.1]{FTT10}. If
$\Theta_{\mathrm{max}}<\infty$, the endpoint trace
$f\mapsto f(\Theta_{\mathrm{max}})$ is continuous as well; if
$\Theta_{\mathrm{max}}=\infty$, the long-end functional
$f\mapsto f(\infty)$ is continuous on $H_\beta$. Thus the corresponding
evaluation functionals can be identified with elements of $V$ through the
Riesz representation theorem. Moreover, the left-shift semigroup
$(S(t))_{t \geq 0}$ is given on the infinite horizon by
$S(t) f(x) = f(x + t)$ and on a finite horizon by the flat-extension convention
$S(t) f(x) = f(x + t)$ for $0\leq x\leq \Theta_{\mathrm{max}}-t$ and
$S(t)f(x)=f(\Theta_{\mathrm{max}})$ for
$x\geq \Theta_{\mathrm{max}}-t$. This semigroup is strongly continuous on $V$,
and its infinitesimal generator is the corresponding realization of the spatial
differentiation operator $\frac{\partial}{\partial x}$ with domain determined by
the chosen shift convention; see~\cite[Section 2.3]{farkas2015isem}.

\subsubsection{The Stochastic Volatility Modulated HJMM
  Equation}\label{sec:HJMM-equation} 

Throughout, we denote by $(V, \langle \cdot, \cdot \rangle_V)$ and
$(H, \langle \cdot, \cdot \rangle_H)$ two separable Hilbert spaces, such that
$V \hookrightarrow H \hookrightarrow V^*$, where $V^*$ is the adjoint space with
respect to $\langle \cdot, \cdot \rangle_H$. As a classical benchmark on bounded domains, one may
consider the Sobolev pair $V = H_0^1(0,\Theta_{\mathrm{max}})$ embedded in the $L^{2}$-space $H = L^2(0,\Theta_{\mathrm{max}})$.
In the forward-curve setting considered in this work, we interpret $V$ as a space
of \emph{regular} forward curves and $H$ as a space of \emph{less regular} forward
curves, and we employ the weighted spaces introduced above with the following
identification convention:
$$
V = H_\beta(0,\Theta_{\mathrm{max}}),
\qquad
H = L^2((0,\Theta_{\mathrm{max}}),\mathrm e^{\gamma x}\mathrm dx)\oplus\mathbb R,
\quad \text{for }0<\gamma<\beta.
$$
If $\Theta_{\mathrm{max}}<\infty$, the embedding $V\hookrightarrow H$ is
understood through the short-end decomposition
$f\mapsto (f-f(0),f(0))$. If $\Theta_{\mathrm{max}}=\infty$, it is understood
through the long-end decomposition
$f\mapsto (f-f(\infty),f(\infty))$, so that the first component is in
$L^2(\mathbb R_+,\mathrm e^{\gamma x}\D x)$.

Both spaces $H$ and $V$ can serve as state spaces for the HJMM equation, depending on
modeling preferences and the desired smoothness of forward curves in the
respective markets. Here, to interpret the HJMM equation~\eqref{eq:HJMM}, we treat it
as a stochastic differential equation (SDE) in the space of \emph{smooth forward
curves} $V$:
\begin{align}\label{eq:HJMM-SDE}
\begin{cases}
    \D f_t = \big(\mathcal{A} f_t + g_t\big) \D t + \sigma_t \D W_t, \\
    f_0 = F(0, \cdot) \in V,
\end{cases}  
\end{align}
where $\mathcal{A} \coloneqq \frac{\partial}{\partial x}$, $(g_t)_{t \geq 0}$
denotes some $V$-valued drift, $(W_t)_{t \geq 0}$ is a cylindrical Brownian
motion on $H$, and $(\sigma_t)_{t \geq 0}$ represents the \emph{instantaneous
  volatility process} of $(f_t)_{t \geq 0}$. The initial value
$x \mapsto f_0(x) = F(0, x)$ is referred to as the \emph{initial forward
  curve}.
In a geometric HJM specification one imposes the HJM drift condition
(see e.g.~\cite[Section~6.3]{CT06}), so that the resulting forward curve
dynamics are arbitrage-free; under such a specification the drift
$(g_t)_{t\ge0}$ is determined by the volatility $(\sigma_t)_{t\ge0}$ provided
the resulting expression is well-defined in the chosen forward-curve state
space. In the arithmetic specification this reduces to $g_t\equiv 0$.\par{}

Note that the noise $(W_t)_{t \geq 0}$ is not modeled as a $V$-valued Brownian
motion, and its values are generally less regular in terms of time-to-maturity
compared to the forward curves $x \mapsto f_t(x) \in V$ at any given time
$t \geq 0$. This is consistent with the observation that forward curves are
typically constructed using smoothing techniques from market data, whereas the
noise is not subject to such smoothing. In order for $(f_t)_{t \geq 0}$ to
remain in the space $V$, the instantaneous volatility process
$(\sigma_t)_{t \geq 0}$ must necessarily map from $H$ into $V$.

Throughout this paper, $(X_t)_{t\ge0}$ denotes the
\emph{instantaneous covariance operator} (taking values in $\cHplus$) and
$(\sigma_t)_{t\ge0}$ denotes the associated \emph{volatility operator} that
modulates the cylindrical noise in~\eqref{eq:HJMM-SDE} below. Let
$J:V\hookrightarrow H$ denote the continuous embedding. In the $V$-valued
formulation, where $\sigma_t:H\to V$, the corresponding covariance operator on
the pivot space $H$ is
\[
  X_t=(J\sigma_t)(J\sigma_t)^*.
\]
In the $H$-valued formulation this reduces to $X_t=\sigma_t\sigma_t^*$.
Conversely, a self-adjoint covariance specification $X_t=X_t^*$ admits the
bounded square root $X_t^{1/2}$ on $H$. This square root is an admissible
$H$-valued volatility for cylindrical noise only when
$X_t^{1/2}\in\cL_2(H,H)$, equivalently when $X_t$ is trace-class. We emphasize
that $\sigma_t$ is the volatility object on which the stochastic-integrability
condition of Assumption~\ref{assump:integration} (or its $H$-valued analogue)
must be imposed.
The volatility operator $\sigma_t$ admits the Hilbert--Schmidt expansion
\begin{align}\label{eq:volatility-mapping}
  \sigma_t f = \sum_{i=1}^{\infty} \langle v_i, f \rangle_H \sigma^{(i)}_t, \quad \text{for } f \in H,
\end{align}
where $(v_i)_{i \in \mathbb{N}}$ is an orthonormal basis (ONB) of $H$ and
$\sigma^{(i)}_t\df\sigma_t v_i$. The components
$(\sigma^{(i)}_t)_{i \in \mathbb{N}} \subseteq V$ are the
\emph{instantaneous volatilities} as in~\eqref{eq:HJMM}. If a finite-factor
specification is imposed, only finitely many of these components are non-zero.
In the general case the series in~\eqref{eq:volatility-mapping} converges in
$V$ whenever $\sigma_t\in\cL_2(H,V)$, as required in
Assumption~\ref{assump:integration} below. At this level of generality, we do not impose a specific structure on the instantaneous volatility process $(\sigma_t)_{t \geq 0}$.
However, to ensure the well-posedness of equation~\eqref{eq:HJMM-SDE}, appropriate integrability
conditions on $(\sigma_t)_{t \geq 0}$ are required. Following
\cite[Section~3.10]{GM11}, a mild solution to~\eqref{eq:HJMM-SDE} is given by the
variation-of-constants formula
\begin{align}\label{eq:solution-HJMM}
  f_t
  =
  S(t) f_0
  +
  \int_0^t S(t-s) g_s \,\D s
  +
  \int_0^t S(t-s) \sigma_s \,\D W_s,
  \qquad t \geq 0,
\end{align}
where $(S(t))_{t \geq 0}$ denotes the left-shift semigroup generated by
$\mathcal{A}$. The stochastic integral in~\eqref{eq:solution-HJMM} is well defined provided that
$(S(t-s)\sigma_s)_{s \leq t}$ is progressively measurable and satisfies
\begin{align}\label{eq:integral-condition}
\mathbb{E}\left[
  \int_0^t
  \norm{S(t-s)\sigma_s}_{\mathcal{L}_2(H,V)}^2
  \D s
\right]
<
\infty,
\qquad \forall t \geq 0.
\end{align}

Here, $\mathcal{L}_{2}(H,V)$ denotes the space of
Hilbert-Schmidt operators from $H$ to $V$ equipped with the
inner product $\langle \cdot, \cdot
\rangle_{\mathcal{L}_{2}(H,V)}$, given by  
\begin{align}\label{eq:inner-product-HS}
  \langle A, B \rangle_{\mathcal{L}_2(H,V)} \coloneqq
  \sum_{n=1}^{\infty} \langle A e_n, B e_n
  \rangle_{V}, \quad  A, B \in \mathcal{L}_{2}(H,V).
\end{align}
The space
$(\mathcal{L}_{2}(H,V),\langle \cdot, \cdot \rangle_{\mathcal{L}_{2}(H,V)})$ is
a separable Hilbert space, where we denote the \emph{adjoint} of $A\in\cL(H,V)$
by $A^{*}$, and the definition of the inner-product~\eqref{eq:inner-product-HS}
is independent of the choice of the basis $(e_{i})_{i\in\MN}\subseteq H$.

Moreover, on the maturity component any positive self-adjoint Hilbert-Schmidt
operator admits a kernel representation with respect to the underlying reference
measure. More precisely, let
\[
  H_0=L^2((0,\Theta_{\mathrm{max}}),w(x)\D x),
\]
where $w(x)=\mathrm e^{\gamma x}$ in the weighted setting and $w\equiv 1$ in the
unweighted setting. Then, for each
$t\ge 0$, the maturity--maturity block of $X_t$ is represented by a symmetric
kernel
$k_t\in L^{2}\big((0,\Theta_{\mathrm{max}})^{2},w(x)w(y)\D x\D y\big)$ such that
\begin{align}\label{eq:kernel-representation}
  (X_tf)(x)
  = \int_{0}^{\Theta_{\mathrm{max}}} k_t(x,y)\,f(y)\,w(y)\,\D y,
  \qquad f\in H_0.
  \end{align}
If the full pivot space is $H=H_0\oplus\mathbb R$, the remaining covariance
terms involving the one-dimensional level component are finite-dimensional
blocks.
The kernel $k_t(x,y)$ admits a natural interpretation as the instantaneous
covariance between forward prices with time-to-maturity $x$ and $y$.
This identification is classical in the theory of Hilbert-Schmidt operators and
provides a flexible, nonparametric description of maturity-wise dependence.
Extensions of this kernel representation to (weighted) Sobolev spaces, including
$H_0^1(0,\Theta_{\mathrm{max}})$ and $H_{\beta}$, are given in
\cite[Section~4]{BK14}. Since the shift semigroup $(S(t))_{t\geq 0}$ is strongly continuous on $V$, the process $(S(t-s)\sigma_s)_{s \leq t}$ is progressively measurable whenever the process
$(\sigma_{t})_{t\geq 0}$ is. A convenient sufficient way to guarantee
condition~\eqref{eq:integral-condition} is to require
$(\sigma_{t})_{t \geq 0}$ to take values in $\mathcal{L}_{2}(H,V)$, since
$(S(t))_{t\geq 0}$ is a semigroup of merely bounded operators on $V$ and
$\norm{S(t-s)\sigma_{s}}_{\mathcal{L}_{2}(H,V)}\leq
\norm{S(t-s)}_{\cL(V)}\norm{\sigma_{s}}_{\mathcal{L}_{2}(H,V)}$. For the
$V$-valued formulation driven by a cylindrical Brownian motion on $H$, we impose
the following sufficient integrability assumption:

\begin{assumption}{A}\label{assump:integration}
  The instantaneous volatility process $(\sigma_{t})_{t\geq 0}$ is progressively
  measurable with respect to some filtration $(\cF_{t})_{t\geq 0}$ to which
  $(W_{t})_{t\geq 0}$ is adapted, and satisfies $\sigma_{t}\in
  \cL_{2}(H,V)$ for all $t\geq 0$ and~\eqref{eq:integral-condition}.
\end{assumption}

\begin{remark}
  If the modeling of the forward curves takes place in the space $H$ rather than
  $V$, with $(W_t)_{t \geq 0}$ still being a cylindrical Brownian motion on $H$,
  we can replace $V$ with $H$ in Assumption~\ref{assump:integration} and still
  obtain an assumption leading to viable instantaneous volatility
  processes, provided that the shift is realized as a strongly continuous
  semigroup on the chosen pivot space $H$. On finite maturity intervals this
  requires an $L^2$-well-defined boundary convention, such as the
  level-preserving stopped shift used in
  Section~\ref{sec:abstr-affine-stoch}; the flat endpoint extension on $V$
  cannot be transferred verbatim to $L^2$-equivalence classes. Moreover, note that a cylindrical Brownian motion on $H$ is not an
  $H$-valued stochastic process, but can be identified with a Brownian motion on
  some larger space $G$, see~\cite{PZ07}. This implies that, even in this case,
  the noise remains less regular in time-to-maturity than the forward curves. 
\end{remark}

\subsection{The affine instantaneous covariance
  process}\label{sec:admissible-irregular-affine-processes}

In this section, we introduce a class of operator-valued stochastic processes
that we propose as the instantaneous covariance process
$(X_{t})_{t \geq 0}$ underlying~\eqref{eq:HJMM-SDE}. The associated
volatility operator $\sigma_t$, which is the object that has to satisfy
Assumption~\ref{assump:integration} (or its $H$-valued analogue), is not
determined by the affine construction alone. When the covariance process is used
to drive~\eqref{eq:HJMM-SDE}, the volatility must be obtained from $X_t$ through
a factorization $X_t=(J\sigma_t)(J\sigma_t)^*$ in the $V$-valued formulation, or
$X_t=\sigma_t\sigma_t^*$ in the $H$-valued formulation, with the required
Hilbert--Schmidt regularity imposed separately. The canonical choice
$\sigma_t=X_t^{1/2}$ on the self-adjoint $H$-valued specification is treated
separately in Remark~\ref{rem:square-root-regularity} below. The affine
construction is for the covariance $X_t$, not directly for the volatility
$\sigma_t$.

We write $\cH$ for all symmetric operators in $\cL_{2}(H,H)$, $\langle\cdot,\cdot\rangle\df \langle \cdot,\cdot
\rangle_{_{\mathcal{L}_2(H,H)}}$ and $\norm{\cdot}=\langle\cdot,\cdot\rangle^{1/2}$. Let $\cH^{+}$ be the set of all
positive operators in $\cH$, i.e.,
$$\cH^{+} \df \{ A \in \cH \colon \langle A v, v \rangle_{H} \geq 0 \text{ for
  all } v \in H \}.$$ Note that $\cH^{+}$ is closed and a convex cone in $\cH$,
meaning that $\cH^{+} + \cH^{+} \subseteq \cH^{+}$,
$\lambda \cH^{+} \subseteq \cH^{+}$ for all $\lambda \geq 0$,
$ \cH^{+}\cap (-\cH^{+}) = \{ 0 \}$. The cone $\cH^{+}$ induces a partial
ordering $\leq_{\cH^{+}}$ on $\cH$, defined by $x \leq_{\cH^{+}} y$ whenever
$y - x \in \cH^{+}$. Lastly, define $$\cV\df \cL_{2}(V^{*},H)\cap\cL_{2}(H,V)$$
with $V^{*}$ the adjoint of $V$ with respect to
$\langle \cdot,\cdot\rangle_{H}$.
We adopt the convention that $\cV$ denotes the full intersection (no
symmetry constraint is built in). Whenever covariance operators are meant,
we write $\cV\cap\cH$ or $\cV\cap\cHplus$ explicitly. Adjoint regularity
statements in Remark~\ref{rem:lyapunov-regularity} below use the adjoint of the
$V^*\to H$ realization of an element of $\cV$, and are imposed only on the
symmetric covariance subspace $\cV\cap\cH$ unless stated otherwise.

\begin{definition}\label{def:admissible-irregular}
  Let  $\chi \colon \cH \to \cH$ be given by
\(\chi(\xi) = \xi \one_{\norm{\xi} \leq 1}(\xi)\). We call the tuple
$(b,\mathbf{B},m,\mu)$ an \textit{admissible parameter set} (with respect to
$\chi$) if the following holds:  
  \begin{defenum}
    \item\label{item:m-2moment} a measure
    $m\colon\cB(\cHpluso)\to [0,\infty]$ such that
    \begin{enumerate}
    \item[(a)] $\int_{\cHpluso} \| \xi \|^2 \,\dm < \infty$ and
    \item[(b)] $\int_{\cHpluso}|\langle\chi(\xi),h\rangle|\,\dm<\infty$ for all
      $h\in\cH$   
   and there exists an element $I_{m}\in \cH$ such that $\langle
   I_{m},h\rangle=\int_{\cHpluso}\langle \chi(\xi),h\rangle\, m(\D\xi)$ for
   every $h\in\cH$;  
    \end{enumerate}   
  \item \label{item:drift} a vector $b\in\cH$ such that
    \begin{align*}
      \langle b, v\rangle - \int_{\cHpluso} \langle
      \chi(\xi), v\rangle \,m(\D\xi) \geq 0\, \quad\text{for all}\;v\in\cHplus;
    \end{align*}    
  \item \label{item:affine-kernel} a $\cHplus$-valued measure $\mu \colon
    \mathcal{B}(\cHpluso) \rightarrow \cHplus$ such that
    $\mu(\cHpluso)\in\cHplus\subset\cH$ (in particular the total vector mass is
    finite in $\cH$) and the kernel
    $M(x,\D\xi)$, for every $x\in\cHplus$ defined on $\mathcal{B}(\cHpluso)$ by
    \begin{align}\label{eq:affine-kernel-M}
      M(x,\D\xi)\df \frac{\langle x, \mu(\D\xi)\rangle }{\norm{\xi}^{2}},
    \end{align}
    satisfies
    \begin{align}\label{eq:affine-kernel-quasi-mono}
      \int_{\cHplus\setminus \{0\}} \langle \chi(\xi), u\rangle\,M(x,\D\xi)< \infty,  
    \end{align}
    for all $u,x\in \cHplus$ such that $\langle u,x \rangle = 0$;
\item\label{item:linear-operator-unbounded}
$\mathbf{B}\in\cL(\cV,\cV^{*})$ is of the form
\begin{align*}
  \mathbf{B}(x)=Bx+xB^{*}+\Gamma(x),\quad x\in \cV,
\end{align*}
where $\cH$ is identified with a subspace of $\cV^*$ through the continuous
embedding $\cV\hookrightarrow\cH$, i.e. $h\in\cH$ acts on $y\in\cV$ by
$y\mapsto\langle h,y\rangle_{\cH}$. In particular, the term
$\Gamma(x)\in\cH$ is interpreted as an element of $\cV^*$ in the above display.
The components are:
\begin{enumerate}
\item[(a)]
$B$ is a \emph{self-adjoint, non-positive} linear operator on $H$
with spectral domain
\[
  \dom(B)=\Big\{h\in H\colon\sum_{i\in\MN}\lambda_i^{2}\,|\langle
  h,e_i\rangle_H|^{2}<\infty\Big\},
\]
where $(e_{i})_{i\in\MN}$ is a complete orthonormal basis of $H$ consisting
of eigenvectors of $B$ satisfying
\begin{align*}
  B e_{i}=-\lambda_{i} e_{i},
\end{align*}
for a sequence $(\lambda_{i})_{i\in\MN}\subseteq\MRplus$ such that
\begin{align*}
  0\le \lambda_{1}\le \lambda_{2}\le \cdots \le \lambda_{n}\to\infty
  \quad\text{as } n\to\infty .
\end{align*}
We further assume that $B,B^{*}\in\cL(V,V^{*})$, where $B^{*}=B$ by
self-adjointness on $H$ extended to the form sense
$\langle Bu,v\rangle_{V^*,V}=\langle u,Bv\rangle_{V^*,V}$ for $u,v\in V$.
The self-adjoint $C_0$-semigroup
$T(t)=\E^{tB}$ on $H$ is defined spectrally by
$T(t)e_i=\E^{-\lambda_i t}e_i$ and extended by linearity and closure;
because all $\lambda_i\ge0$, it is a contraction semigroup. The Lyapunov
semigroup $\cT(t)x\df T(t)xT^{*}(t)=T(t)xT(t)$ on $\cH$ is the
corresponding self-adjoint contraction semigroup induced by conjugation;
its generator $(L,\dom(L))$ is given spectrally by~\eqref{eq:Lyapunov} on
the basis $\{\be_{i,j}\}_{i\le j}$ of $\cH$.

\item[(b)]
$\Gamma\colon \cH\to \cH$ is a bounded linear operator such that
$\Gamma(\cHplus)\subseteq\cHplus$ and
\begin{align*}
  \left\langle \Gamma(x), u \right\rangle
  -
  \int_{\cHpluso}
  \langle \chi(\xi),u\rangle
  \frac{\langle \mu(\D\xi), x \rangle}{\| \xi\|^2 }
  \ge 0,
\end{align*}
for all $x,u \in \cHplus$ such that $\langle u,x\rangle=0$.

\item[(c)]
Moreover, we assume that the operator $B$ satisfies the following coercivity
assumption: there exist $\alpha_{0},\alpha_{1}>0$ and $\lambda\ge 0$ such that
\begin{align*}
  \alpha_{0}\norm{u}^{2}_{V}
  \le
  -2\langle B u,u\rangle_{H}+\lambda\norm{u}_{H}^{2}
  \le
  \alpha_{1}\norm{u}_{V}^{2},
  \qquad u\in V,
\end{align*}
where $\langle Bu,u\rangle_H$ is to be read as the duality
pairing $\langle Bu,u\rangle_{V^*,V}$.
In addition, the forward-curve regularity space $V$ is assumed to be compatible
with the spectral scale of $B$: the eigenvectors $(e_i)_{i\in\MN}$ belong to
$V$, their linear span is dense in $V$, and the norm
\[
  \|u\|_{V,B}^{2}
  \df
  \sum_{i\in\MN}(1+\lambda_i)|\langle u,e_i\rangle_H|^2
\]
is equivalent to $\|u\|_V^2$ on $V$. Equivalently, $V$ is the form domain of
$I-B$ up to equivalent norms. Since $\lambda_i\to\infty$, this spectral
compatibility also implies the compact embedding $V\hookrightarrow H$.
\end{enumerate}
\end{defenum}
\end{definition}

\begin{assumption}{C}\label{assump:cV-compact-containment}
  In addition to admissibility in Definition~\ref{def:admissible-irregular},
  the parameters satisfy the following compact-containment hypotheses:
  \begin{align*}
    b\in\cV\cap\cH,\qquad
    \Gamma(\cV\cap\cH)\subseteq\cV\cap\cH
    \quad\text{and}\quad
    \Gamma:\cV\cap\cH\to\cV\cap\cH
    \text{ is bounded},
  \end{align*}
  \begin{align*}
    \int_{\cHpluso}\|\xi\|_{\cV}^{2}\,m(\D\xi)<\infty,
  \end{align*}
  and there exists a constant $C_{\cV}<\infty$ such that, for every
  $y\in\cV\cap\cHplus$,
  \begin{align}\label{eq:cV-mu-compact-containment}
    \int_{\cHpluso}\|\xi\|_{\cV}^{2}\,
      \frac{\langle y,\mu(\D\xi)\rangle}{\|\xi\|^{2}}
    \le C_{\cV}\bigl(1+\|y\|_{\cV\cap\cH}^{2}\bigr).
  \end{align}
\end{assumption}

Assumption~\ref{assump:cV-compact-containment} is used only for stochastic
compact containment of the Galerkin laws. It is separate from admissibility:
admissibility gives the finite-rank affine processes and Riccati equations,
whereas Assumption~\ref{assump:cV-compact-containment} supplies the uniform
$\cV$-moment estimates needed to extract an infinite-rank weak limit.

\begin{remark}[Roles of (a) and (c)]\label{rem:B-spectral}
The self-adjoint, non-positive realization of $B$ and the spectral
construction of $T(t)$ and $\cT(t)$ are part of the definition above; this
remark only collects two interpretive points used downstream. First, since
$Be_i=-\lambda_i e_i$ with $\lambda_i\ge0$, the operator $B$ is dissipative
on $H$ already (the shift $\lambda\ge0$ in~(c) is not needed for
dissipativity). The shift in~(c), together with the spectral compatibility
assumption, turns the form
$[u,u]_B\df -\langle Bu,u\rangle_{V^*,V}+\lambda\|u\|_H^2$
into one comparable to $\|u\|_V^2$ and identifies $V$ with the Hilbert scale
generated by $B$. The downstream consequence is a clean split: mere
dissipativity ($\lambda$ irrelevant) drives the trace bound of
Lemma~\ref{lem:trace-bound-finite-rank}, while the full $V$-coercivity
and Hilbert-scale compatibility drive the $\cV$-smoothing of
Lemma~\ref{lem:Lyapunov-semigroup-smoothing}.
Second, in Example~\ref{ex:Laplacian_Forward} the combined $H$-spectrum is
indexed with the level component as the first mode: $e_1=(0,1)$ and
$\lambda_1=0$, with shape modes relabelled as $e_{n+1}=(q_n,0)$, where
$(q_n)_{n\ge1}$ denotes the strictly positive shape-spectrum eigenbasis;
thus $0=\lambda_1<\lambda_2\le\lambda_3\le\cdots$, in agreement with
Definition~\ref{def:admissible-irregular}(a).
\end{remark}

\begin{remark}[On the second-moment condition]\label{rem:second-moment}
Condition~\ref{item:m-2moment}(a) imposes the \emph{full} second moment
$\int_{\cHpluso}\|\xi\|^2\,m(\D\xi)<\infty$, which is stronger than the truncated
condition $\int_{\cHpluso}(\|\xi\|^2\wedge1)\,m(\D\xi)<\infty$ used for existence
in the regular affine theory of~\cite{CKK22a}. The strengthening is genuinely
needed here, but only for the \emph{uniform-in-$d$} second-moment and
Burkholder--Davis--Gundy estimates underlying tightness
(Lemma~\ref{lem:irregular-tightness-square-bound}); the existence of each
finite-rank process and the mild Riccati solution use only the truncated
condition.
\end{remark}

\begin{example}[The Laplacian on forward curve spaces]\label{ex:Laplacian_Forward}
The canonical example of an operator $B$ satisfying
\cref{item:linear-operator-unbounded} in the context of forward curve modeling
is the Laplacian.\newline{}
Throughout this example we work on a \emph{finite} maturity interval, i.e.
let $0<\Theta_{\mathrm{max}}<\infty$ and $\beta>0$. As before, we decouple the forward curve space into
a shape component and a level component by setting
$$
V_{0,\beta}=\{u\in H^1(0,\Theta_{\mathrm{max}},\mathrm e^{\beta x}\mathrm dx):u(0)=0\},
\qquad
H_0=L^2\bigl(0,\Theta_{\mathrm{max}},\mathrm e^{\beta x}\mathrm dx\bigr).
$$
Every forward curve $f\in H_{\beta}(0,\Theta_{\mathrm{max}})$ admits the unique
decomposition
$$
f(x)=u(x)+r,\qquad r=f(0),\quad u=f-f(0)\in V_{0,\beta} .
$$
We define the state space $V$ and the pivot space $H$ by
$$
V:=V_{0,\beta}\oplus\mathbb R,
\qquad
H:=H_0\oplus\mathbb R .
$$
	Next, define the bilinear form $a:V_{0,\beta}\times V_{0,\beta}\to\mathbb R$ by
$$
a(u,v)
:=\int_0^{\Theta_{\mathrm{max}}} u'(x)v'(x)\,\mathrm e^{\beta x}\,\mathrm dx,
$$
and let $A$ be the self-adjoint operator on $H_0$ associated with $a$.
Integration by parts in the form identifies $A$ with the weighted
Sturm--Liouville operator
\[
  Au=-\mathrm e^{-\beta x}\bigl(\mathrm e^{\beta x}u'\bigr)'=-u''-\beta u',
\]
with the Dirichlet condition $u(0)=0$ at the left endpoint (inherited from
$V_{0,\beta}$) and the natural (weighted Neumann) condition
$\mathrm e^{\beta\Theta_{\mathrm{max}}}u'(\Theta_{\mathrm{max}})=0$ at the
right endpoint, the latter arising from the form domain through integration
by parts.
The operator $A$ is positive, self-adjoint, and has compact resolvent.
We denote by $\Delta:=-A$ the corresponding negative generator, henceforth referred
to as the Laplacian. We now define the full operator $B$ on $H$ by
$$
B:=\Delta\oplus 0,
\qquad
	\dom(B):=\dom(\Delta)\oplus\mathbb R .
	$$
	Thus $B(u,r)=(\Delta u,0)$, reflecting that the shape evolves while the level is preserved.
	By the spectral theorem for self-adjoint operators with compact resolvent, there exists
an orthonormal basis $(q_n)_{n\in\mathbb N}$ of $H_0$ consisting of eigenvectors of $\Delta$
satisfying
$$
\Delta q_n=-\eta_n q_n,
$$
where $0 < \eta_1 \le \eta_2 \le \dots$ is the sequence of eigenvalues of $A$
accumulating at infinity.
The associated weighted Sturm--Liouville problem has the standard
one-dimensional Weyl asymptotics
\[
  \eta_n\sim (n\pi/\Theta_{\mathrm{max}})^{2},
  \qquad n\to\infty .
\]
We do not use an explicit closed form for the eigenfunctions below. This avoids
the special-case distinction between oscillatory and possible hyperbolic lowest
modes and keeps only the spectral properties needed for the Galerkin analysis.
The Weyl law above describes the $n\to\infty$ (oscillatory) branch only: the
substitution $u=\mathrm e^{-\beta x/2}v$ turns the eigenvalue problem into
$-v''=(\eta-\beta^2/4)v$, so finitely many low modes with
$\eta<\beta^2/4$ may sit on a non-oscillatory branch. These low modes are
resolved only loosely by the Galerkin basis; what enters the convergence plots
below are the \emph{discrete Galerkin eigenvalues}
$\lambda_{d}^{\mathrm{tens}}$ (variational upper bounds on the true spectrum),
not the asymptotic values, and the two agree at the small ranks $d$ shown.
Consequently,
$$
\{(0,1)\}\cup\{(q_n,0):n\in\mathbb N\}
$$
is an orthonormal basis of $H$ consisting of eigenvectors of $B$. Relabel the
combined basis by
\[
  e_1=(0,1),\qquad e_{n+1}=(q_n,0),\qquad
  \lambda_1=0,\qquad \lambda_{n+1}=\eta_n,\quad n\in\mathbb N .
	\]
	Then $B e_i=-\lambda_i e_i$ and
	$0=\lambda_1<\lambda_2\le\lambda_3\le\cdots$.\newline{}
	Since $V_{0,\beta}$ enforces the boundary condition $u(0)=0$, the Poincar\'e inequality implies
that there exist constants $\alpha_0,\alpha_1>0$ such that
$$
\alpha_0\|u\|_{V_{0,\beta}}^2
\le
a(u,u)
\le
\alpha_1\|u\|_{V_{0,\beta}}^2,
\qquad u\in V_{0,\beta} .
$$
Hence the form is coercive on the shape component. For the operator $\Delta$ this
yields
$$
-2\langle \Delta u, u \rangle_{H_0}
\ge c_0 \|u\|_{V_{0,\beta}}^2,
\qquad u\in V_{0,\beta} ,
$$
for some $c_0>0$.
Introducing the shift parameter $\lambda>0$ in
\cref{item:linear-operator-unbounded}(c), we obtain for any $f=(u,r)\in V$,
$$
-2\langle Bf,f\rangle_H+\lambda\|f\|_H^2
=
-2\langle \Delta u,u\rangle_{H_0}
+\lambda(\|u\|_{H_0}^2+|r|^2)
\geq
c_0\|u\|_{V_{0,\beta}}^2+\lambda|r|^2 .
$$
Choosing $\lambda$ sufficiently large yields
$$
-2\langle Bf,f\rangle_H+\lambda\|f\|_H^2
\geq c\,\|f\|_V^2,
$$
for some $c>0$, which is the lower bound in \cref{item:linear-operator-unbounded}(c).
For the upper bound, we use the bilinear form estimate $a(u,u)\le \alpha_1\|u\|_{V_{0,\beta}}^2$ and the continuous embedding $V\hookrightarrow H$ (which gives $\|f\|_H^2\le C_{\mathrm{emb}}^2\|f\|_V^2$ for some constant $C_{\mathrm{emb}}>0$):
$$
-2\langle Bf,f\rangle_H+\lambda\|f\|_H^2
=
2a(u,u)+\lambda\|f\|_H^2
\le
2\alpha_1\|u\|_{V_{0,\beta}}^2+\lambda C_{\mathrm{emb}}^2\|f\|_V^2
\le
(2\alpha_1+\lambda C_{\mathrm{emb}}^2)\|f\|_V^2,
$$
which establishes the upper bound $\alpha_1' \coloneqq 2\alpha_1+\lambda C_{\mathrm{emb}}^2$.
The same spectral representation verifies the Hilbert-scale compatibility in
Definition~\ref{def:admissible-irregular}(iv)(c). Indeed, the form domain of
$I-B$ is $V_{0,\beta}\oplus\MR$, the eigenvectors
$\{(0,1)\}\cup\{(q_n,0):n\in\MN\}$ are dense in this space, and the spectral
	norm
\[
  |r|^2+\sum_{n\ge1}(1+\eta_n)|\langle u,q_n\rangle_{H_0}|^2
\]
is equivalent to the product norm on $V_{0,\beta}\oplus\MR$ by the spectral
	theorem for the closed coercive form $a$.
	Finally, we note that the operator $B$ generates the strongly continuous semigroup
$$
T(t)=\E^{tB}= T_\Delta(t)\oplus I_{\mathbb R}, \qquad t\geq 0,
$$
where $(T_\Delta(t))_{t\geq 0}$ is the heat semigroup of $\Delta$ on $H_0$.
Consistently with
Definition~\ref{def:admissible-irregular}(iv)(a), $T(t)=\E^{tB}$ denotes
throughout the \emph{full} semigroup on $H=H_0\oplus\MR$; its level mode
$e_1=(0,1)$ is undamped.
\end{example}

\medskip{}

Now, let $(e_{i})_{i\in\MN}$ be the sequence of eigenvectors of $B$ as defined in part (a)
of~\cref{item:linear-operator-unbounded}. For every $i\in\MN$, define
$\be_{i,i}\df e_{i}\otimes e_{i}$, and for distinct $i$ and $j$ set
$\be_{i,j}\df \frac{1}{\sqrt{2}}(e_{i}\otimes e_{j}+e_{j}\otimes e_{i})$.
Then $\norm{\be_{i,j}}=1$, $\be_{i,j}=\be_{j,i}$, and the family
$\set{\be_{i,j}}_{i\leq j\in\MN}$ is an orthonormal basis of $\cH$.
Set $\lambda_{i,j}\df\lambda_i+\lambda_j$. In accordance with
Remark~\ref{rem:B-spectral}, we define the Lyapunov operator
$L:\dom(L)\subseteq\cH\to\cH$ spectrally by
\begin{equation}\label{eq:Lyapunov}
  Lx=\sum_{i\le j}-\lambda_{i,j}\langle x,\be_{i,j}\rangle\,\be_{i,j},
  \qquad
  \dom(L)=\left\{x\in\cH:
  \sum_{i\le j}\lambda_{i,j}^{2}|\langle x,\be_{i,j}\rangle|^{2}<\infty
  \right\}.
\end{equation}
On the algebraic span of the rank-one tensors this agrees with the formal
Lyapunov expression $Bx+xB^*$, and in particular
\begin{equation}\label{eq:Lyapunov-eigenvalues}
  L(\be_{i,j})=-\lambda_{i,j}\,\be_{i,j}, \qquad i\le j.
\end{equation}
Moreover, $\lambda_{i,j}=\lambda_{j,i}$, $\lambda_{i,j}\ge 0$, and
$\lambda_{i,j}\to\infty$ as $i\to\infty$ or $j\to\infty$. The operator $L$ is
the linear part of $\mathbf{B}$. We denote by $(\cT(t))_{t\ge 0}$ the strongly
continuous semigroup on $\cH$ generated by $(L,\dom(L))$; equivalently,
$\cT(t)\be_{i,j}=\E^{-\lambda_{i,j}t}\be_{i,j}$.

    \begin{remark}\label{rem:lyapunov-regularity}
      \begin{enumerate}
          \item[i)] Assume the coercivity estimate from~\cref{item:linear-operator-unbounded}(c), i.e.,
there exist $\alpha_0,\alpha_1>0$ and $\lambda\geq 0$ such that for all $v\in V$,
\[
  \alpha_0\|v\|_V^2 \le -2\langle Bv,v\rangle_H + \lambda\|v\|_H^2 \le \alpha_1\|v\|_V^2.
\]
Let $\mathcal{V} \coloneqq \mathcal{L}_2(H,V)\cap \mathcal{L}_2(V^*,H)$ be equipped with the norm
$\|u\|_{\mathcal{V}}^2 \coloneqq \|u\|_{\mathcal{L}_2(H,V)}^2+\|u\|_{\mathcal{L}_2(V^*,H)}^2$.
Fix an ONB $(f_n)_{n\in\mathbb{N}}$ of $H$. For
$u\in\mathcal V\cap\mathcal H$, let
$u^\dagger:H\to V$ denote the Hilbert adjoint of the
$V^*\to H$ realization of $u$, i.e.
\[
  \langle u\eta,h\rangle_H=\langle \eta,u^\dagger h\rangle_{V^*,V},
  \qquad \eta\in V^*,\ h\in H .
\]
Then $u f_n\in V$ and $u^\dagger f_n\in V$, and
\[
  \|u\|_{\mathcal V}^{2}
  =
  \sum_{n=1}^{\infty}\|u f_n\|_V^2
  +
  \sum_{n=1}^{\infty}\|u^\dagger f_n\|_V^2 .
\]
On the finite-rank spectral core of $L$ (and then by closure on
$\dom(L)\cap\mathcal V\cap\mathcal H$), the Lyapunov form identity reads
\begin{align*}
  \langle Lu,u\rangle_{\mathcal{H}}
  =
  \sum_{n=1}^{\infty}\langle B(u f_n),u f_n\rangle_{V^*,V}
  +
  \sum_{n=1}^{\infty}
  \langle B(u^\dagger f_n),u^\dagger f_n\rangle_{V^*,V}.
\end{align*}
Applying the coercivity estimate to $v=u f_n$ and $v=u^\dagger f_n$ and summing over $n$
yields
$$
  \alpha_0\|u\|_{\mathcal{V}}^2
  \le -2\langle Lu,u\rangle_{\mathcal{H}} + 2\lambda\|u\|_{\mathcal{H}}^2
  \le \alpha_1\|u\|_{\mathcal{V}}^2,
$$
which establishes that $L$ satisfies the corresponding coercivity bound on the tensor space.
	        \item[ii)] Let $\mathcal L_1(H)$ and $\mathcal L_2(H)$ denote,
respectively, the Banach space of trace-class operators on $H$ (with norm
$\|\cdot\|_{\mathcal L_1(H)}$) and the Hilbert space of Hilbert--Schmidt
operators on $H$ (with norm $\|\cdot\|_{\mathcal L_2(H)}$).
In the setting of Example~\ref{ex:Laplacian_Forward}, the positive shape
eigenvalues $(\eta_n)_{n\in\mathbb{N}}$ of the weighted Sturm--Liouville operator
$A$ on the bounded interval $(0,\Theta_{\mathrm{max}})$ (with the mixed
Dirichlet/weighted-Neumann boundary conditions)
satisfy the asymptotic
$\eta_n\simeq c\,n^{2}$ for some $c>0$. In particular,
$\sum_{n\ge 1}\mathrm{e}^{-2\eta_n t}<\infty$ for every $t>0$.
Incorporating the one-dimensional contribution of the constant level
component, this yields
$$
\|T(t)\|_{\mathcal{L}_2(H)}^2 = 1 + \sum_{n\ge1}\mathrm{e}^{-2\eta_n t} < \infty,\qquad t>0,
$$
so the semigroup $(T(t))_{t>0}$ generated by $B$ is Hilbert--Schmidt
(and hence compact) for every $t>0$.
Consequently, the associated Lyapunov semigroup
$$
\mathcal{T}(t)X := T(t)XT^{*}(t),\qquad t\geq 0,
$$
maps bounded operators into trace-class operators for every $t>0$, satisfying the estimate
$$
\|\mathcal T(t)X\|_{\mathcal L_1(H)} \leq \|T(t)\|_{\mathcal L_2(H)}^{2}\,\|X\|_{\mathcal L(H)}.
$$
Indeed, since the Hilbert--Schmidt operators form a two-sided ideal
in $\mathcal L(H)$, we have that $T(t)X$ and $XT^{*}(t)$ are in $\mathcal L_{2}(H)$ with
$$
  \|T(t)X\|_{\mathcal L_2(H)}\leq \|T(t)\|_{\mathcal L_2(H)}\|X\|_{\mathcal L(H)},
  \qquad
  \|XT^{*}(t)\|_{\mathcal L_2(H)}\leq\|X\|_{\mathcal L(H)}\|T^{*}(t)\|_{\mathcal L_2(H)}.
$$
Applying the Schatten--H\"older inequality
$\|AB\|_{\mathcal L_1(H)}\leq\|A\|_{\mathcal L_2(H)}\|B\|_{\mathcal L_2(H)}$
with $A=T(t)$ and $B=XT^{*}(t)$ therefore yields
\[
  \|\mathcal T(t)X\|_{\mathcal L_1(H)}
  \leq\|T(t)\|_{\mathcal L_2(H)}\|XT^{*}(t)\|_{\mathcal L_2(H)}
  \leq\|T(t)\|_{\mathcal L_2(H)}^{2}\|X\|_{\mathcal L(H)}.
\]
This regularization is the mechanism used below to obtain trace-class covariance
regularity under the additional trace and moment hypotheses in
Theorem~\ref{thm:main-convergence}. Such trace-class regularity is precisely the
condition needed for cylindrical noise to produce a well-defined $H$-valued
stochastic convolution in the forward curve dynamics.
Moreover, in contrast to~\cite{karbach2023finiterank}, where only bounded drift operators were
covered, the present paper allows for genuinely unbounded generators $B$, such as the Laplacian.
	Such operators arise naturally in forward curve modelling when maturity-wise diffusion and
	smoothing effects are incorporated.

    \end{enumerate}
    \end{remark}

    Let $(b,\mathbf{B},m,\mu)$ be an admissible parameter set as in
    Definition~\ref{def:admissible-irregular}. Next, define the functions
    $F \colon \cHplus \to \MR$ by 
    \begin{align}\label{eq:F}
      F(u)= \langle b, u \rangle - \int_{\cHpluso} \big( \E^{-\langle \xi, u
      \rangle} - 1 + \langle \chi(\xi), u \rangle \big) \, \dm, 
    \end{align}
    and $\hat{R} \colon \cHplus \to \cH$ by
    \begin{align}\label{eq:R-hat}
      \hat{R}(u)= \Gamma^{*}(u) - \int_{\cHpluso} \big( \E^{-\langle \xi, u \rangle}
      - 1 + \langle \chi(\xi), u \rangle \big) \frac{\dmu}{\norm{\xi}^{2}}. 
    \end{align}
    The functions $F$ and $\hat{R}$ are well-defined and locally Lipschitz
    continuous on $\cHplus$, see~\cite[Remark 3.4]{CKK22a}. The functions $F$
    and $\hat{R}$, together with the Lyapunov operator $L$
    from~\eqref{eq:Lyapunov} and an initial value $(0,u)$ with $u \in \cHplus$
    give rise to the following pair of (formal) differential equations,
    which we refer to as the \emph{generalized mild Riccati equations}:
\begin{subequations}
  \begin{empheq}[left=\empheqlbrace]{align}
    \frac{\partial }{\partial t}\phi(t,u) &=
    F(\psi(t,u)), \label{eq:Riccati-phi-mild} &\phi(0,u)=0,\\
    \frac{\partial }{\partial t}\psi(t,u) &= L(\psi(t,u)) +
    \hat{R}(\psi(t,u)), &\psi(0,u)=u. \label{eq:Riccati-psi-mild}
  \end{empheq}
\end{subequations}
The differential display
\eqref{eq:Riccati-phi-mild}--\eqref{eq:Riccati-psi-mild} is formal: since $L$
is unbounded on $\cH$, the right-hand side of~\eqref{eq:Riccati-psi-mild} is
not defined for arbitrary $u\in\cHplus$. We use the variation-of-constants
formula
\begin{align}
  \psi(t,u) &= \cT(t)u + \int_0^t \cT(t-s)\hat R(\psi(s,u))\,\D s,\\
  \phi(t,u) &= \int_0^t F(\psi(s,u))\,\D s,
\end{align}
as the actual solution concept for the limiting Riccati system; $\psi(\cdot,u)$
is required to be continuous from $[0,T]$ to $\cHplus$ and $\phi(\cdot,u)$ to
be continuous from $[0,T]$ to $\MR$ (the integral representation gives
$\phi(\cdot,u)\in C([0,T];\MR)$, and $C^1$ regularity of $\phi$ follows only
when continuity of $s\mapsto F(\psi(s,u))$ is established). The Galerkin
Riccati equations~\eqref{eq:mild-Riccati-Galerkin-phi}--\eqref{eq:mild-Riccati-Galerkin-psi}
below remain in classical (strong) finite-dimensional form because $L$ is
bounded on the finite-dimensional subspace $\cH_d$.

To solve the generalized mild Riccati equations
\eqref{eq:Riccati-phi-mild}–\eqref{eq:Riccati-psi-mild}, we introduce a spectral
Galerkin-type approximation. For every $d \in \mathbb{N}$, let $H_d$ denote the
$d$-dimensional subspace of $H$ spanned by the first $d$ eigenvectors, i.e.,
$$H_d \coloneqq \operatorname{span}\{e_i \colon i = 1, \dots, d\},$$ and we define the
orthogonal projection of $H$ onto $H_d$, with respect to the inner product
$\langle \cdot, \cdot \rangle_H$, by $\sP_d$. In the finite-horizon Laplacian
example this convention includes the level mode $e_1=(0,1)$ and the first
$d-1$ shape modes.

For each $d \in \mathbb{N}$, let $\cH_d$ represent the finite-dimensional subspace of $\cH$ spanned by the family $\{\be_{i,j} \colon 1 \leq i \leq j \leq d\}$, i.e.,
\begin{align*}
\cH_d \coloneqq \operatorname{span}\{\be_{i,j} \colon 1 \leq i \leq j \leq d\},  
\end{align*}
and let $\bP_d$ denote the orthogonal projection of $\cH$ onto $\cH_d$, with
respect to the inner product $\langle \cdot, \cdot \rangle$. For every
$d \in \mathbb{N}$ and $u \in \cH$, we have $\bP_d(u) = \sP_d u
\sP_d$. Moreover, every operator in $\cH_d$ is self-adjoint and of rank at most
$d$. We define $\bP_d^\perp(u) \coloneqq u - \bP_d(u)$ and note that
$\lim_{d \to \infty} \|\bP_d^\perp x\| = 0$ for all $x \in \cH$.

Additionally, we observe that $\cH_d = \{u \sP_d \colon u \in \cL(H_d), u = u^*\}$, and for the cone of all positive self-adjoint operators in $\cH_d$, denoted by $\cH^+_d$, we have:
\begin{align*}
\cH^+_d \coloneqq \{u \sP_d \colon u \in \cL(H_d), u = u^*, \, \langle u h, h\rangle_H \geq 0 \,\, \forall h \in H_d\}.  
\end{align*}

Importantly, for all $d \in \mathbb{N}$, we have the nested inclusion
$\cH^+_d \subseteq \cH^+_{d+1} \subseteq \cHplus$. For more details on the
subspace of finite-rank operators in the space of all Hilbert-Schmidt operators,
see \cite{Ros91}. Following \cite{Goe84}, we refer to a sequence
$(\cH_d, \bP_d)_{d \in \mathbb{N}}$, as defined above, as a \emph{projection
  scheme} in $\cH$ (with respect to the basis $\{\be_{i,j}\}_{i \leq j \in
  \mathbb{N}}$). 

\begin{definition}
We define the
$d^{\mathrm{th}}$-Galerkin approximation of~\eqref{eq:Riccati-phi-mild}-\eqref{eq:Riccati-psi-mild} as the unique solutions
$\phi_{d}(\cdot,\bP_{d}(u))\in C^{1}(\MRplus,\MR)$ and
$\psi_{d}(\cdot,\bP_{d}(u))\in C^{1}(\MRplus,\cHplus_{d})$ to 
  \begin{subequations}
    \begin{empheq}[left=\empheqlbrace]{align}
      \,\frac{\partial}{\partial
        t}\phi_{d}(t,\bP_{d}(u))&=F_{d}(\psi_{d}(t,\bP_{d}(u))), &\phi_{d}(0,\bP_{d}(u))=0, \label{eq:mild-Riccati-Galerkin-phi} \\
      \,\frac{\partial}{\partial
        t}\psi_{d}(t,\bP_{d}(u))&=R_{d}(\psi_{d}(t,\bP_{d}(u))), &\psi_{d}(0,\bP_{d}(u))=\bP_{d}(u),\label{eq:mild-Riccati-Galerkin-psi}
    \end{empheq} 
  \end{subequations}
  where $F_{d}\colon\cHplus_{d}\to \MR$ is defined by $F_{d}(u_{d})\df
  F(u_{d})$ and $R_{d}\colon\cHplus_{d}\to \cH_{d}$ is defined for every $u_{d}\in
  \cHplus_{d}$ as 
\begin{align}\label{eq:R-Galerkin-irregular}
  R_{d}(u_{d})= L(u_{d})+\hat{R}_{d}(u_{d})=Bu_{d}+u_{d}B^{*}+\hat{R}_{d}(u_{d}),
\end{align}
with $\hat{R}_{d}(u_{d})\df\bP_{d}(\hat{R}(u_{d}))$ for
$u_{d}\in \cHplus_{d}$. 
\end{definition}

\begin{proposition}\label{prop:existence-mild-Riccati}
Let $(b,\mathbf{B},m,\mu)$ be an admissible parameter set as defined in
Definition~\ref{def:admissible-irregular}. Then, for every $d\in\MN$, $T>0$
and $u\in\cHplus$, there exists a unique solution
$\big(\phi_{d}(\cdot,\bP_{d}(u)),\psi_{d}(\cdot,\bP_{d}(u))\big)$ to the system
of equations~\eqref{eq:mild-Riccati-Galerkin-phi}-\eqref{eq:mild-Riccati-Galerkin-psi}
such that $\phi_{d}(\cdot,\bP_{d}(u))\in C^{1}([0,T],\MR)$ and
$\psi_{d}(\cdot,\bP_{d}(u))\in C^{1}([0,T],\cHplus_{d})$. Moreover,
the Galerkin approximations $\phi_{d}(\cdot,\bP_{d}(u))$ and $\psi_{d}(\cdot,\bP_{d}(u))$ converge
uniformly on $[0,T]$ to $\phi(\cdot,u)\in C^{1}([0,T],\MR)$
and $\psi(\cdot,u)\in C([0,T],\cHplus)$, the unique mild solution of the
limiting Riccati system in the variation-of-constants sense. More precisely,
the pair satisfies
\begin{align}
  \psi(t,u)&=\cT(t)u + \int_{0}^{t} \cT(t-s)\hat R(\psi(s,u))\D s,\quad t\in [0,T],
  \label{eq:mild-solution-variation}\\
  \phi(t,u)&=\int_0^t F(\psi(s,u))\,\D s,\quad t\in[0,T].
  \label{eq:mild-solution-phi}
\end{align}
Finally,
\begin{align}\label{eq:convergence-mild-Galerkin}
\sup_{t\in[0,T]}\big(|\phi_{d}(t,\bP_{d}(u))-\phi(t,u)|+\norm{\psi_{d}(t,\bP_{d}(u))-\psi(t,u)}\big)\longrightarrow 0,
  \qquad d\to\infty.
\end{align}
\end{proposition}

The well-posedness of the generalized (mild) Riccati equations is closely
related to the existence of an associated class of \emph{affine processes},
which we first define for the Galerkin approximations and subsequently show that
the affine processes associated with the approximations converge to the desired
irregular affine process.

For each $d \in \mathbb{N}$, $F_{d}(u_{d})\df
  F(u_{d})$ and $R_d$ defined as in
\eqref{eq:R-Galerkin-irregular}, we define the set
$\cD^d \coloneqq \{\E^{-\langle \cdot, u \rangle} \colon u \in \cHplus_d\}
\subseteq C(\cHplus, \mathbb{R})$, and the operator
$\cG^d \colon \cD^d \to C(\cHplus, \mathbb{R})$ as:
(The operator $\cG^d$ is well-defined because $\cH_d$ is invariant
under the Lyapunov operator $L$: by~\eqref{eq:Lyapunov-eigenvalues},
$L(\be_{i,j})=-\lambda_{i,j}\be_{i,j}\in\cH_d$ for $1\le i\le j\le d$, and
$\hat R_d=\bP_d\circ\hat R$ maps into $\cH_d$ by construction; hence
$R_d(u_d)=L(u_d)+\hat R_d(u_d)\in\cH_d$ for every $u_d\in\cH^+_d$.)
\begin{align}\label{eq:G-d-operator}
  \cG^d \E^{-\langle \,\cdot\,, u \rangle}(x) \coloneqq \left(-F_d(u) - \langle x, R_d(u) \rangle\right) \E^{-\langle x, u \rangle}, \quad x \in \cHplus.
\end{align}

The following proposition asserts that, for every $d \in \mathbb{N}$, the
$d^{\mathrm{th}}$ Galerkin Riccati solution gives rise to an affine Markov
process $X^d$ with values in $\cHplus_d$, which solves the
martingale problem for $\cG^d$ with $X^d_0 = \bP_d(x)$ on a suitable stochastic
basis. This is the \emph{irregular} version of~\cite[Proposition
3.3]{karbach2023finiterank}.

\begin{proposition}\label{prop:embedding-affine-main}
  Assume that the conditions of Proposition~\ref{prop:existence-mild-Riccati}
  hold. Then for every $d\in\MN$, the following holds:
  \begin{propenum}
  \item[i)]\label{item:embedding-affine-main-1} There exists a unique Markov
    process $(X^{d},(\MP_{x}^{d})_{x\in\cHplus})$, realized on the space
    $D(\MRplus,\cHplus)$ of all c\`adl\`ag paths and where $\MP_{x}^{d}$
    denotes the law of $X^{d}$ given $X_{0}^{d}=\bP_{d}(x)$, such that for all
    $x\in\cHplus$ we have $\MP_{x}^{d}(\set{X^{d}_{t}\in\cHplus_{d}\colon\,
      t\geq 0})=1$ and the following affine transform formula holds true:    
    \begin{align}\label{eq:affine-Galerkin}
      \hspace{10mm} \EXspec{\mathbb{P}^{d}_{x}}{\E^{-\langle X^{d}_{t},
      \bP_{d}(u)\rangle}}=\E^{-\phi_{d}(t,\bP_{d}(u))-\langle
      \bP_{d}(x),\psi_{d}(t,\bP_{d}(u))\rangle},\,\, t\geq 0,\,u\in\cHplus,  
    \end{align}
    for $\big(\phi_{d}(\cdot,\bP_{d}(u)),\psi_{d}(\cdot,\bP_{d}(u))\big)$ the
    unique solution
    of~\eqref{eq:mild-Riccati-Galerkin-phi}-\eqref{eq:mild-Riccati-Galerkin-psi}. 
  \item[ii)]\label{item:embedding-affine-main-2} For every $x\in\cHplus$
    and every $u\in\cHplus$ the process 
    \begin{align}\label{eq:Gd-martingale-problem}
      \hspace{10mm}\Big(\E^{-\langle X_{t}^{d},\bP_{d}(u)\rangle}-\E^{-\langle
      \bP_{d}(x),\bP_{d}(u)\rangle}-\int_{0}^{t}(\cG^{d}\E^{-\langle
      \,\cdot\,,\bP_{d}(u)\rangle})(X_{s}^{d})\,\D s\Big)_{t\geq 0},
    \end{align}
    is a real-valued martingale with respect to the stochastic basis
    \[
      (\Omega,\bar{\cF}^{d},\bar{\MF}^{d},\MP_{x}^{d}).
    \]
    Here
    $\Omega=D(\MRplus,\cHplus)$, and
    $\bar{\MF}^{d}=(\bar{\cF}^{d}_{t})_{t\geq0}$ denotes the augmentation of
    the natural filtration of $X^{d}$ with respect to $\MP_{x}^{d}$ from i).
  \end{propenum}
\end{proposition}

Let $\cL_{1}(H)$ denote the space of all self-adjoint trace-class operators on $H$. Next, we show that any
\textit{admissible parameter set} $(b, \mathbf{B}, m, \mu)$ satisfying
Assumption~\ref{assump:cV-compact-containment} gives rise to an
affine process with values in $\cHplus$. Trace-class regularity is an
additional refinement under the hypotheses recorded in
part~\cref{item:main-convergence-1} of Theorem~\ref{thm:main-convergence}
below. The theorem also states that the sequence of finite-rank operator-valued
processes $(X^{d})_{d\in\MN}$ converges weakly to this limiting process.

\begin{theorem}\label{thm:main-convergence}
  Let $(b,\mathbf{B},m,\mu)$ be an admissible parameter set as defined in
  Definition~\ref{def:admissible-irregular}, and assume
  Assumption~\ref{assump:cV-compact-containment}. Then, the following holds: 
  \begin{theoremenum}
  \item\label{item:main-convergence-1} There exists a unique affine process
    $(X,(\MP_{x})_{x\in\cHplus})$ on $\cHplus$, with paths in
    $D(\MRplus,\cHplus)$ and where $\MP_{x}$ denotes the law of $X$ given
    $X_{0}=x$, such that for every $x\in\cHplus$ we have  
    \begin{align}\label{eq:affine-formula-main}
      \EXspec{\MP_{x}}{\E^{-\langle X_{t}, u\rangle}}=\E^{-\phi(t,u)-\langle
      x,\psi(t,u)\rangle},\quad t\geq 0,\, u\in\cHplus,  
    \end{align}
    for $(\phi(\cdot,u),\psi(\cdot,u))$ the unique mild solution
    of~\eqref{eq:Riccati-phi-mild}-\eqref{eq:Riccati-psi-mild}. Furthermore, the
    following trace-class refinement holds under additional hypotheses on the
    parameter set:
    \begin{itemize}
    \item[(a)] If
      $x\in\cL_{1}(H)\cap\cHplus$, $m(\D\xi)$ and $\mu(\D\xi)$ are supported on
      $\cL_{1}(H)\setminus\set{0}$, $\Gamma(\cL_{1}(H))\subseteq\cL_{1}(H)$,
      $b,\mu(A)\in\cL_{1}(H)$ for all
      $A\in\cB(\cHplus\setminus\set{0})$, and the trace first-moment bounds
      $$
        \int_{\cHpluso}\Tr(\xi)\,m(\D\xi)<\infty,\qquad
        \int_{\cHpluso}\Tr(\xi)\,
          \frac{\langle y,\mu(\D\xi)\rangle}{\|\xi\|^2}
        \le C_\mu \Tr(y),
      $$
      hold for some $C_\mu<\infty$ and all $y\in\cL_1(H)\cap\cHplus$, then for
      every fixed $t\geq 0$,
      $X_{t}\in\cHplus\cap\cL_{1}(H)$ $\MP_x$-a.s.
    \end{itemize}
  \item\label{item:main-convergence-2} Let $(X^{d})_{d\in\MN}$
    be the sequence of finite-rank operator-valued affine processes in
    Proposition~\ref{prop:embedding-affine-main}. Then
    the sequence $(X^{d})_{d\in\MN}$ converges weakly to $X$ on $D(\MRplus,\cHplus)$
    equipped with the Skorohod topology, i.e. it holds
    \begin{align*}
      \EXspec{\MP_{x}^{d}}{f(X^{d})}\to \EXspec{\MP_{x}}{f(X)},\quad\text{as
      }d\to\infty, 
    \end{align*}
    for all bounded continuous $f\colon D(\MRplus,\cHplus)\to\MR$.
  \end{theoremenum}
\end{theorem}

\begin{proposition}\label{prop:cV-refinement}
  Assume the hypotheses of Theorem~\ref{thm:main-convergence}. Then for every
  initial datum $x\in\cHplus$ and every fixed $t>0$, the limiting affine
  process satisfies
  $X_t\in\cV\cap\cHplus$ $\MP_x$-a.s. If moreover $x\in\cV\cap\cHplus$, then
  $X_t\in\cV\cap\cHplus$ $\MP_x$-a.s.\ for every fixed $t\geq 0$.
\end{proposition}
\begin{proof}
  Fix $t>0$ and choose $\delta\in(0,t]$. By
  Lemma~\ref{lem:irregular-smoothing-cV},
  $\sup_d\E_{\MP_x^d}[\|X_t^d\|_{\cV}^{2}]<\infty$. The embedding
  $\cV\hookrightarrow\cH$ is continuous and compact; hence
  $\|\cdot\|_{\cV}$ is lower semicontinuous as an extended function on
  $\cH$. By Theorem~\ref{thm:main-convergence}(ii) and
  Lemma~\ref{lem:no-fixed-time-jumps}, the continuous-mapping theorem gives
  $X_t^d\Rightarrow X_t$ for every deterministic $t\ge0$. Portmanteau's theorem yields
  $$
    \mathbb{E}_{\MP_x}[\|X_t\|_{\cV}^{2}]
    \le \liminf_{d\to\infty}\mathbb{E}_{\MP_x^d}[\|X_t^d\|_{\cV}^{2}]<\infty.
  $$
  Thus $X_t\in\cV\cap\cHplus$ $\MP_x$-a.s. If $x\in\cV\cap\cHplus$, the same
  estimate is available with $\delta=0$, because
  $\sup_d\|\bP_d x\|_{\cV}\le \|x\|_{\cV}$, giving the assertion at $t=0$ as
  well.
\end{proof}

\begin{remark}\label{rem:regularity}
  The affine processes $(X, (\MP_x)_{x \in \cHplus})$ on $\cHplus$ from
  Theorem~\ref{thm:main-convergence} are \emph{irregular}, in the sense that the derivative
  $\frac{\partial}{\partial t}\psi(t,u)\lvert_{t=0}$, cf.~\cite[Definition 2.2]{CFMT11}, is not continuous at
  $u=0$, due to the unboundedness of $L$ on $\cH$. For general regularity
  and path properties of affine processes on (finite-dimensional) state spaces,
  including the role of such regularity assumptions, we refer to~\cite{CT13}. The coercivity assumption is the
  ingredient that allows the Lyapunov estimates underlying the
  $\cV$-smoothing of Proposition~\ref{prop:cV-refinement} above. If in part (a) of
  \cref{item:linear-operator-unbounded} the condition $B\in\cL(H)$ holds, then the
admissible parameter set is also admissible in the regular sense of \cite[Definition 2.3]{CKK22a}
and existence would follow from~\cite{CKK22a}, but Assumption~\ref{assump:integration} would not be satisfied in general.

\end{remark}

\section{Finite-Rank Approximations of Heat-Modulated Affine SV
  Models}\label{sec:the-affine-model}

In this section, we examine the joint stochastic covariance model
$(f_{t}, X_{t})_{t \geq 0}$, where $(f_{t})_{t \geq 0}$ represents the
mild solution to the HJMM equation~\eqref{eq:HJMM-SDE} modulated by a stochastic
volatility process $(\sigma_{t})_{t \geq 0}$ satisfying
$X_t=\sigma_t\sigma_t^*$ in the $H$-valued formulation. The covariance process
$(X_t)_{t\ge0}$ is an irregular affine process on $\cHplus$, associated with the
admissible parameter set $(b,\mathbf B,m,\mu)$, whose existence is guaranteed by
Theorem~\ref{thm:main-convergence}. The main result of this section is that the
joint model $(f_t,X_t)_{t\ge0}$ admits an affine transform formula once the
$H$-valued stochastic convolution is well defined (the accompanying
Feynman--Kac martingale identity for the affine covariance process is provided
by Proposition~\ref{prop:affine-fk}), see
Theorem~\ref{thm:heat-affine-model} below, and that the associated affine
transform is stable under finite-rank approximations of the covariance process
via a spectral Galerkin approximation of the generalized mild Riccati equations
(see Theorem~\ref{thm:finite-dim-approx-ScoV} below). We conclude this section
with an example of heat-modulated affine stochastic volatility models, as well
as an application of the theory to the pricing of options written on
\emph{power flow forwards}.

	Throughout this section it is useful to keep in mind that the
	results below are organized as three independent layers, each of which adds an assumption to the previous one. Layer~(i) is the \emph{driftless $H$-valued
affine transform} of Theorem~\ref{thm:heat-affine-model}: it gives the joint
Fourier--Laplace transform of $(f_t,X_t)$ at imaginary $u_1\in\I H$, under
Assumption~\ref{assump:H-valued-joint} (stochastic integrability of the
square-root volatility); the accompanying limiting Feynman--Kac martingale
identity for the affine covariance process is \emph{not} an extra hypothesis but
is established in Proposition~\ref{prop:affine-fk}. Layer~(ii)
is the \emph{HJM no-arbitrage drift correction} for the geometric
specification: this is not part of the driftless transform theorem and has to
be imposed separately when forward prices are required to be martingales (see
the discussion before Theorem~\ref{thm:heat-affine-model}). Layer~(iii) is the
\emph{analytic-continuation/Carr--Madan pricing layer} used in
Proposition~\ref{prop:robustness}: this is invoked through the
complex-transform result of~\cite[Theorem~2.3]{HeKarbachKhedher2025} and is
subject to its own admissibility hypotheses. The reader interested only in
	finite-rank approximation rates can stop at layer~(i) plus
	Theorem~\ref{thm:finite-dim-approx-ScoV}; the reader interested in
	geometric-payoff option pricing needs layers~(ii) and~(iii) in addition.

 \subsection{The Heat-Modulated Affine Stochastic Covariance
   Model}\label{sec:abstr-affine-stoch}

 As before, we consider the Hilbert space triple
 $V\hookrightarrow H \hookrightarrow V^{*}$ compact, with $V$ the state space of regular
 forward curves. This triple of forward curve spaces induces, on the
 symmetric covariance subspace $\cV\cap\cH\subseteq\cV$, the triple of
 Hilbert--Schmidt operators
 $\cV\cap\cH\hookrightarrow\cH\hookrightarrow(\cV\cap\cH)^{*}$. Throughout
 this section we work consistently with $\cV\cap\cH$ on the covariance side; the
 symbol $\cV$ alone retains its meaning from Section~\ref{sec:admissible-irregular-affine-processes}
 (full intersection $\cL_2(V^*,H)\cap\cL_2(H,V)$, no symmetry built in).
 Let $(b,\mathbf{B},m,\mu)$ be an admissible parameter set as in
 Definition~\ref{def:admissible-irregular} satisfying
 Assumption~\ref{assump:cV-compact-containment}, and denote by
 $(X_{t})_{t\geq 0}$ the associated irregular affine process on $\cHplus$ given by
 Theorem~\ref{thm:main-convergence}. Moreover, denote the filtered
 probability space of $X$ by
 $(\Omega^{(1)},\cF^{(1)},(\mathcal{F}_t^{(1)})_{t\geq 0}, \mathbb{P}^{(1)})$
 and fix an initial value $\mathbb{P}^{(1)}(X_0=x)=1$ for $x\in\cHplus$. We let
 $(\Omega^{(2)},\cF^{(2)}, (\mathcal{F}_t^{(2)})_{t\geq 0}, \mathbb{P}^{(2)})$
 be another filtered probability space satisfying the usual conditions and that
 carries a cylindrical Brownian motion $W\colon [0,\infty)\times H \rightarrow
 L^{2}(\Omega^{(2)},\cF^{(2)},\mathbb{P}^{(2)})$, in the following denoted by $(W_{t})_{t\geq
   0}$. We assume that
 $\mathbb{P}^{(2)}(f_{0}=F(0,\cdot))=1$ for some initial forward curve
 $F(0,\cdot)\in H$ and define the following joint filtered probability space:
\begin{align*}
(\Omega, \cF, \MF, \MP)\df(\Omega^{(1)}\times\Omega^{(2)}, (\cF^{(1)}\otimes \cF^{(2)}),
(\cF^{(1)}_{t}\otimes \cF_{t}^{(2)})_{t\geq 0}, \MP^{(1)}\otimes \MP^{(2)}).
\end{align*}
We denote the expectation with respect to $\MP$ by $\mathbb{E}$ and from now on
consider $(X_{t})_{t\geq 0}$ and the cylindrical Brownian motion $(W_{t})_{t\geq
0}$ to be defined on $(\Omega,\cF,\MF,\MP)$.\par

Throughout this subsection we work with the driftless specification
$g_t\equiv0$ in~\eqref{eq:HJMM-SDE}. The covariance process remains the affine
process from Theorem~\ref{thm:main-convergence}; the additional assumption
below only concerns the well-posedness of the $H$-valued stochastic
convolution.

For the $H$-valued model we fix, once and for all, a strongly
continuous semigroup $(S(t))_{t\ge0}$ on $H$ with generator $\cA$. On the
infinite-horizon space this is the usual Musiela shift on the maturity
component. On the finite-horizon space
$H=H_0\oplus\MR$, $H_0=L^2(0,\Theta_{\max},w(x)\D x)$, used below, we take the
level-preserving stopped shift
\[
  S(t)(u,r)=(u_t,r),
  \qquad
  u_t(x)=u(x+t)\mathbf 1_{\{x+t<\Theta_{\max}\}},
  \qquad (u,r)\in H_0\oplus\MR .
\]
Equivalently, the full curve value is shifted on the available maturity window
and set equal to the level component once the shifted maturity leaves the
window. In financial terms, a forward maturity that drifts past the
horizon $\Theta_{\max}$ is set equal to the finite-horizon level component~$r$;
on the finite interval this level is the chosen scalar anchor of the
decomposition, not a long-end limit. This convention is well defined on
$L^2$-equivalence classes and avoids using the endpoint value
$u(\Theta_{\max})$, which is not defined in the pivot space. Its Hilbert-space adjoint semigroup is denoted by $(S^*(t))_{t\ge0}$.
Both $(S(t))_{t\ge0}$ and $(S^*(t))_{t\ge0}$ are strongly continuous
contractions on $H=L^2(0,\Theta_{\max},w(x)\D x)\oplus\MR$ (the multiplier
$\mathbf 1_{\{x+t<\Theta_{\max}\}}$ is bounded by $1$ in absolute value, so
$\|S(t)\|_{\cL(H)}\le 1$ and $\|S^*(t)\|_{\cL(H)}\le 1$ by duality). On the
regular space $V=V_{0,\beta}\oplus\MR$ the adjoint semigroup additionally
preserves the boundary condition $u(0)=0$ on the shape component (the
truncated shift extends the source curve by the level value, which does not
change the value at $0$ of the shifted curve), and satisfies the growth bound
$\|S^*(s)\|_{\cL(V)}\le M_1\E^{\omega s}$ used in
Theorem~\ref{thm:finite-dim-approx-ScoV}(ii); both properties are established in
Lemma~\ref{lem:adjoint-shift-V} below, in fact with $M_1=1$ and $\omega=0$.

\begin{lemma}[Regularity of the adjoint stopped shift on the regular space]
\label{lem:adjoint-shift-V}
Assume $0<\Theta_{\max}<\infty$, $H=H_0\oplus\MR$ with
$H_0=L^2(0,\Theta_{\max},\E^{\beta x}\D x)$, $V=V_{0,\beta}\oplus\MR$ and let the 
$H$-adjoint of the level-preserving stopped shift $S(t)(u,r)=(u_t,r)$,
$u_t(x)=u(x+t)\mathbf 1_{\{x+t<\Theta_{\max}\}}$ be given by
$$
  S^*(t)(v,\rho)
  =\bigl(\,\E^{-\beta t}\,v(\cdot-t)\,\mathbf 1_{\{\,\cdot\,>t\}}\,,\ \rho\,\bigr),
  \qquad (v,\rho)\in H_0\oplus\MR .
$$
Then $(S^*(t))_{t\geq 0}$ leaves $V$ invariant, i.e. $S^*(t)(V)\subseteq V$ for all $t\geq 0$, and is a contraction on $V$:
$$
  \|S^*(t)(v,\rho)\|_V^2
  \le \E^{-\beta t}\|v\|_{V_{0,\beta}}^2+|\rho|^2
  \le \|(v,\rho)\|_V^2,\qquad t\ge0 .
$$
In particular the growth bound $\|S^*(s)\|_{\cL(V)}\le M_1\E^{\omega s}$ used in
Theorem~\ref{thm:finite-dim-approx-ScoV}(ii) holds with $M_1=1$ and $\omega=0$,
and $(S^*(t))_{t\ge0}$ is a strongly continuous contraction semigroup on $V$.
\end{lemma}

\begin{proof}
For $u,v\in H_0$ the substitution $y=x+t$ gives
\[
  \langle S(t)u,v\rangle_{H_0}
  =\int_0^{\Theta_{\max}}u(x+t)\mathbf 1_{\{x+t<\Theta_{\max}\}}\,v(x)\,\E^{\beta x}\D x
  =\int_t^{\Theta_{\max}}u(y)\,v(y-t)\,\E^{\beta(y-t)}\D y ,
\]
so $(S^*(t)v)(y)=\E^{-\beta t}v(y-t)\mathbf 1_{\{y>t\}}$; the level component is
unchanged because $S(t)$ acts as the identity on $\MR$. Let $v\in V_{0,\beta}$
and $w=S^*(t)v$. Then $w\equiv0$ on $(0,t)$, while on $(t,\Theta_{\max})$ one has
$w=\E^{-\beta t}v(\cdot-t)\in H^1$; at the front $y=t$ the right limit is
$\E^{-\beta t}v(0)=0$ by the anchoring $v(0)=0$, so $w$ is continuous across $t$
and therefore $w\in H^1(0,\Theta_{\max},\E^{\beta x}\D x)$ with $w(0)=0$, i.e.\
$w\in V_{0,\beta}$. This is precisely where the short-end condition built into
$V_{0,\beta}$ is used: it makes the zero-extension of the shifted curve across
the front an $H^1$ function. Differentiating,
$w'(y)=\E^{-\beta t}v'(y-t)\mathbf 1_{\{y>t\}}$ a.e., and with $z=y-t$,
\[
  \int_0^{\Theta_{\max}}\E^{\beta y}|w'(y)|^2\D y
  =\E^{-2\beta t}\int_t^{\Theta_{\max}}\E^{\beta y}|v'(y-t)|^2\D y
  =\E^{-\beta t}\int_0^{\Theta_{\max}-t}\E^{\beta z}|v'(z)|^2\D z .
\]
The identical computation with $v'$ replaced by $v$ gives
$\int\E^{\beta y}|w|^2\D y=\E^{-\beta t}\int_0^{\Theta_{\max}-t}\E^{\beta z}|v|^2\D z$.
Hence, for either the weighted $H^1$ norm or its seminorm on $V_{0,\beta}$,
$\|w\|_{V_{0,\beta}}^2\le\E^{-\beta t}\|v\|_{V_{0,\beta}}^2$, so the contraction
is independent of the choice of equivalent norm. Adding the isometric level
component yields the displayed bound. Strong continuity on $V$ follows from
strong continuity of translation on $H^1(0,\Theta_{\max},\E^{\beta x}\D x)$ and
$\E^{-\beta t}\to1$, together with the uniform bound $\|S^*(t)\|_{\cL(V)}\le1$
and density of smooth functions vanishing near the endpoints.
\end{proof}

\subsubsection{Joint Affine Transform for the Driftless Model}

We define a new set of generalized mild Riccati equations associated to the joint
process $(f_{t},X_{t})_{t\geq 0}$. For this, define $\hat{\cR}\colon \I H \times
\cHplus \to \cH$ as  
\begin{align}
 \hat{\cR}(h,u)&= \Gamma^{*}(u)-\tfrac{1}{2}h\otimes
    h-\int_{\cHpluso}\big(\E^{-\langle
    \xi,u\rangle}-1+\langle \chi(\xi), u\rangle\big)\frac{\mu(\D \xi)}{\norm{\xi}^{2}},  \label{eq:R-intro}
\end{align}
where for two elements $x$ and $y$ in $H$ we define the operator
$x\otimes y\in \cL(H)$ by $x\otimes y(h)=\langle x,h \rangle_{H} y$ for every
$h\in H$ and write $x^{\otimes 2}\df x\otimes x$.
For complex arguments $v=\I h\in\I H$ with $h\in H$ we extend the
notation by complex bilinearity and identify the resulting object with the
real symmetric Hilbert--Schmidt operator
$$v^{\otimes 2}\df (\I h)\otimes(\I h)\equiv -h^{\otimes 2}\in\cH,$$
i.e.\ a real element of $\cH$ with the opposite sign of $h^{\otimes 2}$. Under
this convention $-\tfrac12 v^{\otimes 2}=\tfrac12 h^{\otimes 2}\in\cHplus$,
which is the sign needed for the forcing term $-\tfrac12\psi_1\otimes\psi_1$
in~\eqref{eq:R-intro} to be order-preserving on $\cHplus$ when
$\psi_1\in\I H$.\par{}

Let $T>0$, $u_{1}\in \I H$ and $u_{2}\in\cHplus$. We define the new system of
\textit{generalized joint mild Riccati equations} as follows:
\begin{subequations}
\begin{empheq}[left=\empheqlbrace]{align}
   \,\frac{\partial}{\partial t}\Phi(t,u)&=F(\psi_{2}(t,u)),  &\quad\Phi(0,u)=0,\label{eq:Riccati-phi-psi-1-1}\\
    \,\psi_{1}(t,u)&=u_{1}+\cA^{*}\!\left(\int_{0}^{t}\psi_{1}(s,u)\D
                   s\right),  &\quad\psi_{1}(0,u)=u_{1}, \label{eq:Riccati-phi-psi-1-2}\\
    \,\frac{\partial }{\partial t}\psi_{2}(t,u)&=L(\psi_{2}(t,u))+\hat{\cR}(\psi_{1}(t,u),
    \psi_{2}(t,u)),  & \quad
                                                  \psi_{2}(0,u)=u_{2}.\label{eq:Riccati-phi-psi-1-3}
       \end{empheq}
\end{subequations}
The differential displays
\eqref{eq:Riccati-phi-psi-1-1} and~\eqref{eq:Riccati-phi-psi-1-3} are formal:
$L$ is unbounded on $\cH$, so the right-hand side
of~\eqref{eq:Riccati-phi-psi-1-3} is not defined for arbitrary
$\psi_2(t,u)\in\cHplus$. The actual solution concept is the
\emph{variation-of-constants} formulation
$$
  \psi_{1}(t,u)=S^*(t)u_1,\qquad
  \psi_{2}(t,u)=\cT(t)u_2+\int_0^t\cT(t-s)\hat{\cR}(\psi_1(s,u),\psi_{2}(s,u))\,\D s,
$$
together with $\Phi(t,u)=\int_0^t F(\psi_2(s,u))\,\D s$. We use the term
``mild solution'' below to refer to a triple
$(\Phi(\cdot,u),\psi_1(\cdot,u),\psi_2(\cdot,u))$ that satisfies these
variation-of-constants identities; the differential displays are retained
only for notational convenience.

For every initial value $u=(u_{1},u_{2})\in \I H\times \cHplus$, set
$$
  \Psi(\cdot,u)\df(\psi_{1}(\cdot,u),\psi_{2}(\cdot,u)).
$$
We say that $(\Phi(\cdot, u),\Psi(\cdot, u))$ is a \emph{mild solution} from
$[0,T]$ to $\R\times \I H\times\cH$
to~\eqref{eq:Riccati-phi-psi-1-1}-\eqref{eq:Riccati-phi-psi-1-3} if
$\Phi(\cdot,u)\in C([0,T];\MR)$, $\psi_{1}(\cdot,u)\in C([0,T];\I H)$,
$\psi_{2}(\cdot,u)\in C([0,T];\cHplus)$ and the map
$(\Phi(\cdot,u),\Psi(\cdot,u))$ satisfies the
variation-of-constants identities just displayed.
($C^1$ regularity for $\Phi$ is recovered only once continuity of
$s\mapsto F(\psi_2(s,u))$ is established.)

Here $\cA^{*}\colon\dom(\cA^{*})\subseteq H\to H$ denotes the
$H$-adjoint of the generator $\cA$ of the $H$-level shift semigroup fixed
above; its Hilbert-space adjoint semigroup is $(S^{*}(t))_{t\geq 0}$.
Before invoking the joint Riccati equations in the statement of
Theorem~\ref{thm:heat-affine-model} below, we first establish their
well-posedness.

\begin{lemma}[Well-posedness of the joint mild
Riccati equations]\label{lem:joint-Riccati-wellposed}
  Let $(b,\mathbf{B},m,\mu)$ be an admissible parameter set in the
  sense of Definition~\ref{def:admissible-irregular} and let $T>0$. For every
  initial datum $u=(u_{1},u_{2})\in \I H\times \cHplus$ there exists a unique
  mild solution
  $$
    \bigl(\Phi(\cdot,u),\psi_{1}(\cdot,u),\psi_{2}(\cdot,u)\bigr)
    \in C^{1}([0,T];\MR)\times C([0,T];\I H)\times C([0,T];\cHplus)
  $$
  of~\eqref{eq:Riccati-phi-psi-1-1}-\eqref{eq:Riccati-phi-psi-1-3}, which is
  given by
  \begin{align}
    \psi_{1}(t,u) &= S^{*}(t)u_{1},\label{eq:psi1-explicit}\\
    \psi_{2}(t,u) &= \cT(t)u_{2}
      + \int_{0}^{t}\cT(t-s)\,\hat\cR\bigl(\psi_{1}(s,u),\psi_{2}(s,u)\bigr)\D
      s,\label{eq:psi2-explicit}\\
    \Phi(t,u) &= \int_{0}^{t}F\bigl(\psi_{2}(s,u)\bigr)\D s.
  \end{align}
\end{lemma}

\begin{proof}
  Equation~\eqref{eq:Riccati-phi-psi-1-2} is the linear
  integral form
  $\psi_{1}(t,u)=u_{1}+\cA^{*}\!\bigl(\int_{0}^{t}\psi_{1}(s,u)\D s\bigr)$,
  i.e.\ the mild form of $\partial_{t}\psi_{1}=\cA^{*}\psi_{1}$ with
  $\psi_{1}(0,u)=u_{1}$. By the Hille--Yosida theorem (see
  e.g.~\cite[Chapter 1]{Paz83}), $\cA^{*}$ generates the strongly continuous
  semigroup $(S^{*}(t))_{t\geq 0}$ on $H$, so
  $\psi_{1}(t,u)=S^{*}(t)u_{1}$ is the unique continuous mild solution on
  $[0,T]$; the map $t\mapsto S^{*}(t)u_{1}$ takes values in $\I H$ whenever
  $u_{1}\in\I H$ because $S^{*}(t)$ preserves the real subspace $H$.

  Given $\psi_{1}(\cdot,u)\in C([0,T];\I H)$ from the first step, define  $\hat\cR_{\psi_{1}}(t,u)\df
  \hat\cR(\psi_{1}(t,u),u)=\hat R(u)-\tfrac12\psi_{1}(t,u)^{\otimes 2}$, where
  $\hat R$ is defined in~\eqref{eq:R-hat}. By~\cite[Remark 3.4]{CKK22a} the
  map $u\mapsto \hat R(u)$ is locally Lipschitz on $\cHplus$, and
  $t\mapsto\psi_{1}(t,u)^{\otimes 2}$ is continuous from $[0,T]$ to $\cH$;
  consequently $\hat\cR_{\psi_{1}}$ is locally Lipschitz in $u$ uniformly on
  $[0,T]$ and continuous in $t$. Arguing exactly as in the proof
  of Proposition~\ref{prop:existence-mild-Riccati}
  (quasi-monotonicity of $\hat\cR_{\psi_{1}}$ with respect to $\cHplus$ is
  inherited from $\hat R$ because the additive forcing is independent of the
  $\psi_2$ variable and therefore cancels in differences. Separately, for
  $v=\psi_{1}(t,u)\in\I H$, writing $v=\I h$ with $h\in H$ gives
  $v^{\otimes2}=-h^{\otimes2}$ and hence
  $-\tfrac12v^{\otimes2}=\tfrac12h^{\otimes2}\in\cHplus$, which is the
  subtangential cone-invariance condition; see Lemma~\ref{lem:quasi-cont}), the
  time-dependent version of the local existence theorem of~\cite[Chapter
  8, Theorem 2.1 \& Remark 2.1]{Mar76} yields a unique continuous mild solution
  $\psi_{2}(\cdot,u)\in C([0,T];\cHplus)$ satisfying
  \eqref{eq:psi2-explicit}, which is globally defined on $[0,T]$ by the same
  moment bound used in the proof of
  Proposition~\ref{prop:existence-mild-Riccati}.

  With $\psi_{2}(\cdot,u)$ in hand, the scalar component $\Phi(\cdot,u)$ is
  obtained by direct integration:
  \[
    \Phi(t,u)=\int_{0}^{t}F(\psi_{2}(s,u))\,\D s.
  \]
  Since $F$ is continuous on $\cHplus$ and $\psi_{2}(\cdot,u)$ is continuous,
  we have $\Phi(\cdot,u)\in C^{1}([0,T];\MR)$. Uniqueness of the triple
  follows from the uniqueness established at each of the three steps.
\end{proof}

\begin{assumption}{B}\label{assump:H-valued-joint}
  The square-root volatility process $\sigma_t\df X_t^{1/2}$ is progressively
  measurable with respect to $\MF$ and satisfies the $H$-valued analogue of
  Assumption~\ref{assump:integration}, namely
  \begin{align}\label{eq:integral-condition-H}
    \mathbb{E}\left[
      \int_0^t \norm{S(t-s)\sigma_s}_{\mathcal{L}_2(H,H)}^2 \D s
    \right]
    <
    \infty,
    \qquad t\geq 0.
  \end{align}
  In this case we define the associated $H$-valued heat-modulated stochastic
  covariance model by
  \begin{align}\label{eq:HJMM-SDE-H}
    f_t
    =
    S(t)f_0
    +
    \int_0^t S(t-s)\sigma_s\,\D W_s,
    \qquad t\geq 0.
  \end{align}
  
\end{assumption}

\begin{remark}[Trace-class sufficient condition for Assumption~\ref{assump:H-valued-joint}]
\label{rem:H-valued-trace-sufficient}
If the shift semigroup $(S(t))_{t\ge0}$ is bounded on $H$ on compact time
intervals and
\[
  \int_0^t\mathbb E[\Tr(X_s)]\,\D s<\infty,
  \qquad t\ge0,
\]
then Assumption~\ref{assump:H-valued-joint} holds for the square-root
volatility $\sigma_s=X_s^{1/2}$. Indeed,
\[
  \|S(t-s)X_s^{1/2}\|_{\mathcal L_2(H,H)}^2
  \le
  \|S(t-s)\|_{\mathcal L(H)}^2\,\Tr(X_s).
\]
Thus the trace-class refinement of Theorem~\ref{thm:main-convergence}, together
with the displayed integrated trace moment, gives a direct sufficient condition
for the $H$-valued stochastic convolution.
\end{remark}

\begin{proposition}[Affine Feynman--Kac martingale]\label{prop:affine-fk}
  Let $(b,\mathbf{B},m,\mu)$ be an admissible parameter set and let $X$ be the
  associated affine process on $\cHplus$ from
  Theorem~\ref{thm:main-convergence}; in particular assume
  Assumption~\ref{assump:cV-compact-containment}. For every $t>0$ and every
  $u=(u_1,u_2)\in\I H\times\cHplus$, with $(\Phi,\psi_1,\psi_2)$ the mild
  Riccati solution from Lemma~\ref{lem:joint-Riccati-wellposed}, the process
  \[
    M_s
    =
    \exp\!\left(
      \frac12\int_0^s
      \langle X_r,\psi_1(t-r,u)^{\otimes2}\rangle\,\D r
    \right)
    \exp\!\left(
      -\Phi(t-s,u)-\langle X_s,\psi_2(t-s,u)\rangle
    \right),
    \qquad 0\le s\le t,
  \]
  is a true martingale under $\MP_x$ for every $x\in\cHplus$.
\end{proposition}

This is the infinite-dimensional Feynman--Kac martingale identity underlying the
affine transform formula. It is
proved below for every admissible parameter set satisfying
Assumption~\ref{assump:cV-compact-containment}, and no trace-class
hypothesis is required for this. The proof is the joint analogue of
Proposition~\ref{prop:irregular-weak-limit-mp}: the identity is verified for the
finite-rank Galerkin processes, where the generator is bounded and the classical
affine Feynman--Kac calculation applies, and then passed to the limit. It is
given in Section~\ref{sec:proof-section-2}, in Lemma~\ref{lem:fk-finite-rank} and
the proof of Proposition~\ref{prop:affine-fk} immediately thereafter. The
restriction $u_1\in\I H$ is essential: it makes the Feynman--Kac potential
$\tfrac12\langle\cdot,\psi_1(t-r,u)^{\otimes2}\rangle$ non-positive on the cone,
which is what gives the uniform bound used to pass from a local to a true
martingale; for real $u_1$ the sign is lost and the unconditional statement is
not claimed.

\begin{theorem}\label{thm:heat-affine-model}
  Let $(X_t)_{t \geq 0}$ be the affine process associated with $(b,
  \mathbf{B}, m, \mu)$ as in Theorem~\ref{thm:main-convergence}; in particular
  assume Assumption~\ref{assump:cV-compact-containment}. Assume also
  Assumption~\ref{assump:H-valued-joint}, and let
  $(f_t)_{t\geq 0}$ be the \emph{driftless} $H$-valued stochastic covariance
	  model~\eqref{eq:HJMM-SDE-H} (i.e.\ no HJM drift correction is included; the
	  forward curve evolves purely under the shift semigroup and the stochastic
	  convolution of $\sigma_s\D W_s$). Furthermore, let
	  \[
	    u = (u_1, u_2) \in \I H \times \cHplus
	  \]
	  and let $(\Phi(\cdot, u),(\psi_1(\cdot, u),\psi_2(\cdot, u)))$ be the
	  unique mild solution from Lemma~\ref{lem:joint-Riccati-wellposed} of the
	  generalized joint mild Riccati equations
	  \eqref{eq:Riccati-phi-psi-1-1}--\eqref{eq:Riccati-phi-psi-1-3}. Then, for all
  $t \geq 0$, the following holds: 
  \begin{align}
    \label{eq:extended-affine-formula}
    \mathbb{E} \left[\E^{\langle f_t, u_1 \rangle_H - \langle X_{t}, u_2
    \rangle} \right] = \E^{-\Phi(t, u) + \langle f_0, \psi_1(t, u) \rangle_H -
    \langle x, \psi_2(t, u) \rangle}.
  \end{align}
\end{theorem}

\begin{remark}[Absence of leverage]\label{rem:no-leverage}
  The covariance process $X$ and the cylindrical noise $W$ are constructed on the
  two independent factors $\Omega^{(1)}$ and $\Omega^{(2)}$ of the product
  space; this independence is used in the conditional-Gaussian step
  \eqref{eq:conditional-gaussian-heat} of the proof. Consequently the model has
  \emph{no leverage}: the forward curve and its instantaneous covariance are
  uncorrelated. This is in contrast to the leverage extension of the
  Hilbert-space Barndorff--Nielsen--Shephard model in~\cite{BS24}. Introducing a
  genuine price--volatility correlation (e.g.\ a common driving noise as
  in~\cite{BS24,CKK22b}) would break the conditional Gaussianity and require a
  different argument; we do not pursue it here. For the same reason we describe
  the model as capturing maturity-specific risk rather than as reproducing
  leverage-driven volatility clustering.
\end{remark}

\begin{remark}
  Theorem~\ref{thm:heat-affine-model} is formulated at the natural
  $H\times\cHplus$ level under the explicit stochastic-integrability condition
  from Assumption~\ref{assump:H-valued-joint}. If, in addition, $f_{0}\in V$
  and $X_{t}^{1/2}\in\cV$ for all $t\geq 0$, $\mathbb{P}$-a.s.\ (see
  Remark~\ref{rem:square-root-regularity}(iv) for a sufficient condition), then
  the same model also satisfies Assumption~\ref{assump:integration} on
  $V\times\cV$.
\end{remark}

\begin{example}[The heat-modulated BNS stochastic volatility model]  
As an example of a heat-modulated affine stochastic volatility model, we introduce
the \emph{heat--modulated Barndorff--Nielsen and Shephard (BNS) stochastic
volatility model}. This is a heat-modulated version of the Hilbert
space-valued BNS model presented in~\cite{BRS18,BS24}. In the heat-modulated BNS stochastic volatility model we assume that the forward
curve process $(f_t)_{t \geq 0}$ is given by the driftless $H$-valued
specification~\eqref{eq:HJMM-SDE-H} on the state space 
$$
H = L^2(0,\Theta_{\mathrm{max}},\mathrm e^{\beta x}\mathrm dx)\oplus\mathbb R,
$$
for some $\Theta_{\mathrm{max}}>0$ and $\beta>0$.
Here we deliberately specialize the general forward-curve setting from
Section~\ref{sec:HJMM-equation} to the stronger weighted $L^2$ pivot space
$L^2(0,\Theta_{\mathrm{max}},\mathrm e^{\beta x}\mathrm dx)\oplus\MR$.
This choice is tailored to the diagonal BNS benchmark and the pricing
functional used below; the general existence theory above remains formulated on
the original pair
$V=H_\beta(0,\Theta_{\max})\hookrightarrow L^2((0,\Theta_{\max}),\mathrm e^{\gamma x}\D x)\oplus\MR\hookrightarrow V^*$
with $0<\gamma<\beta$ (i.e.\ the pivot space here is the weighted $L^2$ space,
not the Filipovic-type Sobolev space $H_\gamma$).
The dynamics are modulated by an
instantaneous volatility process $(\sigma_t)_{t\ge0}$, defined as the operator
square root of an Ornstein--Uhlenbeck process $(X_t)_{t\ge0}$ on $\cHplus$,
satisfying
\begin{align}
\begin{cases}
\mathrm dX_t = L_\Delta(X_t)\,\mathrm dt + \mathrm dL_t,\\
X_0 = x\in\cHplus,
\end{cases}
\end{align}
where $(L_t)_{t\ge0}$ is a $\cHplus$-valued Lévy subordinator and
$L_\Delta$ denotes the Lyapunov operator associated with the Laplacian, cf.
Example~\ref{ex:Laplacian_Forward}.
	We denote by $b_{\mathrm{phys}}\in\cHplus$ the physical drift of the
subordinator and by $m$ its Lévy measure on $\cHplus\setminus\{0\}$. Since the
affine admissibility convention in Definition~\ref{def:admissible-irregular}
uses the compensated truncation $\chi(\xi)=\xi\one_{\{\|\xi\|\le1\}}$, the
canonical affine drift parameter is
\[
  b_{\mathrm{adm}}
  \df
  b_{\mathrm{phys}}+I_m,
  \qquad
  I_m\df\int_{\cHplus\setminus\{0\}}\chi(\xi)\,m(\D\xi).
\]
Thus the Laplace exponent of $L$ can be written equivalently as
\begin{align}\label{eq:Laplace-Levysubordinator}
\varphi_L(u)
&= \langle b_{\mathrm{phys}},u\rangle
 + \int_{\cHplus\setminus\{0\}}
 \big(1-\mathrm e^{-\langle\xi,u\rangle}\big)\,m(\mathrm d\xi)
 \nonumber\\
&= \langle b_{\mathrm{adm}},u\rangle
- \int_{\cHplus\setminus\{0\}}
\big(\mathrm e^{-\langle\xi,u\rangle}-1+\langle\chi(\xi),u\rangle\big)
\,m(\mathrm d\xi),
\qquad u\in\cHplus.
\end{align}
	The first line is the physical subordinator exponent; the second is the
	canonical affine coefficient $F$ associated with the admissible drift
	$b_{\mathrm{adm}}$, see~\cite{Det77}.
	The semigroup $(\cT(t))_{t\geq 0}$ generated by $L_{\Delta}$ satisfies
$$
\mathcal T(t)u = T(t)uT^{*}(t),
$$
where $(T(t))_{t\geq 0}=(\E^{tB})_{t\geq 0}$ is the full semigroup on $H$ generated by $B=\Delta\oplus 0$ (cf.\ Example~\ref{ex:Laplacian_Forward}), and we note that $\cT(t)$ maps $\cHplus$ into itself for all $t\ge0$.
Moreover, the process $(X_t)_{t\ge0}$ is a stochastically continuous affine Markov
process with values in $\cHplus$, associated with the admissible parameter set
$(b_{\mathrm{adm}},L_\Delta,m,0)$. This is the parameter set used by the
affine theory; the physical drift $b_{\mathrm{phys}}$ is recovered from
$b_{\mathrm{adm}}-I_m$. This follows from the explicit Ornstein--Uhlenbeck
subordinator construction, equivalently from
Theorem~\ref{thm:main-convergence} specialized to $\Gamma=0$ and $\mu=0$ when
Assumption~\ref{assump:cV-compact-containment} is satisfied; it is
not a consequence of the joint transform formula, which is applied only after
the covariance process has been constructed. In this OU-subordinator
case the Feynman--Kac martingale identity of Proposition~\ref{prop:affine-fk}
can also be seen directly from the explicit mild representation
$X_t=\cT(t)x+\int_0^t\cT(t-s)\,\D L_s$ and the independent-increment structure
of $L$.
In particular, the characteristic function of the forward curve process
$(f_t)_{t\ge0}$ is given by
\begin{align}\label{eq:affine-transform-levy}
\EX{\exp\left(\mathrm i\langle f_t,u_1\rangle_H\right)}
&= \exp\left(\mathrm i\langle f_0,S^*(t)u_1\rangle_H\right) \nonumber\\
&\quad\times
\exp\left(
-\int_0^t
\varphi_L\left(
\tfrac12\int_0^s \mathcal T(s-\tau)
(S^*(\tau)u_1)^{\otimes2}\mathrm d\tau
\right)\mathrm ds
\right) \nonumber\\
&\quad\times
\exp\left(
-\tfrac12
\Big\langle x,
\int_0^t \mathcal T(t-\tau)(S^*(\tau)u_1)^{\otimes 2}\mathrm d\tau
\Big\rangle_{\mathcal H}
\right),
\end{align}
where the arguments of $\varphi_L$ belong to $\cHplus$.

\begin{remark}
Regarding Assumption~\ref{assump:integration} for the heat-modulated BNS model:
since Theorem~\ref{thm:heat-affine-model} is stated at the $H\times\cHplus$
level, the forward curve process $(f_t)_{t\ge0}$ is modeled in $H$ rather than
$V$. The relevant condition is therefore Assumption~\ref{assump:H-valued-joint},
and the general trace-class sufficient condition is recorded in
Remark~\ref{rem:H-valued-trace-sufficient}. In the present BNS specialization it
amounts to
\[
  \mathbb E\!\left[\int_0^t
    \|S(t-s)X_s^{1/2}\|_{\mathcal L_2(H,H)}^2\,\D s\right]<\infty.
\]
A sufficient trace-class criterion is $X_s\in\cL_1(H)$ for almost every $s$ and
\[
  \int_0^t \mathbb E[\Tr(X_s)]\,\D s<\infty,
\]
because then
\[
  \|S(t-s)X_s^{1/2}\|_{\mathcal L_2(H,H)}^2
  \le \|S(t-s)\|_{\cL(H)}^2\,\Tr(X_s).
\]
Under the strengthened trace-class hypotheses of
Theorem~\ref{thm:main-convergence}, together with the corresponding first-moment
integrability of $\Tr(X_s)$, the characteristic-function
formula~\eqref{eq:affine-transform-levy} is therefore rigorous without any
additional square-root regularity hypothesis.

\end{remark}

The formula~\eqref{eq:affine-transform-levy} follows from
Theorem~\ref{thm:heat-affine-model} by specializing to $\mu=0$ (no
state-dependent jumps), $\Gamma=0$, $u_2=0$ (evaluating the $X$-component
marginal out). A small notational point: Theorem~\ref{thm:heat-affine-model}
and Lemma~\ref{lem:joint-Riccati-wellposed} treat $u_1\in\I H$, so that
$\psi_2(t,u_1)=-\frac12\int_0^t\cT(t-s)(S^*(s)u_1)^{\otimes 2}\D s$ with a
negative sign. The display in~\eqref{eq:affine-transform-levy} uses instead a
\emph{real} argument $u_1\in H$ paired against the explicit imaginary unit on
the left-hand side, i.e.\ it is the application of the lemma at the imaginary
argument $\I u_1$. Using the sign identity $(\I h)^{\otimes 2}=-h^{\otimes 2}$
from the convention introduced above, the corresponding loadings simplify to
$\psi_1(t,\I u_1)=\I S^*(t)u_1$ and
$\psi_2(t,\I u_1)=+\frac{1}{2}\int_0^t\cT(t-s)(S^*(s)u_1)^{\otimes 2}\D s$
(\emph{positive} sign for $u_1\in H$), which is exactly the expression
appearing in~\eqref{eq:affine-transform-levy}. Note that for the arguments
of $\varphi_L$ to lie in $\cH^+$, one needs
$\psi_2(t,\I u_1)\in\cH^+$, which follows from positivity of
$(S^*(s)u_1)^{\otimes 2}$ and the positivity-preserving property of
$(\cT(t))_{t\geq 0}$ established in Lemma~\ref{lem:prop-semigroup}.
\end{example}

The following theorem is our main result on the approximation of affine function-valued stochastic
covariance models by affine finite-rank models. Moreover, it provides an explicit convergence rate,
which depends on the speed of convergence of the Galerkin projection. 

\begin{theorem}\label{thm:finite-dim-approx-ScoV}
  Let $(f,X)$ denote the stochastic covariance model from
  Theorem~\ref{thm:heat-affine-model} with $(f_{0},X_{0})=(f_0,x)\in
  H\times\cHplus$. For every $d\in\MN$ let $X^{d}$ denote
  the affine process on $\cHplus_{d}$ as in
  Proposition~\ref{prop:embedding-affine-main} and define $(f^{d},X^{d})$ as the
  stochastic covariance model obtained from~\eqref{eq:HJMM-SDE-H} by replacing
  $X$ with the finite-rank instantaneous covariance process $X^d$. Then the
  following holds: 
  \begin{theoremenum}
  \item\label{item:finite-dim-approx-ScoV-1} For every $f_{0}\in H$,
    $x\in\cHplus$ the stochastic covariance model $(f^{d},X^{d})$ with
    $(f^{d}_{0},X_{0}^{d})=(f_{0},\bP_{d}(x)) $ is well-defined and for every $u=(u_{1},u_{2})\in \I H\times\cHplus$ and $t\geq 0$ the
	    process satisfies
	    \begin{align}\label{eq:stochastic-covariance-affine-formula}
	     \EX{\E^{\langle f^{d}_{t}, u_{1}\rangle_{H}- \langle X^{d}_{t},
	      u_{2}\rangle}}=\E^{-\Phi_{d}(t,u)+\langle f_{0}, 
      \psi_{1}(t,u)\rangle_{H}-\langle \bP_{d}(x),
      \psi_{2,d}(t,u)\rangle}, 
    \end{align}
	    for $(\Phi_{d}(\cdot,u),(\psi_{1}(\cdot,u),\psi_{2,d}(\cdot,u)))$ the
	    unique Galerkin joint Riccati solution with
	    $\psi_{2,d}(0,u)=\bP_d u_2$, $\psi_1(t,u)=S^*(t)u_1$, and
	    with $F$ and $L(\cdot)+\hat{\cR}$ in
	    \eqref{eq:Riccati-phi-psi-1-1}--\eqref{eq:Riccati-phi-psi-1-3} replaced by
	    $F\circ\bP_{d}$ and
	    $\bP_{d}\circ L\circ \bP_{d}+\bP_{d}\circ \hat{\cR}\circ \bP_{d}$,
	    respectively. Since $X_t^d\in\cH_d^+$, the exponent only depends on the
	    projected covariance loading:
	    $\langle X_t^d,u_2\rangle=\langle X_t^d,\bP_d u_2\rangle$.
  \item\label{item:finite-dim-approx-ScoV-2}  When specialized to the canonical
    finite-horizon forward-curve setting of Example~\ref{ex:Laplacian_Forward},
    namely
    \[
      \Theta_{\max}<\infty,\qquad
      V=V_{0,\beta}(0,\Theta_{\max})\oplus\MR,\qquad
      H=L^2(0,\Theta_{\max},\mathrm e^{\beta x}\D x)\oplus\MR,
    \]
    and $B=\Delta\oplus 0$ with the boundary conditions stated there. Moreover, suppose the limiting
    $H$-valued stochastic covariance
    model $(f,X)$ from Theorem~\ref{thm:heat-affine-model} is well defined
	    under Assumptions~\ref{assump:cV-compact-containment} and
	    \ref{assump:H-valued-joint}. Fix $T>0$ and assume moreover that
	    \[
	      S^*(s)V\subseteq V,\qquad
	      \|S^*(s)\|_{\cL(V)}\le M_1\e^{\omega s},\qquad 0\le s\le T,
	    \]
	    and that rank-one tensors satisfy
	    $\|h^{\otimes2}\|_{\cV}\le C_\otimes\|h\|_V^2$ for $h\in V$.
	    If, moreover, the vector measure $\mu$ takes values in $\cV$ and has
	    finite $\cV$-total variation
	    $|\mu|_{\cV}(\cHpluso)<\infty$, and
	    $\Gamma^{*}(\cV\cap \cH)\subseteq\cV\cap \cH$, then for every
    $\tilde{u}\in V$ such that $\norm{\tilde{u}}_{V}\leq 1$ there exists a constant
    $\tilde{C}_{T}$ such that for all $d\in\MN$ we have
    \begin{align}\label{eq:stochastic-covariance-convergence-rate}
      \sup_{t\in [0,T]}
      \left|\EX{\E^{\I \langle f_{t}, \tilde{u}\rangle_{H}}}-\EX{\E^{\I\langle
      f^{d}_{t},\tilde{u}\rangle_{H}}}\right|\leq
      \tilde{C}_{T}\norm{\bP_{d}^{\perp}}_{\cL(\cV,\cH)}(1+\norm{x}).
    \end{align}
	  \end{theoremenum}
	\end{theorem}

\begin{remark}[Explicit convergence rate]\label{rem:explicit-rate}
This remark is a direct corollary of the specialized item~(ii) of
Theorem~\ref{thm:finite-dim-approx-ScoV}: it makes the projection-error rate
$\|\bP_d^\perp\|_{\cL(\cV,\cH)}$ explicit in the canonical setting of
Example~\ref{ex:Laplacian_Forward}, and is therefore not a separate
theorem but the substitution of the Sturm--Liouville eigenvalue asymptotics
into the rate already established in item~(ii).
Concretely, in the heat-modulated setting with covariance generator
$B=\Delta\oplus 0$ on $V=V_{0,\beta}\oplus\MR$, where $\Delta$ is the
weighted Sturm--Liouville operator of Example~\ref{ex:Laplacian_Forward}, the
projection error bound in~\eqref{eq:stochastic-covariance-convergence-rate}
becomes explicit. Let
\[
  \lambda_{d+1}^{\mathrm{tens}}
  \df
  \inf\{\lambda_i+\lambda_j:\max(i,j)>d\}
\]
be the first omitted Lyapunov eigenvalue in the tensor basis
$\{\be_{i,j}\}_{i\le j}$. The standard estimate
\[
\|\bP_d^\perp\|_{\cL(\cV,\cH)}
\leq \frac{C}{\sqrt{\lambda_{d+1}^{\mathrm{tens}}}}
\]
follows from the variational characterisation of eigenvalues and the Hilbert-scale
norm induced by the Lyapunov generator. In the present finite-horizon Laplacian
setting the positive shape eigenvalues satisfy
$\eta_n\sim (n\pi/\Theta_{\max})^{2}$ by
Example~\ref{ex:Laplacian_Forward}. With the one-based convention
$e_1=(0,1)$ and $e_{n+1}=(q_n,0)$, the first omitted Lyapunov eigenvalue
$\lambda_{d+1}^{\mathrm{tens}}$ is the first omitted combined spatial
eigenvalue $\lambda_{d+1}$ (equivalently the first omitted shape eigenvalue
under the chosen rank convention). Substituting into~\eqref{eq:stochastic-covariance-convergence-rate}
gives the explicit rate
\[
\sup_{t\in[0,T]}\left|\mathbb E\big[\E^{\I\langle f_t,\tilde u\rangle_H}\big]
-\mathbb E\big[\E^{\I\langle f_t^d,\tilde u\rangle_H}\big]\right|
\leq
\frac{\tilde C_T}{\sqrt{\lambda_{d+1}^{\mathrm{tens}}}}(1+\|x\|),
\]
which is the theoretical underpinning of the $1/\sqrt{\lambda_{d+1}^{\mathrm{tens}}}$
reference lines used in Figure~\ref{fig:convergence}.
\end{remark}

\begin{remark}[The structural hypotheses hold in the BNS specialization]\label{rem:struct-hyp-bns}
The four structural hypotheses of Theorem~\ref{thm:finite-dim-approx-ScoV}(ii)
are met in the diagonal heat-modulated BNS setting used in the numerics, so the
rate~\eqref{eq:stochastic-covariance-convergence-rate} genuinely applies there.
Indeed: (i) $S^*(s)V\subseteq V$ with the growth bound $\|S^*(s)\|_{\cL(V)}\le
M_1\E^{\omega s}$ holds by Lemma~\ref{lem:adjoint-shift-V}, in fact with
$M_1=1$ and $\omega=0$;
(ii) the rank-one bound
$\|h^{\otimes2}\|_{\cV}\le C_\otimes\|h\|_V^2$ follows from the continuous
embedding $V\hookrightarrow H$ and the standard Hilbert--Schmidt identity
$\|h^{\otimes2}\|_{\cL_2(H,V)}=\|h\|_H\|h\|_V$; (iii) $\mu=0$ gives
$|\mu|_{\cV}(\cHpluso)=0<\infty$; and (iv) $\Gamma=0$ gives
$\Gamma^{*}(\cV\cap\cH)=\{0\}\subseteq\cV\cap\cH$ trivially.
\end{remark}

\begin{remark}[Assumption~\ref{assump:cV-compact-containment} in the diagonal BNS benchmark]\label{rem:bns-assump-C}
The diagonal heat-modulated BNS parameters used in the numerical section also
satisfy Assumption~\ref{assump:cV-compact-containment} and the
$H$-valued stochastic-integrability condition of
Assumption~\ref{assump:H-valued-joint}. Indeed, in that benchmark the physical
subordinator drift is $b_{\mathrm{phys}}=0$, while $\Gamma=0$ and $\mu=0$.
Under the fixed truncation convention, the admissibility drift is
$b_{\mathrm{adm}}=I_m$. Since the reported jump amplitude satisfies
$\|d\|_{\cH}=0.5(\sum_{k\ge1}k^{-4})^{1/2}<1$, we have $\chi(d)=d$ and hence
$b_{\mathrm{adm}}=\beta_L d$. Therefore
$b_{\mathrm{adm}}-\int\chi(\xi)m(\D\xi)=b_{\mathrm{phys}}=0\in\cHplus$, so the
drift admissibility condition is satisfied. The estimate below gives
$d\in\cV$, hence $b_{\mathrm{adm}}\in\cV\cap\cH$ and the drift-regularity part
of Assumption~\ref{assump:cV-compact-containment} holds; the state-dependent
jump conditions are trivial because $\mu=0$.
The compound-Poisson subordinator has finite intensity $\beta_L$ and a single
diagonal jump amplitude $d=(d_k)_{k\ge1}$ with $d_k=0.5/k^2$. Since the
finite-horizon Laplacian eigenvalues satisfy $\lambda_k\simeq k^2$,
\[
  \|d\|_{\cV}^{2}
  \lesssim
  \sum_{k\ge1}(1+\lambda_k)d_k^2
  \lesssim
  \sum_{k\ge1}(1+k^2)k^{-4}<\infty.
\]
Thus $\int\|\xi\|_{\cV}^{2}m(\D\xi)=\beta_L\|d\|_{\cV}^{2}<\infty$. Moreover,
$\Tr(d)=\sum_k d_k<\infty$ and the initial covariance used in the figures has
$x_{0,k}=0.5d_k$, hence $\Tr(x_0)<\infty$. The explicit OU-subordinator
representation gives
\[
  \EX{\Tr(X_t)}
  \le \Tr(x_0)+\beta_L t\,\Tr(d),
  \qquad t\ge0,
\]
because the Lyapunov semigroup is trace-dissipative on positive operators.
Since the stopped-shift semigroup $S$ is a contraction on $H$, this implies
\[
  \EX{\int_0^T\|S(T-s)X_s^{1/2}\|_{\cL_2(H,H)}^2\,\D s}
  \le
  \int_0^T\EX{\Tr(X_s)}\,\D s<\infty,
\]
which is Assumption~\ref{assump:H-valued-joint}.
\end{remark}

\begin{remark}\label{rem:square-root-regularity}
  We make the following remarks on approximating the forward curve
  process $(f_{t})_{t\geq 0}$ by finite-dimensional models, and the modeling in $V$ versus $H$:
  \begin{enumerate}
  \item[i)] Note that since the embedding
    $H_{\beta} \subset\!\subset L^2(\mathbb{R}_+, \E^{\gamma x} \D x) \oplus
    \mathbb{R}$ is compact if $0<\gamma<\beta$, we can approximate the forward curve dynamics
    $(f_t)_{t \geq 0}$ by a finite-dimensional process $(U_d(f_t))_{t \geq 0}$
    in the space $L^2(\mathbb{R}_+, \E^{\gamma x} \D x) \oplus \mathbb{R}$,
    where $(U_d)_{d \in \mathbb{N}}$ is defined as the finite-rank operator
    \begin{align*}
    U_d v = \sum_{k=1}^{d} s_k \langle v, h_k^{(d)} \rangle_V e_k  
    \end{align*}
    for some $(s_k)_{k \in \mathbb{N}} \subseteq \mathbb{R}_+$ and
    $(h_k^{(d)})_{k \leq d} \subseteq \operatorname{dom}(\mathcal{A}^*)$ (see
    \cite{Tap13}). However, the resulting approximations $(f_t^d)_{t\ge0}$ are generally \emph{not}
affine processes, since finite-dimensional projections of affine processes do not
preserve the affine property in general. As a consequence, although
$(U_d(f_t))_{t\ge0}$ is finite-dimensional, it may fail to be tractable in the
sense usually associated with affine models.
More specifically, the principal source of tractability of affine dynamics, the reduction of Fourier-Laplace transforms and characteristic functions to
generalized Riccati equations, is typically lost after projection. In particular,
one cannot in general expect exponential--affine representations of quantities
of the form
$$
  \mathbb E\left[\exp\bigl(\langle u, U_d f_t\rangle\bigr)\right],
$$
nor closed or semi-closed pricing formulas obtained via Fourier methods, even if
such formulas are available for the infinite-dimensional model.
This loss of structure has several practical implications: 
(a) pricing and hedging of European-style claims must typically rely on simulation
or numerical PDE/integro-PDE methods in $\mathbb R^d$; 
(b) calibration becomes substantially more costly, since each evaluation of an
objective function requires repeated Monte Carlo simulations or numerical PDE
solves; and 
(c) preserving the no-arbitrage HJM drift condition under the reduced dynamics
requires additional care, as the projected dynamics need not inherit the same
structural drift restrictions as the original model.
 Instead, using
    Theorem~\ref{thm:finite-dim-approx-ScoV}, we propose approximating the
    forward curve dynamics by finite-rank affine stochastic covariance models
    $(f^d, X^d)_{d \in \mathbb{N}}$, which, for each $d \in \mathbb{N}$, are
    infinite-dimensional, yet affine and exhibit only finite-rank noise.
  
  \item[ii)] If there exists an orthonormal basis (ONB)
    $(e_n)_{n \in \mathbb{N}}$ of eigenvectors for the drift operator
    $\mathcal{A}$ in~\eqref{eq:HJMM-SDE}, we could also approximate the process
    $f^d$ by $H_d$-valued processes $(\tilde{f}^d)_{d \in \mathbb{N}}$ by
    following the same ideas as in this paper, namely approximating
    the unbounded operator $\mathcal{A}$ by $\mathcal{A}^d \coloneqq \sP_d
    \mathcal{A} \sP_d$ and $X$, as presented here, by $X^d$. In this
    case, an analog of Theorem~\ref{thm:finite-dim-approx-ScoV} holds, with
    $(f^d, X^d)$ replaced by $(\tilde{f}^d, X^d)$. However, in applications, we
    typically have $\mathcal{A} = \frac{\partial}{\partial x}$, which does not
    admit an ONB of eigenvectors in the space $H_{\beta}$. One could, motivated
    by~\cite{Con05}, consider the HJM model with
    $\mathcal{A}_1 \coloneqq \frac{\partial}{\partial x} + \Delta$. This
    operator possesses a sequence of eigenvectors that forms an ONB for spaces
    like $L^2(0, \Theta_{\mathrm{max}}, \E^{\beta x} \D x)$, allowing for
    finite-dimensional and finite-rank approximations. However, such a model is not arbitrage-free
    in general. We leave these considerations for future
    work.
    
  \item[iii)] If one wants to model the forward curve dynamics in the regular
    space $V$, then the volatility, not only the covariance, must satisfy the
    regularity required by Assumption~\ref{assump:integration}. In the
    square-root specification this means imposing
    $X_t^{1/2}\in\cL_2(H,V)$, and, if one wants the full covariance-regularity
    convention used for $\cV$, $X_t^{1/2}\in\cV$. This is an additional
    assumption. The fixed-time regularity $X_t\in\cV\cap\cHplus$ from
    Theorem~\ref{thm:main-convergence} does \emph{not} imply
    $X_t^{1/2}\in\cV$.

    The distinction is visible already in the diagonal Hilbert-scale
    case. Suppose that $H$ has an orthonormal basis $(e_n)$ and that
    \[
      \|h\|_V^2\simeq \sum_n \rho_n^2|\langle h,e_n\rangle_H|^2,
      \qquad \rho_n\to\infty .
    \]
    If $X_t e_n=\hat x_n(t)e_n$ with $\hat x_n(t)\ge0$, then
    $X_t\in\cL_2(H,V)$ requires
    \[
      \sum_n \hat x_n(t)^2\rho_n^2<\infty,
    \]
    whereas $X_t^{1/2}\in\cL_2(H,V)$ requires the strictly stronger condition
    \[
      \sum_n \hat x_n(t)\rho_n^2<\infty .
    \]
    Thus trace-class covariance regularity and $\cV$-regularity of $X_t$ do not
    by themselves control the square-root volatility in the regular curve space.

  \item[iv)] A clean sufficient condition for using the square-root
    specification in the regular space $V$ is to impose the square-root
    condition directly:
    \[
      X_t^{1/2}\in\cL_2(H,V)
      \quad\text{and}\quad
      \mathbb E\!\left[\int_0^T
      \|S(T-s)X_s^{1/2}\|_{\cL_2(H,V)}^2\,\D s\right]<\infty
    \]
    for every $T>0$. In the diagonal Hilbert-scale regime above, a sufficient
    pointwise condition is
    $\sum_n\hat x_n(t)\rho_n^2<\infty$. This condition may follow from a
    stronger model specification forcing $X_t$ to take values in
    $\cL_2(V^*,V)$, but it is not a consequence of the abstract admissible
    parameter set. Alternatively, one may formulate the forward equation in the
    pivot space $H$ and use Assumption~\ref{assump:H-valued-joint}; this is the
    route taken in Theorem~\ref{thm:heat-affine-model}.
  \end{enumerate}
\end{remark}

\subsection{Example: Robustness of Pricing Flow Forwards in Energy Markets}\label{sec:exampl-robustn-pric}
\label{ex:robustness-flow-forward}
We study Fourier prices for exponential functionals of the forward
curve, motivated by European Call options written on \emph{geometric power flow
forwards} in electricity and gas markets, cf.~\cite{BK15}.
The stochastic covariance model used here is the driftless
$H$-valued affine benchmark from Theorem~\ref{thm:heat-affine-model}. Thus the
formulas below are mathematically prices of exponential curve functionals under
that benchmark. If they are interpreted as no-arbitrage geometric forward-option
prices, this interpretation requires the additional geometric HJM pricing layer
and complex-transform assumptions stated in Proposition~\ref{prop:robustness};
it is not an automatic consequence of the driftless transform theorem alone.
Fix $t<\tau_1<\tau_2<\Theta_{\max}$. We model the log-forward curve in Musiela
parametrisation as an $H$-valued process $(f_t)_{t\geq 0}$, where
$$
V = V_{0,\beta}(0,\Theta_{\max})\oplus\MR,
\qquad
H = L^2(0,\Theta_{\max},e^{\beta x}\,dx)\oplus \MR.
$$
This is the same finite-horizon weighted pair used in
Theorem~\ref{thm:finite-dim-approx-ScoV}(ii) and in the numerical
implementation. On a bounded interval the exponential weight is equivalent to
the unweighted $L^2$ norm, and the compactness needed for the Galerkin
approximation follows from the Sobolev embedding discussed in
Section~\ref{sec:the-state-space}.
The instantaneous forward price is, formally, given by
$F(t,T)=\exp(f_t(T-t))$. Point evaluation $f\mapsto f(T-t)$ is well-defined
on the regular space $V=H_\beta$ (where it is a bounded linear functional)
but \emph{not} on the larger pivot space $H=L^2$. We therefore use this
pointwise expression only as motivation for the averaged functional below;
all subsequent statements are formulated in terms of $\ell_{t,\tau_1,\tau_2}$,
which has a bounded $H$-Riesz representer and is therefore well-defined on the
$H$-valued forward curve $f_t$.
We define the \emph{geometric flow forward} over $[\tau_1,\tau_2]$ by
\begin{equation}\label{eq:geo-flow-forward}
  F^{\mathrm{geo}}(t;\tau_1,\tau_2)
  :=\exp\Big(
      \frac{1}{\tau_2-\tau_1}
      \int_{\tau_1}^{\tau_2}\log F(t,T)\,dT
     \Big)
  =\exp\big(\ell_{t,\tau_1,\tau_2}(f_t)\big),
\end{equation}
where the averaging functional $\ell_{t,\tau_1,\tau_2}:V\to\R$ is defined by
\begin{equation}\label{eq:avg-functional}
  \ell_{t,\tau_1,\tau_2}(f)
  :=\frac{1}{\tau_2-\tau_1}\int_{\tau_1-t}^{\tau_2-t} f(z)\,dz,
  \qquad f\in V.
\end{equation}
Note that $\ell_{t,\tau_1,\tau_2}$ is a bounded linear
functional on $V$, hence $\ell_{t,\tau_1,\tau_2}\in V^*$.
Moreover, on the finite-horizon pivot space
$H=L^2(0,\Theta_{\max},e^{\beta x}\D x)\oplus\MR$, the same averaging operation
admits an explicit Riesz representer in $H$: writing a curve as $(u,r)$ with
actual value $u(x)+r$ and setting $a_\tau\df\tau_1-t$,
$b_\tau\df\tau_2-t$, the linear
functional
$$
  (u,r)\mapsto \frac{1}{b_\tau-a_\tau}\int_{a_\tau}^{b_\tau}\bigl(u(z)+r\bigr)\D z
$$
is represented in the standard product inner product as
$$
  \ell_H = \Bigl(\frac{e^{-\beta z}}{b_\tau-a_\tau}\mathbf 1_{[a_\tau,b_\tau]}(z),\;1\Bigr)\in H,
$$
so that the averaged functional is $\langle f,\ell_H\rangle_H$ for $f\in H$. Throughout the
proposition below we use this $H$-representer whenever the transform
machinery is applied to the $H$-valued process $f_t$; the notation
$\ell\in V^*$ is reserved for the regular-curve interpretation.\newline{}
A European Call payoff with strike $K>0$ and exercise time $t<\tau_{1}$ written
on this exponential functional is
$$
  \Pi_t := \big(F^{\mathrm{geo}}(t;\tau_1,\tau_2)-K\big)^+.
$$

\begin{proposition}[Robustness bound for forward-flow option prices]\label{prop:robustness}
Let $(f,X)$ be the driftless affine stochastic covariance model on
$L^2(0,\Theta_{\max},\mathrm{e}^{\beta x}\mathrm{d}x)\oplus \mathbb{R}$ and
let $(f^d,X^d)$ be its finite-rank approximation as in
Theorem~\ref{thm:finite-dim-approx-ScoV}, under
Assumptions~\ref{assump:cV-compact-containment} and
\ref{assump:H-valued-joint}. Fix
$t<\tau_1<\tau_2<\Theta_{\max}$ and $K>0$, and set
$\ell_H\in H$ equal to the Riesz representer of the averaging
functional displayed above, and let $\ell$ denote its restriction to $V$.
Define
$$
\pi_0 := \mathbb E\big[(\mathrm e^{\langle f_t,\ell_H\rangle_H}-K)^+\big],
\qquad
  \pi_0^d := \mathbb E\big[(\mathrm e^{\langle f_t^d,\ell_H\rangle_H}-K)^+\big].
$$
Fix $\nu>1$, set $s(\lambda):=\nu+\mathrm{i}\lambda$, and define
\[
g(\lambda)
:=\frac{K^{-(\nu-1+\mathrm{i}\lambda)}}{(\nu+\mathrm{i}\lambda)(\nu-1+\mathrm{i}\lambda)},
\qquad \lambda\in\mathbb{R}.
\]
Assume moreover that the Carr--Madan-type pricing assumptions of
\cite[Proposition~3.2 and Theorem~2.3]{HeKarbachKhedher2025} apply to the
functional represented by $\ell_H$. Finally, assume that on the pricing strip $\operatorname{Re}(z)=\nu$ there exist a
measurable function $M_t:\mathbb{R}\to[0,\infty)$ and functions
$r_d:\mathbb{R}\to[0,\infty)$ such that $|g|M_t\in L^1(\mathbb{R})$,
$\int_{\mathbb{R}}|g(\lambda)|r_d(\lambda)\,\mathrm{d}\lambda\to0$, and
$$
\Big|\mathbb E[\mathrm e^{s(\lambda)\langle f_t,\ell_H\rangle_H}]
-\mathbb E[\mathrm e^{s(\lambda)\langle f_t^d,\ell_H\rangle_H}]\Big|
\leq M_t(\lambda)\,
\|\bP_d^\perp\|_{\cL(\cV,\cH)}(1+\|X_0\|)
+r_d(\lambda)
$$
for all $\lambda\in\mathbb{R}$.
Then:
\begin{enumerate}
\item[(i)]
With the above $g$ and $s$, 
$$
\pi_0
=\frac{1}{2\pi}\int_{\mathbb{R}} g(\lambda)\,
\mathbb E\!\left[\mathrm e^{s(\lambda)\langle f_t,\ell_H\rangle_H}\right]\,\mathrm{d}\lambda,
$$
and, by the assumed complex affine transform formula,
$$
\pi_0
=\frac{1}{2\pi}\int_{\mathbb{R}} g(\lambda)\,
\exp\Big(
-\tilde\Phi(t,s(\lambda)\ell_H)
+\langle f_0,\tilde\psi_1(t,s(\lambda)\ell_H)\rangle_H
-\langle x,\tilde\psi_2(t,s(\lambda)\ell_H)\rangle_{\mathcal{H}}
\Big)\,\mathrm{d}\lambda,
$$
where $(\tilde\Phi,\tilde\psi_1,\tilde\psi_2)$ denotes the complex-valued mild solution
of the associated generalized Riccati equations.
\item[(ii)]
There exists a constant $C(t)>0$ such that for every $d\in\mathbb{N}$,
$$
|\pi_0-\pi_0^d|
\leq C(t)\,\|\bP_d^\perp\|_{\cL(\cV,\cH)}\,(1+\|X_0\|)
+\varepsilon_d,
$$
where $\varepsilon_d\to0$ as $d\to\infty$.
\end{enumerate}
\end{proposition}
\begin{proof}
\noindent\textbf{(i).}
By~\cite[Proposition~3.2]{HeKarbachKhedher2025} applied with the functional $u_\vartheta$
replaced by the averaging functional represented by $\ell_H$, the call price
$\pi_0 = \mathbb{E}[(\mathrm{e}^{\langle f_t,\ell_H\rangle_H}-K)^+]$ admits the Carr-Madan-type representation
$$
\pi_0
=\frac{1}{2\pi}\int_{\mathbb R}g(\lambda)\,
\mathbb{E}\!\left[\mathrm{e}^{s(\lambda)\langle f_t,\ell_H\rangle_H}\right]\mathrm{d}\lambda.
$$
The expression for $\pi_0$ in terms of
$(\tilde\Phi,\tilde\psi_1,\tilde\psi_2)$ follows from the assumed complex
affine transform formula with $u_1=s(\lambda)\ell_H$.

\noindent\textbf{(ii).}
Use the representation from part (i). With $s(\lambda)=\nu+\mathrm i\lambda$,
\[
\pi_0-\pi_0^d
=\frac1{2\pi}\int_{\mathbb R} g(\lambda)\Big(
\mathbb E[\mathrm e^{s(\lambda)\langle f_t,\ell_H\rangle_H}]
-\mathbb E[\mathrm e^{s(\lambda)\langle f_t^d,\ell_H\rangle_H}]
\Big)\,\mathrm d\lambda.
\]
Hence,
\[
|\pi_0-\pi_0^d|
\le \frac1{2\pi}\int_{\mathbb R} |g(\lambda)|\,
\Big|\mathbb E[\mathrm e^{s(\lambda)\langle f_t,\ell_H\rangle_H}]
-\mathbb E[\mathrm e^{s(\lambda)\langle f_t^d,\ell_H\rangle_H}]\Big|
\,\mathrm d\lambda.
\]
Using the dominated stability estimate assumed in the proposition, we obtain
\[
|\pi_0-\pi_0^d|
\le \frac{1}{2\pi}\Big(\int_{\mathbb R}|g(\lambda)|M_t(\lambda)\,\mathrm d\lambda\Big)\,
\|\bP_d^\perp\|_{\mathcal L(\mathcal V,\mathcal H)}\,(1+\|X_0\|)
+\varepsilon_d,
\]
where
\[
\varepsilon_d
:=\frac{1}{2\pi}\int_{\mathbb R}|g(\lambda)|r_d(\lambda)\,\mathrm d\lambda
\to0.
\]
This proves the claim with
$C(t)=(2\pi)^{-1}\int_{\mathbb R}|g(\lambda)|M_t(\lambda)\,\mathrm d\lambda$.
\end{proof}

\noindent
{The numerical illustrations below use $\Theta_{\max}=50$, reference rank $d_{\mathrm{ref}}=12$ in the convergence
benchmark, jump amplitudes $d_k=0.5/k^2$, compound-Poisson jump
intensity $\beta_L=1.5$ of the subordinator $L$ (so that, in the eigenbasis, $L$
has physical Laplace exponent
$\varphi_L(u)=\langle b_{\mathrm{phys}},u\rangle
+\beta_L(1-\E^{-\langle d,u\rangle})$ with $b_{\mathrm{phys}}=0$ in the
experiments; the admissibility drift is
$b_{\mathrm{adm}}=\beta_L d$ under the fixed truncation convention), initial covariance eigenvalues
$x_{0,k}=0.5\,d_k$, initial forward curve
$f_0(\theta)=0.03+0.01(1-\mathrm e^{-\theta})/\theta$, and Fourier settings
$n_{\mathrm{trun}}=40$ and $\Delta_{\lambda}=0.05$. In all numerical
experiments, the physical L\'evy drift is set to $b_{\mathrm{phys}}=0$ (i.e.\
the uncompensated subordinator has no deterministic drift beyond its jumps; the
canonical admissibility drift is supplied by the compensation term, and the
Lyapunov operator itself is always present). For consistency with the
plotting code, the figure captions below use $\alpha$ for the
short-end-anchored weight exponent of the pivot space
$L^2(0,\Theta_{\max},\mathrm e^{\beta x}\D x)$ introduced in
Example~\ref{ex:Laplacian_Forward}; the two symbols refer to the same
parameter, and $\alpha=\beta=0.1$ throughout the reported figures. The choice $d_{\mathrm{ref}}=12$ is
adequate for the reported parameter set as a finite-rank benchmark: the
discrete Galerkin eigenvalues satisfy
$\lambda_{12}^{\mathrm{tens}}\approx 60.16$, so the upper-bound rate
$1/\sqrt{\lambda_{12}^{\mathrm{tens}}}\approx 0.13$ is already small on the scale of the
displayed convergence plot. The last one or two rank errors should be read as
diagnostic finite-rank differences, since they are at or below the numerical
Fourier-truncation sensitivity reported below.
We stress the scope of these experiments: they confirm
\emph{qualitative} convergence within the finite-rank hierarchy, and the fitted
log-log slope $\approx 2.83>1$ shows that the theoretical
$1/\sqrt{\lambda_{d+1}^{\mathrm{tens}}}$ bound is conservative here rather than
sharp. They are run in the BNS specialization ($\mu=0$, $\Gamma=0$); a numerical
study of the genuinely state-dependent case ($\mu\neq0$), for which the affine
transform is now rigorous through Proposition~\ref{prop:affine-fk}, is left for
future work.}
\par
{We illustrate Proposition~\ref{prop:robustness} numerically using the
diagonal heat-modulated BNS specialization introduced above.
Figure~\ref{fig:convergence} shows the rank-$d$ prices $\pi_0^d$ converging to
the rank-$12$ \emph{finite-rank benchmark} $\pi_0^{\mathrm{ref}}$ (we emphasize
that the unknown infinite-dimensional price $\pi_0$ is not available
analytically; the figure does not claim to reproduce $\pi_0$, only that the
rank-$d$ prices converge \emph{within} the finite-rank hierarchy. The absolute
error $|\pi_0-\pi_0^d|$ is bounded by
$|\pi_0-\pi_0^{12}|+|\pi_0^{12}-\pi_0^d|$; the first term is theoretically
controlled under Proposition~\ref{prop:robustness} but is not observed
numerically because $\pi_0$ is unavailable.
The second term is the finite-rank benchmark error displayed in the figure.)
Panel~(b) compares the absolute error against the upper-bound rate
$\|\bP_d^\perp\|_{\mathcal{L}(\mathcal{V},\mathcal{H})}\lesssim 1/\sqrt{\lambda_{d+1}^{\mathrm{tens}}}$
from part~(ii), using the discrete Galerkin eigenvalues
$\lambda_{d+1}^{\mathrm{tens}}$ and a
proportionality constant calibrated from the computed
errors; the log-log regression slope of error against rate is reported in the
panel title (a slope of~$1.00$ would indicate the error tracks the bound
exactly; a slope greater than~$1.00$ indicates faster-than-bound convergence,
i.e.\ the upper bound is conservative).
Panel~(c) displays the convergence constant
$|\pi_0^d-\pi_0^{\mathrm{ref}}|\cdot\sqrt{\lambda_{d+1}^{\mathrm{tens}}}$. In the present
benchmark this quantity decreases with $d$, which is consistent with the fitted
slope $2.83>1$ in panel~(b) and indicates that the theoretical upper bound is
conservative rather than sharp here. The rank-$12$ benchmark value is
$\pi_0^{\mathrm{ref}}=0.05993553$, and the fitted proportionality constant in
the rate guide is $C=5.85\times10^{-4}$.
A quadrature and Fourier sensitivity check shows that varying the
Gauss-Legendre node count in $\{16,32,64\}$ changes the benchmark price by less
than $10^{-10}$, while varying the Fourier half-bandwidth in $\{20,40,60\}$
changes it by about $4.32\times 10^{-4}$ between $20$ and $40$ and by about
$2.75\times 10^{-6}$ between $40$ and $60$. We therefore use
$n_{\mathrm{trun}}=40$ throughout the reported figures.
Figure~\ref{fig:mc} compares Fourier inversion of the affine transform with a
conditional-Gaussian Monte Carlo estimator across four tenors
$\theta\in\{0.5,1,2,5\}$, with per-$N$ standard-error bars shown in
panel~(a).
Panel~(b) displays a per-tenor $C_\theta/\sqrt{N}$ reference line for each
$\theta$, where the constant $C_\theta$ is calibrated from the $N=1\,000$
error of that tenor; all four curves track their respective reference
lines across five sample sizes
$N\in\{1{,}000;\,5{,}000;\,20{,}000;\,100{,}000;\,200{,}000\}$, confirming
$1/\sqrt{N}$ Monte Carlo scaling for each tenor individually.
}
\medskip

\begin{figure}[ht]
  \centering
  \includegraphics[width=\textwidth]{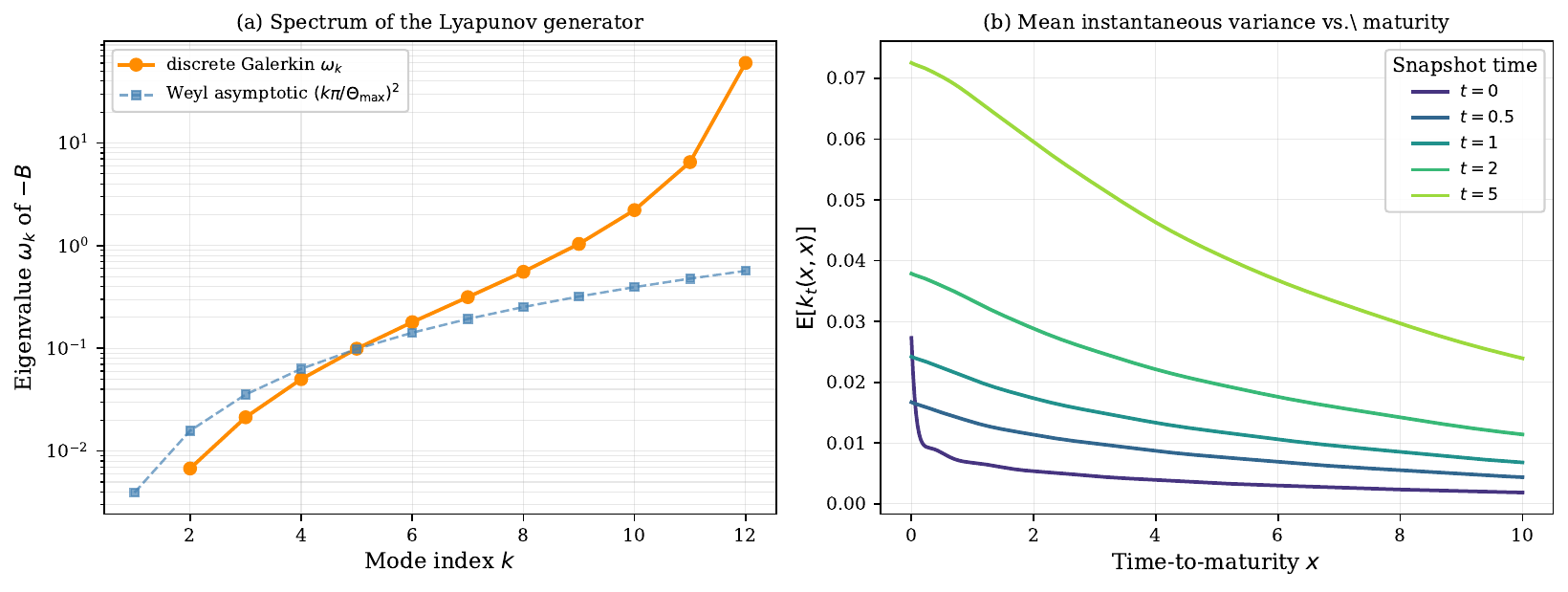}
  \caption{Illustration of the heat regularization mechanism in the
    rank-$12$ heat-modulated BNS specialization.
    \textit{Panel~(a):} discrete Galerkin eigenvalues $\omega_k$ of the
    operator $-B=\Delta\oplus 0$ in the weighted Laguerre basis (orange),
    together with the Weyl asymptotic $(k\pi/\Theta_{\max})^{2}$
    (steelblue, dashed). The two sequences agree for low $k$ where the
    Laguerre basis resolves the oscillatory branch, and diverge at high
    $k$, reflecting the well-known fact that a finite Laguerre basis
    overestimates high-frequency Sturm--Liouville eigenvalues; this
    over-estimate is a variational \emph{upper bound} on the true
    eigenvalues and is the source of the conservative finite-rank rate
    in Figure~\ref{fig:convergence}.
    \textit{Panel~(b):} mean instantaneous covariance kernel diagonal
    $\mathbb{E}[k_t(x,x)]$ as a function of time-to-maturity $x$, at five
    snapshot times $t\in\{0,0.5,1,2,5\}$. The initial profile is sharply
    peaked at the short end; under the Lyapunov decay each mode $k$ is
    damped at rate $2\omega_k$, with the higher modes (small spatial
    scales) damped fastest, so the kernel relaxes towards the smoother
    stationary profile dictated by the Lévy subordinator. The
    finite-horizon level mode $\omega_1=0$ is the only undamped mode and produces
    the linear growth of the integrated variance.
    Parameters: $\Theta_{\max}=50$, $\alpha=0.1$, $\beta_L=1.5$,
    $d_k=0.5/k^{2}$, $x_{0,k}=0.5\,d_k$.}
  \label{fig:heat-smoothing}
\end{figure}

\begin{figure}[ht]
  \centering
  \includegraphics[width=\textwidth]{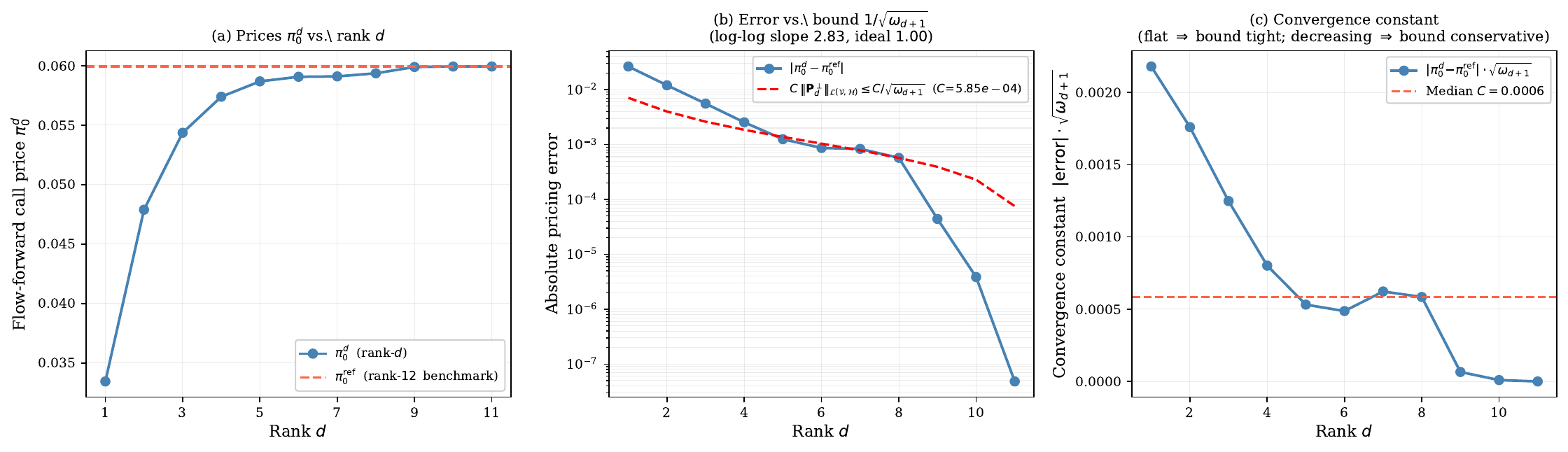}
  \caption{%
    \emph{Within-Galerkin} convergence diagnostic in the
    finite-rank hierarchy. The figure displays the rank-$d$ exponential-functional
    call price $\pi_0^d$ against the rank-$d_{\mathrm{ref}}=12$ benchmark
    $\pi_0^{\mathrm{ref}}$; it does \emph{not} test convergence to the
    infinite-dimensional price $\pi_0$, which is not available in closed
    form. The fitted log-log slope $\approx 2.83$ in panel~(b) is the slope
    of error against the theoretical bound rate
    $1/\sqrt{\lambda_{d+1}^{\mathrm{tens}}}$; a slope of $1$ would indicate
    the error tracks the bound exactly, a slope $>1$ that the bound is
    conservative and the actual within-Galerkin decay is faster.
    \textit{Panel~(a):} prices $\pi_0^d$ vs.\ rank $d$ together with
    the rank-$12$ benchmark price $\pi_0^{\mathrm{ref}}$ (dashed red).
    \textit{Panel~(b):} absolute error $|\pi_0^d-\pi_0^{\mathrm{ref}}|$
    (circles) and the fitted upper-bound rate
    $C\,\|\bP_d^\perp\|_{\mathcal{L}(\mathcal{V},\mathcal{H})}
    \leq C/\sqrt{\lambda_{d+1}^{\mathrm{tens}}}$ (dashed), where
    $\lambda_{d+1}^{\mathrm{tens}}$ is the corresponding discrete Galerkin
    eigenvalue and the proportionality constant $C$ is \emph{fitted from
    the computed finite-rank errors} (median ratio of error to rate over
    the noise-free range), \emph{not} an independently derived theorem
    constant. The dashed line is meant to illustrate the \emph{shape} and
    conservativeness of the bound, not as an independent validation of
    the constant.
    \textit{Panel~(c):} convergence constant
    $|\pi_0^d-\pi_0^{\mathrm{ref}}|\cdot\sqrt{\lambda_{d+1}^{\mathrm{tens}}}$;
    a decreasing profile indicates that, in this benchmark, the observed
    error decays faster than the upper-bound guide
    $1/\sqrt{\lambda_{d+1}^{\mathrm{tens}}}$. Noise-dominated points
    (error $<10^{-8}\,\pi_0^{\mathrm{ref}}$) are excluded from
    panels~(b) and~(c).
    Parameters: $\alpha=0.1$, $\beta_L=1.5$, $b_{\mathrm{phys}}=0$, $t=1$,
    $[\tau_1,\tau_2]=[2,5]$, $K=1$.%
  }
  \label{fig:convergence}
\end{figure}

\begin{figure}[ht]
  \centering
  \includegraphics[width=\textwidth]{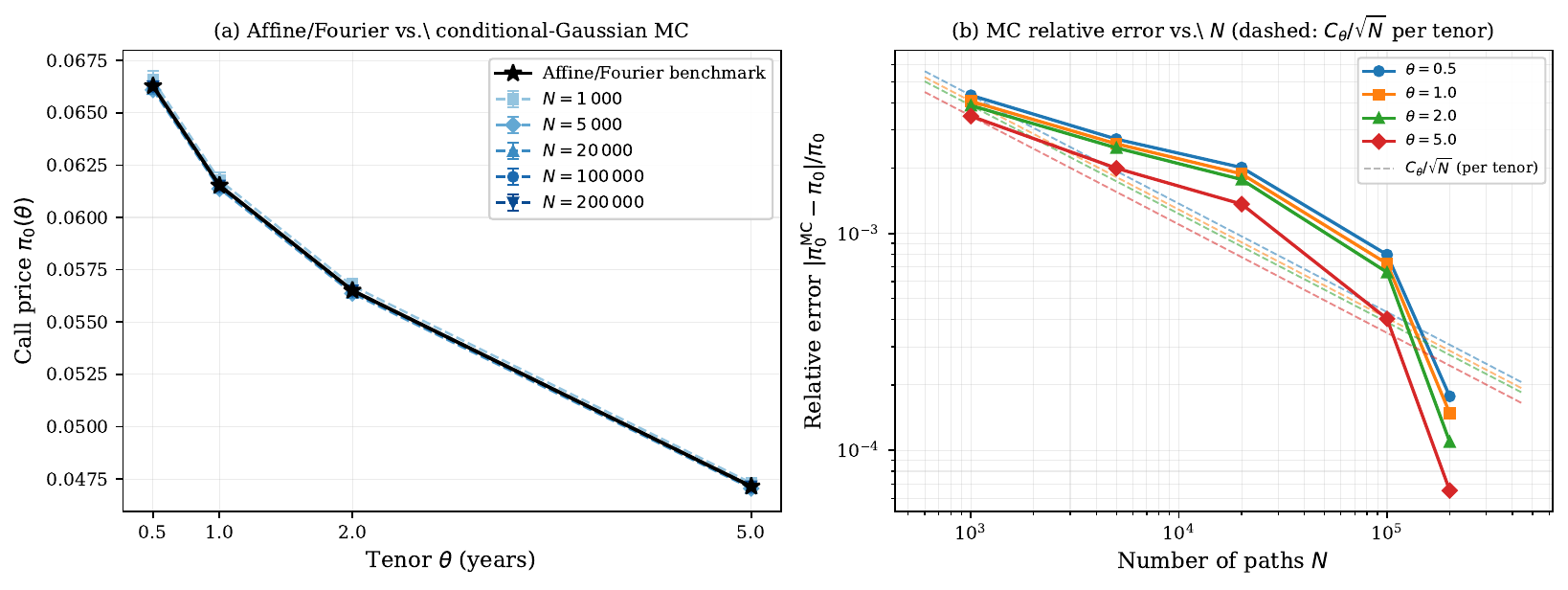}
  \caption{%
    Affine-transform / conditional-Gaussian Monte Carlo comparison
    for the diagonal heat-modulated BNS model. The displayed prices
    $\pi_0(\theta)$ are evaluated using the per-tenor evaluation loadings
    inside the rank-$n=6$ Galerkin discretisation; this is a finite-rank /
    diagonal diagnostic in which the Dirac-type evaluation at maturity
    $\theta$ is represented by the finite Laguerre basis. We do not claim
    that point evaluation $f\mapsto f(\theta)$ is a bounded linear functional
    on the abstract pivot space $H=L^2$; the diagnostic here is meaningful
    because, for the finite-rank model, point evaluation reduces to a bounded
    linear combination of the basis values.
    \textit{Panel~(a):} call prices $\pi_0(\theta)$ via Fourier inversion
    of the affine transform (stars) vs.\ conditional-Gaussian Monte Carlo
    for $N\in\{1{,}000;\,5{,}000;\,20{,}000;\,100{,}000;\,200{,}000\}$ paths, with
    $\pm1$ standard-error bars.
    \textit{Panel~(b):} relative error $|\pi_0^{\mathrm{MC}}-\pi_0|/\pi_0$
    vs.\ $N$ on a log-log scale, with a per-tenor $C_\theta/\sqrt{N}$
    reference line (dashed, same colour as the data); each curve is anchored
    at its own $N=1{,}000$ error so that $1/\sqrt{N}$ convergence can be
    assessed independently of the tenor's price level.
    Parameters: $\alpha=0.1$, $\beta_L=1.5$, $b_{\mathrm{phys}}=0$, $n=6$, $t=0.5$, $K=1$,
    tenors $\theta\in\{0.5,1,2,5\}$,
    $N\in\{1{,}000;\,5{,}000;\,20{,}000;\,100{,}000;\,200{,}000\}$.%
  }
  \label{fig:mc}
\end{figure}

\section{Proofs}\label{sec:proof-theor-eqrefthm}

In this section, we provide the proofs of our main results. The section is
structured as follows: In Section~\ref{sec:mild-solut-gener}, we present the proof of
Proposition~\ref{prop:existence-mild-Riccati}, which establishes the existence
and convergence of the spectral Galerkin approximation. In
Section~\ref{sec:finite-rank-convergence}, we then prove
Proposition~\ref{prop:embedding-affine-main}, concerning the existence of affine
finite-rank operator-valued processes, and Theorem~\ref{thm:main-convergence},
which demonstrates the convergence of the finite-rank affine process to the
class of irregular affine processes, with values in the cone of positive
self-adjoint Hilbert--Schmidt operators and, under stronger hypotheses, in the
trace class. In Section~\ref{sec:proof-section-2}, we
provide the proofs related to the joint model. Specifically, we prove
Theorem~\ref{thm:heat-affine-model}, which establishes the affine transform
formula for the $H$-valued heat-modulated stochastic covariance model under the
explicit stochastic-integrability assumption from
Assumption~\ref{assump:H-valued-joint}, and
Theorem~\ref{thm:finite-dim-approx-ScoV}, which concerns the finite-rank
approximation of the corresponding affine transforms.

Throughout the proof section we use the regular-case arguments from
\cite{karbach2023finiterank} as a template whenever the dynamics are
finite-dimensional or governed by bounded projected operators. The genuinely
new ingredients in the present irregular setting are the construction of the
mild Riccati solution for the unbounded Lyapunov operator $L$, the smoothing
estimate that yields compact containment through the space $\cV$, and the
trace-class regularization induced by the Lyapunov semigroup associated with
the Laplacian-type generator.

\subsection{Galerkin approximation of the generalized mild Riccati
  equations}\label{sec:mild-solut-gener}

Let $(b,\mathbf{B},m,\mu)$ be an admissible parameter set as in
Definition~\ref{def:admissible-irregular}. Recall the operator $L$ from
\cref{eq:Lyapunov} and the function $\hat{R}\colon\cHplus\to \cH$
from~\eqref{eq:R-hat}. We denote the right-hand side function of the mild
generalized Riccati equations~\eqref{eq:Riccati-psi-mild} by $R$, i.e. we define
it as
\begin{align*}
  R(u)\df L(u)+\hat{R}(u),\quad u\in\dom(L).
\end{align*}
In the spirit of~\cite{Mar76}, and going back to the classical study of
Riccati equations associated with unbounded operators in~\cite{Tem69}, we
interpret~\eqref{eq:Riccati-psi-mild} as semi-linear differential equation on
$\cHplus$, with unbounded linear operator $L$ and non-linear continuous
perturbation $\hat{R}$. Note that $\hat{R}$ is even locally Lipschitz continuous,
see~\cite[Remark 3.4]{CKK22a}. In the following lemma, we show that the
semigroup $(\mathcal{T}(t))_{t\geq 0}$, generated by $(L,\dom(L))$, is a
semigroup of \emph{positive operators} on $\cH$ with respect to the cone
$\cHplus$: 
\begin{lemma}\label{lem:prop-semigroup}
  The operator-semigroup $(\mathcal{T}(t))_{t\geq 0}$ is self-adjoint and
  strongly continuous. Moreover,
  $\mathcal{T}(t)(\cHplus)\subseteq \cHplus$ for all $t\geq 0$.
\end{lemma}
\begin{proof}
By Definition~\ref{def:admissible-irregular}(iv)(a), $B$ is
self-adjoint and non-positive on $H$ with spectral domain $\dom(B)$, and
$T(t)=\E^{tB}$ is the associated self-adjoint $C_0$-contraction semigroup;
in particular $T^{*}(t)=T(t)$ for every $t\ge 0$. For $t\geq 0$ the
operator $\cT(t):\cH\to\cH$ satisfies
\begin{equation}\label{eq:Lyapunov-semigroup}
  \cT(t)x = T(t)xT^{*}(t), \qquad x\in\cH,
\end{equation}
Since $T(t)\in\cL(H)$ and $x\in\cL_2(H)$, we have $\cT(t)x\in\cL_2(H)$ and
$$
  \|\cT(t)x\|_{\cH}=\|T(t)xT^{*}(t)\|_{\cH}
  \leq \|T(t)\|_{\cL(H)}^2\,\|x\|,
$$
so each $\cT(t)$ is bounded on $\cH$. Moreover, \eqref{eq:Lyapunov-semigroup}
immediately yields the semigroup property $\cT(t+s)=\cT(t)\cT(s)$ and $\cT(0)=I$.
Strong continuity follows from the strong continuity of $(T(t))_{t\geq 0}$: for any $x\in\cH$ and $t\geq 0$,
\begin{align*}
  \|\cT(t)x - x\|_{\cH}
  = \|T(t)xT^{*}(t) - x\|_{\cH}.
\end{align*}
Choose $\varepsilon>0$. Since finite-rank operators are dense in $\cH$, there
exists a finite-rank operator $x^{(N)}\in\cH$ such that
$\|x-x^{(N)}\|_{\cH}<\varepsilon$. Using that $(T(t))_{t\geq 0}$ is a
contraction semigroup,
\begin{align*}
  \|\cT(t)x-x\|_{\cH}
  &\leq \|\cT(t)(x-x^{(N)})\|_{\cH}
    + \|\cT(t)x^{(N)}-x^{(N)}\|_{\cH}
    + \|x^{(N)}-x\|_{\cH}\\
  &\leq 2\varepsilon + \|\cT(t)x^{(N)}-x^{(N)}\|_{\cH}.
\end{align*}
Write the finite-rank operator as
$x^{(N)}=\sum_{k=1}^{N} u_k\otimes v_k$, where
$(u\otimes v)h=\langle h,v\rangle_H u$. Then
\begin{align*}
  \cT(t)x^{(N)}-x^{(N)}
  = \sum_{k=1}^{N}\bigl(T(t)u_k\otimes T(t)v_k-u_k\otimes v_k\bigr),
\end{align*}
and each summand converges to zero in $\cH$ because
\begin{align*}
  \|T(t)u_k\otimes T(t)v_k-u_k\otimes v_k\|_{\cH}
  &\leq \|(T(t)u_k-u_k)\otimes T(t)v_k\|_{\cH}
    + \|u_k\otimes (T(t)v_k-v_k)\|_{\cH}\\
  &\leq \|T(t)u_k-u_k\|_H\,\|v_k\|_H
    + \|u_k\|_H\,\|T(t)v_k-v_k\|_H,
\end{align*}
which tends to zero as $t\to 0^+$ by strong continuity of $T(t)$ on
$H$. Since the sum is finite, $\|\cT(t)x^{(N)}-x^{(N)}\|_{\cH}\to 0$ as
$t\to 0^+$. Letting first $t\to 0^+$ and then $\varepsilon\downarrow 0$ proves
strong continuity of $(\cT(t))_{t\geq 0}$ on $\cH$.
Next, $\cT(t)$ is self-adjoint on $\cH$: for $x,y\in\cH$, using the
trace cyclicity $\mathrm{Tr}(ABC)=\mathrm{Tr}(BCA)$ together with $T^{*}(t)=T(t)$
from the first paragraph of this proof,
$$
  \langle \cT(t)x, y\rangle
  = \mathrm{Tr}\big((T(t)xT^{*}(t))\,y\big)
  = \mathrm{Tr}\big(x\,T^{*}(t)\,y\,T(t)\big)
  = \mathrm{Tr}\big(x\,(T(t)yT^{*}(t))\big)
  = \langle x, \cT(t)y\rangle.
$$
Finally, $\cT(t)$ is positive with respect to the cone $\cHplus$: if
$x\in\cHplus$, then for all $h\in H$, using $T^{*}(t)=T(t)$,
$$
  \langle \cT(t)x\,h, h\rangle_H
  = \langle T(t)\,x\,T^{*}(t)h,\, h\rangle_H
  = \langle x\,T^{*}(t)h,\, T^{*}(t)h\rangle_H
  = \langle x\,T(t)h, T(t)h\rangle_H \ge 0,
$$
since $x$ is positive self-adjoint. Hence $\cT(t)x\in\cHplus$ and therefore
$\cT(t)(\cHplus)\subseteq \cHplus$ for all $t\ge 0$.
It remains to identify the generator. This follows spectrally, avoiding any
formal differentiation of an unbounded product. If
$x=\sum_{i\le j}x_{i,j}\be_{i,j}$, then
\[
  \cT(t)x=\sum_{i\le j}\E^{-\lambda_{i,j}t}x_{i,j}\be_{i,j}.
\]
Hence the generator of $(\cT(t))_{t\ge0}$ is the spectral operator with domain
\[
  \left\{x\in\cH:
  \sum_{i\le j}\lambda_{i,j}^{2}|x_{i,j}|^{2}<\infty\right\},
\]
acting by
$x\mapsto-\sum_{i\le j}\lambda_{i,j}x_{i,j}\be_{i,j}$, which is precisely
$(L,\dom(L))$ from~\eqref{eq:Lyapunov}.
\end{proof}

\begin{lemma}\label{lem:Lyapunov-semigroup-smoothing}
  The Lyapunov semigroup $(\cT(t))_{t\geq 0}$ regularizes $\cH$ into
  $\cV$. More precisely, equipping
  \[
    \cV=\cL_2(V^{*},H)\cap\cL_2(H,V)
  \]
  with the natural Hilbert norm
  \[
    \|x\|_{\cV}^{2}
    \df
    \|x\|_{\cL_{2}(H,V)}^{2}
    +
    \|x\|_{\cL_{2}(V^{*},H)}^{2},
  \]
  there exists a constant $C_{\cV}>0$ such that
  \begin{align*}
    \|\cT(t)\|_{\cL(\cH,\cV)}
    \leq C_{\cV}(1+t^{-1/2}),\qquad t>0.
  \end{align*}
  In particular, for every $0<\delta\leq T<\infty$ we have
  \begin{align*}
    \sup_{t\in[\delta,T]}\|\cT(t)\|_{\cL(\cH,\cV)}<\infty,
  \end{align*}
  and the map $t\mapsto \cT(t)x$ is continuous from $(0,\infty)$ into
  $\cV$ for every $x\in\cH$.
\end{lemma}
\begin{proof}
  Let $(\be_{i,j})_{i\leq j}$ be the eigenbasis from
  \eqref{eq:Lyapunov-eigenvalues}. By the spectral Hilbert-scale compatibility
  in Definition~\ref{def:admissible-irregular}(iv)(c), every eigenvector $e_i$
  of $B$ belongs to $V$, and the norm equivalence gives
  \begin{align*}
    \|e_i\|_V^2\le C(1+\lambda_i).
  \end{align*}
  For
  $i<j$, the symmetrized rank-one operator
  $\be_{i,j}=\frac{1}{\sqrt 2}(e_i\otimes e_j+e_j\otimes e_i)$ therefore
  satisfies
  \begin{align*}
    \|\be_{i,j}\|_{\cL_2(H,V)}^{2}
    =\frac12\bigl(\|e_i\|_V^2+\|e_j\|_V^2\bigr)
    \leq C_{1}(1+\lambda_i+\lambda_j),
  \end{align*}
  The same spectral weight controls the $V^*\to H$ realization. Indeed, for
  $\eta\in V^*$ and the rank-one convention
  $(u\otimes v)\eta=\langle u,\eta\rangle_{V,V^*}v$,
  \[
    \|e_i\otimes e_j\|_{\cL(V^*,H)}
    \le \|e_i\|_V\|e_j\|_H=\|e_i\|_V .
  \]
  Equivalently, using the Hilbert adjoint characterization from
  Remark~\ref{rem:lyapunov-regularity}, the Hilbert--Schmidt norm of
  $e_i\otimes e_j:V^*\to H$ agrees with the
  $\cL_2(H,V)$ norm of its adjoint $e_j\otimes e_i:H\to V$. Hence
  \begin{align*}
    \|\be_{i,j}\|_{\cL_2(V^{*},H)}^{2}
    \le C_{2}\bigl(1+\lambda_i+\lambda_j\bigr).
  \end{align*}
  For the diagonal case $i=j$, the operator $\be_{i,i}=e_i\otimes e_i$
  is the unsymmetrised rank-one tensor and satisfies the same bounds:
  $\|\be_{i,i}\|_{\cL_2(H,V)}^{2}=\|e_i\|_V^{2}\le C(1+\lambda_i)
  \le C(1+\lambda_{i,i})$ (using $\lambda_{i,i}=2\lambda_i$),
  and, by the same adjoint argument,
  $\|\be_{i,i}\|_{\cL_2(V^{*},H)}^{2}\le \|e_i\|_V^2
  \le C(1+\lambda_{i,i})$.
  Thus, with $\lambda_{i,j}=\lambda_i+\lambda_j$, there exists a
  constant $C_{3}>0$ such that
  \begin{align}\label{eq:cV-weight-bound}
    \|\be_{i,j}\|_{\cV}^{2}\leq C_{3}(1+\lambda_{i,j}),
    \qquad i\leq j,
  \end{align}
  covering both the off-diagonal ($i<j$) and diagonal ($i=j$) cases.

  Here, and in the spectral estimates of
  Lemma~\ref{lem:irregular-smoothing-cV}, we use the equivalent Hilbert-scale
  norm on $\cV$ for which the eigenbasis $\{\be_{i,j}\}_{i\le j}$ is orthogonal.
  On the spectral core set
  \[
    \|x\|_{\cV,B}^{2}
    \df \sum_{i\le j}(1+\lambda_{i,j})|\langle x,\be_{i,j}\rangle|^{2}.
  \]
  Equivalently, this is the intersection norm built from the coercive
  form $[u,v]_B\df\langle -Bu,v\rangle_{V^*,V}+\lambda\langle u,v\rangle_H$ of
  Definition~\ref{def:admissible-irregular}(iv)(c). Coercivity yields constants
  $0<c\le C$ such that
  \[
    c\|x\|_{\cV}\le \|x\|_{\cV,B}\le C\|x\|_{\cV},
    \qquad x\in\cV\cap\cH.
  \]
  Moreover, $\bP_d$ is contractive in $\|\cdot\|_{\cV,B}$ and
  $\cT_d(t)\df\bP_d\cT(t)\bP_d$ is a contraction semigroup in this norm. The
  estimates below are taken in $\|\cdot\|_{\cV,B}$ and transferred back to the
  original $\cV$ norm by this equivalence; constants are absorbed into
  $C_{\cV}$ and $C_3$.

  Now write $x=\sum_{i\leq j}x_{i,j}\be_{i,j}\in\cH$. Since
  $\cT(t)\be_{i,j}=e^{-\lambda_{i,j}t}\be_{i,j}$, we obtain
  \begin{align*}
    \|\cT(t)x\|_{\cV}^{2}
    \leq
    C_{3}\sum_{i\leq j}(1+\lambda_{i,j})e^{-2\lambda_{i,j}t}|x_{i,j}|^{2}
    \leq
    C_{3}\sup_{r\geq 0}(1+r)e^{-2rt}\,\|x\|_{\cH}^{2}.
  \end{align*}
  Since $\sup_{r\geq 0}(1+r)e^{-2rt}\leq C(1+t^{-1})$, this yields
  \begin{align*}
    \|\cT(t)x\|_{\cV}\leq C_{\cV}(1+t^{-1/2})\|x\|_{\cH},\qquad t>0.
  \end{align*}
The continuity of $t\mapsto \cT(t)x$ as a $\cV$-valued map on
  $(0,\infty)$ follows from the dominated convergence theorem applied to the
  coefficient expansion above.
\end{proof}

As before, let $(e_{i})_{i\in\MN}$ be the orthonormal basis of $H$ consisting of
the eigenvectors of $B$ and for every $d\in\MN$, denote by $H_{d}$ the
$d$-dimensional subspace of $H$ spanned by the first $d$ basis vectors,
i.e. $H_{d}\df \lin\set{e_{i}\colon i=1,\ldots,d}$ and let
$\cH_{d}\df \lin\set{\be_{i,j}\colon\, 1\leq i\leq j\leq d\,}$.  To prove the
existence and convergence of the $d^{\mathrm{th}}$-Galerkin approximations
$\phi_{d}(\cdot,\bP_{d}(u))$ and $\psi_{d}(\cdot,\bP_{d}(u))$, we collect some
important properties of the operators $L$ and $B$, and the function $\hat{R}$ in
the following lemma:

\begin{lemma}\label{lem:irregular-properties}
  For every $d\in\MN$ the following holds: 
  \begin{lemenum}
  \item $\cH_{d}\subseteq \dom(L)$ for all $d\in\MN$.
  \item $B\rvert_{H_{d}}$ is a bounded operator.
  \end{lemenum}
\end{lemma}
\begin{proof}
  The family $(\be_{i,j})_{i\le j}$ is the complete orthonormal basis of
  $\cH$ consisting of eigenvectors of $L$, with $L\be_{i,j}=-\lambda_{i,j}\be_{i,j}$
  (see~\eqref{eq:Lyapunov-eigenvalues}). Since
  $\cH_{d}=\lin\set{\be_{i,j}\colon 1\le i\le j\le d}$ is a finite linear
  combination of these eigenvectors, $\cH_{d}\subseteq\dom(L)$ for every
  $d\in\MN$. (We work with the spectral operator $L$ rather than with an
  unbounded Hilbert-space realization of the form-level map $\mathbf B$; only
  $\cH_d\subseteq\dom(L)$ and the boundedness of the finite-rank restriction are
  used below.)
   
   The restriction of operator $B$ on the finite-dimensional subspace $H_{d}$,
   i.e. $B\rvert_{H_{d}}$, is bounded
   as $H_{d}=\lin\set{e_{i}\colon i=1,\ldots,d}$ for
   $\set{e_{i}\colon i=1,\ldots,d}$ the set of eigenvectors of $B$ and hence
   $\mathrm{ran}(B\rvert_{H_{d}})\subseteq H_{d}$.
 \end{proof}

To ensure that the solution $\psi(t, u)$ remains within the cone $\cHplus$ for
all $t\geq 0$ and $u\in\cHplus$, we introduce the concept of
\emph{quasi-monotonicity}: 

\begin{lemma}\label{lem:quasi-cont}
   For any $k \in \mathbb{N}$, we define $\dmuk \coloneqq \one_{\{\|\xi\| > 1/k\}} \, \dmu$ and, for each $u \in \dom(L)$, we define the map
   \begin{align*}
     \hat{R}^{(k)}(u) \coloneqq \Gamma^{*}(u) - \int_{\cHpluso} \left( \E^{-\langle \xi, u \rangle} - 1 + \langle \chi(\xi), u \rangle \right) \frac{\dmuk}{\|\xi\|^2}.  
   \end{align*}
   Then, for every $k \in \mathbb{N}$, the map $\hat{R}^{(k)}$ is Lipschitz
   continuous and quasi-monotone with respect to $\cHplus$. Since
   $\hat R^{(k)}(0)=0$, its Lipschitz continuity also gives
   $\|\hat R^{(k)}(u)\|\le K_k\|u\|$ for some constant $K_k<\infty$ and all
   $u\in\cHplus$. That is, for all
   $v_1, v_2 \in \cH$ satisfying $v_1 \leq_{\cHplus} v_2$ and $u \in \cHplus$
   such that $ \langle v_2 - v_1, u \rangle = 0$, we have
   \begin{align*}
   \langle \hat{R}^{(k)}(v_2) - \hat{R}^{(k)}(v_1), u \rangle \geq 0.  
   \end{align*}
\end{lemma}
\begin{proof}
   The Lipschitz continuity of $\hat{R}^{(k)}$ on $\cHplus$ follows from
   \cite[Lemma 3.3]{CKK22a}, and the quasi-monotonicity is established by
   \cite[Lemma 3.2]{CKK22a}. 
\end{proof}

\begin{remark}\label{rem:cone-regularity}
We record here, for use in the proof below, that the cone $\cHplus$ is \emph{regular} in $\cH$: every sequence $(x_n)_{n\in\mathbb{N}}\subset\cHplus$ that is monotone decreasing ($x_{n+1}\leq_{\cHplus} x_n$) and bounded below by $0$ converges strongly in $\cH$. This follows from the general theory of ordered Hilbert spaces: the inner product $\langle\cdot,\cdot\rangle_{\cH}$ is an order-compatible inner product on $\cH$ (i.e.\ $\langle x,y\rangle_{\cH}\geq 0$ for $x,y\in\cHplus$), and the cone $\cHplus$ is self-dual and generating. In this setting, \cite[Theorem~2.1]{Ros91} (see also \cite[Chapter~V, \S7]{Sch74}) implies that bounded monotone sequences converge strongly. In particular, any monotone decreasing sequence in $\cHplus$ that is bounded below by $0\in\cHplus$ must be norm-convergent in $\cH$.
\end{remark}

With these two lemmas in hand, we can now proceed to prove
Proposition~\ref{prop:existence-mild-Riccati}.

\begin{proof}[Proof of Proposition~\ref{prop:existence-mild-Riccati}]
  Let $d \in \mathbb{N}$, $T > 0$, and $u \in \cHplus$.
  We split the proof into a finite-dimensional and an
  infinite-dimensional part. On each invariant subspace $\cH_d$, the
  restricted operator $L_d=\bP_dL\bP_d$ is bounded, so the Galerkin system can
  be treated by the same argument as in the regular case; compare
  \cite[Proposition~2.7]{karbach2023finiterank}. The genuinely new part in the
  present paper is the construction of the mild solution to
  \eqref{eq:Riccati-psi-mild} for the unbounded Lyapunov operator $L$ and the
  derivation of a convergence estimate that is uniform in $d$.
  The uniqueness of the
  solutions $\phi(\cdot, u)$ and $\psi(\cdot, u)$ to equations
  \eqref{eq:Riccati-phi-mild} and \eqref{eq:Riccati-psi-mild}, as well as the
  solutions $\phi_{d}(\cdot,\mathbf{P}_{d}(u))$ and
  $\psi_{d}(\cdot,\mathbf{P}_{d}(u))$ to equations
  \eqref{eq:mild-Riccati-Galerkin-phi} and \eqref{eq:mild-Riccati-Galerkin-psi},
  follows from the fact that the mappings $F$, $\hat{R}$, $F_{d}$, and
  $\hat{R}_{d}$ are Lipschitz continuous on bounded sets, which follows from~\cite[Lemma 4.1]{karbach2023finiterank}.

  Additionally, note that due to the continuity of $F$ and $F^{d}$, given
  solutions $\psi(\cdot, u)$ and $\psi_{d}(\cdot, \mathbf{P}_{d}(u))$ of
  \eqref{eq:Riccati-psi-mild} and \eqref{eq:mild-Riccati-Galerkin-psi},
  respectively, we can define:
   \begin{align*}
      \phi(t, u) = \int_0^t F(\psi(s, u)) \, \D s, \quad t \geq 0,
   \end{align*}
   and
   \begin{align*}
      \phi_{d}(t, \mathbf{P}_{d}(u)) = \int_0^t F^{d}(\psi^{d}(s, \mathbf{P}_{d}(u))) \, \D s, \quad t \geq 0,
   \end{align*}
   which are the unique continuously differentiable solutions to
   \eqref{eq:Riccati-phi-mild} and \eqref{eq:mild-Riccati-Galerkin-phi},
   respectively. Therefore, it remains to show the existence of a solution to
   equations \eqref{eq:Riccati-psi-mild} and
   \eqref{eq:mild-Riccati-Galerkin-psi}.

   By Lemma~\ref{lem:irregular-properties}, the operator $B\rvert_{H_d}$ is
   bounded, hence the Galerkin approximation of the generalized Riccati equations
along the eigenvectors are continuous in the sense that the linear part  $B\rvert_{H_d}$ is bounded. Consequently, the existence of a solution to the Galerkin
   approximation follows directly from the corresponding result in the regular
   affine case, established in \cite[Proposition
   2.7]{karbach2023finiterank}.\par{} 
   
   It remains to prove the existence of a mild solution $\psi(\cdot,u)$ to the
   semi-linear equation~\eqref{eq:Riccati-psi-mild} on $[0,T]$ with the initial
   condition $\psi(0,u)=u\in\cHplus$, satisfying the variation-of-constants
   formula~\eqref{eq:mild-solution-variation}.

   First, by Lemma~\ref{lem:prop-semigroup}, the operator semigroup
   $(\cT(t))_{t \geq 0}$ is strongly continuous with generator $(L, \dom(L))$,
   and satisfies $\cT(t)(\cHplus)\subseteq \cHplus$ for all $t \geq
   0$. Moreover, by Lemma~\ref{lem:quasi-cont}, the function
   $\hat{R}^{(k)}$ is locally Lipschitz continuous on $\cHplus$ and
   quasi-monotone with respect to $\cHplus$. The quasi-monotonicity of
   $\hat{R}^{(k)}$ is equivalent to the condition
   $\lim_{h \to 0+} \inf_{y \in \cHplus} \frac{\norm{x+h\hat{R}^{(k)}(x)-y}}{h}
   = 0$, as shown in \cite[Lemma 4.1 \& Example 4.1]{Dei77}. Therefore, we can
   apply the results of \cite[Chapter 8, Theorem 2.1 \& Remark 2.1]{Mar76},
   which ensures the existence of a $T_1 > 0$ and a mapping
   $\psi^{(k)}(\cdot,u) \colon [0,T_1) \to \cHplus$, which is the unique local
   solution to equation~\eqref{eq:Riccati-psi-mild} when replacing $\hat{R}$ with $\hat R^{(k)}$, satisfying 
   \begin{align}\label{eq:variation-of-constant-k}
  \psi^{(k)}(t,u)=\cT(t)u+\int_{0}^{t}\cT(t-s)\hat{R}^{(k)}(\psi^{(k)}(s,u))\,\D
  s,\quad t\in [0,T_{1}).   
   \end{align}

   Next, we show that this unique local solution $\psi^{(k)}(\cdot,u)$ can be
   extended to $[0,T]$ for any $T > T_1$. To do this, note that the map
   $u \mapsto \hat{R}^{(k)}(u)$ maps bounded sets of $\cHplus$ to bounded sets
   of $\cH$. By \cite[Chapter 8, Theorem 2.1]{Mar76}, it suffices to show that
   $\norm{\psi^{(k)}(t,u)} < \infty$ for every $t \in [0,T_1)$. This implies
   that the solution $\psi^{(k)}(\cdot,u)$ on $[0,T_1)$ can be extended to $[0,T]$,
   thus proving that $\psi^{(k)}(\cdot,u)$ exists on the entire interval
   $[0,T]$, as desired.

   We now check that $\norm{\psi^{(k)}(t,u)}<\infty$ for every $t\in[0,T_{1})$
   holds. Note first that by part i) of Lemma~\ref{lem:irregular-properties}, it
   holds $\be_{i,j}\in\dom(L)$ for every $i,j\in\MN$. The coordinate
   functional $t\mapsto\langle\psi^{(k)}(t,u),\be_{i,j}\rangle$ is differentiable
   on $[0,T_1)$: differentiating the variation-of-constants
   formula~\eqref{eq:variation-of-constant-k} against $\be_{i,j}\in\dom(L)$
   yields
   \[
     \langle\psi^{(k)}(t,u),\be_{i,j}\rangle
     =\E^{-\lambda_{i,j}t}\langle u,\be_{i,j}\rangle
     +\int_0^t \E^{-\lambda_{i,j}(t-s)}\langle\hat R^{(k)}(\psi^{(k)}(s,u)),\be_{i,j}\rangle\,\D s,
   \]
   so the coordinate is absolutely continuous in $t$ (the integrand
   $s\mapsto\langle\hat R^{(k)}(\psi^{(k)}(s,u)),\be_{i,j}\rangle$ is bounded
   by $\|\hat R^{(k)}\circ\psi^{(k)}\|_{C([0,T_1);\cH)}<\infty$ on $[0,T_1)$).
   Differentiating this identity directly gives the displayed ODE.
   Moreover, on $[0,T_{1})$ we have
  \begin{align*}
    \langle \frac{\partial}{\partial t} \psi^{(k)}(t,u), \be_{i,j}\rangle&=\langle
                                                                \psi^{(k)}(t,u),L(\be_{i,j})\rangle+\langle
                                                                 \hat{R}^{(k)}(\psi^{(k)}(t,u)),\be_{i,j}\rangle\\
    &=-\lambda_{i,j}\langle \psi^{(k)}(t,u),\be_{i,j}\rangle+\langle
      \hat{R}^{(k)}(\psi^{(k)}(t,u)),\be_{i,j}\rangle.
  \end{align*}
  Thus, we obtain
  \begin{align*}
    \frac{\partial}{\partial t}\langle \psi^{(k)}(t,u), \be_{i,j}\rangle^{2}&=2\langle
                                                                   \psi^{(k)}(t,u), \be_{i,j}\rangle \langle \frac{\partial}{\partial t}\psi^{(k)}(t,u),\be_{i,j}\rangle \\
                                                                            &=-2\lambda_{i,j}\langle \psi^{(k)}(t,u), \be_{i,j}\rangle^{2}\\
   &\quad + 2\langle \psi^{(k)}(t,u),
      \be_{i,j}\rangle \langle \hat{R}^{(k)}(\psi^{(k)}(t,u)),\be_{i,j}\rangle,
  \end{align*}
  and hence by the variation-of-constant formula we see that
  \begin{align*}
    \langle \psi^{(k)}(t,u), \be_{i,j}\rangle^{2}&=\E^{-2\lambda_{i,j}t}\langle u,
                                                   \be_{i,j}\rangle^{2}\\
    &\quad+2\int_{0}^{t}\E^{-2\lambda_{i,j}(t-s)}\langle \psi^{(k)}(s,u),
      \be_{i,j}\rangle \langle \hat{R}^{(k)}(\psi^{(k)}(s,u)),\be_{i,j}\rangle \D s. 
  \end{align*}
Given that 
\begin{align*}
\norm{\psi^{(k)}(t,u)}^{2} = \sum_{i,j\in\mathbb{N}, i \leq j} \langle \psi^{(k)}(t,u), \be_{i,j} \rangle^{2}  
\end{align*}
and
\begin{align*}
\langle \psi^{(k)}(t,u), \hat{R}^{(k)}(\psi^{(k)}(t,u)) \rangle = \sum_{i,j\in\mathbb{N}, i \leq j} \langle \psi^{(k)}(t,u), \be_{i,j} \rangle \langle \hat{R}^{(k)}(\psi^{(k)}(t,u)), \be_{i,j} \rangle,  
\end{align*}
it follows from the estimate $\norm{\hat{R}^{(k)}(u)} \leq K_k \norm{u}$ from
Lemma~\ref{lem:quasi-cont} that: 
\begin{align*}
\norm{\psi^{(k)}(t,u)}^{2} \leq \norm{u}^{2} + M_k \int_0^t \norm{\psi^{(k)}(s,u)}^{2} \, \D s,  
\end{align*}
for some constant $M_k > 0$ (depending on the fixed truncation level $k$). Thus,
by applying Gronwall's inequality, we conclude that for all $t \in [0, T_1)$:
\begin{align*}
  \norm{\psi^{(k)}(t,u)}^{2} \leq \norm{u}^{2}\E^{M_k t}<\infty.   
\end{align*}

Therefore, $\norm{\psi^{(k)}(t,u)} < \infty$ for every $t \in [0,T_1)$. Hence,
following \cite[Chapter 8, Proposition 4.1]{Mar76}, we conclude that there
exists a unique solution $\psi^{(k)}(\cdot,u)$ to
equation~\eqref{eq:Riccati-psi-mild} on the entire interval $[0,T]$ given
by~\eqref{eq:variation-of-constant-k}.

It follows by the same reasoning as in \cite[Proposition~3.7]{CKK22a} that
\[
  \psi^{(k+1)}(t,u)\;\le_{\cHplus}\; \psi^{(k)}(t,u),
  \qquad t\in[0,T],\ u\in\cHplus.
\]
Hence, for fixed $t\in[0,T]$ and $u\in\cHplus$, the sequence
$\bigl(\psi^{(k)}(t,u)\bigr)_{k\in\mathbb N}$ is monotone decreasing with respect
to the partial order induced by the cone $\cHplus$ and bounded from below by $0$.
Since the cone $\cHplus$ is regular, i.e.\ every monotone and order-bounded
sequence in $\cHplus$ converges strongly in $\cH$, there exists for every
$t\in[0,T]$ and $u\in\cHplus$ a limit
\[
  \psi(t,u) \df \lim_{k\to\infty} \psi^{(k)}(t,u)
  \quad \text{in } \cH.
\]
By the same reasoning as in \cite[Lemma~3.6]{CKK22a}, the truncated mappings
$\hat{R}^{(k)}$ converge to $\hat{R}$ uniformly on bounded subsets of $\cHplus$. Thus by the monotonicity of the sequence $ \psi^{(k)}(t,u)$ in
$k\in\MN$ and by inserting formula~\eqref{eq:variation-of-constant-k}, we obtain
  \begin{equation*}
  \begin{aligned}
   \psi(t,u) &= \lim_{k\rightarrow \infty} \psi^{(k)}(t,u)
   \\
   & = \cT(t)u+\lim_{k\rightarrow
     \infty}\int_{0}^{t}\cT(t-s)\hat{R}^{(k)}(\psi^{(k)}(s,u))\,\D s
   \\
   & = \cT(t)u + \lim_{k\rightarrow \infty} \int_{0}^{t} \cT(t-s)\bigl( \hat{R}^{(k)}(\psi^{(k)}(s,u)) - \hat{R}(\psi^{(k)}(s,u))\bigr)\D s
   \\ &
   \quad + \lim_{k\rightarrow \infty} \int_{0}^{t} \cT(t-s)\hat{R}(\psi^{(k)}(s,u))\D s
   \\
   & = \cT(t)u + \int_{0}^{t} \cT(t-s)\hat{R}(\psi(s,u))\D s.
   \end{aligned}
  \end{equation*}
  To justify the interchange of the limit and the integral in the last equality, we argue as follows. By monotonicity in the cone, $0\le_{\cHplus}\psi^{(k)}(s,u)\le_{\cHplus}\psi^{(1)}(s,u)$ for all $k$ and $s\in[0,T]$; since the cone is self-dual, each order interval is norm bounded, and the Gronwall estimate for the fixed truncation $k=1$ gives $\sup_{s\in[0,T]}\|\psi^{(1)}(s,u)\|<\infty$. Hence there exists $\tilde{M}_{T}>0$ such that $\sup_{k\in\mathbb{N}}\sup_{s\in[0,T]}\|\psi^{(k)}(s,u)\|\leq \tilde{M}_{T}$. On the ball $\{v\in\cHplus:\|v\|\le\tilde M_T\}$, local Lipschitz continuity and $\hat R(0)=0$ imply that $\hat R$ is bounded, and the uniform convergence $\hat R^{(k)}\to\hat R$ on bounded subsets of $\cHplus$ (from \cite[Lemma~3.6]{CKK22a}) gives a uniform bound for $\hat R^{(k)}$ on the same ball, after absorbing finitely many initial $k$. Thus $\|\cT(t-s)\hat{R}^{(k)}(\psi^{(k)}(s,u))\|\leq C_T$ uniformly in $k$ and $s\in[0,t]$. The dominated convergence theorem therefore justifies passing the limit through the integral, using pointwise convergence $\psi^{(k)}(s,u)\to\psi(s,u)$ for each $s$ and the same uniform convergence of $\hat R^{(k)}$ on bounded subsets.

  It remains to record that the limiting path is continuous. Set
  $g(s)\df\hat R(\psi(s,u))$. The preceding uniform bound and the linear growth
  of $\hat R$ imply $\sup_{s\le T}\|g(s)\|<\infty$. Let $t_n\to t$ in
  $[0,T]$. Strong continuity of $(\cT(r))_{r\ge0}$ gives
  $\cT(t_n)u\to\cT(t)u$. For the integral term, assume first $t_n\ge t$; the
  other case is identical. Then
  \begin{align*}
    &\left\|
      \int_0^{t_n}\cT(t_n-s)g(s)\,\D s
      -\int_0^t\cT(t-s)g(s)\,\D s
    \right\|\\
    &\quad\le
    \int_0^t\|(\cT(t_n-s)-\cT(t-s))g(s)\|\,\D s
    +\int_t^{t_n}\|g(s)\|\,\D s .
  \end{align*}
  The first term tends to zero by dominated convergence, using contraction of
  $\cT$ and pointwise strong continuity; the second tends to zero by boundedness
  of $g$. Thus $t\mapsto\psi(t,u)$ is continuous as an $\cH$-valued map. Since
  $\cHplus$ is closed, $\psi(\cdot,u)\in C([0,T],\cHplus)$.
  This proves the existence of a unique solution to
  equation~\eqref{eq:Riccati-psi-mild} satisfying the variation-of-constants
  formula~\eqref{eq:mild-solution-variation}.\par{}

This is the only place where the irregular setting genuinely differs
from the regular case: the bounded-generator ODE argument from
\cite{karbach2023finiterank} must here be replaced by a semilinear mild
equation argument for the unbounded Lyapunov operator $L$. Once the mild
solution is available, the convergence estimate follows by the same perturbation
strategy as in the regular case.

Lastly, we prove the uniform convergence stated
in~\eqref{eq:convergence-mild-Galerkin}.
First, note that the semigroup $(\cT(t))_{t \geq 0}$ commutes with the projection
$\bP_d$ for every $d \in \mathbb{N}$, since both operators are diagonal with
respect to the orthonormal eigenbasis $\{\be_{i,j}\}_{i\le j}$ of the generator
$L$, and $\bP_d$ projects onto the invariant subspace
$\cH_d=\mathrm{span}\{\be_{i,j}:1\le i\le j\le d\}$.

At this point one can follow the same perturbation argument as in the
proof of \cite[Proposition~2.7]{karbach2023finiterank}, with the bounded
semigroup $\mathrm e^{tB^*}$ there replaced here by the Lyapunov semigroup
$(\cT(t))_{t\geq 0}$ and using the commutation relation
$\bP_d\cT(t)=\cT(t)\bP_d$.
By setting $M_T \df \sup_{t\in[0,T]}\|\cT(t)\|_{\cL(\cH)}$, we obtain for
$t\in[0,T]$ and $u\in\cHplus$ the estimate
\begin{align*}
  \|\psi(t,u)-\psi_{d}(t,\bP_{d}u)\|
  &\le \|\cT(t)\,u-\cT(t)\bP_d u\|
  \\
  &\quad +\int_0^t \Big\|\cT(t-s)\Big(\hat R(\psi(s,u))-\hat R_d(\psi_d(s,\bP_d u))\Big)\Big\|\,\D s \\
  &\le M_T \| \bP_d^\perp u\|
      + M_T \int_0^t \Big\|\hat R(\psi(s,u))-\hat R_d(\psi_d(s,\bP_d u))\Big\|\,\D s.
\end{align*}
Next, we split the difference by adding and subtracting $\bP_d\hat R(\psi(s,u))$:
\begin{align*}
  \Big\|\hat R(\psi(s,u))-\hat R_d(\psi_d(s,\bP_d u))\Big\|
  &\le \Big\|\hat R(\psi(s,u))-\bP_d\hat R(\psi(s,u))\Big\| \\
  &\quad + \Big\|\bP_d\hat R(\psi(s,u))-\hat R_d(\psi_d(s,\bP_d u))\Big\|.
\end{align*}
Combining the two displays yields
\begin{align*}
  \|\psi(t,u)-\psi_{d}(t,\bP_{d}u)\|
  &\le M_T \|\bP_d^\perp u\|
   + M_T\int_0^t \Big\|\hat R(\psi(s,u))-\bP_d\hat R(\psi(s,u))\Big\|\,\D s \\
  &\quad + M_T\int_0^t \Big\|\bP_d\hat R(\psi(s,u))-\hat R_d(\psi_d(s,\bP_d u))\Big\|\,\D s .
\end{align*}

  We claim that there exists $M>0$, independent of $d\in\mathbb{N}$, such that $\|\psi(t,u)\|\leq M$ and $\|\psi_d(t,\bP_d(u))\|\leq M$ for all $d\in\mathbb{N}$ and $t\in[0,T]$. For $\psi(\cdot,u)$ this follows from the Gronwall bound established above. For $\psi_d(\cdot,\bP_d(u))$, the same Gronwall argument applies uniformly in $d$: the local Lipschitz constant of $\hat{R}_d=\bP_d\hat{R}\bP_d$ on bounded sets is bounded by the same constant as that of $\hat{R}$ (since $\|\bP_d\|\leq 1$), and the initial data $\|\bP_d(u)\|\leq\|u\|$ is uniformly bounded. Thus $M$ can be chosen independently of $d$. Now, let $M>0$ be as above, then by the local
  Lipschitz property of $\hat{R}$, we find a constant $K_{M}>0$ such that
  \begin{align*}
   \norm{\bP_{d}(\hat{R}(\psi(s,u)))-\hat{R}_{d}(\psi_{d}(s,\bP_{d}(u)))}\leq
    K_{M}\norm{\psi(s,u)-\psi_{d}(s,\bP_{d}(u))}, 
  \end{align*}
  and thus by setting
  \begin{align*}
   K^{d}_{t}=M_{T}\Big(\norm{\bP_{d}^{\perp}(u)}+\int_{0}^{t}\norm{\hat{R}(\psi(s,u))-\bP_{d}\hat{R}(\psi(s,u))}\,\D
    s\Big), 
  \end{align*}
  we obtain
  \begin{align*}
    \norm{\psi(t,u)-\psi_{d}(t,\bP_{d}(u))}\leq K_{t}^{d}+M_{T}K_{M}\int_{0}^{t}\norm{\psi(s,u)-\psi_{d}(s,\bP_{d}(u))}\,\D
    s.
  \end{align*}
  
  Again, by an application of Gronwall's inequality, we find that
  \begin{align}\label{eq:Riccati-convergence-3}
    \norm{\psi(\cdot,u)-\psi_{d}(\cdot,\bP_{d}(u))}\leq K_{t}^{d}\E^{M_{T} K_{M} t}. 
  \end{align}
  Hence taking the supremum over all $t\in [0,T]$ on both sides
  of~\eqref{eq:Riccati-convergence-3} yields
  \begin{align}\label{eq:Riccati-convergence-sup}
    \sup_{t\in [0,T]}\norm{\psi(t,u)-\psi_{d}(t,\bP_{d}(u))}\leq
    K_{T}^{d}\E^{M_{T} K_{M} T}.  
  \end{align}
  
  Similarly, to the proof~\cite[Proposition 3.1]{karbach2023finiterank} we can
  then also bound the Galerkin approximation error of the first
  equation~\eqref{eq:Riccati-phi-mild} as 
  \begin{align}\label{eq:phi-error}
    \sup_{t\in [0,T]}|\phi(t,u)-\phi_{d}(t,\bP_{d}(u))|&\leq K_{2,M}K_{T}^{d}\E^{M_{T} K_{M} T },
                                       \quad t\in [0,T], u\in\cHplus.
  \end{align}
Combining \eqref{eq:Riccati-convergence-sup} and \eqref{eq:phi-error}, and
using that $\norm{\bP_d^\perp z}\to0$ for every $z\in\cH$, we conclude that the
right-hand sides tend to zero as $d\to\infty$. This proves
\eqref{eq:convergence-mild-Galerkin}. Thus the approximation argument
has the same perturbative structure as in the regular case, but the
convergence is expressed here qualitatively at the level of the projected mild
Riccati solutions.
\end{proof}

\begin{lemma}[Continuity of the mild Riccati flow in the initial datum]
\label{lem:riccati-initial-continuity}
  Let $u_n,u\in\cHplus$ with $\|u_n-u\|\to0$. Then for every
  $T<\infty$,
  \[
    \sup_{t\in[0,T]}
    \left(
      |\phi(t,u_n)-\phi(t,u)|
      +\|\psi(t,u_n)-\psi(t,u)\|
    \right)
    \to0 .
  \]
\end{lemma}
\begin{proof}
  Since $u_n\to u$, the set $\{u\}\cup\{u_n:n\in\MN\}$ is bounded in $\cH$.
  The Gronwall estimate used in the proof of
  Proposition~\ref{prop:existence-mild-Riccati} therefore gives a constant
  $M_T<\infty$ such that
  \[
    \sup_n\sup_{t\le T}\|\psi(t,u_n)\|
    +\sup_{t\le T}\|\psi(t,u)\|\le M_T .
  \]
  Let $L_T$ be a Lipschitz constant of $\hat R$ on
  $\{z\in\cHplus:\|z\|\le M_T\}$. Subtracting the mild equations and using
  contractivity of $\cT$ gives
  \[
    \|\psi(t,u_n)-\psi(t,u)\|
    \le
    \|u_n-u\|
    +L_T\int_0^t\|\psi(s,u_n)-\psi(s,u)\|\,\D s .
  \]
  Gronwall's lemma yields
  \[
    \sup_{t\le T}\|\psi(t,u_n)-\psi(t,u)\|
    \le \|u_n-u\|\,\E^{L_TT}\to0 .
  \]
  The functions $F$ are locally Lipschitz on $\cHplus$ by
  \cite[Remark~3.4]{CKK22a}; with $L_T^F$ a Lipschitz constant on the same
  bounded ball,
  \[
    \sup_{t\le T}|\phi(t,u_n)-\phi(t,u)|
    \le
    L_T^F\int_0^T\|\psi(s,u_n)-\psi(s,u)\|\,\D s\to0 .
  \]
  This proves the claim.
\end{proof}

\subsection{Affine Finite-Rank Approximations and Their Weak Convergence}
\label{sec:finite-rank-convergence}

In this section, we provide the proofs for
Proposition~\ref{prop:embedding-affine-main} and
Theorem~\ref{thm:main-convergence}. We build on the finite-rank approximation
framework for \emph{regular affine processes} introduced in
\cite{karbach2023finiterank}, making the necessary adjustments to handle the
\emph{irregular case}.

\subsubsection{Finite-Rank Admissible Parameters}
First, we introduce finite-rank approximations for the irregular admissible
parameter set $(b, \mathbf{B}, m, \mu)$. Then, we demonstrate that a finite-rank
affine process can be associated with this approximate parameter set.

We use the following notation: for any two measurable spaces $(\cE_1, \cB_1)$
and $(\cE_2, \cB_2)$, and a measurable function $f: (\cE_1, \cB_1) \to (\cE_2,
\cB_2)$, we denote the push-forward of a measure $\mu_1: \cE_1 \to [0, \infty]$
with respect to $f$ as $f_* \mu_1$. Specifically, we have $f_*\mu_1(A) =
\mu_1(f^{-1}(A))$ for any $A \in \cB_2$. Note that $f_*\mu_1$ is a proper
measure on $(\cE_2, \cB_2)$. For every $d \in \mathbb{N}$, we define the Borel sets
$$E_d \coloneqq \{\xi \in \cHplus \colon 0 < \|\bP_d(\xi)\| \leq 1, \|\xi\| > 1\},$$
and $E_d^0 \coloneqq \cHplus \cap \{\xi \colon \bP_d(\xi) \neq 0\}$. Note that
$E_d \subseteq E_d^0 \subseteq \cHpluso$ for every $d \in \mathbb{N}$. For any
(vector-valued) measure $\lambda$ on $\cB(\cHpluso)$, the Borel $\sigma$-algebra
on $\cHpluso$, we denote the restriction of $\lambda$ to the trace of the Borel
$\sigma$-algebra generated by the open sets in $E_d^0$ as
$\lambda\vert_{E_d^0}$. We now introduce the following notion, which differs from the corresponding definition in~\cite[Definition
5.1]{karbach2023finiterank} only in the part iv):
 
\begin{definition}\label{def:admissible-Galerkin}
  For every $d\in\MN$, and admissible parameter set $(b, \mathbf{B}, m, \mu)$ we
  define the parameters $(b_{d},\mathbf{B}_{d},m_{d},\mu_{d})$ and $M_{d}$ as follows:
  \begin{defenum}
  \item\label{item:md} The measure $m_{d}\colon \cB(\cHplus_{d}\setminus\set{0})\to
    [0,\infty]$ is defined as the push-forward of $m\vert_{E_{d}^{0}}$ with
    respect to $\bP_{d}$, i.e.
    \begin{align*}
      m_{d}(\D\xi)\df (\bP_{d\,*}m\vert_{E_{d}^{0}})(\D\xi). 
    \end{align*}
  \item The vector \label{item:bd} $b_{d}\in \cH_{d}$ is given by
    \begin{align}\label{eq:bd}
      b_{d}\df \bP_{d}(b)+\int_{\bP_{d}(E_{d})}\xi\,m_{d}(\D\xi). 
    \end{align}
  \item\label{item:mud} The $\cHplus_{d}$-valued measure $\mu_{d}\colon
    \cB(\cHplus_{d}\setminus\set{0})\to \cHplus_{d}$ is defined as the
    $\bP_{d}$-projection of the push-forward of $\mu\vert_{E_{d}^{0}}$ with respect to
    $\bP_{d}$, i.e.
    \begin{align*}
      \mu_{d}(\D\xi)\df \bP_{d}(\bP_{d\,*}\mu\vert_{E_{d}^{0}})(\D\xi).
    \end{align*}
    Moreover, we define the $\cHplus_{d}$-valued measure $M_{d}$ on
    $\cHplus_{d}\setminus\set{0}$ as follows: For every
    $A\in\cB(\cHplus_{d}\setminus\set{0})$ we set
    \begin{align}\label{eq:Md} 
      M_{d}(A)\df
      \int_{E_{d}^{0}}\one_{A}(\bP_{d}(\xi))\frac{1}{\norm{\xi}^{2}}\bP_{d}(\mu\vert_{E_{d}^{0}}(\D\xi)).
    \end{align}
\item[iv)] The linear operator $\mathbf{B}_{d}\colon \cH_{d}\to\cH_{d}$ is given by
    \begin{align}\label{eq:Bd-irregular}
      \quad\mathbf{B}_{d}(y)&\df
      \bP_{d}(\mathbf{B}(y))+\int_{\bP_{d}(E_{d})}\xi\,
      \langle y,M_{d}(\D\xi)\rangle\nonumber\\
                              &=B y+yB^{*}+\Gamma_{d}(y)+
                              \int_{\bP_{d}(E_{d})}\xi\,
      \langle y,M_{d}(\D\xi)\rangle,\qquad y\in\cH_d,
    \end{align} 
    The displayed formula defines the linear extension on all of
    $\cH_d$; the admissibility checks below are applied on the cone
    $\cH_d^+$.
    Here $\Gamma_{d}\colon \cHplus_{d}\to\cH_{d}$ is given by
    \begin{align}\label{eq:kGamma-d}
     \Gamma_{d}(y)\df\bP_{d}(\Gamma(y)). 
    \end{align}
  \end{defenum}
\end{definition}

The idea behind the approximate parameter set is that we can use the parameters
$b_d$, $m_d$, $M_d$, and $\mathbf{B}_d$ to express the functions $F_d$ and
$\hat{R}_d$ as follows: 
\begin{align}
   F_d(u) &= \langle b_d, u \rangle - \int_{\cHplus_d \setminus \{0\}} \left( \E^{-\langle \xi, u \rangle} - 1 + \langle \chi(\xi), u \rangle \right) m_d(\D \xi), \label{eq:F-Galerkin}
\end{align}
and by denoting $K(\xi, u)\df\E^{-\langle
\xi,u\rangle}-1+\langle\chi(\xi),u\rangle$ and writing
$\Gamma_d^*(u_d)\df\bP_d(\Gamma^*(u_d))$ for the adjoint action on
$\cH_d$, we can express $\hat{R}_d$ for every $u_d \in
\cHplus_d$ as
\begin{align*}
   \hat{R}_d(u_d) &\coloneqq \bP_d(\Gamma^{*}(u_d)) - \int_{\cHpluso} \left( \E^{-\langle \xi, u_d \rangle} - 1 + \langle \chi(\xi), u_d \rangle \right) \bP_d(M(\D \xi)) \\
                  &= \Gamma_d^*(u_d) + \int_{\bP_d(E_d)} \langle \xi, u_d \rangle \, M_d(\D \xi) - \int_{\cHplus_d \setminus \{0\}} K(\xi, u_d) \, M_d(\D \xi).
\end{align*}

\subsubsection{Matrix-Valued Admissible Parameters}
Next, we relate the finite-rank approximate admissible parameters $b_d$, $m_d$,
$M_d$, and $\mathbf{B}_d$ to their matrix-valued counterparts. This is necessary
in order to apply the existence results for matrix-valued affine processes
from~\cite{CFMT11}. 
 
For $d \in \mathbb{N}$, we denote by $\{v_1, \ldots, v_d\}$ the standard basis
of $\mathbb{R}_d$, and define the coordinate system $\Xi_d \colon
\mathbb{R}_d \to H_d$ associated with the basis $\{e_1, \ldots, e_d\}$ of $H_d$ by
\begin{align}\label{eq:coordinate-system}
  \Xi_d(v_i) = e_i, \quad \text{for } i = 1, \ldots, d.
\end{align}
(We use the symbol $\Xi_d$ here rather than the more natural
$\Phi_d$ in order to avoid a clash with the joint Galerkin cumulant
$\Phi_d(t,u)$ from~\eqref{eq:stochastic-covariance-affine-formula}.)
The coordinate system $\Xi_d$ identifies the $d$-dimensional subspace $H_d$ with $\mathbb{R}_d$, allowing us to represent every linear operator $A \in \mathcal{L}(H_d)$ as a $d \times d$-matrix through the mapping $i_d \colon \mathcal{L}(H_d) \to \mathbb{S}_d$, defined by
\begin{align}\label{eq:canonical-identification-d}
  i_d(A) \coloneqq \Xi_d^{-1} \circ A \circ \Xi_d,
\end{align}
where, under the usual matrix identification, we interpret $i_d(A)$ as an element in $\mathbb{M}_d$. Note that when $A$ is self-adjoint, its matrix representation $i_d(A)$ is also self-adjoint. This can be verified by taking $x, y \in \mathbb{R}_d$ and performing the following computation:
\begin{align*}
  (i_d(A)x, y)_{\mathbb{R}^d} = (A \circ \Xi_d(x), \Xi_d(y))_H = (\Xi_d(x), A^* \circ \Xi_d(y))_H = (x, i_d(A)y)_{\mathbb{R}^d}.
\end{align*}
Under the mapping $i_d$ in~\eqref{eq:canonical-identification-d}, we identify
$\mathcal{H}_d\vert_{H_d} \subseteq \mathcal{L}(H_d)$ with $\mathbb{S}_d$ and
note that $i_d$ is an isometry between $\mathbb{S}_d$ and
$\mathcal{H}_d\vert_{H_d}$, meaning that it identifies the Frobenius norm with
the Hilbert-Schmidt norm. In the following, we may sometimes omit writing the
restriction $\rvert_{H_d}$ when it is clear from the context. Moreover, we
denote by $\mathbb{S}_d^+$ the convex cone of all symmetric positive
semi-definite $d \times d$-matrices, and observe that positivity is preserved
under $i_d$, i.e. $i_d(\mathcal{H}_d^+) = \mathbb{S}_d^+$.

\begin{definition}\label{def:admissible-matrix}
    Let $(b,\mathbf{B},m,\mu)$ be an admissible parameter set as in
  Definition~\ref{def:admissible-irregular} and for $d\in\MN$ let
  $(b_{d},\mathbf{B}_{d},m_{d},\mu_{d})$ and $M_{d}$ be as in
  Definition~\ref{def:admissible-Galerkin}. For every $d\in\MN$ we define the parameters 
  $(\tilde{b}_{d},\mathbf{\tilde{B}}_{d},\tilde{m}_{d},\tilde{\mu}_{d})$ and $\tilde{M}_{d}$ as follows:
  \begin{defenum}
  \item\label{item:b-d} The matrix $\tilde{b}_{d}\in\MS_{d}^{+}$ is
    defined as $\tilde{b}_{d}\df i_{d}(b_{d})$.
  \item\label{item:B-d} 
  Define the linear operator
  $\mathbf{\tilde{B}}_{d}\colon \MS_{d}\to\MS_{d}$ by
 \begin{align}\label{eq:B-d}
 \mathbf{\tilde{B}}_{d}\df i_{d}\circ\mathbf{B}_{d}\circ i^{-1}_{d}.
 \end{align}
  \item\label{item:m-d} The measure $\tilde{m}_{d}\colon\cB(\MS_{d}^{+}\setminus\set{0})\to
    [0,\infty]$ is defined as the push-forward of $m_{d}$ with respect to $i_{d}$, i.e.
    \begin{align*}
     \tilde{m}_{d}(\D\xi)\df (i_{d\,*}m_{d})(\D\xi). 
    \end{align*}
  \item\label{item:mu-d} The matrix-valued measure $\tilde{\mu}_{d}\colon
    \cB(\MS_{d}^{+}\setminus\set{0})\to \MS_{d}^{+}$ is defined as the
    composition of $i_{d}$ and the push-forward of $\mu_{d}$ with respect to
    $i_{d}$, i.e. 
    \begin{align*}
      \tilde{\mu}_{d}(\D\xi)=i_{d}((i_{d\,*}\mu_{d})(\D\xi)). 
    \end{align*}
    Moreover, we define the $\MS_{d}^{+}$-valued measure $\tilde{M}_{d}(\D\xi)$
    as follows: For every
    $A\in\cB(\MS_{d}^{+}\setminus\set{0})$ we set 
  \begin{align*} 
    \tilde{M}_{d}(A)=\int_{\cHpluso}\one_{A}(i_{d}(\bP_{d}(\xi)))\frac{1}{\norm{\xi}^{2}}i_{d}(\bP_{d}(\mu(\D\xi))),
  \end{align*}
  and for every $x\in\MS_{d}^{+}$ we write $\tilde{M}_{d}(x,\D\xi)\df\langle x,\tilde{M}_{d}(\D\xi)\rangle$.
\end{defenum}
\end{definition}

With the definition of the matrix-valued admissible parameters at hand we can
now prove Proposition~\ref{prop:embedding-affine-main}. 

\begin{proof}[Proof of Proposition~\ref{prop:embedding-affine-main}]

Throughout this proof we use the following notation, consistent with the
identification $i_d:\cL(H_d)\to\MS_d$ introduced in
\eqref{eq:canonical-identification-d}. For matrices $x,u\in\MS_d$, we write
$x_d\df i_d^{-1}(x)\in\cH_d\subseteq\cH$ and $u_d\df i_d^{-1}(u)\in\cH_d\subseteq\cH$
for their operator-valued counterparts on $H_d$; the Frobenius pairing on $\MS_d$
is denoted
$\langle x,u\rangle_d\df \mathrm{Tr}(xu)$. Since $i_d$ is an isometry between $\MS_d$
(with Frobenius inner product) and $\cH_d\rvert_{H_d}$ (with Hilbert--Schmidt inner
product inherited from $\cH$), we have $\langle x,u\rangle_d=\langle
x_d,u_d\rangle_{\cH}$ for all $x,u\in\MS_d$. We also write
$\chi_d(\zeta)\df\zeta\one_{\{\|\zeta\|_d\le1\}}$ for the matrix truncation.

  Let $d \in \mathbb{N}$ and $\tilde{b}_d$, $\mathbf{\tilde{B}}_{d}$, $\tilde{m}_d$,
  $\tilde{\mu}_d$, and $\tilde{M}_d$ be as defined in parts (i) to (iv) of
  Definition~\ref{def:admissible-matrix}. The positivity and
  integrability conditions for $\tilde b_d$, $\tilde m_d$, $\tilde\mu_d$, and
  $\tilde M_d$ are inherited from the projected parameters exactly as in
  \cite[Lemma~5.6 and Proposition~5.7]{karbach2023finiterank}. These checks
  involve only positivity of the cone $\cHplus$, the orthogonal projection
  $\bP_d$, the pushforward of $m$ and $\mu$ under $i_d\circ\bP_d$, and the
  moment and integrability assumptions from
  Definition~\ref{def:admissible-irregular}; the unboundedness of the full
  Lyapunov operator $L$ on $\cH$ plays no role here, since in the
  finite-rank problem the linear drift enters only through the restriction
  $B\rvert_{H_d}$, which is bounded on
  $H_d=\operatorname{span}\{e_1,\dots,e_d\}$ by
  Lemma~\ref{lem:irregular-properties}. The remaining
  point is the matrix admissibility condition for the linear drift. Namely, for
  all $x,u\in\MS_d^+$ such that $\langle x,u\rangle_d=0$, we have to prove
  \begin{align}\label{eq:admissible-matrix-check}
    \langle \mathbf{\tilde{B}}_{d}(x), u \rangle_d
    - \int_{\mathbb{S}_d^+ \setminus \{0\}}
      \langle \chi_d(\zeta), u \rangle_d\,\tilde{M}_d(x, \D \zeta) \geq 0.
  \end{align}

  To prove this, let $x, u \in \mathbb{S}_d^+$ such that
  $\langle x, u \rangle_d = 0$. The orthogonality
  $\langle x,u\rangle_d=0$ transfers to the operator level: since $i_d$
  identifies the Frobenius and Hilbert--Schmidt inner products,
  $\langle x_d,u_d\rangle_{\cH}=\langle x,u\rangle_d=0$ with
  $x_d,u_d\in\cH_d^+\subseteq\cHplus$. Since $x_d$ and $u_d$ are positive
  finite-rank operators and $\Tr(x_du_d)=0$, we have $x_du_d=u_dx_d=0$. As
  $H_d$ is invariant under $B$ by construction, the Lyapunov part therefore
  satisfies
  \begin{align*}
    \langle Lx_d,u_d\rangle_{\cH}
    =\Tr(Bx_du_d)+\Tr(B^*u_dx_d)=0.
  \end{align*}
  Moreover, by the isometry property of $i_d$,
  $\langle\Gamma_d(x_d),u_d\rangle_{\cH}
  =\langle\bP_d\Gamma(x_d),u_d\rangle_{\cH}
  =\langle\Gamma(x_d),u_d\rangle_{\cH}$, because $u_d\in\cH_d$ and
  $\bP_d$ is self-adjoint. For the jump integral we use the
  pushforward identity for $\tilde M_d$: for every bounded Borel
  $g\colon\MS_d^+\to\MR$,
  \begin{align*}
    \int_{\MS_d^+\setminus\{0\}} g(\zeta)\,\tilde M_d(x,\D\zeta)
    &= \int_{\MS_d^+\setminus\{0\}} g(\zeta)\bigl\langle x,\tilde
      M_d(\D\zeta)\bigr\rangle_d\\
    &= \int_{\cHpluso} g(i_d(\bP_d(\eta)))\,
      \frac{\bigl\langle
      x_d,\bP_d(\mu(\D\eta))\bigr\rangle_{\cH}}{\|\eta\|^{2}}
      \\
    &= \int_{\cHpluso} g(i_d(\bP_d(\eta)))\,
      \frac{\langle\mu(\D\eta),x_d\rangle_{\cH}}{\|\eta\|^{2}},
  \end{align*}
  where the second equality is the definition of $\tilde M_d$ in
  Definition~\ref{def:admissible-matrix}(iv) and the third uses
  $\bP_d x_d=x_d$ together with self-adjointness of $\bP_d$. Specialising
  this identity to $g(\zeta)=\langle\chi_d(\zeta),u\rangle_d$ and using the
  standard radial truncation $\chi_d(\zeta)=\zeta\,\mathbf 1_{\{\|\zeta\|_d\le 1\}}$
  together with the isometry identity
  $\langle i_d(\bP_d\eta),u\rangle_d=\langle\bP_d\eta,u_d\rangle_{\cH}=\langle\eta,u_d\rangle_{\cH}$
  (by $u_d\in\cH_d$ and self-adjointness of $\bP_d$), we obtain
  \begin{align*}
    \int_{\mathbb{S}_d^+ \setminus \{0\}} \langle \chi_d(\zeta), u \rangle_d
    \tilde{M}_d(x, \D \zeta)
    = \int_{\cHpluso}
      \langle\eta,u_d\rangle_{\cH}\,
      \one_{\{0<\|\bP_d\eta\|\le 1\}}\,
      \frac{\langle \mu(\D \eta), x_d \rangle_{\cH}}{\|\eta\|^2}.
  \end{align*}
  The compensation term built into $\mathbf B_d$ gives, again by the
  same pushforward calculation,
  \begin{align*}
    \left\langle
      \int_{\bP_d(E_d)}\zeta\,\langle x_d,M_d(\D\zeta)\rangle,
      u_d
    \right\rangle_{\cH}
    &=
    \int_{E_d}
      \langle\eta,u_d\rangle_{\cH}
      \frac{\langle\mu(\D\eta),x_d\rangle_{\cH}}{\|\eta\|^2}.
  \end{align*}
  Combining the preceding identities with the definition of
  $E_d=\{\eta\in\cHplus:0<\|\bP_d\eta\|\le1,\ \|\eta\|>1\}$ yields
  \begin{align*}
    &\langle \mathbf{\tilde{B}}_{d}(x), u \rangle_d
    - \int_{\mathbb{S}_d^+ \setminus \{0\}}
      \langle \chi_d(\zeta), u \rangle_d\,\tilde{M}_d(x, \D \zeta)
    \\
    &\quad =
    \langle \Gamma(x_d),u_d\rangle_{\cH}
    +\int_{\cHpluso}
      \langle\eta,u_d\rangle_{\cH}
      \bigl(\one_{E_d}(\eta)-\one_{\{0<\|\bP_d\eta\|\le1\}}(\eta)\bigr)
      \frac{\langle\mu(\D\eta),x_d\rangle_{\cH}}{\|\eta\|^2}
    \\
    &\quad =
    \langle \Gamma(x_d),u_d\rangle_{\cH}
    -\int_{\cHpluso}
      \langle\chi(\eta),u_d\rangle_{\cH}
      \frac{\langle\mu(\D\eta),x_d\rangle_{\cH}}{\|\eta\|^2}
    \ge 0.
  \end{align*}
  In the last equality we used
  $\langle\eta,u_d\rangle_{\cH}=\langle\bP_d\eta,u_d\rangle_{\cH}$, so the
  integrand vanishes on $\{\bP_d\eta=0\}$, and on
  $\{\bP_d\eta\ne0\}$ the identity
  $\one_{E_d}-\one_{\{0<\|\bP_d\eta\|\le1\}}
  =-\one_{\{0<\|\eta\|\le1\}}$ holds. The final inequality is exactly
  Definition~\ref{def:admissible-irregular}\cref{item:linear-operator-unbounded}(b)
  applied to $x_d,u_d\in\cHplus$ with
  $\langle x_d,u_d\rangle_{\cH}=0$. This proves
  \eqref{eq:admissible-matrix-check}.\par{}

  Together with the already verified finite-dimensional positivity and
  integrability conditions, this shows that the corresponding matrix-valued
  parameter tuple is admissible in the sense of \cite[Definition~3.1]{May12}.
  More explicitly, with
  \begin{align*}
    \tilde c_d
    &\df
    \tilde b_d-\int_{\MS_d^+\setminus\{0\}}\chi_d(\zeta)\,\tilde m_d(\D\zeta),
    \\
    \tilde D_d(v)
    &\df
    \mathbf{\tilde B}_d^*(v)
    -\int_{\MS_d^+\setminus\{0\}}
      \langle\chi_d(\zeta),v\rangle_d\,\tilde M_d(\D\zeta),
    \qquad v\in\MS_d,
  \end{align*}
  the seven-parameter tuple
  $(0,\tilde c_d,\tilde D_d,0,0,\tilde m_d,\tilde M_d)$ satisfies the
  admissibility conditions of \cite[Definition~3.1]{May12}; the inequality
  \eqref{eq:admissible-matrix-check} is precisely the quasi-monotonicity
  condition for $\tilde D_d$.
  Therefore, by \cite[Theorem
  2.4]{CFMT11} and \cite[Theorem 3.2]{May12}, there exists a unique affine
  process $(\tilde{X}^d)_{t \geq 0}$ with values in $\mathbb{S}_d^+$.

  Note further that the existence of a c\`adl\`ag version of
  $(\tilde{X}^{d})_{t \geq 0}$ follows from \cite{CT13}, and we will denote this
  version again by $(\tilde{X}^{d})_{t \geq 0}$. Moreover, we denote the law of
  $\tilde{X}^{d}$, given that $\tilde{X}^{d}_0 = x$, by $\tilde{\MP}^{d}_x$. By
  \cite[Remark 2.5]{CFMT11}, combined with the second moment assumption in
  \cref{item:m-d} and \cref{item:mu-d}, the
  $\mathbb{S}_d^+$-valued affine process $(\tilde{X}^{d})_{t \geq 0}$ satisfies
  \begin{align*}
    \tilde{\MP}^{d}_x\left(\left\{ \tilde{X}^{d}_t \in \mathbb{S}_d^+ \colon t \geq 0 \right\}\right) = 1.    
  \end{align*}

  For every $d \in \mathbb{N}$ and $x \in \mathbb{S}_d^+$, the law
  $\tilde{\MP}^{d}_x$ is thus defined on $\cB(D(\mathbb{R}^+,
  \mathbb{S}_d^+))$. Furthermore, since the diffusion part is zero, it follows
  from \cite{May12} that $(\tilde{X}^{d})_{t \geq 0}$ is a finite-variation
  process.

  Additionally, from \cite[Theorem 2.6]{CFMT11}, we know that the process
  $\tilde{X}^{d}$ is a semimartingale with characteristics given by: 
\begin{align*}
  \tilde{C}^{d}_t = 0, \quad \tilde{\nu}^{d}(\D t, \D \xi) = \big(\tilde{m}_d(\D \xi) + \tilde{M}_d(\tilde{X}^{d}_t, \D \xi)\big) \D t,  
\end{align*}
and
\begin{align*}
\tilde{A}^{d}_t = \int_0^t \left( \tilde{b}_d + \mathbf{\tilde{B}}_d(\tilde{X}^{d}_s) \right) \D s.  
\end{align*}

Finally, since the process $(\tilde{X}^{d}_t)_{t \geq 0}$ is of finite
variation, it is locally bounded and by the admissibility in the matrix sense,
we conclude that
\begin{align*}
\int_0^t \int_{\mathbb{S}_d^+ \setminus \{0\}} \|\xi\|_d^2 \, \tilde{\nu}^{d}(\D
  s, \D \xi) < \infty,
\end{align*}
for all $t \geq 0$. By \cite[Chapter II, Proposition 2.29(b)]{JS03},
this finiteness of the predictable compensator of the square-jump measure
ensures that the compensated jump integral $\bar J^{d}$ in
\eqref{eq:decomp-martingale} is a purely discontinuous locally
square-integrable martingale on $\MS_d^+$. Combined with the admissibility
moment bound in~\cref{item:m-d}, we obtain the second-moment estimate
\begin{align*}
\EXspec{\tilde{\MP}^{d}_x}{\|\tilde{X}^{d}_t\|_d^2} < \infty,
\end{align*}
for all $t \geq 0$.

Now, for every $d \in \mathbb{N}$, let
$(\tilde{X}^d, (\tilde{\MP}^d_x)_{x \in \mathbb{S}_d^+})$ be the
$\mathbb{S}_d^+$-valued affine process given as above. More precisely, let
$\tilde{X}^d$ be a version with paths in
$\Omega = D(\mathbb{R}^+, \mathbb{S}_d^+)$, and denote by $\tilde{\MP}^d_x$ the
law of $\tilde{X}^d$, defined on $\cB(\Omega)$, given that
$\tilde{X}_0^d = x \in \mathbb{S}_d^+$. Moreover, let
$(\tilde{\cF}^d_t)_{t \geq 0}$ denote the natural filtration of the process
$\tilde{X}^d$.

By identifying the cones $\mathbb{S}_d^+$ and $\cH_d^+$ under the mapping
$i_d^{-1}$, we define the process $X^d = (X^d_t)_{t \geq 0}$ as
\begin{align*}
  X^d_t \coloneqq i_d^{-1}(\tilde{X}^d_t) = \Xi_d \circ \tilde{X}^d_t \circ \Xi_d^{-1}, \quad t \geq 0.
\end{align*}

Note that the process $(X^d_t)_{t \geq 0}$ has paths in
$D(\mathbb{R}^+, \cH_d^+)$, and the law of $X^d$ is given by the push-forward
measure $(i_d^{-1})_* \tilde{\MP}^d_x$ for $x \in \mathbb{S}_d^+$. Here,
$i_d^{-1}$ acts pointwise on the functions in
$D(\mathbb{R}^+, \mathbb{S}_d^+)$, such that
\begin{align*}
i_d^{-1}(D(\mathbb{R}^+, \mathbb{S}_d^+)) = D(\mathbb{R}^+, \cH_d^+).  
\end{align*}

Moreover, we observe that
$D(\mathbb{R}^+, \cH_d^+) \subseteq D(\mathbb{R}^+, \cHplus)$ for all
$d \in \mathbb{N}$, see \cite[Remark 4.5]{Jak86}. For every $x \in \cHplus$, we
define the measure $\MP^d_x$ on $D(\mathbb{R}^+, \cHplus)$ as
\begin{align*}
  \MP^d_x(A) = (i_d^{-1})_* \tilde{\MP}^d_{i_d(\bP_d(x))}(A \cap D(\mathbb{R}^+, \cH_d^+)), \quad A \in \cB(D(\mathbb{R}^+, \cHplus)).  
\end{align*}

Note that
\begin{align*}
  \MP^d_x\big(X^d_0 = \bP_d(x)\big) = \tilde{\MP}^d_{i_d(\bP_d(x))}\big(\tilde{X}^d_0 = i_d(\bP_d(x))\big) = 1,  
\end{align*}
and the process $(X^d, (\MP^d_x)_{x \in \cHplus})$ is a Markov process on the
ambient space $\cHplus$ with respect to its natural filtration
$\MF^d = (\cF^d_t)_{t \geq 0}$. We set $\cF^d = \cF^d_\infty$. Moreover, for
every $x \in \cHplus$, denoting by $\mathcal{N}^d_x$ the collection of all
$\MP^d_x$-null sets of $\cF^d$, we define
\begin{align*}
\bar{\cF}_t \coloneqq \cF^d_t \vee \mathcal{N}^d_x \quad \text{for every } t \geq 0,  
\end{align*}

and set $\bar{\MF}^d \coloneqq (\bar{\cF}_t)_{t \geq 0}$. In other words,
$\bar{\MF}^d$ is the usual augmented filtration of $X^d$, and $X^d$ remains a
Markov process with respect to $\bar{\MF}^d$.

It remains to verify that the Galerkin approximation
$(X^d,(\MP_x^d)_{x\in\cHplus})$ is an affine process on $\cHplus$ satisfying the
affine transform formula~\eqref{eq:affine-Galerkin} associated with the
finite-rank Riccati
equations~\eqref{eq:mild-Riccati-Galerkin-phi}-\eqref{eq:mild-Riccati-Galerkin-psi},
and that the canonical process of $\MP^d_x$ on $D(\MR^+,\MS_d^+)$ is a
semimartingale with respect to $(\Omega,\bar\cF^d,\bar\MF^d,\MP^d_x)$.

For the affine transform formula, we apply the matrix-valued affine
transform theorem~\cite[Theorem~3.2]{May12} to
$(\tilde X^d,(\tilde\MP^d_{\tilde x})_{\tilde x\in\MS_d^+})$: for every
$\tilde u\in\MS_d^+$,
\[
  \EXspec{\tilde\MP_{\tilde x}^{d}}{\E^{-\langle\tilde X_t^d,\tilde u\rangle_d}}
  =\E^{-\tilde\phi_d(t,\tilde u)-\langle\tilde x,\tilde\psi_d(t,\tilde u)\rangle_d},
  \qquad t\ge 0,
\]
where $(\tilde\phi_d,\tilde\psi_d)$ is the unique solution of the matrix-valued
Riccati system associated with
$(\tilde b_d,\mathbf{\tilde B}_d,\tilde m_d,\tilde\mu_d)$ in the sense
of~\cite{CFMT11}. Pulling the exponent back under $i_d$ and using the identities
$\langle\tilde X_t^d,\tilde u\rangle_d=\langle X_t^d,u\rangle_{\cH}$ and
$\langle\tilde x,\tilde\psi_d(t,\tilde u)\rangle_d=\langle
\bP_d(x),i_d^{-1}\tilde\psi_d(t,i_d(u))\rangle_{\cH}$ for
$\tilde u=i_d(u)$, one obtains~\eqref{eq:affine-Galerkin}; the identification
of $i_d^{-1}\tilde\psi_d(t,i_d(\cdot))$ with the mild solution
$\psi^{(d)}(t,\cdot)$ of the finite-rank operator Riccati
equation~\eqref{eq:mild-Riccati-Galerkin-psi} is uniqueness in the
ODE-system~\cite[Theorem 2.4]{CFMT11}.

For the semimartingale and martingale claims, we reuse the computation
already done above: the characteristics of $X^d$ under $\MP_x^d$ are obtained
by pulling back those of $\tilde X^d$ via $i_d^{-1}$, which preserves the
finite-variation and predictable-compensator structure. The compensated jump
integral in~\eqref{eq:decomp-martingale} is then a locally square-integrable
martingale: the predictable compensator of the square-jump measure is bounded
by $C(1+\|y\|_{\cH}^{2})$ by
Lemma~\ref{lem:jump-bracket-finite-rank}; the local boundedness of
$\tilde X^d$ (from finite variation of $\tilde X^d$) then implies the
square-integrability of $\bar J^d$ on compact intervals via
\cite[Chapter~II, Proposition~2.29(b)]{JS03}. Only the admissibility of the
Galerkin parameters and the second-moment bound from~\cref{item:m-d} are
used; both are preserved under the unbounded generator $B$ considered
here. Hence
$\MP_x^{d}(\{X_t^d\in\cHplus_d\colon t\ge0\})=1$
and~\eqref{eq:affine-Galerkin} holds. In particular, defining the compensated
jump process $(\bar J^d_t)_{t\ge 0}$ by
\begin{align}\label{eq:decomp-martingale}
  \bar{J}^d_t \coloneqq X^d_t \!-\! \bP_d(x) \!-\! \!\int_0^t \!\left(b_d + \mathbf{B}_d(X^d_s) + \int_{\cHplus_d \cap \{\|\xi\| > 1\}} \xi \, (m_d(\D\xi) + M_d(X^d_s, \D\xi)) \right) \D s,
\end{align}
the process $\bar J^d$ is a locally square-integrable martingale on
$\cH_d$ by the argument above.
 \end{proof}

 \subsubsection{Weak Convergence of the Affine Finite-Rank Processes}
 
 Let $(b, B, m, \mu)$ be an admissible parameter set, and for each
 $d \in \mathbb{N}$, let $X^d$ denote the associated affine finite-rank
 operator-valued process given by
 Proposition~\ref{prop:embedding-affine-main}. In this section, we study the
 tightness and weak convergence of the sequence $(X^d)_{d \in \mathbb{N}}$ in
 the space $D(\mathbb{R}^+, \cHplus)$ equipped with the Skorohod topology. More
 precisely, for each $x \in \cHplus$, we consider the sequence
 $(\MP^d_x)_{d \in \mathbb{N}}$, which represents the laws of $X^d$ given that
 $X^d_0 = \bP_d(x)$, defined on the Borel $\sigma$-algebra
 $\cB(D(\mathbb{R}^+, \cHplus))$, and examine its weak convergence as
 $d \to \infty$.

	 To make the weak-convergence argument rigorous, we now separate the
	 three ingredients that are needed in the irregular setting: a uniform
	 square-moment bound for the projected semimartingales, a positive-time
	 smoothing estimate in $\cV$, and the identification of the weak limit via the
	 martingale problem. Once these ingredients are established, the proof follows
	 the same template as in \cite[Section~6]{karbach2023finiterank}.

\begin{lemma}[Maximal inequality for jump stochastic convolutions]
\label{lem:jump-convolution-maximal}
  Let $E$ be a Hilbert space and let $(S(t))_{t\ge0}$ be a contraction
  $C_0$-semigroup on $E$. Let $\mu$ be an integer-valued random measure on
  $[a,b]\times E$ with predictable compensator $\nu$, and assume
  \[
    \mathbb E\int_a^b\int_E \|\zeta\|_E^2\,\nu(\D s,\D\zeta)<\infty .
  \]
  Then the stochastic convolution
  \[
    Z_r=\int_a^r\int_E S(r-s)\zeta\,(\mu-\nu)(\D s,\D\zeta),
    \qquad a\le r\le b,
  \]
  satisfies
  \[
    \mathbb E\sup_{a\le r\le b}\|Z_r\|_E^2
    \le
    C\,
    \mathbb E\int_a^b\int_E \|\zeta\|_E^2\,\nu(\D s,\D\zeta),
  \]
  with a universal constant $C$ independent of $E$, $S$, and the length of the
  interval.
\end{lemma}
\begin{proof}
  This is the Kotelenez maximal inequality for Hilbert-space stochastic
  convolutions with contraction semigroups; see, for instance,
  \cite[Theorem~9.15]{PZ07}. Applied to the square-integrable purely
  discontinuous martingale
  $\int_a^\cdot\int_E\zeta\,(\mu-\nu)(\D s,\D\zeta)$, the predictable quadratic
  variation is exactly
  $\int_a^\cdot\int_E\|\zeta\|_E^2\,\nu(\D s,\D\zeta)$, giving the stated
  estimate.
\end{proof}

\begin{corollary}[Endpoint estimate for stopped jump stochastic convolutions]
\label{cor:stopped-jump-convolution-maximal}
  Let $E$ be a Hilbert space and let $(S(t))_{t\ge0}$ be a contraction
  $C_0$-semigroup on $E$. Let $\mu$ be an integer-valued random measure on
  $[0,T+1]\times E$ with predictable compensator $\nu$, let $\tau\le T$ be a
  stopping time, and let $0\le\theta\le1$ be deterministic. Assume
  \[
    \mathbb E\int_\tau^{\tau+\theta}\int_E
    \|\zeta\|_E^2\,\nu(\D s,\D\zeta)<\infty .
  \]
  Then
  \[
    N_{\tau,\theta}
    \df
    \int_\tau^{\tau+\theta}\int_E
    S(\tau+\theta-s)\zeta\,(\mu-\nu)(\D s,\D\zeta)
  \]
  satisfies
  \[
    \mathbb E\|N_{\tau,\theta}\|_E^2
    \le
    C\,
    \mathbb E\int_\tau^{\tau+\theta}\int_E
    \|\zeta\|_E^2\,\nu(\D s,\D\zeta),
  \]
  with the same universal constant convention as in
  Lemma~\ref{lem:jump-convolution-maximal}.
\end{corollary}
\begin{proof}
  The process
  $(\omega,s,\zeta)\mapsto
  \mathbf 1_{\{\tau(\omega)<s\le\tau(\omega)+\theta\}}
  S(\tau(\omega)+\theta-s)\zeta$
  is predictable; this follows first for finitely valued stopping times and
  then by the standard decreasing approximation of bounded stopping times. The
  Hilbert-space isometry for compensated jump integrals gives
  \[
    \mathbb E\|N_{\tau,\theta}\|_E^2
    =
    \mathbb E\int_\tau^{\tau+\theta}\int_E
    \|S(\tau+\theta-s)\zeta\|_E^2\,\nu(\D s,\D\zeta).
  \]
  Since $S$ is a contraction, the right-hand side is bounded by the displayed
  expression.
\end{proof}

\begin{lemma}\label{lem:jump-bracket-finite-rank}
  For every $d\in\MN$, let $\mu^{X^d}(\D s,\D\xi)$ denote the jump
  measure of $X^d$ and define
  \begin{align*}
    \nu^d(\omega;\D s,\D\xi)
    \df
    \bigl(m_d(\D\xi)+M_d(X_{s-}^d(\omega),\D\xi)\bigr)\,\D s.
  \end{align*}
  Then $\nu^d$ is the compensator of $\mu^{X^d}$, the compensated jump
  martingale from \eqref{eq:decomp-martingale} admits the representation
  \begin{align*}
    \bar J_t^d
    =
    \int_0^t\int_{\cH_d^+\setminus\{0\}}
    \xi\,\bigl(\mu^{X^d}-\nu^d\bigr)(\D s,\D\xi),\qquad t\ge0,
  \end{align*}
  and its predictable quadratic variation is given by
  \begin{align*}
    \langle \bar J^d\rangle_t
    =
    \int_0^t\int_{\cH_d^+\setminus\{0\}}
    \|\xi\|_{\cH}^{2}\,
    \bigl(m_d(\D\xi)+M_d(X_{s-}^d,\D\xi)\bigr)\,\D s.
  \end{align*}
	  Moreover, there exists a constant $C>0$, independent of
	  $d\in\MN$, such that
	  \begin{align*}
	    \int_{\cH_d^+\setminus\{0\}}\|\xi\|_{\cH}^{2}\,
	    \bigl(m_d(\D\xi)+M_d(y,\D\xi)\bigr)
	    \le C\bigl(1+\|y\|_{\cH}^{2}\bigr),\qquad y\in\cH_d^+.
	  \end{align*}
	  The state-dependent part satisfies the sharper homogeneous estimate
	  \begin{align*}
	    \int_{\cH_d^+\setminus\{0\}}\|\xi\|_{\cH}^{2}\,
	    M_d(y,\D\xi)
	    \le C\|y\|_{\cH},\qquad y\in\cH_d^+ .
	  \end{align*}
	\end{lemma}
\begin{proof}
  The form of the compensator follows directly from the semimartingale
  characteristics in Proposition~\ref{prop:embedding-affine-main}. The
  representation of $\bar J^d$ is therefore the standard compensated jump
  integral representation of the purely discontinuous martingale part in
  \eqref{eq:decomp-martingale}. The formula for the predictable quadratic
  variation is the corresponding Hilbert-space bracket formula for square-
  integrable purely discontinuous martingales.

  It remains to prove the uniform bound. Since $\bP_d$ is an
  orthogonal projection on $\cH$, it is contractive, and the projected kernels
  inherit the second-moment bounds from Definition~\ref{def:admissible-irregular}.
  In particular,
  \begin{align*}
    \int_{\cH_d^+\setminus\{0\}}\|\xi\|_{\cH}^{2}\,m_d(\D\xi)
    \le
    \int_{\cHpluso}\|\xi\|_{\cH}^{2}\,m(\D\xi)\df C_0,
  \end{align*}
  while the state-dependent part is first order in $y$. Indeed, using the
  definition of the projected kernel and contractivity of $\bP_d$,
  \begin{align*}
    \int_{\cH_d^+\setminus\{0\}}\|\xi\|_{\cH}^{2}\,M_d(y,\D\xi)
    &\le
    \int_{E_d^0}\|\bP_d\eta\|_{\cH}^{2}
      \frac{\langle y,\bP_d\mu(\D\eta)\rangle}{\|\eta\|_{\cH}^{2}}\\
    &\le
	    \int_{E_d^0}\langle y,\mu(\D\eta)\rangle
	    \le \|y\|_{\cH}\,\|\mu(\cHpluso)\|_{\cH}\\
	    &\le C_\mu\bigl(1+\|y\|_{\cH}^{2}\bigr),
	    \qquad y\in\cH_d^+,
	  \end{align*}
	  for a constant $C_\mu$ independent of $d$. The penultimate bound is the
	  sharper homogeneous estimate for the state-dependent part, and combining it
	  with the $m_d$ estimate proves the displayed total bound.
	\end{proof}

\begin{lemma}\label{lem:irregular-tightness-square-bound}
  Let $x\in\cHplus$ and $T>0$. Then there exists a constant
  $C_T(x)>0$, independent of $d\in\MN$, such that
  \begin{align*}
    \EXspec{\MP_x^d}{\sup_{0\le t\le T}\|X_t^d\|_{\cH}^2}
    +\EXspec{\MP_x^d}{\sup_{0\le t\le T}\|\bar J_t^d\|_{\cH}^2}
    \le C_T(x).
  \end{align*}
\end{lemma}
\begin{proof}
  We keep the decomposition strategy of
  \cite[Lemma~6.1]{karbach2023finiterank}, but spell out the estimates in the
  present irregular setting. Define
  \begin{align*}
    \beta_d(y)
    \df
    b_d+\Gamma_d(y)
    +\int_{\cH_d^+\cap\{\|\xi\|>1\}}
    \xi\,\bigl(m_d(\D\xi)+M_d(y,\D\xi)\bigr),
    \qquad y\in\cH_d^+.
  \end{align*}
  By contractivity of $\bP_d$, boundedness of $\Gamma$, the second-moment bound
  for $m_d$, and the preceding state-dependent estimate, the large-jump drift
  \[
    a_d(y)\df
    \int_{\cH_d^+\cap\{\|\xi\|>1\}}\xi\,
    \bigl(m_d(\D\xi)+M_d(y,\D\xi)\bigr)
  \]
  satisfies
  \[
    \|a_d(y)\|_{\cH}
    \le C(1+\|y\|_{\cH}),\qquad y\in\cH_d^+.
  \]
  Indeed, the $m_d$-part is bounded by
  $\int_{\{\|\xi\|>1\}}\|\xi\|\,m_d(\D\xi)
  \le \int\|\xi\|^2\,m_d(\D\xi)\le C$. For the state-dependent part,
  $\|\xi\|>1$ implies $\|\xi\|\le\|\xi\|^2$, and
  Lemma~\ref{lem:jump-bracket-finite-rank} gives the sharper homogeneous bound
  \[
    \int_{\cH_d^+\setminus\{0\}}\|\xi\|_{\cH}^{2}M_d(y,\D\xi)
    \le
    \int_{E_d^0}\langle y,\mu(\D\eta)\rangle
    \le \|y\|_{\cH}\,\|\mu(\cHpluso)\|_{\cH},
  \]
  so this part is bounded by $C\|y\|_{\cH}$.
  Hence there exists a constant $C_\beta>0$, independent of $d$, such that
  \begin{align}\label{eq:beta-growth-bound}
    \|\beta_d(y)\|_{\cH}^{2}\le C_\beta(1+\|y\|_{\cH}^{2}),
    \qquad y\in\cH_d^+.
  \end{align}

  \smallskip
  \noindent\textbf{Step 1: One-time second moments.}
  Let $h(y)\df\|y\|_{\cH}^{2}$. Dynkin's formula gives
  \begin{align*}
    \EXspec{\MP_x^d}{\|X_t^d\|_{\cH}^{2}}
    =
    \|\bP_d x\|_{\cH}^{2}
    +\EXspec{\MP_x^d}{\int_0^t I_d(X_s^d)\,\D s},
  \end{align*}
  where
  \begin{align*}
    I_d(y)
    &=
    2\langle y,b_d+L_d y+\Gamma_d(y)\rangle_{\cH}\\
    &\quad
    +2\Big\langle y,\int_{\cH_d^+\cap\{\|\xi\|>1\}}\xi\,
      \bigl(m_d(\D\xi)+M_d(y,\D\xi)\bigr)\Big\rangle_{\cH}\\
    &\quad
    +\int_{\cH_d^+\setminus\{0\}}\|\xi\|_{\cH}^{2}\,
      \bigl(m_d(\D\xi)+M_d(y,\D\xi)\bigr).
  \end{align*}
  The individual terms are estimated as follows. First,
  \begin{align*}
    2\langle y,b_d\rangle_{\cH}
    \le \|b_d\|_{\cH}^{2}+\|y\|_{\cH}^{2}
    \le C(1+\|y\|_{\cH}^{2}),
  \end{align*}
  and the Lyapunov part is dissipative:
  \begin{align*}
    \langle y,L_d y\rangle_{\cH}
    =\langle y,L y\rangle_{\cH}\le0.
  \end{align*}
  Next,
  \begin{align*}
    2\langle y,\Gamma_d(y)\rangle_{\cH}
    \le 2\|\Gamma\|_{\cL(\cH)}\|y\|_{\cH}^{2}.
  \end{align*}
  For the large-jump drift term, Young's inequality and the large-jump drift
  bound above yield
  \begin{align*}
    2\langle y,a_d(y)\rangle_{\cH}
    \le \|y\|_{\cH}^{2}
    +\|a_d(y)\|_{\cH}^{2}
    \le C(1+\|y\|_{\cH}^{2}).
  \end{align*}
  Together with Lemma~\ref{lem:jump-bracket-finite-rank}, this gives
  \begin{align*}
    I_d(y)\le C_0+C_1\|y\|_{\cH}^{2},\qquad y\in\cH_d^+,
  \end{align*}
  with constants independent of $d$. Hence
  \begin{align*}
    \EXspec{\MP_x^d}{\|X_t^d\|_{\cH}^{2}}
    \le \|x\|_{\cH}^{2}
    +C_0 t
    +C_1\int_0^t\EXspec{\MP_x^d}{\|X_s^d\|_{\cH}^{2}}\,\D s,
  \end{align*}
  and Gronwall's lemma yields
  \begin{align}\label{eq:uniform-second-moment-Xd}
    \sup_{d\in\MN}\sup_{t\in[0,T]}
    \EXspec{\MP_x^d}{\|X_t^d\|_{\cH}^{2}}<\infty.
  \end{align}

  \smallskip
  \noindent\textbf{Step 2: Maximal estimate.}
  Let $\cT_d(t)\df \bP_d\cT(t)\bP_d$. Since $(\cT(t))_{t\ge0}$ is a
  contraction semigroup on $\cH$, the same holds for $(\cT_d(t))_{t\ge0}$.
  Writing the mild form of \eqref{eq:decomp-martingale}, we obtain
  \begin{align*}
    X_t^d
    =
    \cT_d(t)\bP_d x
    +\int_0^t \cT_d(t-s)\beta_d(X_s^d)\,\D s
    +M_t^{d,\cT},
  \end{align*}
  where
  \begin{align*}
    M_t^{d,\cT}
    \df
    \int_0^t\int_{\cH_d^+\setminus\{0\}}
    \cT_d(t-s)\xi\,\bigl(\mu^{X^d}-\nu^d\bigr)(\D s,\D\xi).
  \end{align*}
  Hence, for every $t\in[0,T]$,
  \begin{align*}
    \sup_{0\le r\le t}\|X_r^d\|_{\cH}^{2}
    \le
    3\|x\|_{\cH}^{2}
    +3t\int_0^t \|\beta_d(X_s^d)\|_{\cH}^{2}\,\D s
    +3\sup_{0\le r\le t}\|M_r^{d,\cT}\|_{\cH}^{2}.
  \end{align*}
  Taking expectations and using \eqref{eq:beta-growth-bound} gives
  \begin{align*}
    \EXspec{\MP_x^d}{\sup_{0\le r\le t}\|X_r^d\|_{\cH}^{2}}
    \le
    C(1+\|x\|_{\cH}^{2})
    +C\int_0^t\EXspec{\MP_x^d}{\|X_s^d\|_{\cH}^{2}}\,\D s
    +3\EXspec{\MP_x^d}{\sup_{0\le r\le t}\|M_r^{d,\cT}\|_{\cH}^{2}}.
  \end{align*}
  By Lemma~\ref{lem:jump-convolution-maximal} and the
  contractivity of $\cT_d$,
  \begin{align*}
    \EXspec{\MP_x^d}{\sup_{0\le r\le t}\|M_r^{d,\cT}\|_{\cH}^{2}}
    \le
    C_{\mathrm{BDG}}
    \EXspec{\MP_x^d}{\int_0^t\int_{\cH_d^+\setminus\{0\}}
      \|\cT_d(t-s)\xi\|_{\cH}^{2}\,\nu^d(\D s,\D\xi)}
  \end{align*}
  \begin{align*}
    \le
    C_{\mathrm{BDG}}
    \EXspec{\MP_x^d}{\int_0^t\int_{\cH_d^+\setminus\{0\}}
      \|\xi\|_{\cH}^{2}\,\nu^d(\D s,\D\xi)}.
  \end{align*}
  Using Lemma~\ref{lem:jump-bracket-finite-rank} and
  \eqref{eq:uniform-second-moment-Xd}, we conclude that
  \begin{align*}
    \EXspec{\MP_x^d}{\sup_{0\le r\le t}\|M_r^{d,\cT}\|_{\cH}^{2}}
    \le C_t(x),
  \end{align*}
  with a constant independent of $d$. Therefore
  \begin{align*}
    \EXspec{\MP_x^d}{\sup_{0\le r\le t}\|X_r^d\|_{\cH}^{2}}
    \le C_t(x),
  \end{align*}
  uniformly in $d$. Finally, applying the same BDG estimate directly to
  $\bar J^d$ and using Lemma~\ref{lem:jump-bracket-finite-rank} once more gives
  \begin{align*}
    \EXspec{\MP_x^d}{\sup_{0\le r\le t}\|\bar J_r^d\|_{\cH}^{2}}
    \le C_t(x),
  \end{align*}
  uniformly in $d$. This proves the claim.
\end{proof}

\begin{lemma}\label{lem:trace-bound-finite-rank}
  Assume that $x\in\cL_1(H)\cap\cHplus$, that $m(\D\xi)$ and
  $\mu(\D\xi)$ are supported on $\cL_1(H)\setminus\{0\}$, that
  $\Gamma(\cL_1(H))\subseteq \cL_1(H)$, that
  $b,\mu(A)\in\cL_1(H)$ for all $A\in\cB(\cHplus\setminus\{0\})$, and that the
  trace first-moment bounds
  \[
    \int_{\cHpluso}\Tr(\xi)\,m(\D\xi)<\infty,\qquad
    \int_{\cHpluso}\Tr(\xi)\,
      \frac{\langle y,\mu(\D\xi)\rangle}{\|\xi\|^2}
    \le C_\mu \Tr(y),\quad y\in\cL_1(H)\cap\cHplus,
  \]
  hold for some $C_\mu<\infty$. Then for every $T>0$ there exists a constant
  $C_T^{\Tr}(x)>0$, independent of $d\in\MN$, such that
  \begin{align*}
    \sup_{d\in\MN}\sup_{t\in[0,T]}
    \EXspec{\MP_x^d}{\Tr(X_t^d)}\le C_T^{\Tr}(x).
  \end{align*}
\end{lemma}
\begin{proof}
  For $n\in\MN$, define
  \begin{align*}
    h_n(y)\df \sum_{k=1}^n\langle y e_k,e_k\rangle_H,\qquad y\in\cHplus,
  \end{align*}
  where $(e_k)_{k\in\MN}$ is the eigenbasis of $B$. Fix $d\in\MN$ and
  choose $n\ge d$. Since $X_t^d\in\cH_d^+$, we have
  $h_n(X_t^d)=\Tr(X_t^d)$ for all $t\ge0$. Applying Dynkin's formula to the
  linear functional $h_n$ yields
  \begin{align*}
    \EXspec{\MP_x^d}{\Tr(X_t^d)}
    =
    \Tr(\bP_d x)
    +\EXspec{\MP_x^d}{\int_0^t K_d(X_s^d)\,\D s},
  \end{align*}
  where
  \begin{align*}
    K_d(y)
    &=
    \Tr(b_d)+h_n(L_d y)+h_n(\Gamma_d(y))\\
    &\quad
    +\int_{\cH_d^+\setminus\{0\}}\Tr(\xi)\,
      \bigl(m_d(\D\xi)+M_d(y,\D\xi)\bigr).
  \end{align*}

  We estimate the individual terms. Since $T(t)e_k=e^{-\lambda_k t}e_k$,
  one has
  \begin{align*}
    h_n(\cT_d(t)y)
    =
    \sum_{k=1}^d e^{-2\lambda_k t}\langle y e_k,e_k\rangle_H
    \le
    \sum_{k=1}^d \langle y e_k,e_k\rangle_H
    = h_n(y),
  \end{align*}
  for every $y\in\cH_d^+$. Differentiating at $t=0$ therefore gives
  \begin{align*}
    h_n(L_d y)\le0,\qquad y\in\cH_d^+.
  \end{align*}

  Next, the restriction of $\Gamma$ to $\cL_1(H)$ is bounded. To see
  this, recall that $(\cL_1(H),\|\cdot\|_{\cL_1(H)})$ is a Banach space and
  that the embedding $\cL_1(H)\hookrightarrow\cH$ is continuous, with
  $\|\cdot\|_{\cH}\leq\|\cdot\|_{\cL_1(H)}$. By hypothesis,
  $\Gamma(\cL_1(H))\subseteq\cL_1(H)$ and $\Gamma\colon\cH\to\cH$ is bounded.
  Suppose $y_n\to y$ in $\cL_1(H)$ and $\Gamma(y_n)\to z$ in $\cL_1(H)$. Then
  $y_n\to y$ and $\Gamma(y_n)\to z$ also in $\cH$ by continuity of the
  embedding, and the boundedness of $\Gamma$ on $\cH$ gives
  $\Gamma(y_n)\to\Gamma(y)$ in $\cH$. Hence $z=\Gamma(y)$, so the graph of
  $\Gamma\rvert_{\cL_1(H)}$ is closed in $\cL_1(H)\times\cL_1(H)$. The closed
  graph theorem now yields $C_{\Gamma,1}>0$ such that
  \begin{align*}
    h_n(\Gamma_d(y))
    \le \Tr(\Gamma_d(y))
    \le C_{\Gamma,1}\Tr(y)
    = C_{\Gamma,1}h_n(y),
  \end{align*}
  for all $y\in\cH_d^+$. Furthermore, the definition of $b_d$ includes the
  projected compensation drift, and the trace first-moment assumption controls
  it uniformly. The positivity of $b$ used below follows from admissibility:
  for $v\in\cHplus$ the integrand $\langle\chi(\xi),v\rangle$ is non-negative,
  so the drift condition in Definition~\ref{def:admissible-irregular}(ii)
  implies $\langle b,v\rangle\ge0$ for all $v\in\cHplus$; by self-duality of
  $\cHplus$, $b\in\cHplus$. Hence
  \begin{align*}
    \Tr(b_d)
    &=\Tr(\bP_d b)
      +\int_{\bP_d(E_d)}\Tr(\xi)\,m_d(\D\xi)\\
    &=\Tr(\bP_d b)
      +\int_{E_d}\Tr(\bP_d\eta)\,m(\D\eta)\\
    &\le \Tr(b)+\int_{\cHpluso}\Tr(\eta)\,m(\D\eta)
    \df C_b^{\Tr}<\infty .
  \end{align*}
  Here we used positivity of $b$ and $\eta$, the compression inequality
  $\Tr(\bP_d z)\le\Tr(z)$ for positive trace-class $z$, and the push-forward
  definition of $m_d$. The trace first-moment assumptions in the lemma imply
  \begin{align*}
    \int_{\cH_d^+\setminus\{0\}}\Tr(\xi)\,m_d(\D\xi)\le C_m,
  \end{align*}
  \begin{align*}
    \int_{\cH_d^+\setminus\{0\}}\Tr(\xi)\,M_d(y,\D\xi)
    \le C_\mu\,\Tr(y)
    = C_\mu\,h_n(y),
  \end{align*}
  with constants independent of $d$. Altogether,
  \begin{align*}
    K_d(y)\le C_0+C_1\Tr(y),\qquad y\in\cH_d^+,
  \end{align*}
  for constants $C_0,C_1$ independent of $d$. Since
  $\Tr(\bP_d x)\le \Tr(x)$, we obtain
  \begin{align*}
    \EXspec{\MP_x^d}{\Tr(X_t^d)}
    \le
    \Tr(x)+C_0 t
    +C_1\int_0^t\EXspec{\MP_x^d}{\Tr(X_s^d)}\,\D s.
  \end{align*}
  Gronwall's lemma now yields
  \begin{align*}
    \sup_{d\in\MN}\sup_{t\in[0,T]}
    \EXspec{\MP_x^d}{\Tr(X_t^d)}\le C_T^{\Tr}(x),
  \end{align*}
  which proves the claim.
\end{proof}

\begin{lemma}\label{lem:irregular-smoothing-cV}
  Assume Assumption~\ref{assump:cV-compact-containment}. Let
  $x\in\cHplus$, $T>0$ and $\delta\in(0,T]$. Then there exists a constant
  $C_{T,\delta}(x)>0$, independent of $d\in\MN$, such that
  \begin{align*}
    \sup_{d\in\MN}\EXspec{\MP_x^d}{\sup_{\delta\le t\le T}
    \|X_t^d\|_{\cV}^2}\le C_{T,\delta}(x).
  \end{align*}
\end{lemma}
\begin{proof}
	  We first prove the estimate for initial data already in
	  $\cV\cap\cHplus$. Work in the equivalent adapted norm
	  $\|\cdot\|_{\cV,B}$ introduced in
	  Lemma~\ref{lem:Lyapunov-semigroup-smoothing}. By norm equivalence,
	  contractivity of $\bP_d$ in $\|\cdot\|_{\cV,B}$, and
	  Assumption~\ref{assump:cV-compact-containment}, the projected kernels satisfy
	  \[
	    \int_{\cH_d^+\setminus\{0\}}\|\xi\|_{\cV,B}^{2}\,m_d(\D\xi)\le C
	  \]
	  and
	  \[
	    \int_{\cH_d^+\setminus\{0\}}\|\xi\|_{\cV,B}^{2}\,M_d(y,\D\xi)
	    \le C(1+\|y\|_{\cV,B}^{2}),\qquad y\in\cH_d^+ .
	  \]
	  Since the left-hand side in the second display is homogeneous and linear in
	  $y$, applying the estimate to $y/\|y\|_{\cV,B}$ when $y\ne0$ gives the
	  sharper homogeneous bound
	  \begin{align}\label{eq:cV-kernel-homogeneous}
	    \int_{\cH_d^+\setminus\{0\}}\|\xi\|_{\cV,B}^{2}\,M_d(y,\D\xi)
	    \le C\|y\|_{\cV,B},\qquad y\in\cH_d^+ ,
	  \end{align}
	  with the convention that the left-hand side is zero at $y=0$.
	  Consequently, the large-jump drift
	  \[
	    a_d^{\cV}(y)\df
	    \int_{\cH_d^+\cap\{\|\xi\|_{\cH}>1\}}\xi\,
	      \bigl(m_d(\D\xi)+M_d(y,\D\xi)\bigr)
	  \]
	  has linear growth in the adapted norm:
	  \begin{align}\label{eq:cV-large-jump-drift}
	    \|a_d^{\cV}(y)\|_{\cV,B}\le C(1+\|y\|_{\cV,B}),
	    \qquad y\in\cH_d^+ .
	  \end{align}
	  Indeed, on $\{\|\xi\|_{\cH}>1\}$ the continuous embedding
	  $\cV\hookrightarrow\cH$ implies
	  $\|\xi\|_{\cV,B}\le C\|\xi\|_{\cV,B}^{2}$ after changing constants; the
	  $m_d$-part is therefore bounded by the $\cV$ second moment, and the
	  $M_d$-part by~\eqref{eq:cV-kernel-homogeneous}. For
	  $q_d(y)\df\|y\|_{\cV,B}^{2}$, $y\in\cH_d^+$, the finite-dimensional generator
	  satisfies the uniform growth estimate
	  \begin{align}\label{eq:cV-generator-growth}
    \cG^d q_d(y)
    +
    \int_{\cH_d^+\setminus\{0\}}\|\xi\|_{\cV,B}^{2}
      \bigl(m_d(\D\xi)+M_d(y,\D\xi)\bigr)
    \le C\bigl(1+\|y\|_{\cV,B}^{2}\bigr),
  \end{align}
  with a constant independent of $d$. To see this, expand
  $y=\sum y_{i,j}\be_{i,j}$ on $\cH_d$. The Lyapunov part is dissipative:
  \[
    2\langle L_dy,y\rangle_{\cV,B}
    =
    -2\sum_{1\le i\le j\le d}
      \lambda_{i,j}(1+\lambda_{i,j})|y_{i,j}|^{2}
    \le0.
  \]
  The projection $\bP_d$ is contractive in $\|\cdot\|_{\cV,B}$, so
  $\sup_d\|b_d\|_{\cV,B}<\infty$. Boundedness of
  $\Gamma:\cV\cap\cH\to\cV\cap\cH$ gives
  $\|\Gamma_d(y)\|_{\cV,B}\le C\|y\|_{\cV,B}$. Hence Cauchy--Schwarz yields
  \[
    2\langle y,b_d+\Gamma_d(y)\rangle_{\cV,B}
    \le C(1+\|y\|_{\cV,B}^{2}).
  \]
	  For the large-jump drift term, Young's inequality and
	  \eqref{eq:cV-large-jump-drift} give
	  \[
	    2\langle y,a_d^{\cV}(y)\rangle_{\cV,B}
	    \le \|y\|_{\cV,B}^{2}+\|a_d^{\cV}(y)\|_{\cV,B}^{2}
	    \le C(1+\|y\|_{\cV,B}^{2}).
	  \]
	  The quadratic jump term is controlled by Assumption
	  \ref{assump:cV-compact-containment}, transferred to
	  $\|\cdot\|_{\cV,B}$ by norm equivalence, together with the weaker form of
	  \eqref{eq:cV-kernel-homogeneous}. Thus
	  \[
	    \int_{\cH_d^+\setminus\{0\}}\|\xi\|_{\cV,B}^{2}
	      (m_d(\D\xi)+M_d(y,\D\xi))
    \le C(1+\|y\|_{\cV,B}^{2}),
  \]
  which proves \eqref{eq:cV-generator-growth}. Dynkin's formula and
  Gronwall's lemma give the corresponding pointwise estimate
  \[
    \sup_d\sup_{t\le T}
    \EXspec{\MP_y^d}{\|X_t^d\|_{\cV,B}^{2}}
    \le C_T(1+\|y\|_{\cV,B}^{2}).
  \]

  We now prove the maximal estimate. Write the finite-variation drift as
  \[
    \beta_d(z)\df b_d+\Gamma_d(z)
    +\int_{\cH_d^+\cap\{\|\xi\|_{\cH}>1\}}\xi\,
      \bigl(m_d(\D\xi)+M_d(z,\D\xi)\bigr).
  \]
	  The large-jump drift estimate~\eqref{eq:cV-large-jump-drift}, together with
	  boundedness of $b_d$ and $\Gamma_d$ in the adapted norm, gives
	  \[
	    \|\beta_d(z)\|_{\cV,B}^{2}
    \le C\bigl(1+\|z\|_{\cV,B}^{2}\bigr),
    \qquad z\in\cH_d^+,
  \]
  with $C$ independent of $d$. For $y\in\cV\cap\cHplus$ define
  \[
    F_d(t)\df
    \EXspec{\MP_y^d}{\sup_{0\le r\le t}\|X_r^d\|_{\cV,B}^{2}},
    \qquad 0\le t\le T.
  \]
  The mild representation in the finite-dimensional Hilbert space
  $(\cH_d,\|\cdot\|_{\cV,B})$ is
  \[
    X_r^d=\cT_d(r)y+\int_0^r\cT_d(r-s)\beta_d(X_s^d)\,\D s
    +M_r^{d,\cT},
  \]
  where
  \[
    M_r^{d,\cT}\df
    \int_0^r\int_{\cH_d^+\setminus\{0\}}
    \cT_d(r-s)\xi\,(\mu^{X^d}-\nu^d)(\D s,\D\xi).
  \]
  Since $\cT_d$ is contractive in $\|\cdot\|_{\cV,B}$,
  \[
    \sup_{r\le t}\|\cT_d(r)y\|_{\cV,B}^{2}\le\|y\|_{\cV,B}^{2}.
  \]
  For the finite-variation term, Jensen's inequality and the drift-growth bound
  yield
  \begin{align*}
    &\EXspec{\MP_y^d}{\sup_{r\le t}
      \left\|\int_0^r\cT_d(r-s)\beta_d(X_s^d)\,\D s\right\|_{\cV,B}^{2}}\\
    &\qquad\le
    t\int_0^t
      \EXspec{\MP_y^d}{\|\beta_d(X_s^d)\|_{\cV,B}^{2}}\,\D s
    \le
    C\int_0^t\bigl(1+F_d(s)\bigr)\,\D s.
  \end{align*}
  For the martingale convolution, Lemma~\ref{lem:jump-convolution-maximal}
  gives
  \begin{align*}
    \EXspec{\MP_y^d}{\sup_{r\le t}\|M_r^{d,\cT}\|_{\cV,B}^{2}}
    &\le
    C
    \EXspec{\MP_y^d}{\int_0^t\int_{\cH_d^+\setminus\{0\}}
      \|\xi\|_{\cV,B}^{2}
      \bigl(m_d(\D\xi)+M_d(X_s^d,\D\xi)\bigr)\,\D s}\\
    &\le
    C\int_0^t\bigl(1+F_d(s)\bigr)\,\D s,
  \end{align*}
  by \eqref{eq:cV-generator-growth}. Combining the three estimates gives
  \[
    F_d(t)\le
    C\bigl(1+\|y\|_{\cV,B}^{2}\bigr)
    +C\int_0^t F_d(s)\,\D s,\qquad 0\le t\le T,
  \]
  with constants independent of $d$. Gronwall's lemma yields, whenever the
  initial datum is in $\cV$,
  \begin{align}\label{eq:cV-uniform-initial}
    \sup_{d\in\MN}
    \EXspec{\MP_y^d}{\sup_{0\le t\le T}\|X_t^d\|_{\cV,B}^{2}}
    \le C_T\bigl(1+\|y\|_{\cV,B}^{2}\bigr),
    \qquad y\in\cV\cap\cHplus .
  \end{align}
  By equivalence of $\|\cdot\|_{\cV}$ and $\|\cdot\|_{\cV,B}$, the same
  estimate holds with $\|\cdot\|_{\cV}$, after changing $C_T$.

  For a general $x\in\cHplus$ we do not bound
  $\EXspec{\MP_x^d}{\|X_\delta^d\|_{\cV}^{2}}$ pointwise directly; we first
  establish a time-averaged $\cV$ estimate and then restart the process at a
  positive time. We claim that for every $\delta\in(0,T]$,
  \begin{align}\label{eq:cV-integrated}
    \sup_{d\in\MN}\int_0^\delta
      \EXspec{\MP_x^d}{\|X_s^d\|_{\cV}^{2}}\,\D s
    \le C_\delta(x)<\infty .
  \end{align}
  The mechanism is the elementary spectral identity: for every Lyapunov mode
  $\lambda_{i,j}\ge0$ and every $\tau\in[0,\delta]$,
  \begin{align}\label{eq:spectral-time-integral}
    (1+\lambda_{i,j})\int_0^{\tau}\E^{-2\lambda_{i,j}\sigma}\,\D\sigma
    =(1+\lambda_{i,j})\,\frac{1-\E^{-2\lambda_{i,j}\tau}}{2\lambda_{i,j}}
    \le \tau+\tfrac12 \le \delta+\tfrac12 ,
  \end{align}
  where the left-hand side is read as $\tau$ when $\lambda_{i,j}=0$. The factor
  $1+\lambda_{i,j}$ is exactly the $\cV$-weight of~\eqref{eq:cV-weight-bound},
  and \eqref{eq:spectral-time-integral} shows that integrating in time trades it
  for a bounded constant; this is the precise sense in which the heat smoothing
  absorbs the time singularity.

  We use the finite-rank mild representation, with $\cT_d(t)\df\bP_d\cT(t)\bP_d$,
  \[
    X_t^d=\cT_d(t)\bP_d x
          +\int_0^t\cT_d(t-s)\beta_d(X_s^d)\,\D s
          +M_t^{d,\cT},\qquad
    M_t^{d,\cT}=\int_0^t\!\!\int_{\cH_d^{+}\setminus\{0\}}
       \cT_d(t-s)\xi\,(\mu^{X^d}-\nu^d)(\D s,\D\xi).
  \]

  \emph{Deterministic term.} By \eqref{eq:cV-weight-bound} and
  \eqref{eq:spectral-time-integral},
  \[
    \int_0^\delta\|\cT_d(s)\bP_d x\|_{\cV}^{2}\,\D s
    \le C_3\sum_{i\le j}(1+\lambda_{i,j})
       |\langle\bP_d x,\be_{i,j}\rangle|^{2}
       \int_0^\delta \E^{-2\lambda_{i,j}s}\,\D s
    \le C_3\bigl(\delta+\tfrac12\bigr)\|x\|_{\cH}^{2}.
  \]

  \emph{Martingale term.} By the It\^o isometry for the compensated jump
  integral, Tonelli's theorem, \eqref{eq:cV-weight-bound} and
  \eqref{eq:spectral-time-integral},
  \begin{align*}
    \int_0^\delta\EXspec{\MP_x^d}{\|M_s^{d,\cT}\|_{\cV}^{2}}\,\D s
    &=\EXspec{\MP_x^d}{\int_0^\delta\!\!\int_{\cH_d^{+}\setminus\{0\}}
       \Bigl(\int_r^\delta\|\cT_d(s-r)\xi\|_{\cV}^{2}\,\D s\Bigr)
       \nu^d(\D r,\D\xi)}\\
    &\le C_3\bigl(\delta+\tfrac12\bigr)\,
       \EXspec{\MP_x^d}{\int_0^\delta\!\!\int_{\cH_d^{+}\setminus\{0\}}
          \|\xi\|_{\cH}^{2}\,\nu^d(\D r,\D\xi)}\\
    &\le C_3\bigl(\delta+\tfrac12\bigr)
       \int_0^\delta\bigl(C_0+C_1\EXspec{\MP_x^d}{\|X_r^d\|_{\cH}^{2}}\bigr)\D r ,
  \end{align*}
  which is finite uniformly in $d$ by
  Lemma~\ref{lem:jump-bracket-finite-rank} and
  Lemma~\ref{lem:irregular-tightness-square-bound}. This step uses only the
  $\cH$ jump second moment, hence is valid in the state-dependent case
  $\mu\neq0$ as well.

  \emph{Drift term.} Set $A_\delta\df\int_0^\delta C(1+\sigma^{-1/2})\,\D\sigma
  =C(\delta+2\sqrt\delta)$. Using $\|\cT(r)\|_{\cL(\cH,\cV)}\le C(1+r^{-1/2})$,
  the growth bound~\eqref{eq:beta-growth-bound}, Cauchy--Schwarz in $r$, and the
  $\cH$-moment bound,
  \[
    \int_0^\delta\EXspec{\MP_x^d}{\Bigl\|\int_0^s\cT_d(s-r)\beta_d(X_r^d)\,\D r
      \Bigr\|_{\cV}^{2}}\D s
    \le A_\delta^{2}\,\delta\,
       \sup_{r\le\delta}\EXspec{\MP_x^d}{\|\beta_d(X_r^d)\|_{\cH}^{2}}<\infty
  \]
  uniformly in $d$. Summing the three contributions
  proves~\eqref{eq:cV-integrated}.

  It follows, applying \eqref{eq:cV-integrated} on $[\delta/2,\delta]$, that for
  every $d$ there is a deterministic time $s_d\in[\delta/2,\delta]$ with
  \[
    \EXspec{\MP_x^d}{\|X_{s_d}^d\|_{\cV}^{2}}
    \le \frac{2}{\delta}\,C_{\delta/2}(x)\df K_\delta(x),
  \]
  uniformly in $d$; in particular $X_{s_d}^d\in\cV\cap\cHplus$
  $\MP_x^d$-almost surely. Restarting the finite-rank Markov process at the
  deterministic time $s_d\le\delta$ and using the $\cV$-initial-datum maximal
  estimate~\eqref{eq:cV-uniform-initial} together with the Markov property,
  \[
    \EXspec{\MP_x^d}{\sup_{\delta\le t\le T}\|X_t^d\|_{\cV}^{2}}
    \le \EXspec{\MP_x^d}{\sup_{s_d\le t\le T}\|X_t^d\|_{\cV}^{2}}
    = \EXspec{\MP_x^d}{\Psi_d(X_{s_d}^d)}
    \le C_T\bigl(1+K_\delta(x)\bigr)\df C_{T,\delta}(x),
  \]
  where $\Psi_d(y)\df\EXspec{\MP_y^d}{\sup_{0\le t\le T}\|X_t^d\|_{\cV}^{2}}
  \le C_T(1+\|y\|_{\cV}^{2})$ by~\eqref{eq:cV-uniform-initial}. This is the
  assertion of the lemma.
\end{proof}

\begin{lemma}\label{lem:tightness-irregular}
Assume Assumption~\ref{assump:cV-compact-containment}. For every
$x\in\cHplus$ the sequence $(\MP^{d}_{x})_{d\in\MN}$ of laws of
$(X^{d})_{d\in\MN}$ is a tight sequence of measures on
$\cB(D(\MRplus,\cHplus))$.
\end{lemma}
\begin{proof}
We use Aldous' criterion on $D([0,T],\cHplus)$ for each fixed $T>0$. The state
space is the closed cone of the separable Hilbert space $\cH$. We write
$\cT_d(t)=\bP_d\cT(t)\bP_d$ and use the mild form
\[
  X_t^d=\cT_d(t)\bP_d x
        +\int_0^t\cT_d(t-s)\beta_d(X_s^d)\,\D s
        +M_t^{d,\cT},
\]
where $\beta_d$ is defined in~\eqref{eq:beta-growth-bound} and
\[
  M_t^{d,\cT}
  =\int_0^t\int_{\cH_d^+\setminus\{0\}}
    \cT_d(t-s)\xi\,(\mu^{X^d}-\nu^d)(\D s,\D\xi).
\]

\smallskip
\noindent\textbf{Step 1: Compact containment.}
We use the standard compactness criterion in Hilbert spaces: a family is
relatively compact if it is bounded and its tails with respect to the finite-rank
orthogonal projections $\bP_N$ vanish uniformly. The maximal
$\cH$-moment estimate in Lemma~\ref{lem:irregular-tightness-square-bound}
provides the boundedness. It remains to control the spectral tails.

Fix $\alpha>0$. On every interval $[\delta,T]$ with $\delta>0$,
Lemma~\ref{lem:irregular-smoothing-cV} and compactness of
$\cV\hookrightarrow\cH$ imply
\begin{align}\label{eq:tightness-tail-positive-time}
  \lim_{N\to\infty}\sup_{d\in\MN}
  \MP_x^d\!\left(
    \sup_{\delta\le s\le T}\|\bP_N^\perp X_s^d\|_{\cH}>\alpha
  \right)=0 .
\end{align}
Indeed
$\|\bP_N^\perp z\|_{\cH}\le
\|\bP_N^\perp\|_{\cL(\cV,\cH)}\|z\|_{\cV}$ and
$\|\bP_N^\perp\|_{\cL(\cV,\cH)}\to0$. This compactness follows from the
spectral Hilbert-scale compatibility in
Definition~\ref{def:admissible-irregular}(iv)(c), since
$\lambda_i\to\infty$ makes omitted spectral tails vanish uniformly on bounded
subsets of $\cV$.

On the short interval $[0,\delta]$ we work in $\cH$ and keep the Lyapunov part
inside the semigroup. The deterministic term satisfies
\[
  \sup_{0\le s\le\delta}
  \|\bP_N^\perp\cT_d(s)\bP_d x\|_{\cH}
  \le \|\bP_N^\perp x\|_{\cH},
\]
because the spectral projections commute with $\cT$ and $\cT$ is contractive.
For the drift convolution, contraction, Jensen's inequality,
\eqref{eq:beta-growth-bound}, and
Lemma~\ref{lem:irregular-tightness-square-bound} give
\[
  \sup_d
  \EXspec{\MP_x^d}{\sup_{0\le s\le\delta}
  \left\|\int_0^s\cT_d(s-r)\beta_d(X_r^d)\,\D r\right\|_{\cH}^{2}}
  \le C\delta^2 .
\]
For the martingale convolution, Lemma~\ref{lem:jump-convolution-maximal},
Lemma~\ref{lem:jump-bracket-finite-rank}, and the same moment estimate give
\[
  \sup_d
  \EXspec{\MP_x^d}{\sup_{0\le s\le\delta}\|M_s^{d,\cT}\|_{\cH}^{2}}
  \le C\delta .
\]
Consequently,
\begin{align}\label{eq:tightness-tail-short-time}
  \lim_{\delta\downarrow0}\limsup_{N\to\infty}\sup_{d\in\MN}
  \MP_x^d\!\left(
    \sup_{0\le s\le\delta}\|\bP_N^\perp X_s^d\|_{\cH}>\alpha
  \right)=0 .
\end{align}
We now make the diagonal compact set explicit. Fix $\varepsilon>0$. By
Lemma~\ref{lem:irregular-tightness-square-bound}, choose $R<\infty$ such that
\[
  \sup_d\MP_x^d\!\left(\sup_{s\le T}\|X_s^d\|_{\cH}>R\right)
  \le \varepsilon/2 .
\]
Let $\alpha_k=2^{-k}$, $k\ge1$. For each $k$, first choose
$\delta_k\in(0,T]$ small enough and then choose $N_k$ recursively increasing so
that \eqref{eq:tightness-tail-short-time} and
\eqref{eq:tightness-tail-positive-time} give
\[
  \sup_d\MP_x^d\!\left(
    \sup_{0\le s\le T}\|\bP_{N_k}^{\perp}X_s^d\|_{\cH}>\alpha_k
  \right)
  \le \varepsilon\,2^{-k-1}.
\]
Indeed, the probability on $[0,T]$ is bounded by the sum of the probabilities
on $[0,\delta_k]$ and $[\delta_k,T]$, and the two preceding tail estimates make
these terms arbitrarily small by first choosing $\delta_k$ and then $N_k$.
Define
\[
  K\df
  \left\{z\in\cHplus:\|z\|_{\cH}\le R,\ 
  \|\bP_{N_k}^{\perp}z\|_{\cH}\le\alpha_k
  \text{ for all }k\ge1\right\}.
\]
The set $K$ is closed. It is compact in $\cH$ by the Hilbert compactness
criterion: it is bounded, and for every $\eta>0$ one can choose $k$ with
$\alpha_k<\eta$, after which all elements of $K$ have spectral tail beyond the
finite-dimensional space $\cH_{N_k}$ bounded by $\eta$. Consequently,
\[
  \sup_{d\in\MN}\MP_x^d\bigl(X_s^d\notin K
  \text{ for some }s\in[0,T]\bigr)
  \le \varepsilon/2+\sum_{k\ge1}\varepsilon\,2^{-k-1}
  \le\varepsilon .
\]
This proves compact containment.

\smallskip
\noindent\textbf{Step 2: Aldous' increment condition.}
Let $(\tau_d)_{d\in\MN}$ be stopping times bounded by $T$ and let
$\theta_d\downarrow0$. The mild representation gives
\begin{align*}
  X_{\tau_d+\theta_d}^d-X_{\tau_d}^d
  &=(\cT_d(\theta_d)-I)X_{\tau_d}^d\\
  &\quad+\int_{\tau_d}^{\tau_d+\theta_d}
    \cT_d(\tau_d+\theta_d-s)\beta_d(X_s^d)\,\D s\\
  &\quad+N_d(\tau_d,\theta_d),
\end{align*}
where
\[
  N_d(\tau,\theta)
  \df
  \int_\tau^{\tau+\theta}\int_{\cH_d^+\setminus\{0\}}
    \cT_d(\tau+\theta-s)\xi\,(\mu^{X^d}-\nu^d)(\D s,\D\xi).
\]
For the semigroup term, compact containment from Step~1 and uniform strong
continuity of $\cT(h)$ on compact subsets of $\cH$ imply
\[
  \lim_{\theta\downarrow0}\sup_d\sup_{\tau_d\le T}
  \MP_x^d\!\left(
    \|(\cT_d(\theta)-I)X_{\tau_d}^d\|_{\cH}>\alpha
  \right)=0 .
\]
For the drift term, using contraction and the supremum moment estimate,
\[
  \EXspec{\MP_x^d}{
  \left\|\int_{\tau_d}^{\tau_d+\theta_d}
    \cT_d(\tau_d+\theta_d-s)\beta_d(X_s^d)\,\D s
  \right\|_{\cH}^{2}}
  \le
  C\theta_d^2\,
  \EXspec{\MP_x^d}{1+\sup_{r\le T+1}\|X_r^d\|_{\cH}^{2}}
  \le C\theta_d^2 .
\]
For the compensated jump convolution,
Corollary~\ref{cor:stopped-jump-convolution-maximal} gives
\begin{align*}
  \EXspec{\MP_x^d}{\|N_d(\tau_d,\theta_d)\|_{\cH}^{2}}
  &\le
  C\EXspec{\MP_x^d}{\int_{\tau_d}^{\tau_d+\theta_d}
    \int_{\cH_d^+\setminus\{0\}}\|\xi\|_{\cH}^{2}\,\nu^d(\D s,\D\xi)}\\
  &\le
  C\theta_d\,
  \EXspec{\MP_x^d}{1+\sup_{r\le T+1}\|X_r^d\|_{\cH}^{2}}
  \le C\theta_d .
\end{align*}
Chebyshev's inequality yields Aldous' increment condition. Since $T>0$ was
arbitrary, the sequence is tight on $D(\MRplus,\cHplus)$.
\end{proof}

\begin{proposition}\label{prop:irregular-weak-limit-mp}
  Let
  \begin{align*}
    \cD_0\df \lin\{f_u\colon f_u(y)\df \E^{-\langle y,u\rangle},
    \ u\in\dom(L)\cap\cHplus\}.
  \end{align*}
  Let $x\in\cHplus$ and let $\MP$ be a weak limit point of the sequence
  $(\MP_x^d)_{d\in\MN}$. Then $\MP$ solves the martingale problem for
  $(\cG,\cD_0)$ on $D(\MRplus,\cHplus)$, where on the exponential generators
  $f_u(y)=\E^{-\langle y,u\rangle}$ with $u\in\dom(L)\cap\cHplus$ the
  operator $\cG$ is given pointwise by
  \begin{align*}
    \cG \E^{-\langle \cdot,u\rangle}(y)
    \df \bigl(-F(u)-\langle y,R(u)\rangle\bigr)\E^{-\langle y,u\rangle},
    \qquad y\in\cHplus,\ u\in\dom(L)\cap\cHplus.
  \end{align*}
  The restriction $u\in\dom(L)\cap\cHplus$ is essential: $R(u)=L(u)+\hat R(u)$
  involves the unbounded Lyapunov operator $L$ and is therefore not defined for
  arbitrary $u\in\cHplus$. The class $\dom(L)\cap\cHplus$ is dense in $\cHplus$:
  for every $u\in\cHplus$, the spectral truncation $\bP_d(u)$ lies in
  $\cH_d^{+}\subseteq\dom(L)\cap\cHplus$ and converges to $u$ in $\cH$.
\end{proposition}
\begin{proof}
  This is the irregular analogue of
  \cite[Proposition~6.3]{karbach2023finiterank}, restricted to test loadings
  on which the limiting generator $\cG$ is well defined. By linearity it
  suffices to verify the martingale property on the generators
  $f_u(y)=\E^{-\langle y,u\rangle}$ for fixed $u\in\dom(L)\cap\cHplus$. Set
  $u_d\df \bP_d(u)\in\cH_d^{+}\subseteq\dom(L)\cap\cHplus$ and
  $f_u^d(y)\df\E^{-\langle y,u_d\rangle}$. By
  Proposition~\ref{prop:embedding-affine-main}, the process
  \begin{align*}
    M_t^{d,u}\df f_u^d(X_t^d)-f_u^d(\bP_d(x))
    -\int_0^t \cG^d f_u^d(X_s^d)\,\D s
  \end{align*}
  is a martingale under $\MP_x^d$. We pass to the limit only on this
  restricted test domain, where both $L(u)$ and $L(u_d)\to L(u)$ in $\cH$
  are defined (the latter by spectral truncation: if
  $u=\sum_{i\le j}u_{i,j}\be_{i,j}\in\dom(L)$, then
  $L(u_d)=-\sum_{i\le j\le d}\lambda_{i,j}u_{i,j}\be_{i,j}\to L(u)$ in
  $\cH$). Combined with the continuity of $F$ and $\hat R$ on $\cHplus$
  (\cite[Remark 3.4]{CKK22a}), this gives
  \begin{align*}
    f_u^d\longrightarrow f_u,\qquad
    \cG^d f_u^d\longrightarrow \cG f_u,
  \end{align*}
  locally uniformly on $\cHplus$. The generator terms have linear growth, so we
  add the uniform-integrability step explicitly. Since
  $F_d(u_d)=F(u_d)\to F(u)$ and
  $R_d(u_d)=L(u_d)+\bP_d\hat R(u_d)\to L(u)+\hat R(u)=R(u)$ in $\cH$, there is
  a constant $C_u<\infty$ such that, for all $d$ and all $y\in\cHplus$,
  \[
    |\cG^d f_u^d(y)|+|\cG f_u(y)|
    \le C_u(1+\|y\|_{\cH}).
  \]
  Indeed, $0\le f_u^d,f_u\le1$ on $\cHplus$, and the only unbounded term is the
  linear pairing with $y$. Lemma~\ref{lem:irregular-tightness-square-bound}
  implies, in fact,
  \[
    \sup_d
    \EXspec{\MP_x^d}{\int_0^T(1+\|X_s^d\|_{\cH}^{2})\,\D s}<\infty,
    \qquad T<\infty .
  \]
  The same estimate is inherited by any weak limit by Fatou's lemma, first for
  bounded continuous truncations and then by monotone convergence.

  Let $\rho_R\colon\cHplus\to[0,1]$ be continuous with $\rho_R(y)=1$ for
  $\|y\|\le R$ and $\rho_R(y)=0$ for $\|y\|\ge R+1$. The locally uniform
  convergence above gives convergence of the bounded continuous functions
  $\rho_R\cG^d f_u^d$ to $\rho_R\cG f_u$. The corresponding time-integral
  functionals are continuous at paths except possibly at their jump times, which
  have zero Lebesgue measure for every c\`adl\`ag path. Hence they pass to the
  weak limit. The error on $\{\|y\|>R\}$ is bounded uniformly by the displayed
  linear-growth estimate and
  $(1+\|y\|)\mathbf 1_{\{\|y\|>R\}}\le R^{-1}(1+\|y\|^2)$, and therefore tends
  to zero as $R\to\infty$. Thus the hypotheses of
  \cite[Chapter~4, Lemma~5.1]{EK86} are satisfied, and we conclude that
  \begin{align*}
    f_u(X_t)-f_u(x)-\int_0^t \cG f_u(X_s)\,\D s
  \end{align*}
is a martingale under $\MP$. Since $\cD_0$ is the linear span of these
  generators, the same holds for every test function in $\cD_0$.
\end{proof}

\begin{lemma}[No fixed-time jumps of weak limits]\label{lem:no-fixed-time-jumps}
  Assume Assumption~\ref{assump:cV-compact-containment}. Let $\MP_x$
  be any weak limit point of $(\MP_x^d)_{d\in\MN}$ on
  $D(\MRplus,\cHplus)$. Then, for every deterministic $t\ge0$,
  \[
    \MP_x(\Delta X_t=0)=1,
  \]
  where $\Delta X_t=X_t-X_{t-}$ for $t>0$ and $\Delta X_0=0$.
  Consequently, the coordinate map $\omega\mapsto\omega(t)$ is
  $\MP_x$-a.s.\ continuous on $D(\MRplus,\cHplus)$.
\end{lemma}
\begin{proof}
  The assertion is trivial at $t=0$ by convention. Fix a deterministic
  $t>0$, choose $T>t+1$, and let $\eta>0$. It is enough to prove the
  deterministic-time left and right oscillation estimates
  \begin{align}\label{eq:no-fixed-jump-forward}
    \lim_{h\downarrow0}\limsup_{d\to\infty}
    \MP_x^d\!\left(
      \sup_{0\le\theta\le h}\|X_{t+\theta}^d-X_t^d\|_{\cH}>\eta
    \right)=0,
  \end{align}
  and
  \begin{align}\label{eq:no-fixed-jump-backward}
    \lim_{h\downarrow0}\limsup_{d\to\infty}
    \MP_x^d\!\left(
      \sup_{0\le\theta\le h}\|X_t^d-X_{t-\theta}^d\|_{\cH}>\eta
    \right)=0.
  \end{align}
  We prove the right estimate; the left estimate is identical with
  $[t,t+h]$ replaced by $[t-h,t]$. Let $h<t/2$. The mild representation gives,
  for $0\le\theta\le h$,
  \[
    X_{t+\theta}^d-X_t^d
    =
    (\cT_d(\theta)-I)X_t^d
    +\int_t^{t+\theta}\cT_d(t+\theta-s)\beta_d(X_s^d)\,\D s
    +N_\theta^{d,h},
  \]
  where
  \[
    N_\theta^{d,h}\df
    \int_t^{t+\theta}\int_{\cH_d^+\setminus\{0\}}
    \cT_d(t+\theta-s)\xi\,(\mu^{X^d}-\nu^d)(\D s,\D\xi).
  \]

  First consider the semigroup term. By Lemma~\ref{lem:irregular-smoothing-cV},
  \[
    \sup_d
    \EXspec{\MP_x^d}{\sup_{t/2\le r\le T}\|X_r^d\|_{\cV}^{2}}<\infty .
  \]
  Hence, for every $\varepsilon>0$, we can choose $R$ such that
  \[
    \sup_d\MP_x^d\!\left(
      \sup_{t/2\le r\le T}\|X_r^d\|_{\cV}>R
    \right)<\varepsilon .
  \]
  The set $K_R=\{z\in\cV\cap\cHplus:\|z\|_{\cV}\le R\}$ is compact in
  $\cH$, and strong continuity of $(\cT(\theta))_{\theta\ge0}$ is uniform on
  compact subsets. Since $\cT_d(\theta)z=\cT(\theta)z$ for
  $z\in\cH_d$, it follows that
  \[
    \lim_{h\downarrow0}\sup_d
    \MP_x^d\!\left(
      \sup_{0\le\theta\le h}\|(\cT_d(\theta)-I)X_t^d\|_{\cH}>\eta
    \right)=0.
  \]

  For the drift contribution, \eqref{eq:beta-growth-bound} and
  Lemma~\ref{lem:irregular-tightness-square-bound} yield
  \[
    \EXspec{\MP_x^d}{
      \sup_{0\le\theta\le h}
      \left\|\int_t^{t+\theta}\cT_d(t+\theta-s)\beta_d(X_s^d)\,\D s
      \right\|_{\cH}^{2}}
    \le
    h\int_t^{t+h}
      \EXspec{\MP_x^d}{\|\beta_d(X_s^d)\|_{\cH}^{2}}\,\D s
    \le C h^2 ,
  \]
  uniformly in $d$. For the compensated jump convolution,
  Lemma~\ref{lem:jump-convolution-maximal} gives
  \[
    \EXspec{\MP_x^d}{\sup_{0\le\theta\le h}\|N_\theta^{d,h}\|_{\cH}^{2}}
    \le
    C
    \EXspec{\MP_x^d}{\int_t^{t+h}\int_{\cH_d^+\setminus\{0\}}
      \|\xi\|_{\cH}^{2}\,\nu^d(\D s,\D\xi)}
    \le C h,
  \]
  where the last inequality follows from
  Lemma~\ref{lem:jump-bracket-finite-rank} and
  Lemma~\ref{lem:irregular-tightness-square-bound}. Chebyshev's inequality gives
  \eqref{eq:no-fixed-jump-forward}; the same argument on $[t-h,t]$ gives
  \eqref{eq:no-fixed-jump-backward}.

  The deterministic-time oscillation criterion for weak convergence in the
  Skorohod topology (Ethier--Kurtz \cite[Chapter~3, Section~7]{EK86}) now
  implies that every weak limit has zero jump at the fixed time $t$. Since
  $t>0$ was arbitrary, the assertion holds for all deterministic times. The last
  statement is the usual characterization of continuity points of the coordinate
  projection on Skorohod space.
\end{proof}

\begin{lemma}[Riccati stability for moving Galerkin loadings]
  \label{lem:riccati-moving-loading-stability}
  Let $w_d\in\cH_d^+$ and $w\in\cHplus$ with
  $\|w_d-w\|_{\cH}\to0$. Then, for every $T<\infty$,
  \[
    \sup_{t\in[0,T]}
    \left(
      |\phi_d(t,w_d)-\phi(t,w)|
      +\|\psi_d(t,w_d)-\psi(t,w)\|
    \right)\to0.
  \]
\end{lemma}
\begin{proof}
  The proof is the same Gronwall argument as in
  Proposition~\ref{prop:existence-mild-Riccati}, with the initial loadings
  allowed to vary. Since $w_d\to w$, the family $\{w_d\}_{d\in\MN}\cup\{w\}$ is
  bounded in $\cH$. The local growth estimates for the mild Riccati equations
  therefore give a constant $M_T$ such that
  \[
    \sup_{d\in\MN}\sup_{t\le T}\|\psi_d(t,w_d)\|
    +\sup_{t\le T}\|\psi(t,w)\|\le M_T.
  \]
  Let $L_T$ be a Lipschitz constant of $\hat R$ on the ball
  $\{z\in\cHplus:\|z\|\le M_T\}$. Subtracting the mild equations and using
  $\bP_d\cT(t)=\cT(t)\bP_d$ gives
  \begin{align*}
    \|\psi_d(t,w_d)-\psi(t,w)\|
    &\le \|\cT(t)(w_d-w)\|\\
    &\quad+\int_0^t
      \|\bP_d\cT(t-s)
        [\hat R(\psi_d(s,w_d))-\hat R(\psi(s,w))]\|\,\D s\\
    &\quad+\int_0^t
      \|\bP_d^\perp\cT(t-s)\hat R(\psi(s,w))\|\,\D s\\
    &\le \|w_d-w\|
      +L_T\int_0^t\|\psi_d(s,w_d)-\psi(s,w)\|\,\D s+a_d(t),
  \end{align*}
  where
  \[
    a_d(t)\df
    \int_0^t\|\bP_d^\perp\cT(t-s)\hat R(\psi(s,w))\|\,\D s.
  \]
  The map $(t,s)\mapsto\cT(t-s)\hat R(\psi(s,w))$ has compact range in
  $\cH$ on the triangle $0\le s\le t\le T$, and
  $\bP_d^\perp\to0$ strongly; hence $\sup_{t\le T}a_d(t)\to0$. Gronwall's lemma
  yields $\sup_{t\le T}\|\psi_d(t,w_d)-\psi(t,w)\|\to0$. Finally,
  \[
    |\phi_d(t,w_d)-\phi(t,w)|
    \le\int_0^t|F(\psi_d(s,w_d))-F(\psi(s,w))|\,\D s,
  \]
  and the continuity of $F$ on bounded subsets of $\cHplus$ gives the
  claimed uniform convergence of $\phi_d$ as well.
\end{proof}

\begin{lemma}[Riccati semiflow]\label{lem:riccati-semiflow}
  For all $s,t\ge0$ and $u\in\cHplus$,
  \[
    \psi(t+s,u)=\psi(t,\psi(s,u)),\qquad
    \phi(t+s,u)=\phi(s,u)+\phi(t,\psi(s,u)).
  \]
  Equivalently, since the equation is autonomous,
  \[
    \psi(t+s,u)=\psi(s,\psi(t,u)),\qquad
    \phi(t+s,u)=\phi(t,u)+\phi(s,\psi(t,u)).
  \]
\end{lemma}
\begin{proof}
  Fix $s\ge0$ and set
  $\tilde\psi(t)\df\psi(t,\psi(s,u))$ and
  $\tilde\phi(t)\df\phi(s,u)+\phi(t,\psi(s,u))$. By the
  variation-of-constants formula and the semigroup property of $\cT$,
  $(\tilde\phi,\tilde\psi)$ solves the same mild Riccati system on
  $[0,\infty)$ as $(\phi(t+s,u),\psi(t+s,u))$, with the same initial value at
  $t=0$. Uniqueness in Proposition~\ref{prop:existence-mild-Riccati} gives the
  identities.
\end{proof}

\begin{lemma}[Laplace transforms determine laws on the cone]
\label{lem:laplace-measure-determining}
  Let $\nu_1$ and $\nu_2$ be Borel probability measures on $\cHplus$. If
  \[
    \int_{\cHplus}\E^{-\langle y,u\rangle}\,\nu_1(\D y)
    =
    \int_{\cHplus}\E^{-\langle y,u\rangle}\,\nu_2(\D y),
    \qquad u\in\cHplus,
  \]
  then $\nu_1=\nu_2$. The same conclusion holds on product cones
  $(\cHplus)^n$: the functions
  \[
    (y_1,\ldots,y_n)\mapsto
    \exp\!\left(-\sum_{k=1}^n\langle y_k,u_k\rangle\right),
    \qquad u_1,\ldots,u_n\in\cHplus,
  \]
  are measure determining.
\end{lemma}
\begin{proof}
  Since $\cH$ is a separable Hilbert space and $\cHplus$ is closed,
  $\cHplus$ is Polish. The cone $\cHplus$ is self-dual and generating in the
  Hilbert space of self-adjoint Hilbert--Schmidt operators: if
  $z=z^+-z^-$ is the spectral decomposition into positive and negative parts,
  then $z^\pm\in\cHplus$ and $z^+z^-=0$. Hence positive loadings separate
  points of $\cHplus$: if $\langle x-y,u\rangle=0$ for all $u\in\cHplus$, then
  testing with the positive and negative parts of $x-y$ gives $x=y$.

  Choose a countable dense set $(u_m)_{m\ge1}$ in $\cHplus$. By continuity of
  the pairings, this countable family still separates points. Moreover, the
  Borel $\sigma$-field of $\cHplus$ is generated by the coordinates
  $y\mapsto\langle y,u_m\rangle$: rational linear combinations of the
  $u_m$ generate a countable weakly separating family, and on a separable
  Hilbert space the weak and norm Borel $\sigma$-fields coincide.

  For any finite set $m_1,\ldots,m_k$ and any
  $\lambda_1,\ldots,\lambda_k\ge0$, equality of the cone Laplace transforms gives
  equality of the finite-dimensional Laplace transforms
  \[
    \int \exp\!\left(-\sum_{j=1}^k\lambda_j
      \langle y,u_{m_j}\rangle\right)\nu_1(\D y)
    =
    \int \exp\!\left(-\sum_{j=1}^k\lambda_j
      \langle y,u_{m_j}\rangle\right)\nu_2(\D y),
  \]
  because $\sum_j\lambda_j u_{m_j}\in\cHplus$. Classical uniqueness of
  Laplace transforms on $\MR_+^k$ therefore identifies all finite-dimensional
  distributions of the coordinate map
  $y\mapsto(\langle y,u_m\rangle)_{m\ge1}$. Since these coordinates generate the
  Borel $\sigma$-field, $\nu_1=\nu_2$. Applying the same argument to the product
  cone $(\cHplus)^n$ proves the final assertion.
\end{proof}

\begin{lemma}[Finite-dimensional Laplace identification]\label{lem:fdd-laplace-convergence}
  Assume the hypotheses of Theorem~\ref{thm:main-convergence}. Let
  $0<t_1<\cdots<t_n$, $u_1,\ldots,u_n\in\cHplus$, and $x\in\cHplus$. For
  each $d\in\MN$ define recursively
  \[
    v_n^d\df \bP_d u_n,\qquad a_n^d\df0,
  \]
  and, for $k=n-1,\ldots,1$,
  \[
    v_k^d\df \bP_d u_k+\psi_d(t_{k+1}-t_k,v_{k+1}^d),
    \qquad
    a_k^d\df a_{k+1}^d+\phi_d(t_{k+1}-t_k,v_{k+1}^d).
  \]
  Define $(v_k,a_k)$ by the same recursion with $\bP_d u_k$,
  $\phi_d$, and $\psi_d$ replaced by $u_k$, $\phi$, and $\psi$. Then
  \begin{align}\label{eq:fdd-laplace-galerkin}
    \EXspec{\MP_x^d}{\exp\!\left(-\sum_{k=1}^n
      \langle X_{t_k}^d,\bP_d u_k\rangle\right)}
    =
    \exp\!\left(
      -a_1^d-\phi_d(t_1,v_1^d)
      -\langle\bP_d x,\psi_d(t_1,v_1^d)\rangle
    \right),
  \end{align}
  and the right-hand side converges to
  \begin{align}\label{eq:fdd-laplace-limit}
    \exp\!\left(
      -a_1-\phi(t_1,v_1)-\langle x,\psi(t_1,v_1)\rangle
    \right).
  \end{align}
\end{lemma}
\begin{proof}
  The identity \eqref{eq:fdd-laplace-galerkin} follows by iterating
  the finite-rank affine Markov property from
  Proposition~\ref{prop:embedding-affine-main}(i), starting at time $t_n$ and
  moving backwards. Since $\cHplus$ is a cone and the Riccati flow preserves
  $\cHplus$, all recursively defined loadings stay in $\cHplus$. The convergence
  $v_k^d\to v_k$ and $a_k^d\to a_k$ follows by backward induction. The claim is
  immediate for $k=n$ because $\bP_d u_n\to u_n$. If
  $v_{k+1}^d\to v_{k+1}$ and $a_{k+1}^d\to a_{k+1}$, then
  Lemma~\ref{lem:riccati-moving-loading-stability}, applied at the time
  $t_{k+1}-t_k$, gives
  $\psi_d(t_{k+1}-t_k,v_{k+1}^d)\to
  \psi(t_{k+1}-t_k,v_{k+1})$ and
  $\phi_d(t_{k+1}-t_k,v_{k+1}^d)\to
  \phi(t_{k+1}-t_k,v_{k+1})$. Since $\bP_d u_k\to u_k$, the recursive
  definitions imply $v_k^d\to v_k$ and $a_k^d\to a_k$. Substituting these
  convergences in \eqref{eq:fdd-laplace-galerkin} gives
  \eqref{eq:fdd-laplace-limit}.
\end{proof}

    \begin{proof}[Proof of Theorem~\ref{thm:main-convergence}]
Fix $x\in\cHplus$. By Lemma~\ref{lem:tightness-irregular}, the family
$(\MP_x^d)_{d\in\MN}$ is tight on $D(\MRplus,\cHplus)$. Let $\MP$ be any weak
limit point. We identify $\MP$ directly through the finite-rank affine
transform formula; this avoids applying the limiting unbounded generator
$\cG$ to time-dependent test functions $g_{T,u}(t,\cdot)$ whose loadings
$\psi(T-t,u)$ need not lie in $\dom(L)$.

\emph{Step 1: bounded convergence on the marginal transform.}
Fix $t\ge0$ and $u\in\cHplus$. Set $u_d\df \bP_d(u)\in\cH_d^{+}$; then
$u_d\to u$ in $\cH$ as $d\to\infty$ and
$|\E^{-\langle X_t^d,u_d\rangle}|\le 1$ uniformly in $d$. By
Proposition~\ref{prop:embedding-affine-main}(i) the finite-rank affine
transform formula reads
\begin{align}\label{eq:affine-Galerkin-recall}
  \EXspec{\MP_x^d}{\E^{-\langle X_t^d,u_d\rangle}}
  =\E^{-\phi_d(t,u_d)-\langle \bP_d(x),\psi_d(t,u_d)\rangle}.
\end{align}
Let $\MP_x$ denote any weak limit point of $(\MP_x^d)_{d\in\MN}$
(along a subsequence). By Lemma~\ref{lem:tightness-irregular} such a limit
exists; the goal of this step is to compute its marginal Laplace transform.
For $t=0$ the assertion follows directly from $\bP_d x\to x$, so assume
$t>0$. By Lemma~\ref{lem:no-fixed-time-jumps}, the coordinate map
$\omega\mapsto\omega(t)$ is $\MP_x$-a.s.\ continuous. Hence the
continuous-mapping theorem gives
\[
  X_t^d\Rightarrow X_t
  \quad\text{under }\MP_x^d\Rightarrow\MP_x .
\]
The remaining issue is that the loading also varies with $d$. Since
$a,b\ge0$ imply $|\E^{-a}-\E^{-b}|\le |a-b|$, we have
\begin{align*}
&\left|
\EXspec{\MP_x^d}{\E^{-\langle X_t^d,u_d\rangle}}
-\EXspec{\MP_x}{\E^{-\langle X_t,u\rangle}}
\right|\\
&\quad\le
\EXspec{\MP_x^d}{\left|
\E^{-\langle X_t^d,u_d\rangle}
-\E^{-\langle X_t^d,u\rangle}
\right|}
+
\left|
\EXspec{\MP_x^d}{\E^{-\langle X_t^d,u\rangle}}
-\EXspec{\MP_x}{\E^{-\langle X_t,u\rangle}}
\right|\\
&\quad\le
\|u_d-u\|_{\cH}\,
\EXspec{\MP_x^d}{\|X_t^d\|_{\cH}}
+
\left|
\EXspec{\MP_x^d}{\E^{-\langle X_t^d,u\rangle}}
-\EXspec{\MP_x}{\E^{-\langle X_t,u\rangle}}
\right|.
\end{align*}
The first term tends to zero by $u_d\to u$ and
Lemma~\ref{lem:irregular-tightness-square-bound}; the second tends to zero by
the continuous-mapping theorem applied to the bounded continuous function
$y\mapsto\E^{-\langle y,u\rangle}$ on $\cHplus$. Thus
$\EXspec{\MP_x^d}{\E^{-\langle X_t^d,u_d\rangle}}
\to \EXspec{\MP_x}{\E^{-\langle X_t,u\rangle}}$.

\emph{Step 2: Riccati limit on the restricted class.}
For $u\in\dom(L)\cap\cHplus$, Proposition~\ref{prop:existence-mild-Riccati}
gives
\[
  \sup_{s\in[0,T]}
  \Big(
    |\phi_d(s,\bP_d u)-\phi(s,u)|
    +\|\psi_d(s,\bP_d u)-\psi(s,u)\|
  \Big)\to0 .
\]
Hence the right-hand side of~\eqref{eq:affine-Galerkin-recall}
converges to $\E^{-\phi(t,u)-\langle x,\psi(t,u)\rangle}$. Combining
Steps~1 and~2,
\begin{align}\label{eq:transform-on-domL}
  \EXspec{\MP_x}{\E^{-\langle X_t,u\rangle}}
  =\E^{-\phi(t,u)-\langle x,\psi(t,u)\rangle},
  \qquad u\in\dom(L)\cap\cHplus.
\end{align}

\emph{Step 3: extension by density.}
For any $u\in\cHplus$, the spectral truncations $u_n\df \bP_n(u)$ satisfy
$u_n\in\dom(L)\cap\cHplus$ and $u_n\to u$ in $\cH$. By
Lemma~\ref{lem:riccati-initial-continuity}, $\phi(t,u_n)\to\phi(t,u)$ and
$\psi(t,u_n)\to\psi(t,u)$ in $\cH$; hence both sides
of~\eqref{eq:transform-on-domL} are continuous in $u\in\cHplus$ (the
left side by bounded convergence, $|\E^{-\langle X_t,u_n\rangle}|\le 1$
and $\E^{-\langle y,u_n\rangle}\to\E^{-\langle y,u\rangle}$ pointwise).
The identity~\eqref{eq:transform-on-domL} therefore extends to every
$u\in\cHplus$, proving the affine transform
formula~\eqref{eq:affine-formula-main}.

\emph{Step 4: finite-dimensional distributions and uniqueness.}
The one-time identity is not by itself enough to determine a path law. We
therefore use Lemma~\ref{lem:fdd-laplace-convergence}. Let
$0<t_1<\cdots<t_n$ and $u_1,\ldots,u_n\in\cHplus$. Along any weakly convergent
subsequence, Lemma~\ref{lem:no-fixed-time-jumps} makes the coordinate map
$\omega\mapsto(\omega(t_1),\ldots,\omega(t_n))$ $\MP_x$-a.s.\ continuous.
Together with the uniform second-moment bound and the replacement
$\bP_d u_k\to u_k$, the continuous-mapping theorem gives
\[
  \EXspec{\MP_x^d}{\exp\!\left(-\sum_{k=1}^n
      \langle X_{t_k}^d,\bP_d u_k\rangle\right)}
  \longrightarrow
  \EXspec{\MP_x}{\exp\!\left(-\sum_{k=1}^n
      \langle X_{t_k},u_k\rangle\right)}
\]
for arbitrary deterministic times $0<t_1<\cdots<t_n$. The right-hand
side is then identified by \eqref{eq:fdd-laplace-limit}. The class
\[
  \left\{\exp\!\left(-\sum_{k=1}^n\langle y_k,u_k\rangle\right):
  u_1,\ldots,u_n\in\cHplus\right\}
  \]
  is multiplicatively closed and separates points on
  $(\cHplus)^n$. By Lemma~\ref{lem:laplace-measure-determining}, the
  finite-dimensional distributions of every weak limit point are uniquely
  determined.

It remains to make the Markov-affine structure explicit. For $t\ge0$ and
$x\in\cHplus$, let $p_t(x,\cdot)$ denote the unique probability measure on
$\cHplus$ determined by
\[
  \int_{\cHplus}\E^{-\langle y,u\rangle}\,p_t(x,\D y)
  =
  \E^{-\phi(t,u)-\langle x,\psi(t,u)\rangle},
  \qquad u\in\cHplus.
  \]
  Existence follows from the one-time weak limit above (with initial
  state $x$), and uniqueness follows from
  Lemma~\ref{lem:laplace-measure-determining}. This family is a Borel
  transition kernel. Indeed, $\cHplus$ is a closed subset of the separable
  Hilbert space $\cH$, hence is Polish. Choose a countable dense subset
  $D\subseteq\cHplus$ and set $e_u(y)=\E^{-\langle y,u\rangle}$ for
  $u\in D$. The map
  \[
    \Theta:\mathcal P(\cHplus)\to[0,1]^D,\qquad
    \Theta(\nu)=\left(\int e_u(y)\,\nu(\D y)\right)_{u\in D},
  \]
  is Borel and injective: density of $D$ and continuity in $u$ recover the
  Laplace transform for all $u\in\cHplus$, and
  Lemma~\ref{lem:laplace-measure-determining} then identifies the measure. By
  the Lusin--Souslin theorem, $\Theta^{-1}$ is Borel on $\Theta(\mathcal
  P(\cHplus))$. Since
  \[
    x\mapsto\Theta(p_t(x,\cdot))
    =
    \left(\E^{-\phi(t,u)-\langle x,\psi(t,u)\rangle}\right)_{u\in D}
  \]
  is Borel, the map $x\mapsto p_t(x,\cdot)$ is Borel as a map into
  $\mathcal P(\cHplus)$.
  By
Lemma~\ref{lem:riccati-semiflow},
\begin{align*}
  &\int_{\cHplus}\int_{\cHplus}
    \E^{-\langle z,u\rangle}\,p_t(y,\D z)\,p_s(x,\D y)\\
  &\qquad=
    \E^{-\phi(t,u)}
    \int_{\cHplus}\E^{-\langle y,\psi(t,u)\rangle}\,p_s(x,\D y)\\
  &\qquad=
    \E^{-\phi(t,u)-\phi(s,\psi(t,u))
      -\langle x,\psi(s,\psi(t,u))\rangle}
    =
    \E^{-\phi(t+s,u)-\langle x,\psi(t+s,u)\rangle}.
\end{align*}
The Laplace class therefore gives the Chapman--Kolmogorov identity
$p_{t+s}(x,\cdot)=\int p_t(y,\cdot)\,p_s(x,\D y)$. Iterating this transition
kernel yields exactly the finite-dimensional Laplace transforms identified
above, so the canonical process under $\MP_x$ is Markov and has affine
transition transform $p_t$. Hence all weak limit points coincide, and tightness
implies that the whole sequence $(\MP_x^d)$ converges weakly to the single
Markov affine law $\MP_x$ on $D(\MRplus,\cHplus)$. This proves the existence and
uniqueness of the affine process $(X,(\MP_x)_{x\in\cHplus})$. The martingale
problem of Proposition~\ref{prop:irregular-weak-limit-mp} (on the restricted
class $\cD_0$) is satisfied by every weak limit point; the law itself is
identified by the Markov affine finite-dimensional distributions just
constructed.

It remains to verify the trace-class assertion. Assume in addition
that $x\in\cL_1(H)\cap\cHplus$, that $m$ and $\mu$ are supported on
$\cL_1(H)\setminus\{0\}$, that $\Gamma(\cL_1(H))\subseteq\cL_1(H)$, and that
$b,\mu(A)\in\cL_1(H)$ for all $A\in\cB(\cHplus\setminus\{0\})$, and that the
trace first-moment bounds stated in
Theorem~\ref{thm:main-convergence}(a) hold. By
Lemma~\ref{lem:no-fixed-time-jumps}, the coordinate map at each deterministic
$t\ge0$ is $\MP_x$-a.s.\ continuous. Therefore the continuous-mapping theorem
gives $X_t^d\Rightarrow X_t$ for every fixed $t\ge0$. Moreover,
Lemma~\ref{lem:trace-bound-finite-rank} gives
\begin{align*}
  \sup_{d\in\MN}\sup_{t\in[0,T]}
  \EXspec{\MP_x^d}{\Tr(X_t^d)}<\infty,\qquad T>0.
\end{align*}
For each $n\in\MN$, define
\begin{align*}
  h_n(y)\df \sum_{k=1}^n\langle ye_k,e_k\rangle_H,\qquad y\in\cHplus.
\end{align*}
The map $h_n$ is continuous and non-negative on $\cHplus$, so
Portmanteau's theorem implies
\begin{align*}
  \EXspec{\MP_x}{h_n(X_t)}
  \le \liminf_{d\to\infty}\EXspec{\MP_x^d}{h_n(X_t^d)}
  \le \liminf_{d\to\infty}\EXspec{\MP_x^d}{\Tr(X_t^d)}
  <\infty,\qquad t\ge0.
\end{align*}
Letting $n\to\infty$ and using monotone convergence yields
\[
  \EXspec{\MP_x}{\Tr(X_t)}
  =\EXspec{\MP_x}{\sum_{k=1}^{\infty}\langle X_t e_k,e_k\rangle_H}<\infty,
  \qquad t\ge0.
\]
Hence $X_t\in\cL_1(H)$ almost surely for every $t\ge0$, which proves the
claim.
\end{proof}

\subsection{Finite-Rank Approximation of Heat Modulated Affine Stochastic Covariance
  Models}\label{sec:proof-section-2} 
 
\begin{lemma}\label{lem:joint-Riccati-convergence}
  Let $u=(u_1,u_2)\in \I H\times\cHplus$ and $T>0$. Then
  \begin{align*}
    \sup_{t\in[0,T]}
    \Big(
      |\Phi_d(t,u)-\Phi(t,u)|
      +\|\psi_{2,d}(t,u)-\psi_2(t,u)\|
    \Big)\longrightarrow 0,\qquad d\to\infty.
  \end{align*}
\end{lemma}
\begin{proof}
  Fix $u=(u_1,u_2)\in \I H\times\cHplus$ and $T>0$. Since
  $\psi_1(\cdot,u)$ solves the linear equation
  \eqref{eq:Riccati-phi-psi-1-2}, it is continuous on $[0,T]$ and therefore
  bounded. The map $v\mapsto \hat{\cR}(\psi_1(t,u),v)$ is locally Lipschitz on
  $\cHplus$, with Lipschitz constants independent of $t\in[0,T]$, because the
  dependence on the first argument enters only through the explicit additive
  term $-\frac12\psi_1(t,u)^{\otimes2}$. Arguing exactly as in the proof of
  Proposition~\ref{prop:existence-mild-Riccati}, one therefore obtains a
  constant $M_{T,u}>0$, independent of $d$, such that
  \begin{align*}
    \sup_{t\in[0,T]}
    \bigl(\|\psi_2(t,u)\|\vee \|\psi_{2,d}(t,u)\|\bigr)\le M_{T,u}.
  \end{align*}
  Subtracting the mild equations for $\psi_2$ and $\psi_{2,d}$, and
  using that $\cT(t)$ commutes with $\bP_d$, yields
  \begin{align*}
    \|\psi_2(t,u)-\psi_{2,d}(t,u)\|
    \le a_d(t)
    +C_{T,u}\int_0^t \|\psi_2(s,u)-\psi_{2,d}(s,u)\|\,\D s,
  \end{align*}
  where $C_{T,u}$ is a Lipschitz constant on the ball of radius
  $M_{T,u}$ and
  \begin{align*}
    a_d(t)\df
    \|\bP_d^\perp\cT(t)u_2\|
    +\int_0^t
    \|\bP_d^\perp\cT(t-s)\hat{\cR}(\psi_1(s,u),\psi_2(s,u))\|\,\D s.
  \end{align*}
  Both $t\mapsto\cT(t)u_2$ and
  $(t,s)\mapsto\cT(t-s)\hat{\cR}(\psi_1(s,u),\psi_2(s,u))$ are continuous on the
  compact sets $[0,T]$ and $\Delta_T\df\{(t,s):0\le s\le t\le T\}$ respectively,
  so their ranges are compact subsets $K_1,K_2$ of $\cH$. Since
  $(\bP_d^\perp)_{d\in\MN}$ is a decreasing sequence of orthogonal projections
  converging strongly to $0$ on $\cH$, the convergence
  $\bP_d^\perp z\to 0$ is uniform on every compact subset of $\cH$. Applied to
  $K_1$ and $K_2$ this yields
  $\sup_{t\in[0,T]}\|\bP_d^\perp\cT(t)u_2\|\to 0$ and
  $\sup_{(t,s)\in\Delta_T}\|\bP_d^\perp\cT(t-s)
  \hat{\cR}(\psi_1(s,u),\psi_2(s,u))\|\to 0$,
  hence $\sup_{t\in[0,T]}a_d(t)\to 0$. Gronwall's lemma therefore gives
  \begin{align*}
    \sup_{t\in[0,T]}
    \|\psi_{2,d}(t,u)-\psi_2(t,u)\|\longrightarrow 0.
  \end{align*}
  Finally, since
  $\Phi(t,u)=\int_0^t F(\psi_2(s,u))\,\D s$ and
  $\Phi_d(t,u)=\int_0^t F(\psi_{2,d}(s,u))\,\D s$, continuity of $F$ on bounded
  subsets of $\cHplus$ yields
  \begin{align*}
    \sup_{t\in[0,T]}|\Phi_d(t,u)-\Phi(t,u)|\longrightarrow 0,
  \end{align*}
  which proves the claim.
\end{proof}

\begin{lemma}[Finite-rank Feynman--Kac martingale]\label{lem:fk-finite-rank}
Let $(b,\mathbf B,m,\mu)$ be an admissible parameter set, fix $t>0$ and
$u=(u_1,u_2)\in\I H\times\cHplus$, and let $X^d$ be the finite-rank affine
process of Proposition~\ref{prop:embedding-affine-main}. With
$(\Phi_d(\cdot,u),\psi_1(\cdot,u),\psi_{2,d}(\cdot,u))$ the Galerkin joint
mild Riccati solution from Theorem~\ref{thm:finite-dim-approx-ScoV}(i), define
\[
  M^d_s
  \df
  \exp\!\Big(\tfrac12\!\int_0^s\!\langle X^d_r,\psi_1(t-r,u)^{\otimes2}\rangle\,\D r\Big)\,
  \exp\!\big(-\Phi_d(t-s,u)-\langle X^d_s,\psi_{2,d}(t-s,u)\rangle\big),
  \quad 0\le s\le t.
\]
Then $(M^d_s)_{0\le s\le t}$ is a true martingale under $\MP^d_x$, and there is
a constant $C_{t,u}>0$, independent of $d$, with $|M^d_s|\le C_{t,u}$ for all
$s\in[0,t]$.
\end{lemma}

\begin{proof}
On $\cH_d$ the process $X^d$ is a finite-dimensional affine semimartingale with
bounded generator $\cG^d$ (Proposition~\ref{prop:embedding-affine-main}(ii) and
eq.~\eqref{eq:G-d-operator}). The map
$(s,y)\mapsto g_d(s,y)\df\E^{-\Phi_d(t-s,u)-\langle y,\psi_{2,d}(t-s,u)\rangle}$
is $C^{1}$ in $s$ on $[0,t]$ and exponential-affine in $y$: indeed
$\Phi_d(\cdot,u)\in C^1$ and $\psi_{2,d}(\cdot,u)\in C^1$, the latter solving the
finite-rank Riccati system
\eqref{eq:Riccati-phi-psi-1-1}--\eqref{eq:Riccati-phi-psi-1-3} with $L$ replaced
by the bounded $\bP_dL\bP_d$. The potential
$V(s,y)\df\tfrac12\langle y,\psi_1(t-s,u)^{\otimes2}\rangle$ is continuous in $s$,
since $\psi_1(t-s,u)=S^*(t-s)u_1$ is.

Write $M^d_s=E_s\,g_d(s,X^d_s)$ with the finite-variation factor
$E_s\df\E^{\int_0^s V(r,X^d_r)\,\D r}$, so $\D E_s=E_s\,V(s,X^d_s)\,\D s$. Since
$g_d(s,\cdot)$ is a scalar multiple of $\E^{-\langle\cdot,\psi_{2,d}(t-s,u)\rangle}$
with $\psi_{2,d}(t-s,u)\in\cHplus_d$, it lies in the domain of $\cG^d$ for every
$s$. Applying the finite-dimensional It\^o formula for jump semimartingales to
the $C^1$ time-dependent exponential-affine map $(s,y)\mapsto g_d(s,y)$ gives
$\D g_d(s,X^d_s)=(\partial_s g_d+\cG^d g_d)(s,X^d_s)\,\D s+\D N_s$, with $N$ a
local martingale; here the operator $\cG^d$ is the finite-dimensional
integro-differential generator associated with the characteristics of $X^d$ from
Proposition~\ref{prop:embedding-affine-main}. The product rule for the
finite-variation $E_s$ then yields
\[
  \D M^d_s
  = E_s\Big[(\partial_s g_d+\cG^d g_d+V(s,\cdot)g_d)(s,X^d_s)\Big]\D s
  + E_s\,\D N_s.
\]
Using $\partial_s[-\Phi_d(t-s,u)]=F_d(\psi_{2,d}(t-s,u))$,
$\partial_s[-\langle y,\psi_{2,d}(t-s,u)\rangle]
=\langle y,\,R_d(\psi_{2,d}(t-s,u))-\tfrac12\bP_d(\psi_1(t-s,u)^{\otimes2})\rangle$
(the Galerkin form of \eqref{eq:Riccati-phi-psi-1-3} with the forcing
$-\tfrac12\psi_1^{\otimes2}$ from \eqref{eq:R-intro}), and
$\cG^d g_d(s,y)=\big(-F_d(\psi_{2,d}(t-s,u))-\langle y,R_d(\psi_{2,d}(t-s,u))\rangle\big)g_d(s,y)$,
the drift coefficient collapses to
\[
  \partial_s g_d+\cG^d g_d+V g_d
  = \tfrac12\langle y,\,\psi_1^{\otimes2}-\bP_d(\psi_1^{\otimes2})\rangle\,g_d
  = \tfrac12\langle y,\,\bP_d^\perp(\psi_1^{\otimes2})\rangle\,g_d
  = 0,\qquad y\in\cH_d,
\]
because $\langle y,\bP_d^\perp(\psi_1^{\otimes2})\rangle
=\langle\bP_d^\perp y,\psi_1^{\otimes2}\rangle=0$ for $y\in\cH_d$. Hence $M^d$ is
a local martingale.

For the uniform bound, write $u_1=\I h$. Then
$\psi_1(t-r,u)^{\otimes2}=(\I S^*(t-r)h)^{\otimes2}=-(S^*(t-r)h)^{\otimes2}$, so
for $X^d_r\in\cHplus$,
$\tfrac12\langle X^d_r,\psi_1(t-r,u)^{\otimes2}\rangle
=-\tfrac12\langle X^d_r,(S^*(t-r)h)^{\otimes2}\rangle\le0$, whence
$E_s\le1$. By Lemma~\ref{lem:joint-Riccati-convergence} there is $M_{t,u}$,
independent of $d$, with $\sup_{\tau\le t}\|\psi_{2,d}(\tau,u)\|\le M_{t,u}$;
since $F$ is continuous, $\sup_{\tau\le t}|\Phi_d(\tau,u)|\le C'_{t,u}$ uniformly
in $d$. As $\psi_{2,d}(t-s,u)\in\cHplus_d$ and $X^d_s\in\cHplus$ give
$\langle X^d_s,\psi_{2,d}(t-s,u)\rangle\ge0$, we obtain
$g_d(s,X^d_s)\le \E^{C'_{t,u}}$ and therefore
$|M^d_s|\le \E^{C'_{t,u}}\df C_{t,u}$ uniformly in $d$ and $s\in[0,t]$. A bounded
local martingale is a true martingale.
\end{proof}

\begin{proof}[Proof of Proposition~\ref{prop:affine-fk}]
Define the $d$-independent functional $\Lambda_s\colon D([0,t],\cHplus)\to\MR$,
\[
  \Lambda_s(\omega)
  \df
  \exp\!\Big(\tfrac12\!\int_0^s\!\langle\omega(r),\psi_1(t-r,u)^{\otimes2}\rangle\D r\Big)
  \E^{-\Phi(t-s,u)-\langle\omega(s),\psi_2(t-s,u)\rangle},
\]
so that $M_s=\Lambda_s(X)$, while $M^d_s=\Lambda_s^d(X^d)$ with $\Lambda^d_s$
the same functional carrying the Galerkin loadings $(\Phi_d,\psi_{2,d})$. As in
Lemma~\ref{lem:fk-finite-rank}, $0\le\Lambda_s,\Lambda^d_s\le C_{t,u}$.

\emph{Step 1 (loadings).} By Lemma~\ref{lem:joint-Riccati-convergence},
$\sup_{\tau\le t}\big(|\Phi_d(\tau,u)-\Phi(\tau,u)|
+\|\psi_{2,d}(\tau,u)-\psi_2(\tau,u)\|\big)\to0$. Hence, uniformly over paths in
$\{\omega:\sup_{r\le t}\|\omega(r)\|\le R\}$,
$\sup_{s\le t}|\Lambda^d_s(\omega)-\Lambda_s(\omega)|\to0$ as $d\to\infty$.

\emph{Step 2 (continuity).} For fixed $s$, the evaluation
$\omega\mapsto\omega(s)$ is continuous in the Skorohod topology at every $\omega$
continuous at $s$, and $\omega\mapsto\int_0^s\langle\omega(r),\cdot\rangle\D r$ is
Skorohod-continuous. Thus $\Lambda_s$ is bounded and continuous at every path
continuous at $s$. By Lemma~\ref{lem:no-fixed-time-jumps},
$\MP_x(\Delta X_s=0)=1$ for every deterministic $s$, hence $\Lambda_s$ is
$\MP_x$-a.s.\ continuous for every fixed $s$.

\emph{Step 3 (martingale property).} Fix deterministic times
$0\le s_1<\dots<s_k\le s'_1<s'_2\le t$ and a bounded continuous
$\varphi\colon(\cHplus)^k\to\MR$. By Lemma~\ref{lem:fk-finite-rank},
$\EXspec{\MP^d_x}{(M^d_{s'_2}-M^d_{s'_1})\,\varphi(X^d_{s_1},\dots,X^d_{s_k})}=0$
for every $d$. Writing
$M^d_{s'}=\Lambda_{s'}(X^d)+(\Lambda^d_{s'}-\Lambda_{s'})(X^d)$, the error term
vanishes in the limit: by Step~1 and the uniform bound
$\sup_d\EXspec{\MP^d_x}{\sup_{r\le t}\|X^d_r\|}<\infty$
(Lemma~\ref{lem:irregular-tightness-square-bound}), for any $\eps>0$ choose $R$
with $\sup_d\MP^d_x(\sup_r\|X^d_r\|>R)<\eps$; on the complementary event
$|\Lambda^d_{s'}-\Lambda_{s'}|\to0$ uniformly, while on the exceptional event the
integrand is bounded by $2C_{t,u}\|\varphi\|_\infty$. The remaining terms involve
only the $d$-independent, bounded, $\MP_x$-a.s.\ continuous functionals
$\Lambda_{s'_1},\Lambda_{s'_2},\varphi$ of the path; since $X^d\Rightarrow X$
(Theorem~\ref{thm:main-convergence}(2)) and Lemma~\ref{lem:no-fixed-time-jumps}
makes all fixed coordinate evaluations $\MP_x$-a.s.\ continuous, the
continuous-mapping theorem with the uniform bound
gives
\[
  \EXspec{\MP_x}{(M_{s'_2}-M_{s'_1})\,
    \varphi(X_{s_1},\dots,X_{s_k})}=0 .
\]
Cylinder functionals at deterministic times generate $\cF_{s'_1}$ by a
monotone-class argument, and these coordinate maps are $\MP_x$-a.s.\ continuous
by Lemma~\ref{lem:no-fixed-time-jumps}. Hence $M$ is a martingale with respect to the
$\MP_x$-augmented natural filtration of $X$. Evaluating
$\EXspec{\MP_x}{M_t}=M_0$ with $\psi_2(0,u)=u_2$ and $X_0=x$ yields
eq.~\eqref{eq:affine-additive-functional}.
\end{proof}

\begin{proof}[Proof of Theorem~\ref{thm:heat-affine-model}]
Fix $(f_0,x)\in H\times\cHplus$, let $u=(u_1,u_2)\in \I H\times\cHplus$,
and write $\sigma_t\df X_t^{1/2}$. By Assumption~\ref{assump:H-valued-joint},
the stochastic convolution in~\eqref{eq:HJMM-SDE-H} is well defined. Moreover,
the first Riccati equation~\eqref{eq:Riccati-phi-psi-1-2} is precisely the mild
dual shift equation, hence
\begin{align}\label{eq:psi1-dual-shift}
  \psi_1(t,u)=S^*(t)u_1,\qquad t\ge0.
\end{align}
Conditioning on $\cF_t^{(1)}=\sigma(X_s\colon 0\le s\le t)$, the random
variable $\langle f_t,u_1\rangle_H$ is Gaussian with mean
$\langle f_0,\psi_1(t,u)\rangle_H$ and covariance determined by
$\sigma_s^*S^*(t-s)u_1$. Therefore, using~\eqref{eq:psi1-dual-shift},
\begin{align}\label{eq:conditional-gaussian-heat}
  \EX{\E^{\langle f_t,u_1\rangle_H}\,\big|\,X}
  =
  \E^{\langle f_0,\psi_1(t,u)\rangle_H
  +\frac12\int_0^t\langle X_s,\psi_1(t-s,u)^{\otimes2}\rangle\,\D s}.
\end{align}
Hence
\begin{align}\label{eq:heat-affine-reduction}
  \EX{\E^{\langle f_t,u_1\rangle_H-\langle X_t,u_2\rangle}}
  =
  \E^{\langle f_0,\psi_1(t,u)\rangle_H}
  \EX{\E^{-\langle X_t,u_2\rangle
  +\frac12\int_0^t\langle X_s,\psi_1(t-s,u)^{\otimes2}\rangle\,\D s}}.
\end{align}

	To identify the remaining expectation, we use
	Proposition~\ref{prop:affine-fk}. Evaluating the martingale $M_s$ of that
	proposition at $s=t$ and $s=0$, and using $\psi_2(0,u)=u_2$ and $X_0=x$,
	gives
	\begin{align}\label{eq:affine-additive-functional}
	  \EX{\E^{-\langle X_t,u_2\rangle
	  +\frac12\int_0^t\langle X_s,\psi_1(t-s,u)^{\otimes2}\rangle\,\D s}}
  =
  \E^{-\Phi(t,u)-\langle x,\psi_2(t,u)\rangle}.
\end{align}
Combining~\eqref{eq:heat-affine-reduction} and
\eqref{eq:affine-additive-functional} yields
\begin{align*}
  \EX{\E^{\langle f_t,u_1\rangle_H-\langle X_t,u_2\rangle}}
  =
  \E^{-\Phi(t,u)+\langle f_0,\psi_1(t,u)\rangle_H-\langle x,\psi_2(t,u)\rangle},
  \qquad t\ge0,
\end{align*}
which is exactly~\eqref{eq:extended-affine-formula}.
\end{proof}

\begin{proof}[Proof of Theorem~\ref{thm:finite-dim-approx-ScoV}]
We begin with the proof of~\cref{item:finite-dim-approx-ScoV-1}. Fix
$d\in\MN$ and consider the finite-dimensional approximation
$(f_t^d,X_t^d)_{t\ge0}$, where $X^d$ is as in
Proposition~\ref{prop:embedding-affine-main} and, in particular,
$X_0^d=\bP_d(x)$. By part~(ii) of Proposition~\ref{prop:embedding-affine-main},
$X^d$ is a square-integrable affine Markovian semimartingale on $\cH_d^+$ on
some stochastic basis $(\Omega^d,\bar\cF^d,\bar\MF^d,\MP_x^d)$. Moreover, using
the semimartingale representation~\eqref{eq:decomp-martingale}, we see that
$(f^d,X^d)$ solves the stochastic covariance system
\begin{align}\label{eq:stochastic-covariance-model-approximation}
  \D \begin{bmatrix}
    f^{d}_{t}\\
    X^{d}_{t}
  \end{bmatrix}
  &=
  \left(
  \begin{bmatrix}
    \cA f^{d}_{t} \\
    \mathbf{\hat B}_{d}(X^{d}_{t})
  \end{bmatrix}
  +
  \begin{bmatrix}
    g_t \\
    \hat b_{d}
  \end{bmatrix}
  \right)\D t
  +
  \begin{bmatrix}
    (X^{d}_{t})^{1/2} & 0 \\
    0 & 0
  \end{bmatrix}
  \D \begin{bmatrix}
    W_{t} \\
    0
  \end{bmatrix}
  +
  \D \begin{bmatrix}
    0 \\
    \bar J^{d}_{t}
\end{bmatrix},
\end{align}
with $(f^{d}_{0},X^{d}_{0})=(f_0,\bP_{d}(x))\in H\times \cHplus$, drift
\[
  \hat b_{d}
  \df b_{d}+\int_{\cH_{d}^{+}\cap\{\|\xi\|>1\}}\xi\,m_{d}(\D\xi),
\]
and linear drift operator
\[
  \mathbf{\hat B}_{d}(u)
  \df \mathbf{B}_{d}(\bP_{d}u)
  +\int_{\cH_{d}^{+}\cap\{\|\xi\|>1\}}\xi\,\langle u,M_{d}(\D\xi)\rangle,
  \qquad u\in\cH.
\]
(Here $g_t$ denotes the HJM drift of the forward curve component; throughout
this subsection we work with the driftless specification $g_t\equiv0$.) Hence, the
finite-dimensional model is of the form in~\cite[Definition~2.9]{CKK22b}, and
\cref{item:finite-dim-approx-ScoV-1} follows from~\cite[Theorem~3.3]{CKK22b}.
Since $X_t^d\in\cH_d^+$ for every $t\ge0$, the operator
$(X_t^d)^{1/2}$ is finite-rank and therefore Hilbert--Schmidt on $H$.
Consequently the stochastic integral in the first component is well defined on
the $H$-level. In this cylindrical-noise formulation no additional background
covariance-operator commutativity condition is being imposed; the finite rank of
$(X_t^d)^{1/2}$ is the relevant Hilbert--Schmidt condition. The finite-rank
affine transform can also be recovered from the Feynman--Kac identity in
Lemma~\ref{lem:fk-finite-rank}.
Let us verify that the hypotheses of \cite[Theorem~3.3]{CKK22b} are
met. That theorem requires: (i) the state space is a finite-dimensional convex
cone with non-empty interior (here $\cH_d^+$, which is a closed convex cone in
the finite-dimensional space $\cH_d\cong\mathbb{R}^{d^2}$ with non-empty
relative interior consisting of strictly positive operators); (ii) the
parameters $(\hat{b}_d, \mathbf{\hat{B}}_d, m_d, M_d)$ form an admissible set
in the sense of \cite[Definition~2.6]{CKK22b}, which holds because $\hat{b}_d\in\cH_d^+$,
the map $\mathbf{\hat{B}}_d$ is a bounded linear operator on $\cH_d$ (since
$\mathbf{B}_d = \bP_d \mathbf{B}\bP_d$ is bounded in finite dimensions), and
$m_d$, $M_d$ satisfy the integrability conditions inherited from the
original admissible parameters $(b, B, m, \mu)$ via the projections
$\bP_d$; (iii) there exists a unique square-integrable affine Markovian
semimartingale with these parameters, which is guaranteed by
Proposition~\ref{prop:embedding-affine-main}(ii). All three conditions are
therefore satisfied, and the affine transform
formula~\eqref{eq:stochastic-covariance-affine-formula} follows.

We turn to~\cref{item:finite-dim-approx-ScoV-2}. Throughout this
proof we work in the finite-horizon canonical weighted forward-curve setting of
Example~\ref{ex:Laplacian_Forward}: $\Theta_{\max}<\infty$,
$V=V_{0,\beta}(0,\Theta_{\max})\oplus\MR$,
$H=L^2(0,\Theta_{\max},\mathrm e^{\beta x}\D x)\oplus\MR$, and
$B=\Delta\oplus 0$, with the corresponding finite-horizon adjoint shift
semigroup $S^*(s)$. Two structural facts of this canonical setting are used
below: the invariance $S^*(s)(V)\subseteq V$ with the growth bound recalled at
Step~3, and the eigenvalue asymptotics $\lambda_n\sim
(n\pi/\Theta_{\max})^2$ used to control the Lyapunov semigroup $\cT(t-s)$ on
$\cV$. We do not claim the rate
in~\eqref{eq:stochastic-covariance-convergence-rate} for the abstract
$V\hookrightarrow H\hookrightarrow V^*$ framework outside this canonical
setting.
Assume from now on that the
limiting $H$-valued stochastic covariance model $(f,X)$ from
Theorem~\ref{thm:heat-affine-model} is well defined under
Assumptions~\ref{assump:cV-compact-containment} and
\ref{assump:H-valued-joint}, as in the theorem statement. Assumption
\ref{assump:cV-compact-containment} is used here to invoke the limiting affine
covariance process and Proposition~\ref{prop:affine-fk}; the quantitative
estimate below uses the structural hypotheses displayed in item~(ii). We now
establish the convergence rate in \eqref{eq:stochastic-covariance-convergence-rate}.
Fix $T>0$ and $\tilde u\in V$, and set $u\df (\I\tilde u,0)$.
Throughout the argument, constants may depend on $T$ and $\|\tilde u\|_V$, but
are independent of $d\in\mathbb N$.\newline{}
The argument is the inhomogeneous analogue of the convergence estimate
in Proposition~\ref{prop:existence-mild-Riccati}; compare also the regular-case
proof of \cite[Proposition~2.7]{karbach2023finiterank}. The only new feature is
the forcing term $-\frac12(\psi_1\otimes\psi_1)$ in the equation for $\psi_2$,
which is explicit and therefore leaves the Gronwall structure unchanged.
\noindent\textbf{Step 1: Uniform bounds for the Riccati solutions.}
Let $M\ge0$ and consider the set
$\{v\in\cHplus:\|v\|\le M\}$.
For fixed $u=(\I\tilde u,0)$ and $s\in[0,T]$, set
\[
  G_s(v)\df \hat{\cR}(\psi_1(s,u),v)
  =\hat R(v)-\frac12(\psi_1(s,u))^{\otimes2},\qquad v\in\cHplus.
\]
The additive forcing is independent of $v$, hence the local
Lipschitz constant of $G_s$ on this set is the local Lipschitz constant of
$\hat R$, uniformly in $s\in[0,T]$ and independently of $d$. The same
Lipschitz bound applies to the projected map
$\bP_d\circ G_s\circ\bP_d$.

As a consequence, the mild solutions $\psi_2(\cdot,u)$ and
$\psi_{2,d}(\cdot,u)$ of the second Riccati equation satisfy the uniform bound
$$
  \|\psi_2(t,u)\|\vee\|\psi_{2,d}(t,u)\|
  \leq \tilde H_M,
  \qquad t\in[0,T],
$$
for some constant $\tilde H_M>0$ independent of $d$.
Note that, in contrast to Proposition~\ref{prop:existence-mild-Riccati},
the equation for $\psi_2$ contains the additional inhomogeneous term
$-\frac12(\psi_1\otimes\psi_1)$, where $\psi_1(t,u)=\I S^{*}(t)\tilde
u$ by Lemma~\ref{lem:joint-Riccati-wellposed}.
Consequently, $\tilde H_M$ depends on $\tilde u$ so the bound may differ from the bound used
earlier.

\noindent\textbf{Step 2: Error representation via affine transform formulas.}
Using the affine transform representations
\eqref{eq:extended-affine-formula} and
\eqref{eq:stochastic-covariance-affine-formula} for $(f,X)$ and $(f^d,X^d)$,
respectively, we note that the right-hand sides are of the form
$\E^{z(t)}$ and $\E^{z_d(t)}$ with
$z(t)=-\Phi(t,u)+\langle f_0,\psi_1(t,u)\rangle_H-\langle x,\psi_2(t,u)\rangle$
and $z_d(t)$ defined analogously; since $\psi_1(t,u)=\I S^{*}(t)\tilde u\in\I H$
by Lemma~\ref{lem:joint-Riccati-wellposed}, the pairing
$\langle f_0,\psi_1(t,u)\rangle_H$ is already purely imaginary, and no extra
factor of $\I$ is needed. Since $\Phi(t,u)\ge 0$ and $\langle
x,\psi_2(t,u)\rangle\ge 0$ (because $x\in\cHplus$, $\psi_2\in
C([0,T];\cHplus)$), and likewise $\Phi_d(t,u),\langle x,\psi_{2,d}(t,u)\rangle\ge
0$, the real parts of $z(t)$ and $z_d(t)$ are non-positive. The elementary bound
$|\E^{a+\I b}-\E^{a'+\I b'}|\le|a-a'|+|b-b'|$ for $a,a'\le 0$ and $b,b'\in\MR$
(which follows from $\E^{a}\le 1$, $\E^{a'}\le 1$ and $|\E^{\I b}-\E^{\I b'}|\le
|b-b'|$) therefore yields, without any extra constant,
\begin{align}\label{eq:proof-convergence-rate-scov}
  \sup_{t\in[0,T]}
  \Big|
    \EX{\E^{\I\langle f_t,\tilde u\rangle_H}}
    -
    \EX{\E^{\I\langle f_t^d,\tilde u\rangle_H}}
  \Big|
  &\le
  \sup_{t\in[0,T]}
  \Big(
    |\Phi(t,u)-\Phi_d(t,u)|
  \nonumber\\
  &\quad +\|x\|\,
      \|\psi_2(t,u)-\psi_{2,d}(t,u)\|
  \Big).
\end{align}
Thus, it remains to control the difference
$\psi_2-\psi_{2,d}$ uniformly on $[0,T]$; the corresponding scalar estimate for
$\Phi-\Phi_d$ then follows from local Lipschitz continuity of $F$.\newline{}
\noindent\textbf{Step 3: Stability estimate for $\psi_2-\psi_{2,d}$.}
Using the mild formulations of $\psi_2$ and $\psi_{2,d}$, the
commutation relation $\bP_d\cT(t)=\cT(t)\bP_d$, and the fact that both equations
start from $u_2=0$ and share the same projected explicit forcing
$-\tfrac12\bP_d((\psi_1)^{\otimes2})$, we obtain
\begin{align*}
  \|\bP_d\psi_2(t,u)-\psi_{2,d}(t,u)\|
  &\le
  \int_0^t
    \|G_s(\psi_2(s,u))-G_s(\psi_{2,d}(s,u))\|
    \,\D s \\
  &\le
  \tilde L_M^{(1)}
    \int_0^t
      \|\psi_2(s,u)-\psi_{2,d}(s,u)\|
    \,\D s,
\end{align*}
where $\tilde L_M^{(1)}$ is a Lipschitz constant of $\hat R$, and hence
of $G_s$, on $\{v\in\cHplus:\|v\|\le\tilde H_M\}$. Therefore, by the triangle
inequality,
\begin{align*}
  \|\psi_2(t,u)-\psi_{2,d}(t,u)\|
  &\le
  \|\bP_d^\perp\psi_2(t,u)\|
  +\|\bP_d\psi_2(t,u)-\psi_{2,d}(t,u)\| \\
  &\le
  \|\bP_d^\perp\psi_2(t,u)\|
  +\tilde L_M^{(1)}
    \int_0^t
      \|\psi_2(s,u)-\psi_{2,d}(s,u)\|
    \,\D s.
\end{align*}

\noindent\textbf{Step 4: Control of the projection error.}
Denote the integrand appearing in~\eqref{eq:R-hat} by
\[
  K(\xi,u)\df \E^{-\langle\xi,u\rangle}-1+\langle\chi(\xi),u\rangle,
\]
so that $\hat R$ admits the representation
\[
  \hat R(v)=\Gamma^{*}(v)-\int_{\cHpluso}K(\xi,v)\,\frac{\mu(\D\xi)}{\|\xi\|^{2}}.
\]

\emph{Jump-integrand bound (small/large jump separation).} For
$v\in\cHplus$ with $\|v\|\le M$, the integrand
$K(\xi,v)/\|\xi\|^{2}$ admits the bound $|K(\xi,v)|/\|\xi\|^{2}\le C_M$
uniformly in $\xi\in\cHpluso$, where $C_M$ depends only on $M$. On the
small-jump set $\{\|\xi\|\le 1\}$, a second-order Taylor expansion of
$\E^{-\langle\xi,v\rangle}$ around $0$ gives $|K(\xi,v)|\le
\tfrac12\langle\xi,v\rangle^{2}\le\tfrac12 M^{2}\|\xi\|^{2}$, so
$|K(\xi,v)|/\|\xi\|^{2}\le M^{2}/2$. On the large-jump set
$\{\|\xi\|>1\}$, the truncation $\chi(\xi)=0$ and the bound
$|\E^{-\langle\xi,v\rangle}-1|\le 2$ give
$|K(\xi,v)|/\|\xi\|^{2}\le 2/\|\xi\|^{2}\le 2$. Combining the two,
$|K(\xi,v)|/\|\xi\|^{2}\le \max(M^{2}/2,2)\df C_M$ uniformly in $\xi$.

Since $u_2=0$, the $\cT(t)u_2$-term in the mild variation-of-constants formula
vanishes, the unbounded generator contribution is absorbed into the semigroup
$\cT(t-s)$, and we obtain
\[
  \psi_2(t,u)
  =
  \int_0^t
  \cT(t-s)\!\left(
    \Gamma^{*}(\psi_2(s,u))
    -\!\int_{\cHpluso}\!\!K(\xi,\psi_2(s,u))\,\frac{\mu(\D\xi)}{\|\xi\|^2}
    -\frac12(\psi_1(s,u))^{\otimes2}
  \right)\!\D s.
\]
For the jump term, the strengthened hypothesis on $\mu$ is used exactly here.
If
\[
  J(v)\df\int_{\cHpluso}K(\xi,v)\,\frac{\mu(\D\xi)}{\|\xi\|^2},
  \qquad \|v\|\le M,
\]
then the Bartle total-variation estimate for $\cV$-valued vector measures gives
\[
  \|J(v)\|_{\cV}
  \le
  \sup_{\xi\in\cHpluso}\frac{|K(\xi,v)|}{\|\xi\|^2}\,
  |\mu|_{\cV}(\cHpluso)
  \le
  C_M\,|\mu|_{\cV}(\cHpluso).
\]
Thus the proof uses finite $\cV$-total variation of $\mu$, not merely the
$\cV$-norm of its total mass.
Consequently,
\begin{align}\label{eq:scov-convergence-rate-1}
  \|\bP_d^\perp\psi_2(t,u)\|
  &\le
  \int_{0}^{t}
    \|\bP_d^\perp\cT(t-s)\Gamma^{*}(\psi_2(s,u))\|\,\D s
  \nonumber\\
  &\quad+
  C_{\tilde H_M}\,|\mu|_{\cV}(\cHpluso)
  \int_0^t
    \|\bP_d^\perp\cT(t-s)\|_{\cL(\cV,\cH)}\,\D s
  \nonumber\\
  &\quad+\frac{t}{2}
    \sup_{s\in[0,t]}
    \|\bP_d^\perp\cT(t-s)(S^{*}(s)\tilde u)^{\otimes2}\|.
\end{align}
Here the three contributions correspond to the three forcing terms in
the mild formula: the drift-type term $\Gamma^{*}(\psi_2)$, the compensated
jump-integral term $K(\xi,\psi_2)\mu(\D\xi)/\|\xi\|^{2}$ (controlled by
the total-variation bound just established), and the inhomogeneous forcing
$-\tfrac12 \psi_1^{\otimes 2}$. For the last, we have kept
$\psi_1(s,u)=\I S^{*}(s)\tilde u$ from
Lemma~\ref{lem:joint-Riccati-wellposed}, and used the identity
$(\I v)^{\otimes 2}=-v^{\otimes 2}$ which only contributes a sign, absorbed
into the modulus on the right-hand side.
For the drift-type term, we use that $\cT(r)$ maps $\cH$ into
$\cV$ for every $r>0$: using the $B$-eigenbasis expansion established below,
for $x=\sum_{i,j}x_{i,j}\be_{i,j}\in\cH$ we have
$\|\cT(r)x\|_{\cV}^{2}\le \sum_{i,j}(\lambda_{i,j}+2\lambda)\mathrm e^{-2\lambda_{i,j}r}|x_{i,j}|^{2}
\le \kappa(r)^{2}\|x\|_{\cH}^{2}$, where
$\kappa(r)\df\sup_{i,j}\sqrt{\lambda_{i,j}+2\lambda}\,\mathrm e^{-\lambda_{i,j}r}<\infty$
for every $r>0$ and is integrable on $[0,T]$ (by Weyl asymptotics
$\lambda_{n}\simeq cn^{2}$, cf.~Remark~\ref{rem:lyapunov-regularity}(ii)).
Since $\Gamma^{*}\in\cL(\cH)$ and $\|\psi_2(s,u)\|\le\tilde H_M$ by
Step~1, we obtain
$\|\bP_d^\perp\cT(t-s)\Gamma^{*}(\psi_2(s,u))\|
\le\|\bP_d^\perp\|_{\cL(\cV,\cH)}\,\kappa(t-s)\,
\|\Gamma^{*}\|_{\cL(\cH)}\tilde H_M$, and the $s$-integral is finite.
Since $\tilde u\in V$ and $S^{*}(s)(V)\subseteq V$ with the growth
bound $\|S^{*}(s)\|_{\cL(V)}\le M_1\mathrm e^{w s}$
(by Lemma~\ref{lem:adjoint-shift-V}, with $M_1=1$ and $w=0$, in the
finite-horizon canonical specialization
$V=V_{0,\beta}(0,\Theta_{\max})\oplus\MR$ of
Example~\ref{ex:Laplacian_Forward}), we have $(S^{*}(s)\tilde u)^{\otimes 2}\in\cV$.
Moreover, $\cT(t-s)$ maps $\cV$ into $\cV$: by the spectral Hilbert-scale
compatibility in Definition~\ref{def:admissible-irregular}(iv)(c), the
eigenbasis $(e_n)_{n\in\MN}$ of $B$ lies in $V$, is dense in $V$, and satisfies
$\|v\|_V^2\simeq\sum_n(1+\lambda_n)|\langle v,e_n\rangle_H|^2$. Consequently,
if $u=\sum_{n}u_n e_n\in V$ then
$\sum_n(1+\lambda_n)|u_n|^{2}<\infty$ and
\[
  T(t)u=\sum_{n\in\MN}\E^{-\lambda_n t}u_n e_n
\]
still lies in $V$, because
$\sum_n(1+\lambda_n)\E^{-2\lambda_n t}|u_n|^{2}
\le\sum_n(1+\lambda_n)|u_n|^{2}$;
by norm equivalence this yields
$\|T(t)u\|_V\le C\|u\|_V$, hence $T(t)(V)\subseteq V$ uniformly in $t\geq0$.

For the dual component, by self-adjointness $T^{*}(t)=T(t)$ and duality,
\[
  \|T(t)\eta\|_{V^{*}}
  =
  \sup_{\|v\|_V\le 1}|\langle T(t)\eta,v\rangle_{V^{*},V}|
  =
  \sup_{\|v\|_V\le 1}|\langle\eta,T(t)v\rangle_{V^{*},V}|
  \le C\,\|\eta\|_{V^{*}},
\]
for $\eta\in V^{*}$. Thus
$T(t)(V^{*})\subseteq V^{*}$ uniformly in $t\ge 0$ as well.
Consequently, for $x\in\cV=\cL_2(V^{*},H)\cap\cL_2(H,V)$ both components
are preserved: $T(t)xT^{*}(t)\in\cL_2(H,V)$ (using $T(t)V\subseteq V$ on
the range) and $T(t)xT^{*}(t)\in\cL_2(V^{*},H)$ (using $T^{*}(t)V^{*}\subseteq V^{*}$
on the input). Hence $\cT(t-s)(\cV)\subseteq\cV$ uniformly in $t-s\ge 0$.
Set
\[
  M_{T,\cV}\df\sup_{0\le r\le T}\|\cT(r)\|_{\cL(\cV,\cV)}<\infty .
\]
The jump contribution in~\eqref{eq:scov-convergence-rate-1} is therefore
bounded by
\[
  T\,C_{\tilde H_M}\,M_{T,\cV}\,
  |\mu|_{\cV}(\cHpluso)\,\|\bP_d^\perp\|_{\cL(\cV,\cH)}.
\]
Moreover, $\cT(t-s)(S^{*}(s)\tilde u)^{\otimes2}\in\cV$.
Using $\|\bP_d^\perp\|_{\cL(\cV,\cH)}\to0$ and the growth bound
$\|S^{*}(s)\|_{\cL(V)}\le M_1\e^{ws}$, we obtain
\[
  \sup_{s\in[0,t]}
  \|\bP_d^\perp\cT(t-s)(S^{*}(s)\tilde u)^{\otimes2}\|
  \le
  \|\bP_d^\perp\|_{\cL(\cV,\cH)}
  M_T M_1^2\e^{2tw}\|\tilde u\|_V^2.
\]

\noindent\textbf{Step 5: Gronwall argument and rate.}
Define the error function
$e_d(t)\df \sup_{r\in[0,t]}\|\psi_2(r,u)-\psi_{2,d}(r,u)\|$. Combining the
bound on $\|\bP_d^\perp\psi_2(t,u)\|$ from Step~4 (each of the three terms
on the right-hand side of \eqref{eq:scov-convergence-rate-1} is bounded by a
$d$-uniform constant times $\|\bP_d^\perp\|_{\cL(\cV,\cH)}$) with the
Lipschitz stability bound from Step~3, we obtain
\[
  e_d(t)
  \le
  A_d
  +
  \tilde L_M^{(1)}\int_0^t e_d(s)\,\D s,
\]
where the inhomogeneous term collects the three projection
contributions:
\[
  A_d
  \df
  \big(
    T\,\kappa_{\Gamma}\,\|\Gamma^{*}\|_{\cL(\cH)}\,\tilde H_M
    +
    T\,C_{\tilde H_M}\,M_{T,\cV}\,|\mu|_{\cV}(\cHpluso)
    +
    \tfrac{T}{2}\,M_T M_1^{2}\E^{2Tw}\|\tilde u\|_V^{2}
  \big)\,
  \|\bP_d^\perp\|_{\cL(\cV,\cH)},
\]
with $\kappa_{\Gamma}=\int_0^T\kappa(r)\D r<\infty$ and $\tilde L_M^{(1)}$
the local Lipschitz constant of $\hat R$ on $\{\|v\|\le\tilde H_M\}$.
Gronwall's inequality therefore gives
\[
  \sup_{t\in[0,T]}\|\psi_2(t,u)-\psi_{2,d}(t,u)\|
  \le
  A_T\,\E^{\tilde L_M^{(1)}T}
  \df
  \tilde C_T\,\|\bP_d^\perp\|_{\cL(\cV,\cH)},
\]
with $\tilde C_T$ independent of $d$.
The scalar Riccati component has the same rate. Let $L_F$ be a Lipschitz
constant of $F$ on the ball containing
$\{\psi_2(s,u),\psi_{2,d}(s,u):0\le s\le T,\ d\in\MN\}$. Since
\[
  \Phi(t,u)=\int_0^t F(\psi_2(s,u))\,\D s,
  \qquad
  \Phi_d(t,u)=\int_0^t F(\psi_{2,d}(s,u))\,\D s,
\]
we obtain
\[
  \sup_{t\in[0,T]}|\Phi(t,u)-\Phi_d(t,u)|
  \le
  L_F T
  \sup_{s\in[0,T]}\|\psi_2(s,u)-\psi_{2,d}(s,u)\|
  \le
  C_T^\Phi\,\|\bP_d^\perp\|_{\cL(\cV,\cH)}
\]
for a constant $C_T^\Phi$ independent of $d$.
Inserting the two bounds into \eqref{eq:proof-convergence-rate-scov} and
increasing $\tilde C_T$ if necessary proves
\eqref{eq:stochastic-covariance-convergence-rate}, with the final
$(1+\|x\|)$ factor from \eqref{eq:proof-convergence-rate-scov} understood
in the $\cH$ norm (i.e.\ $\|x\|=\|x\|_{\cH}$, consistent with the
statement of Theorem~\ref{thm:finite-dim-approx-ScoV}(ii)).

\end{proof}

\section*{Acknowledgments}
The author thanks an anonymous referee for a careful reading of the manuscript and
for numerous constructive comments and suggestions that significantly improved
the clarity and quality of this paper.

\bibliographystyle{abbrv}

\begin{thebibliography}{10}
    \bibitem{Aldous1978}
D.~J.~Aldous,
\newblock Stopping times and tightness,
\newblock \emph{Annals of Probability}, 6(2):335--340, 1978.

    
\bibitem{AKW10}
A.~Andresen, S.~Koekebakker, and S.~Westgaard.
\newblock Modeling electricity forward prices using the multivariate normal
  inverse gaussian distribution.
\newblock {\em Journal of Energy Markets}, 3(3):3--25, sep 2010.

\bibitem{BNS02}
O.~E. Barndorff-Nielsen and N.~Shephard.
\newblock {Econometric analysis of realized volatility and its use in
  estimating stochastic volatility models}.
\newblock {\em J. R. Stat. Soc., Ser. B, Stat. Methodol.}, 64(2):253--280,
  2002.

\bibitem{BNS13}
O.~E. Barndorff-Nielsen and R.~Stelzer.
\newblock The multivariate supou stochastic volatility model.
\newblock {\em Mathematical Finance}, 23(2):275--296, 2013.

\bibitem{BBK08}
F.~E. Benth, J.~\v{S}altyt\.e Benth, and S.~Koekebakker.
\newblock {\em Stochastic Modeling of Electricity and Related Markets}.
\newblock World Scientific Publishing Co. Pte. Ltd., Singapore, 2008.

\bibitem{BDL21}
F.~E. Benth, N.~Detering, and S.~Lavagnini.
\newblock Accuracy of deep learning in calibrating {HJM} forward curves.
\newblock {\em Digital Finance}, 3(3):209--248, 2021.

\bibitem{BE24}
F.~E. Benth and H.~Eyjolfsson.
\newblock Robustness of {Hilbert} space-valued stochastic volatility models.
\newblock {\em Finance Stoch.}, 28(4):1117--1146, 2024.
\newblock doi:10.1007/s00780-024-00542-4.

\bibitem{BSB12}
F.~E. {Benth} and {J\=urat\.e \v{S}altyt\.e Benth}.
\newblock {\em {Modeling and pricing in financial markets for weather
  derivatives}}, volume~17.
\newblock Hackensack, NJ: World Scientific, 2012.

\bibitem{BK23}
F.~E. Benth and S.~Karbach.
\newblock Multivariate continuous-time autoregressive moving-average processes
  on cones.
\newblock {\em Stochastic Processes Appl.}, 162:299--337, 2023.

\bibitem{BK14}
F.~E. {Benth} and P.~{Kr\"uhner}.
\newblock {Representation of infinite-dimensional forward price models in
  commodity markets}.
\newblock {\em {Commun. Math. Stat.}}, 2(1):47--106, 2014.

\bibitem{BK15}
F.~E. {Benth} and P.~{Kr\"uhner}.
\newblock {Derivatives pricing in energy markets: an infinite-dimensional
  approach}.
\newblock {\em {SIAM J. Financ. Math.}}, 6:825--869, 2015.

\bibitem{BenKru23}
F.~E. Benth and P.~Kr{\"u}hner.
\newblock {\em Stochastic models for prices dynamics in energy and commodity
  markets. {An} infinite-dimensional perspective}.
\newblock Springer Finance. Cham: Springer, 2023.

\bibitem{BLDP22}
F.~E. Benth, G.~J. Lord, G.~Di~Nunno, and A.~E. Petersson.
\newblock The heat modulated infinite dimensional Heston model and its
  numerical approximation.
\newblock {\em Stochastics}, 97(8):1038--1078, 2025.
\newblock doi:10.1080/17442508.2024.2424867.

\bibitem{BRS18}
F.~E. {Benth}, B.~{R\"udiger}, and A.~{S\"uss}.
\newblock {Ornstein-Uhlenbeck processes in Hilbert space with non-Gaussian
  stochastic volatility}.
\newblock {\em {Stochastic Processes Appl.}}, 128(2):461--486, 2018.

\bibitem{BSV22a}
F.~E. Benth, D.~Schroers, and A.~E.~D. Veraart.
\newblock {A weak law of large numbers for realised covariation in a Hilbert
  space setting}.
\newblock {\em Stochastic Processes and their Applications}, 145:241--268,
  2022.

\bibitem{BSV22b}
F.~E. Benth, D.~Schroers, and A.~E.~D. Veraart.
\newblock {A feasible central limit theorem for realised covariation of {SPDEs}
  in the context of functional data}.
\newblock {\em Ann. Appl. Probab.}, 34(2):2208--2242, 2024.
\newblock doi:10.1214/23-AAP2019.

\bibitem{BS24}
F.~E. Benth and C.~Sgarra.
\newblock A {Barndorff-Nielsen} and {Shephard} model with leverage in {Hilbert}
  space for commodity forward markets.
\newblock {\em Finance Stoch.}, 28(4):1035--1076, 2024.
\newblock doi:10.1007/s00780-024-00546-0.

\bibitem{BS18}
F.~E. {Benth} and I.~C. {Simonsen}.
\newblock {The Heston stochastic volatility model in Hilbert space}.
\newblock {\em {Stochastic Anal. Appl.}}, 36(4):733--750, 2018.

\bibitem{BS73}
F.~Black and M.~Scholes.
\newblock The pricing of options and corporate liabilities.
\newblock {\em Journal of Political Economy}, 81(3):637--654, 1973.

\bibitem{CT06}
R.~A. {Carmona} and M.~R. {Tehranchi}.
\newblock {\em {Interest rate models: an infinite dimensional stochastic
  analysis perspective}}.
\newblock Berlin: Springer, 2006.

\bibitem{Con05}
{\sc {Cont}, R.}
\newblock {Modeling term structure dynamics: an infinite dimensional approach}.
\newblock {\em {Int. J. Theor. Appl. Finance} 8}, 3 (2005), 357--380.

\bibitem{CKK22a}
S.~Cox, S.~Karbach, and A.~Khedher.
\newblock Affine pure-jump processes on positive {Hilbert}-{Schmidt} operators.
\newblock {\em Stochastic Processes Appl.}, 151:191--229, 2022.

\bibitem{CKK22b}
S.~Cox, S.~Karbach, and A.~Khedher.
\newblock {An infinite-dimensional affine stochastic volatility model}.
\newblock {\em Math. Finance}, 32(3):878--906, 2022.

\bibitem{CrumpGospodinov2022}
R.~K.~Crump and N.~Gospodinov,
\newblock On the factor structure of bond returns,
\newblock \emph{Econometrica} \textbf{90} (2022), no.~1, 295--314.

\bibitem{CFMT11}
C.~Cuchiero, D.~Filipovi{\'c}, E.~Mayerhofer, and J.~Teichmann.
\newblock {Affine processes on positive semidefinite matrices}.
\newblock {\em {Ann. Appl. Probab.}}, 21(2):397--463, 2011.

\bibitem{CT13}
C.~Cuchiero and J.~Teichmann.
\newblock Path properties and regularity of affine processes on general state
  spaces.
\newblock In {\em S\'eminaire de probabilit\'es XLV}, pages 201--244. Cham:
  Springer, 2013.

\bibitem{DFGT07}
J.~Da~Fonseca, M.~Grasselli, and C.~Tebaldi.
\newblock Option pricing when correlations are stochastic: an analytical
  framework.
\newblock {\em Rev. Deriv. Res.}, 10(2):151--180, 2007.

\bibitem{Dei77}
K.~{Deimling}.
\newblock {\em {Ordinary differential equations in Banach spaces}}, volume 596
  of {\em Lecture Notes in Mathematics}.
\newblock Springer-Verlag, Berlin--Heidelberg--New York, 1977.

\bibitem{Det77}
E.~{Dettweiler}.
\newblock {Infinitely divisible measures on the cone of an ordered locally
  convex vector spaces}.
\newblock {\em Ann. Sci. Univ. Clermont-Ferrand II, Math.},
  61(14):11--17, 1976.

\bibitem{Dou14}
R.~Douady.
\newblock Yield curve smoothing and residual variance of fixed income
  positions.
\newblock In Y.~Kabanov, M.~Rutkowski, and T.~Zariphopoulou, editors,
  {\em Inspired by Finance: The Musiela Festschrift}, pages 221--256.
\newblock Springer International Publishing, Cham, 2014.

\bibitem{Dupire1994}
B.~Dupire.
\newblock Pricing with a smile.
\newblock {\em Risk}, 7(1):18--20, 1994.

\bibitem{EK86}
S.~N. Ethier and T.~G. Kurtz.
\newblock {\em Markov Processes: Characterization and Convergence}.
\newblock Wiley Series in Probability and Mathematical Statistics. John Wiley
  \& Sons, Inc., New York, 1986.

\bibitem{farkas2015isem}
A.~B\'atkai, B.~Farkas, P.~Csom\'os, and A.~Ostermann.
\newblock {\em Operator Semigroups for Numerical Analysis}.
\newblock Lecture notes for the 15th Internet Seminar on Evolution Equations,
  2011/12.

\bibitem{Fil01}
D.~{Filipovi\'c}.
\newblock {\em {Consistency problems for Heath-Jarrow-Morton interest rate
  models}}, volume 1760.
\newblock Berlin: Springer, 2001.

\bibitem{FTT10}
D.~{Filipovi\'c}, S.~{Tappe}, and J.~{Teichmann}.
\newblock {Term structure models driven by Wiener processes and Poisson
  measures: existence and positivity}.
\newblock {\em {SIAM J. Financ. Math.}}, 1:523--554, 2010.

\bibitem{Fre08}
D.~Frestad.
\newblock Common and unique factors influencing daily swap returns in the
  nordic electricity market, 1997-2005.
\newblock {\em Energy Economics}, 30(3):1081--1097, May 2008.

\bibitem{FK24}
M.~Friesen and S.~Karbach.
\newblock Stationary covariance regime for affine stochastic covariance models
  in {Hilbert} spaces.
\newblock {\em Finance Stoch.}, 28(4):1077--1116, 2024.

\bibitem{GJR18}
J.~Gatheral, T.~Jaisson, and M.~Rosenbaum.
\newblock Volatility is rough.
\newblock {\em Quant. Finance}, 18(6):933--949, 2018.

\bibitem{GM11}
L.~{Gawarecki} and V.~{Mandrekar}.
\newblock {\em {Stochastic differential equations in infinite dimensions with
  applications to stochastic partial differential equations}}.
\newblock Berlin: Springer, 2011.

\bibitem{GS10}
C.~{Gourieroux} and R.~{Sufana}.
\newblock {Derivative pricing with Wishart multivariate stochastic volatility}.
\newblock {\em {J. Bus. Econ. Stat.}}, 28(3):438--451, 2010.

\bibitem{Goe84}
R.~{Göthel}.
\newblock {Faedo-Galerkin approximations in equations of evolution}.
\newblock {\em {Math. Methods Appl. Sci.}}, 6:41--54, 1984.

\bibitem{HKLW02}
P.~S. Hagan, D.~Kumar, A.~Lesniewski, and D.~E. Woodward.
\newblock Managing smile risk.
\newblock {\em Wilmott Magazine}, September:84--108, 2002.

\bibitem{HJM92}
D.~{Heath}, R.~{Jarrow}, and A.~{Morton}.
\newblock {Bond pricing and the term structure of interest rates: a new
  methodology for contingent claims valuation}.
\newblock {\em {Econometrica}}, 60(1):77--105, 1992.

\bibitem{HeKarbachKhedher2025}
J.~He, S.~Karbach, and A.~Khedher.
\newblock Pricing options on forwards in function-valued affine stochastic volatility models.
\newblock {\em arXiv preprint}, arXiv:2508.14813, 2025.


\bibitem{Heston1993}
S.~L. Heston.
\newblock A closed-form solution for options with stochastic volatility with
  applications to bond and currency options.
\newblock {\em Review of Financial Studies}, 6(2):327--343, 1993.


\bibitem{JS03}
J.~{Jacod} and A.~N. {Shiryaev}.
\newblock {\em {Limit theorems for stochastic processes}}, volume 288.
\newblock Berlin: Springer, 2003.

\bibitem{Jak86}
A.~Jakubowski.
\newblock On the {S}korokhod topology.
\newblock {\em Annales de l'I.H.P. Probabilit\'es et statistiques},
  22(3):263--285, 1986.

\bibitem{KP15}
J.~Kallsen and P.~Kr{\"u}hner.
\newblock On a {Heath}-{Jarrow}-{Morton} approach for stock options.
\newblock {\em Finance Stoch.}, 19(3):583--615, 2015.

\bibitem{Kar22}
S.~Karbach.
\newblock {\em {Stochastic covariance models in Hilbert spaces with jumps}}.
\newblock PhD thesis, University of Amsterdam, 2022.

\bibitem{karbach2023finiterank}
S.~Karbach.
\newblock Finite-rank approximation of affine processes on positive
  {Hilbert--Schmidt} operators.
\newblock {\em J. Math. Anal. Appl.}, 553(2), Article 129852, 2026.
\newblock doi:10.1016/j.jmaa.2025.129852.

\bibitem{KO05}
S.~Koekebakker and F.~Ollmar.
\newblock Forward curve dynamics in the nordic electricity market.
\newblock {\em Managerial Finance}, 31(6):73--94, 2005.

\bibitem{LT08}
M.~Leippold and F.~Trojani.
\newblock Asset pricing with matrix jump diffusions.
\newblock SSRN working paper, 2008.
\newblock doi:10.2139/ssrn.1274482.

\bibitem{Mar76}
R.~H. {Martin}.
\newblock {\em {Nonlinear operators and differential equations in Banach
  spaces}}.
\newblock John Wiley\&Sons, 1976.

\bibitem{May12}
E.~{Mayerhofer}.
\newblock {Affine processes on positive semidefinite \(d \times d\) matrices
  have jumps of finite variation in dimension \(d > 1\)}.
\newblock {\em {Stochastic Processes Appl.}}, 122(10):3445--3459, 2012.

\bibitem{MKPS12}
J.~Muhle-Karbe, O.~Pfaffel, and R.~Stelzer.
\newblock Option pricing in multivariate stochastic volatility models of {OU}
  type.
\newblock {\em SIAM J. Financ. Math.}, 3:66--94, 2012.

\bibitem{PZ07}
S.~{Peszat} and J.~{Zabczyk}.
\newblock {\em {Stochastic partial differential equations with L\'evy noise. An
  evolution equation approach.}}, volume 113.
\newblock Cambridge: Cambridge University Press, 2007.

\bibitem{Rel30}
F.~Rellich.
\newblock Ein {Satz} {\"u}ber mittlere {Konvergenz}.
\newblock {\em Nachr. Ges. Wiss. G{\"o}ttingen, Math.-Phys. Kl.}, 1930:30--35,
  1930.

\bibitem{Ros91}
I.~G. {Rosen}.
\newblock {Convergence of Galerkin approximations for operator Riccati
  equations - a nonlinear evolution equation approach}.
\newblock {\em {J. Math. Anal. Appl.}}, 155(1):226--248, 1991.

\bibitem{STY20}
T.~Schmidt, S.~Tappe, and W.~Yu.
\newblock {Infinite dimensional affine processes}.
\newblock {\em {Stochastic Processes Appl.}}, 130(12):7131--7169, 2020.

\bibitem{schroers2024dynamicallyconsistentanalysisrealized}
D.~Schroers.
\newblock Dynamically consistent analysis of realized covariations in term
  structure models.
\newblock {\em Math. Finance}, 36(1):203--236, 2026.
\newblock doi:10.1111/mafi.70011.

\bibitem{schroers2024robustfunctionaldataanalysis}
D.~Schroers.
\newblock Robust functional data analysis for stochastic evolution equations in
  infinite dimensions.
\newblock Forthcoming in {\em Bernoulli}; arXiv:2401.16286, 2024.

\bibitem{Tap13}
S.~{Tappe}.
\newblock {Compact embeddings for spaces of forward rate curves}.
\newblock {\em Abstr. Appl. Anal.}, vol.~2013, Article ID 709505,
  6 pages, 2013.
\newblock doi:10.1155/2013/709505.

\bibitem{Tem69}
R.~{Temam}.
\newblock {\'Etude directe d'une \'equation d'\'evolution du type de Riccati,
  associ\'ee \`a des op\'erateurs non born\'es}.
\newblock {\em {C. R. Acad. Sci., Paris, S\'er. A}}, 268:1335--1338, 1969.

\bibitem{Paz83}
A.~Pazy.
\newblock {\em Semigroups of Linear Operators and Applications to Partial
  Differential Equations}.
\newblock Springer, New York, 1983.

\bibitem{Sch74}
H.~H. Schaefer.
\newblock {\em Banach Lattices and Positive Operators}.
\newblock Springer-Verlag, Berlin, 1974.

\end{thebibliography}

\end{document}